\documentclass[runningheads]{llncs}

\usepackage{url}

\usepackage[english]{babel}
\usepackage{booktabs} 

\usepackage{amsthm}
\makeatletter
\def\th@plain{%
  \thm@notefont{}
  \itshape 
}
\def\th@definition{%
  \thm@notefont{}
  \normalfont 
}
\makeatother
\usepackage{amsmath}
\usepackage{amssymb}

\usepackage{centernot}
\usepackage[inline]{enumitem}
\usepackage[a]{esvect} 
\usepackage{hyperref}  
\usepackage{mathtools}
\usepackage{prooftree}
\usepackage{stmaryrd}
\usepackage{xspace}

\usepackage{apptools} 
\AtAppendix{\counterwithin{lemma}{section}}
\AtAppendix{\counterwithin{definition}{section}}
\AtAppendix{\counterwithin{corollary}{section}}

\newcommand{\remove}[1]{}

\newlist{invariants}{enumerate}{1}
\setlist[invariants]{left=\parindent,label=I\arabic*.,ref=I\arabic*}

\usepackage[textwidth=3.2cm]{todonotes}

\newtheorem{principle}{Principle}
\theoremstyle{remark}
\newtheorem{notation}{Notation}

\newcommand{\Bstates}[1]{\mathcal{#1}}	
\newcommand{\companions}[1]{\node{C}_{#1}} 
\newcommand{\dia}[1]{\langle #1 \rangle}		
\newcommand{\epathto}{\xrightarrow[\dot{}]{.}}
\newcommand{\node}[1]{\mathbf{#1}} 
\newcommand{\pathto}{\xrightarrow[\dot{}]{}}

\newcommand{\act}{\textit{act}}               
\newcommand{\bottom}{{\perp}}                   
\newcommand{\cleaves}[1]{\mathbf{CL}_{#1}}      
\newcommand{\cleavesT}{\cleaves{\mathbb{T}}}    
\newcommand{\cnodes}[1]{\mathbf{C}_{#1}}        
\newcommand{\cnodesT}{\cnodes{\mathbb{T}}}      
\newcommand{\dom}{\operatorname{dom}}
\renewcommand{\div}{{\perp}}                    
\newcommand{\emptyL}{\varepsilon}               
\newcommand{\false}{\texttt{f\!f}}
\newcommand{\img}[2]{#1(#2)}                    
\newcommand{\inr}[1]{\mathrel{#1}}              
\newcommand{\iupd}[2]{\langle #1 := #2 \rangle}		
\newcommand{\lts}[1]{(\states{#1}, \xrightarrow{})}
\newcommand{\muforms}{\mathbb{F}}				
\newcommand{\muformsSV}{\muforms^{\Sigma}_{\Var}}	
\newcommand{\nats}{\mathbb{N}}
\newcommand{\pre}{\mathit{pred}}
\newcommand{\preimg}[2]{#1^{-1}(#2)}            
\newcommand{\pto}{\to_\div}                     
\newcommand{\restrict}[2]{#1 \lfloor #2}        
\newcommand{\rcomp}{\mathbin{;}}                  
\newcommand{\rn}{\textit{rn}}                   
\newcommand{\RuleAppl}{\textnormal{RAppl}}      
\newcommand{\RuleApplT}{\RuleAppl_T}		
\newcommand{\semop}[1]{||\, #1 \,||}            
\newcommand{\sem}[3]{\semop{#1}^{#2}_{#3}}		

\newcommand{\semfT}[3]{\semT{\synf{#1}{#2}}{#3}}
\newcommand{\semfTV}[2]{\semTV{\synf{#1}{#2}}}  
\newcommand{\semfZTV}[1]{\semTV{\synfZ{#1}}}    
\newcommand{\semT}[2]{\sem{#1}{\T}{#2}}			
\newcommand{\semTV}[1]{\semT{#1}{\V}}			
\newcommand{\set}{\textit{set}}                 
\newcommand{\Seq}[2]{\mathbf{S}^{#1}_{#2}}	

\newcommand{\SeqTV}{\Seq{\T}{\V}}
\newcommand{\seq}[1]{\mathbf{#1}}	
\newcommand{\seqdl}{\textit{dl}}    
\newcommand{\seqfm}{\textit{fm}}    
\newcommand{\seqst}{\textit{st}}    
\newcommand{\setdiff}{\setminus}

\newcommand{\states}[1]{\mathcal{#1}}
\newcommand{\synf}[2]{#1.#2}        
\newcommand{\synfZ}[1]{\synf{Z}{#1}} 
\newcommand{\T}{\mathcal{T}}
\newcommand{\tableau}[5]{(\tree{#1}, #2, #3, #4, #5)}		
\newcommand{\tableauTsV}[3]{\tableau{#1}{#2}{\T}{\V}{#3}}	
\newcommand{\tableauT}[2]{\tableauTsV{T}{#1}{#2}}		
\newcommand{\tableauTrl}{\tableauT{\rho}{\lambda}}	
\newcommand{\tdelay}{\delta} 
\newcommand{\tdelays}{\mathit{del}} 
\newcommand{\tpre}[2]{#1 \lceil #2}    
\newcommand{\tnx}[3]{\vdash^{{#1},{#2}}_{#3}} 
\newcommand{\tnxTVD}{\tnx{\T}{\V}{\Delta}} 
\newcommand{\tnxTV}[1]{\tnx{\T}{\V}{#1}} 
\newcommand{\tnxT}[2]{\tnx{\T}{#1}{#2}}

\newcommand{\tn}[1]{\vdash_{#1}}	
\newcommand{\tnD}{\tn{\Delta}}		
\newcommand{\tR}[2]{\forall_{#1}(#2)}   
\newcommand{\tree}[1]{\mathbf{#1}}	
\newcommand{\true}{\texttt{t\!t}}
\newcommand{\tsucc}{\mathit{succ}} 
\newcommand{\tU}[2]{\exists_{#1}(#2)}   
\newcommand{\V}{\val{V}}
\newcommand{\val}[1]{\mathcal{#1}}
\newcommand{\Var}{\textnormal{Var}\xspace}
\renewcommand{\vec}[1]{\vv{#1}} 

\newcommand{\proofrule}[3][]{#1 \frac{\raisebox{0.7ex}{\normalsize{$#2$}}}
  {\raisebox{-1.6ex}{\normalsize{$#3$}}}}

\newcommand{\Rnneg}{\mathbb{R}_{\geq 0}}

\newcommand{\seeappendix}{The detailed proof is included in the appendix.}

\begin{document}

\title{Extensible Proof Systems for Infinite-State Systems%
    \thanks{Research supported by US National Science Foundation grant CNS-1446365 and US Office of Naval Research grant N00014-17-1-2622.}
}
%
%
\author{Rance Cleaveland\inst{1}\orcidID{0000-0002-4952-5380} \and
    Jeroen J.A. Keiren\inst{2}\orcidID{0000-0002-5772-9527}}
\authorrunning{R. Cleaveland \and J.J.A. Keiren}
%
\institute{Department of Computer Science, University of Maryland, College Park, Maryland, USA\\
    \email{rance@cs.umd.edu}
    \and
    Department of Mathematics and Computer Science, Eindhoven University of Technology, PO Box 513, 5600 MB Eindhoven, The Netherlands\\
    \email{j.j.a.keiren@tue.nl}
}
\maketitle              
\begin{abstract}
    This paper revisits soundness and completeness of proof systems for proving that sets of states in infinite-state labeled transition systems satisfy formulas in the modal mu-calculus.
    Our results rely on novel results in lattice theory, which give constructive characterizations of both greatest and least fixpoints of monotonic functions over complete lattices. 
    We show how these results may be used to reconstruct the sound and complete tableau method for this problem due to Bradfield and Stirling.
    We also show how the flexibility of our lattice-theoretic basis simplifies reasoning about tableau-based proof strategies for alternative classes of systems.
    In particular, we extend the modal mu-calculus with timed modalities, and prove that the resulting tableaux method is sound and complete for timed transition systems.
\end{abstract}

\section{Introduction}

Proof systems provide a means for proving sequents in formal logics, and are intended to reduce reasoning about objects in a given theory to syntactically checkable proofs consisting of applications of proof rules to the sequents in question.  When a proof system is sound, every provable sequent is indeed semantically valid; when it is in addition complete, it follows that every semantically valid sequent can be proved within the proof system.  Because they manipulate syntax, the construction of proofs within a proof system can be automated; proof assistants such as Coq~\cite{bertot2013interactive} and Nuprl~\cite{constable1986implementing} are built around this observation.  Within the model-checking community, the fully automatic construction of proofs based on sound and complete proof systems for decidable theories provides a basis for establishing the correctness, or incorrectness, of systems \emph{vis \`a vis} properties they are expected to satisfy.  In addition, proof systems can be used as a basis for different approaches to model checking.  In classical \emph{global} model-checking techniques~\cite{clarke2018model}, one uses the proof rules to prove properties of states with respect to larger and larger subformulas of the given formula, until one shows that the start state(s) of the system either do, or do not, satisfy the original formula.  In \emph{local}~\cite{SW1991}, or \emph{on-the-fly}~\cite{bhat1995efficient} methods, in contrast, one uses the proof rules to conduct backward reasoning  from the start states and original formula, applying proof rules ``in reverse'' to generate subgoals that require proving in order for the original sequents to be true.  The virtue of on-the-fly model checking is that proofs can often be completed, or be shown not to exist, without having to examine all states and all subformulas.

Driven by applications in model checking, proof systems have been developed for establishing that finite-state systems satisfy formulas captured in a very expressive temporal logic, the \emph{(propositional) modal mu-calculus}~\cite{Koz1983}.  These in turn have been used as a basis for efficient model-checking procedures for fragments of this logic~\cite{CS1993}.  Work has also shown that for interesting fragments of the mu-calculus, global and on-the-fly techniques exhibit the same worst-case complexity~\cite{CS1993,mateescu2003efficient}, meaning that the early-termination feature of on-the-fly approaches does not incur additional overhead in the worst case.

Researchers have also developed sound and complete proof systems for infinite-state systems and the modal mu-calculus~\cite{BS1992,Bra1991}, a prototype implementation of which was described in~\cite{Bra1993}.  In this case, the sequents, instead of involving single states and formulas in the logic as in the finite-state case, refer to potentially infinite sets of states. A variation of these proof systems that localizes validity by annotating fixed points, and that explicitly provides well-founded order on states satisfying least fixed points, was described in~\cite{And1993}. 
Ultimately, the success conditions in proof systems for infinite state systems are based on the fundamental theorem of the modal mu-calculus due to Streett and Emerson~\cite{SE1989}. For a detailed yet accessible account of the fundamental theorem, the reader is referred to~\cite{BS2001}. 

While in general infinite-state model checking in the modal mu-calculus is undecidable, specialized proof systems for modifications of the so-called \emph{alternation-free} fragment of the mu-calculus~\cite{FC2014,DC2005} can lead to efficient on-the-fly model checkers for \emph{timed automata}~\cite{AD1994}, a class of infinite-state systems whose model-checking problem is decidable.

Ideally, one should be able to prove soundness and completeness of these timed-automata proof systems by referring to the general results for infinite-state systems and the modal mu-calculus.  One should also be able to develop checkers for larger fragments than the alternation-free fragment so that more types of properties can be processed.  However, there are several obstacles to this desirable state of affairs.

\begin{enumerate}
    \item The intricacy of the proof system in~\cite{BS1992} means that modifications to it in essence require re-proofs of soundness and completeness from scratch.
    \item The proof systems used for on-the-fly model checking of timed automata require several modifications to the modal mu-calculus:  modalities for reasoning about time, and computable proof-termination criteria to enable detection of when a proof attempt is complete.
    \item For efficiency reasons, construction of proofs in the on-the-fly model checkers must also use different proof-construction strategies, and this prevents applying reasoning from the infinite-state mu-calculus proof system to establishing the correctness of these procedures.
\end{enumerate}
These issues have limited the practical application of the proof system in~\cite{BS1992}, although its theoretical contribution is rightfully very highly regarded.

In this paper the goal is to revisit the proof system for general infinite-state systems and the modal mu-calculus with a view toward developing new, extensible proofs of soundness and completeness. 
Concretely, our contributions are the following.

We first introduce \emph{support orderings}, along with general lattice-theoretical results, that formalize the dependencies between states that satisfy given fixpoint formulas. This, in essence, gives a purely semantic, constructive account of the least and greatest fixpoints of monotonic functions over subset lattices. These support orderings are closely related to Streett and Emerson's \emph{regeneration relations}~\cite{SE1989}, as well as Bradfield and Stirling's \emph{(extended) paths}~\cite{Bra1991,BS1992}, although unlike those works our results do not rely on infinitary syntactic manipulations of mu-calculus formulas such as ordinal unfoldings.

We next recall the proof system of~\cite{BS1992}, and show that the soundness of the proof system follows from the lattice-theoretical results. In particular, a syntactic ordering, the \emph{extended dependency ordering}, derived from extended paths is a support ordering, and from this observation soundness of the proof system follows straightforwardly.
In a similar way, given a support ordering induced by our lattice-theoretical results, we construct a tableau whose extended dependencies respect the given support ordering. This establishes completeness. To facilitate the completeness proof, we first establish a novel notion of well-founded induction in the context of mutually recursive fixpoints.

Finally, we show that these results also permit extensions to the mu-calculus to be easily incorporated in a soundness- and completeness-preserving manner.  They also simplify reasoning about new proof-termination criteria.
To show this, we first modify the termination criterion of the proof system, such that sequents for least-fixpoint formulas with a non-empty set of states are always unfolded. This results in a proof systems that is (evidently) not complete, but whose soundness follows trivially from our earlier soundness result.
Second, we consider a proof system for timed transition systems, and a mu-calculus with two additional, timed modalities~\cite{FC2014}. The proofs of soundness and completeness are, indeed, straightforward extensions of our earlier results. To the best of our knowledge, ours is also the first sound and complete proof system for a timed mu-calculus.

The rest of the paper develops along the following lines.  In the next section we review mathematical preliminaries used in the rest of the paper. We then state and prove our lattice-theoretical results in Section~\ref{sec:support-orderings}, while Section~\ref{sec:mu-calculus} introduces the syntax and semantics of the modal mu-calculus and establishes some properties that will prove useful.
We and present the proof system of~\cite{BS1992} in Section~\ref{sec:base-proof-system}. In Section~\ref{sec:Soundness-via-support-orderings}, we show that the soundness of the proof system follows from the lattice-theoretical results, which we consider completeness in Section~\ref{sec:Completeness}.
In Sections~\ref{sec:proof-search} and~\ref{sec:timed-mu-calculus} we illustrate how our new approach accommodates changes to the proof system that are needed for efficient on-the-fly model checking for timed automata and other decidable formalisms. Conclusions and future work are discussed in Section~\ref{sec:Conclusions}.

\section{Mathematical preliminaries}\label{sec:preliminaries}

This section defines basic terminology used in the sequel for finite sequences, (partial) functions, binary relations, lattices and fixpoints, and finite trees.

\subsection{Sequences}\label{subsec:sequences}

As usual, sequences are ordered collections of elements $x_1 \cdots x_n$, where each $x_i$ is taken from a given set $X$.
\begin{notation}[Sequences]
    Let $X$ be a set.
    \begin{itemize}
    \item $X^*$ is the set of finite, possibly empty, sequences of elements from $X$.
    \item The \emph{empty sequence} in $X^*$ is denoted $\emptyL$.
    \item We take $X \subseteq X^*$, with each $x \in X$ being a single-element sequence in $X^*$.
    \item 
        Suppose $\vec{w} = x_1 \cdots x_n \in X^*$, where each $x_i \in X$.  Then $|\vec{w}| = n$ denotes the \emph{length} of $\vec{w}$.  Note that $| \emptyL | = 0$, and $|x| = 1$ if $x \in X$.
    \item 
        If $\vec{w}_1, \vec{w}_2 \in X^*$ then $\vec{w}_1 \cdot \vec{w}_2 \in X^*$ is the \emph{concatenation} of $\vec{w}_1$ and $\vec{w}_2$.  We often omit $\cdot$ and write e.g.\/ $\vec{w}_1\vec{w}_2$ for $\vec{w}_1 \cdot \vec{w}_2$.
    \item 
        Suppose $\vec{w}_1, \vec{w}_2 \in X^*$.  Then we write $\vec{w}_1 \preceq \vec{w}_2$ if $\vec{w}_1$ is a (not necessarily strict) \emph{prefix} of $\vec{w}_2$, and $\vec{w}_1 \npreceq \vec{w}_2$ if $\vec{w}_1$ is not a prefix of $\vec{w}_2$.
    \item
        Let $\vec{w} = x_1 \cdots x_n \in X^*$.  Then $\set(\vec{w}) \subseteq X$, the \emph{set associated with} $\vec{w}$, is defined to be $\{x_1, \ldots, x_n\}$.  Note that $\set(\emptyL) = \emptyset$, and $|\set(\vec{w})| \leq |\vec{w}|$.
    \item
        Sequence $\vec{w} \in X^*$ is \emph{duplicate-free} iff $|\vec{w}| = |\set(\vec{w})|$.
    \item
        Sequence $\vec{w} \in X^*$ is a \emph{permutation}, or \emph{ordering}, of $X$ iff $\vec{w}$ is duplicate-free and $\set(\vec{w}) = X$.  
    \end{itemize}
\end{notation}

\noindent
Note that only finite sets can have permutations / orderings in this definition.

\subsection{Partial functions}

In this paper we make significant use of partial as well as total functions.  This section introduces notation we use for such functions.

\begin{notation}[Partial functions]
Let $X$ and $Y$ be sets.
\begin{itemize}
    \item
    Relation $f \subseteq X \times Y$ is \emph{functional} iff for all $x \in X$ and $y_1, y_2 \in Y$, if $(x,y_1) \in R$ and $(x,y_2) \in R$ then $y_1 = y_2$.  We call $f$ a \emph{partial function} from $X$ to $Y$ and use $X \pto Y$ to denote the set of all partial functions from $X$ to $Y$.
    \item
    Suppose $f \in X \pto Y$ and $x \in X$.  If there is $y \in Y$ such that $(x,y) \in Y$ then we write $f(x)$ as usual to denote this $y$ and say $f$ is \emph{defined} for $x$ in this case.  We will also write $f(x) \in Y$ to denote that $f$ is defined for $x$.  If there exists no $y \in Y$ such that $(x,y) \in f$ then we say that $f$ is \emph{undefined} for $x$ and write $f(x)\div$.
    \item
    If $f \in X \pto Y$ then we call $\dom(f) = \{ x \in X \mid f(x) \in Y\}$ the \emph{domain of definition} of $f$.
    \item
    $f \in X \pto Y$ is \emph{total} iff $\dom(f) = X$.  We write $X \to Y$ as usual for the set of total functions from $X$ to $Y$.  Note that $X \to Y \subseteq X \pto Y$.
    \item
    If $f, g \in X \pto Y$ then $f = g$ iff $\dom(f) = \dom(g)$ and for all $x \in \dom(f)$, $f(x) = g(x)$.
\end{itemize}
\end{notation}

\noindent
Note that partial functions are equal exactly when they are defined on the same elements and return the same values when they are defined.  
We also use the following standard operations on partial functions.
\begin{notation}[Function operations]
Let $X$ and $Y$ be sets.
\begin{itemize}
    \item
    The \emph{everywhere undefined} function $f_\emptyset \in X \pto Y$ is defined as $f_\emptyset = \emptyset \subseteq X \times Y$.  Note that $f(x) \div$ for all $x \in X$ and thus $\dom(f_\emptyset) = \emptyset$.
    \item
    Let $f \in X \pto X$ and $i \in \nats$.  Then $f^i \in X \pto X$ is defined as follows.
    \[
    f^i(x) =
    \begin{cases}
        x               & \text{if $i = 0$} \\
        f(f^{i-1}(x))   & \text{otherwise}
    \end{cases}
    \]
    Note that $f^0$ is total and that $f^i(x)\div$ iff $f(f^j(x)) \div$ some $j < i$.  Also note that if $f$ is total then so is $f^i$ for all $i \geq 0$.
    \item 
    Suppose $f \in X \pto Y$ , and let $x \in X$ and $y \in Y$.  Then $f[x:=y] \in X \rightarrow Y$ is defined as follows.
    \[
    f[x:=y](x') =
    \begin{cases}
        y       & \text{if $x'=x$} \\
        f(x')   & \text{if $x' \neq x$ and $f(x') \in Y$}
    \end{cases}
    \]
    Note that even if $f(x)\bottom$, $f[x:=y]$ is nevertheless defined for $x$, and that if $f$ is total then so is $f[x:=y]$.
    This notion can generalized to $f[\vec{x} := \vec{y}]$, where $\vec{x} \in X^*$ is duplicate-free and $\vec{y} \in Y^*$ is such that $|\vec{x}| = |\vec{y}|$, in the obvious fashion.
    \item 
    Let $f \in X \pto Y$ and $\vec{w} = x_1 \cdots x_n \in X^*$.  Then $f \in X^* \pto Y^*$ is defined to by $f(\vec{w}) = f(x_1) \cdots f(x_n)$.  Note that if $f(x) \div$ for any $x \in \set(\vec{w})$ then $f(\vec{w}) \div$.
\end{itemize}
\end{notation}

\subsection{Binary relations}

Later in this paper we refer extensively to the theory of binary relations over a given set $X$.  This section summarizes some of the concepts used in what follows.

\begin{definition}[Binary relations]
Let $X$ be a set.  Then a \emph{binary relation} over $X$ is a subset $R \subseteq X \times X$.
\end{definition}

\noindent
When $R$ is a binary relation over $X$ we usually write $x_1 \inr{R} x_2$ in lieu of $(x_1, x_2) \in R$ and $x_1 \inr{\centernot R} x_2$ instead of $(x_1, x_2) \not\in R$.
We now recall the following terminology.

\begin{definition}[Preorders, partial orders and equivalence relations]~\label{def:relations}
Let $R \subseteq X \times X$ be a binary relation over $X$.
\begin{enumerate}
    \item
    $R$ is \emph{reflexive} iff $x \inr{R} x$ for all $x \in X$.
    \item
    $R$ is \emph{symmetric} iff whenever $x_1 \inr{R} x_2$ then $x_2 \inr{R} x_1$.
    \item
    $R$ is \emph{anti-symmetric} iff whenever $x_1 \inr{R} x_2$ and $x_2 \inr{R} x_1$ then $x_1 = x_2$.
    \item
    $R$ is \emph{transitive} iff whenever $x_1 \inr{R} x_2$ and $x_2 \inr{R} x_3$ then $x_1 \inr{R} x_3$.
    \item
    $R$ is a \emph{preorder} iff $R$ is reflexive and transitive.
    \item
    $R$ is a \emph{partial order} iff $R$ is reflexive, anti-symmetric and transitive.
    \item
    $R$ is an \emph{equivalence relation} iff $R$ is reflexive, symmetric and transitive.
\end{enumerate}
\end{definition}

\noindent
We also use the following standard, if less well-known, definitions.

\begin{definition}[Irreflexive and total relations]\label{def:irreflexive-total}
Let $R \subseteq X \times X$.
\begin{enumerate}
    \item
    $R$ is \emph{irreflexive} iff for every $x \in X$, $x \inr{\centernot R} x$.
    \item
    $R$ is a \emph{(strict) total order} iff it is irreflexive and transitive and satisfies:  for all $x_1 \neq x_2 \in X$, either $x_1 \inr{R} x_2$ or $x_2 \inr{R} x_1$.
\end{enumerate}
\end{definition}

\noindent
A relation $R$ over $X$ is irreflexive iff no element in $X$ is related to itself.  It is total exactly when any distinct $x_1, x_2 \in X$ are \emph{comparable}, one way or another, via $R$.  This version of totality is often called \emph{strict totality}, although we drop the qualifier ``strict" in this paper.
The following relations are used later.

\begin{definition}[Identity, universal relations]\label{def:identity-universal-relations}
Let $X$ be a set.
\begin{enumerate}
    \item 
    The \emph{identity relation over $X$} is defined as $\mathit{Id}_X = \{ (x,x) \mid x \in X\}$.
    \item
    The \emph{universal relation over $X$} is defined as $\mathit{U}_X = X \times X$.
\end{enumerate}
\end{definition}

\noindent
We also use the following operations on binary relations.
\begin{definition}[Relational operations]\label{def:relation-operations}
Let $R, R'$ be binary relations over $X$.
\begin{enumerate}
    \item
    $R'$ \emph{extends} $R$ iff $R \subseteq R'$.
    \item
    Let $X' \subseteq X$.  Then the \emph{restriction} of $R$ with respect to $X'$ is the binary relation $\restrict{R}{X'}$ over $X'$ defined as follows:
    \[
    \restrict{R}{X'} = R \cap (X' \times X') = \{(x_1, x_2) \in X' \times X' \mid x_1 \inr{R} x_2\}.
    \]
    \item
    The \emph{relational composition} of $R$ and $R'$ is the binary relation $R \rcomp R'$ over $X$ defined as follows.
    \[
    R \rcomp R' =
    \{(x_1, x_3) \in X \times X \mid \exists x_2 \colon x_2 \in X \colon x_1 \inr{R} x_2 \land x_2 \inr{R'} x_3\}
    \]
    \item
    The \emph{inverse} of $R$ is the binary relation $R^{-1}$ over $X$ defined by:
    \[
    R^{-1} = \{(x_2, x_1) \in X \times X \mid x_1 \inr{R} x_2\}
    \]
    \item\label{it:relation-image}
    Let $X' \subseteq X$.  Then the \emph{image}, $\img{R}{X'}$, of $R$ with respect to $X'$ is defined by
    \[
    \img{R}{X'} = \{x \in X \mid \exists x' \colon x' \in X' \colon x' \inr{R} x\}.
    \]
    If $x \in X$ then we write $\img{R}{x}$ in lieu of $\img{R}{\{x\}}$.
    \item\label{it:relation-preimage}
    Let $X' \subseteq X$.  Then the \emph{pre-image} of $R$ with respect to $X'$ is the image $\preimg{R}{X'}$ of $R^{-1}$ with respect to $X'$.
    \item\label{item:irreflexive-core}
    The \emph{irreflexive core} of $R$ is the binary relation $R^-$ over $X$ defined by
    \[
    R^- = R \setminus \mathit{Id}_X.
    \]
    \item\label{item:reflexive-closure}
    The \emph{reflexive closure} of $R$ is the binary relation $R^=$ given by
    \[
    R^= = R \cup \mathit{Id}_X.
    \]
    \item\label{subdef:transitive-closure}
    The transitive closure, $R^+$, of $R$ is the least transitive relation extending $R$.
    \item\label{subdef:reflexive-transitive-closure}
    The reflexive and transitive closure, $R^*$ of $R$, is defined by
    \[
    R^* = (R^+)^=.
    \]
\end{enumerate}
\end{definition}

\noindent
Note that based on the definition of $R^{-1}$, the following holds.
    \[
    \preimg{R}{X'} = \{x \in X \mid \exists x' \colon x' \in X' \colon x' \inr{R^{-1}} x \}
                  = \{x \in X \mid \exists x' \colon x' \in X' \colon x \inr{R} x'\}
    \]

Relation $R^+$ is guaranteed to exist for arbitrary set $X$ and relation $R$ over $X$, and from the definition it is immediate that $R$ itself is transitive iff $R^+ = R$.  $R^+$ also has the following alternative characterization.  Define
\begin{align*}
    R^1     &= R\\
    R^{i+1} &= R \rcomp R^i.
\end{align*}
Then $R^+ = \bigcup_{i=1}^\infty R^i$.
It may also be shown that for any relation $R \subseteq X \times X$, $R^*$ is the unique smallest preorder that extends $R$.

The transitive and reflexive closure of a relation also induces an associated equivalence relation and partial order over the resulting equivalence classes.

\begin{definition}[Quotient of a relation]\label{def:relation-quotient}
Let $R \subseteq X \times X$ be a relation.
\begin{enumerate}
    \item
    Relation $\sim_R$ is defined as $x_1 \sim_R x_2$ iff $x_1 \inr{R^*} x_2$ and $x_2 \inr{R^*} x_1$.
    \item
    Let $x \in X$.  Then $[x]_R \subseteq X$ is defined as
    $$
    [x]_R = \{x' \in X \mid x \sim_R x'\}.
    $$
    We use $Q_R = \{ \,[x]_R \mid x \in X \}$.
    \item
    $P(R) \subseteq Q_R \times Q_R$, the \emph{ordering induced by $R$ on $Q_R$}, is defined as
    $$
    P(R) = \left\{ \left( [x]_R, [x']_R \right) \mid x \inr{R^*} x' \right\}.
    $$
    \item
    $(Q_R, P(R))$ is called the \emph{quotient} of $R$.
\end{enumerate}
\end{definition}
It is easy to verify that $\sim_R$ is an equivalence relation for any $R$, and thus that $[x]_R$ is the equivalence class of $R$ and that $P(R)$ is a partial order over $Q_R$.  Note that $Q_R$ defines a partition of set $X$:  $X = \bigcup_{Q \in Q_R} Q$, and either $Q = Q'$ or $Q \cap Q' = \emptyset$ for any $Q, Q' \in Q_R$.

In this paper we also make extensive use of \emph{well-founded relations} and \emph{well-orderings}.  To define these, we first introduce the following.

\begin{definition}[Extremal elements]\label{def:minimal-minimum-elements}
Let $R \subseteq X \times X$, and let $X' \subseteq X$.
\begin{enumerate}
    \item 
    $x' \in X'$ is \emph{$R$-minimal in $X'$} iff for all $x \neq x' \in X'$, $x \inr{\centernot R} x'$.
    \item
    $x' \in X'$ is \emph{$R$-minimum in $X'$} iff $x'$ is the only $R$-minimal element in $X'$.
    \item 
    $x' \in X'$ is \emph{$R$-maximal in $X'$} iff for all $x \neq x' \in X'$, $x' \inr{\centernot R} x$.
    \item
    $x' \in X'$ is \emph{$R$-maximum in $X'$} iff $x'$ is the only $R$-maximal element in $X'$.
    \item
    $x \in X$ is an \emph{$R$-lower bound} of $X'$ iff for all $x' \neq x \in X'$, $x \inr{R} x'$.
    \item
    $x \in X$ is the \emph{$R$-greatest lower bound} of $X'$ iff it is the $R$-maximum of the set of $R$-lower bounds of $X'$.
    \item
    $x \in X$ is an \emph{$R$-upper bound} of $X'$ iff for all $x' \neq x \in X'$, $x' \inr{R} x$.
    \item
    $x \in X$ the \emph{$R$-least upper bound} of $X'$ iff it is the $R$-minimum of the set of $R$-upper bounds of $X'$.
\end{enumerate}
\end{definition}

\noindent
In what follows we often omit $R$ when it is clear from context and instead write minimal rather than $R$-minimal, etc.  Note that minimal / minimum / maximal / maximum elements for $X'$ must themselves belong to $X'$; this is not the case for upper and lower bounds.  We can now defined well-foundness and well-orderings.

\begin{definition}[Well-founded relations, well-orderings]\label{def: well-founded}
Let $R \subseteq X \times X$.
\begin{enumerate}
    \item
    $R$ is \emph{well-founded} iff every non-empty $X' \subseteq X$ has an $R$-minimal element.
    \item
    $R$ is a \emph{well-ordering} iff it is total and well-founded.
\end{enumerate}
\end{definition}

We close the section by remarking on some noteworthy properties of well-founded relations and well-orderings.\footnote{These results generally rely on the inclusion of additional axioms beyond the standard ones of Zermelo-Fraenkel (ZF) set theory.  The Axiom of Choice~\cite{jech2008axiom} is one such axiom, and in the rest of the paper we assume its inclusion in ZF.}  
The first result is a well-known alternative characterization of well-foundedness.  If $X$ is a set and $R \subseteq X \times X$, call $\ldots, x_2, x_1$ an \emph{infinite descending chain in} $R$ iff for all $i \geq 1$, $x_{i+1} \inr{R} x_i$.

\begin{lemma}[Descending chains and well-foundedness]\label{lem:well-foundedness-chains}
Let $R \subseteq X \times X$.  Then $R$ is well-founded iff $R$ contains no infinite descending chains.
\end{lemma}

\noindent
Transitive closures also preserve well-foundedness.

\begin{lemma}[Transitive closures of well-founded relations]\label{lem:transitive-closure-well-founded}
Let $R \subseteq X \times X$ be well-founded.  Then $R^+$ is also well-founded.
\end{lemma}

\noindent
Any well-founded relation can be extended to a well-ordering.

\begin{lemma}[Total extensions of well-founded relations]\label{lem:well-ordering-extension}
Let $R \subseteq X \times X$ be well-founded.  Then there exists a well-ordering $R' \subseteq X \times X$ extending $R$.
\end{lemma}

\noindent
Note that if $R = \emptyset$ then the above lemma reduces to the Well-Ordering Theorem~\cite{Haz01}, which states that every set can be well-ordered.

The next result is immediate from the definition of well-ordering.

\begin{lemma}[Minimum elements and well-orderings]\label{lem:well-ordering-minimum}
Let $R \subseteq X \times X$ be a well-ordering.  Then every non-empty $X' \subseteq X$ contains an $R$-minimum element.
\end{lemma}

\subsection{Complete lattices, monotonic functions and fixpoints}\label{subsec:lattices}

The results in this paper rely heavily on the basic theory of fixpoints of monotonic functions over complete lattices, as developed by Tarski and Knaster~\cite{Tar1955}.  We review the relevant parts of the theory here.

\begin{definition}[Complete lattice]\label{def:complete-lattice}
A \emph{complete lattice} is a tuple $(X, \sqsubseteq, \bigsqcup, \bigsqcap)$ satisfying the following.
\begin{enumerate}
    \item $X$ is a set (the \emph{carrier set}).
    \item
    Relation ${\sqsubseteq}$ is a partial order over $X$.
    \item
    Function ${\bigsqcup} \in 2^X \to X$, the \emph{join} operation, is total and satisfies:  for all $X' \subseteq X$, $\bigsqcup (X')$ is the least upper bound of $X'$.
    \item
    Function ${\bigsqcap} \in 2^X \to X$, the \emph{meet} operation, is total and satisfies:  for all $X' \subseteq X$, $\bigsqcap (X')$ is the greatest lower bound of $X'$.
\end{enumerate}
\end{definition}

\noindent
In what follows we write $\bigsqcup X'$ and $\bigsqcap X'$ instead of $\bigsqcup(X')$ and $\bigsqcap(X')$.

\begin{definition}[Fixpoint]\label{def:fixpoint}
Let $X$ be a set and $f \in X \to X$ be a function.  Then $x \in X$ is a \emph{fixpoint} of $f$ iff $f(x) = x$.
\end{definition}

As usual, $f \in X \rightarrow X$ is monotonic over complete lattice $(X, \sqsubseteq, \bigsqcup,\bigsqcap)$ iff $f$ is total and whenever $x_1 \sqsubseteq x_2$, $f(x_1) \sqsubseteq f(x_2)$.
The next result follows from~\cite{Tar1955}.

\begin{lemma}[Extremal fixpoint characterizations]\label{lem:tarski-knaster}
Let $(X, \sqsubseteq, \bigsqcup, \bigsqcap)$ be a complete lattice, and let $f \in X \to X$ be monotonic over it. Then $f$ has least and greatest fixpoints $\mu f, \nu f \in X$, respectively, characterized as follows.
\begin{align*}
    \mu f &= \bigsqcap \{x \in X \mid f(x) \sqsubseteq x \}\\
    \nu f &= \bigsqcup \{x \in X \mid x \sqsubseteq f(x) \}
\end{align*}
\end{lemma}

\noindent
Elements $x \in X$ such that $f(x) \sqsubseteq x$ are sometimes called \emph{pre-fixpoints} of $f$, while those satisfying $x \sqsubseteq f(x)$ are referred to as \emph{post-fixpoints} of $f$.

In this paper we focus on specialized complete lattices called \emph{subset lattices}.

\begin{definition}[Subset lattice]\label{def:subset-lattice}
Let $S$ be a set.   Then the \emph{subset lattice generated by $S$} is the tuple $(2^S, \subseteq, \bigcup, \bigcap)$.
\end{definition}

\noindent
It is straightforward to establish for any $S$, the subset lattice generated by $S$ is indeed a complete lattice.

\subsection{Finite non-empty trees}\label{subsec:trees}

The proof objects considered later in this paper are finite trees whose nodes are labeled by logical sequents.  As we wish to reason about mathematical constructions on these proof objects we need formal accounts of such trees.

\begin{definition}[Finite non-empty unordered tree]~\label{def:unordered-tree}
    A \emph{finite non-empty unordered tree} is a triple $\tree{T} = (\node{N}, \node{r}, p)$, where:
    \begin{enumerate}
        \item $\node{N}$ is a finite, non-empty set of \emph{nodes};
        \item $\node{r} \in \node{N}$ is the \emph{root} node; and
        \item $p \in \node{N} \pto \node{N}$, the (partial) \emph{parent} function, satisfies: $p(\node{n}) \div$ iff $\node{n} = \node{r}$, and for all $\node{n} \in \node{N}$ there exists $i \geq 0$ such that $p^i(\node{n}) = \node{r}$.
    \end{enumerate}
\end{definition}

\noindent
Note that each non-root node $\node{n} \neq \node{r}$ has a parent node $p(\node{n}) \in \node{N}$.  If $p(\node{n}') = \node{n}$ then we call $\node{n}'$ a \emph{child} of $\node{n}$; we use $c(\node{n}) = \{ \node{n}' \in \node{N} \mid p(\node{n}') = \node{n} \}$ to denote all the children of $\node{n}$.  If $c(\node{n}) = \emptyset$ then $\node{n}$ is a \emph{leaf}; otherwise it is \emph{internal}.  We call node $\node{n}$ an \emph{ancestor} of node $\node{n}'$, or equivalently, $\node{n}'$ a \emph{descendant} of $\node{n}$, iff there exists an $i \geq 0$ such that $p^i(\node{n}') = \node{n}$; we also say in this case that there is a \emph{path from $\node{n}$ to $\node{n}'$}.  
We write $A(\node{n})$ and $D(\node{n})$ for the ancestors and descendants of $\node{n}$, respectively, and 
note that $\node{r} \in A(\node{n})$, $\node{n} \in D(\node{r})$, $\node{n} \in A(\node{n})$ and $\node{n} \in D(\node{n})$ for all $\node{n} \in \node{N}$.
We use $A_s(\node{n}) = A(\node{n}) \setminus \{\node{n}\}$ and $D_s(\node{n}) = D(\node{n}) \setminus \{\node{n}\}$ for the \emph{strict} ancestors and descendants of $\node{n}$.  
We also define the notions of \emph{depth}, $d(\node{n})$, and \emph{height}, $h(\node{n})$ of $\node{n}$ as the length of the unique path from the root $\node{r}$ to $\node{n}$, and the length of the longest path starting at $\node{n}$, respectively.  Specifically, $d(\node{n}) = i$ if $i$ is (the unique $i \in \nats$) such that $p^i(\node{n}) = \node{r}$, while
\[
h(\node{n}) = \max \{ i \mid \exists \node{n}' \colon \node{n}' \in \node{N} \colon p^i(\node{n}') = \node{n} \}.
\]
Note that for any leaf $\node{n}$, $h(\node{n}) = 0$.  We define $h(\tree{T})$ to be $h(\node{r})$.

Finite unordered trees also admit the following induction and co-induction principles.


\begin{principle}[Tree induction]
    Let $\tree{T} = (\node{N}, \node{r}, p)$ be a finite unordered tree, and let $Q$ be a predicate over $\node{N}$.  To prove that $Q(\node{n})$ holds for every $\node{n} \in \node{N}$, it suffices to prove $Q(\node{n})$ under the assumption that $Q(\node{n}')$ holds for every $\node{n}' \in D_s(\node{n})$.  The assumption is referred to as the \emph{induction hypothesis}.
\end{principle}

\begin{principle}[Tree co-induction\footnote{Treatments of co-induction in theoretical computer science tend to focus on its use in reasoning about co-algrebras.  The setting of finite trees in this paper is not explicitly co-algebraic, but the principle of co-induction as articulated in e.g.~\cite{jacobs2011introduction} is easily seen to correspond what is given here.}]
    Let $\tree{T} = (\node{N}, \node{r}, p)$ be a finite unordered tree, and let $Q$ be a predicate over $\node{N}$.  To prove that $Q(\node{n})$ holds for every $\node{n} \in \node{N}$, it suffices to prove that $Q(\node{n})$ holds under the assumption that $Q(\node{n}')$ holds for every $\node{n}' \in A_s(\node{n})$.  The assumption is referred to as the \emph{co-induction hypothesis}.
\end{principle}

\noindent
Tree induction is an instance of standard strong induction on the height of nodes, while tree co-induction is an instance of standard strong induction on the depth of nodes.  Note that the co-induction principle also applies to discrete rooted infinite trees~\cite{jech:1978,kunen:1980} as well as finite ones, although we do not use this capability in this paper.  We also note that in the case of co-induction, reasoning about the root node $\node{r}$ is handled differently than the other nodes, owing to the fact that $\node{r}$ is the only node with no strict ancestors.  Consequently, in the co-inductive arguments given in the paper, we will often single out a special \emph{root case} for dealing with this node, with reasoning about other nodes covered in a so-called \emph{co-induction step}.

Finite ordered trees can now be defined as follows.

\begin{definition}[Finite non-empty ordered tree]\label{def:ordered-tree}
    A \emph{finite non-empty ordered tree} is a tuple $\tree{T} = (\node{N}, \node{r}, p, cs)$, where:
    \begin{enumerate}
        \item
              $(\node{N}, \node{r}, p)$ is a finite non-empty unordered tree; and
        \item
              $cs \in \node{N} \rightarrow \node{N}^*$, the \emph{child ordering}, satisfies: $cs(\node{n})$ is an ordering of $c(\node{n})$ for all $\node{n} \in \node{N}$.
    \end{enumerate}
\end{definition}

The definition of ordered tree extends that of unordered tree by incorporating a function, $cs(\node{n})$, that returns the children of $\node{n}$ in left-to-right order.  Ordered trees inherit the definitions given for unordered trees (height, children, etc.), as well as the tree-induction and co-induction principles.  The notion of \emph{subtree} rooted at a node in a given ordered tree can now be defined.

\begin{definition}[Subtree]\label{def:subtree}
    Let $\tree{T} = (\node{N}, \node{r}, p, cs)$ be a finite ordered tree, and let $\node{n} \in \node{N}$.  Then $\tree{T}_\node{n}$, the \emph{subtree of $\tree{T}$ rooted at $\node{n}$}, is defined to be $\tree{T}_\node{n} = (D(\node{n}), \node{n}, p_\node{n}, cs)$, where $p_\node{n}$ satisfies:  $p_\node{n}(\node{n})\div$, and $p_\node{n}(\node{n}') = p(\node{n}')$ if $\node{n}' \neq \node{n}$.
\end{definition}

\noindent
It is straightforward to verify that $\tree{T}_\node{n}$ is itself a finite ordered tree.  We also use the notion of \emph{tree prefix} later.

\begin{definition}[Tree prefix]\label{def:tree-prefix}
Let $\tree{T}_1 = (\node{N}_1, \node{r}_1, p_1, cs_1)$, $\tree{T}_2 = (\node{N}_2, \node{r}_2, p_2, cs_2)$ be finite ordered trees.  Then $\tree{T}_1$ is a \emph{tree prefix} of $\tree{T}_2$, notation $\tree{T}_1 \preceq \tree{T}_2$, iff:
\begin{enumerate}
\item
    $\node{N}_1 \subseteq \node{N}_2$;
\item
    $\node{r}_1 = \node{r}_2$;
\item
    For all $\node{n} \in \node{N}_1$, $p_1(\node{n}) = p_2(\node{n})$; and
\item
    For all $\node{n} \in \node{N}_1$, either $cs_1(\node{n}) = cs_2(\node{n})$, or $cs_1(\node{n}) = \emptyL$.
\end{enumerate}
\end{definition}

\noindent
Intuitively, if $\tree{T}_1 \preceq \tree{T}_2$ then the two trees share the same root and tree structure, except that some internal nodes in $\tree{T}_2$ are leaves in $\tree{T}_1$.
Note that if $\tree{T}_1 \preceq \tree{T}_2$ and $\node{n} \in \node{N}_1$ is such that $p_2^i(\node{n}) \in \node{N}_2$, then $p_2^i(\node{n}) = p_1^i(\node{n}) \in \node{N}_1$; that is, if $\node{n}$ is a node in $\tree{T}_1$ then it has the same ancestors in $\tree{T}_1$ as in $\tree{T}_2$.  We can also specify a prefix of a given tree by giving the set of nodes in the tree that should be turned into leaves.

\begin{definition}[Tree-prefix generation]\label{def:tree-prefix-generation}
Let $\tree{T} = (\node{N}, \node{r}, p, cs)$ be a finite ordered tree, and let $\node{L} \subseteq \node{N}$.  Then $\tpre{\tree{T}}{\node{L}}$, the \emph{tree prefix of $\tree{T}$ generated by $\node{L}$}, is the tree $(\node{N}', \node{r}, p', cs')$ given as follows.
\begin{itemize}
    \item 
    $\node{N}' = \node{N} \setminus \left( \bigcup_{\node{l} \in \node{L}} D_s(\node{l}) \right)$.
    \item
    Let $\node{n}' \in \node{N}'$. Then $p'(\node{n}') = p(\node{n}')$.
    \item
    Let $\node{n}' \in \node{N}'$.  Then $cs'(\node{n}') = \emptyL$ if $D_s(\node{n}) \cap \node{N}' = \emptyset$, and $cs(\node{n}')$ otherwise.
\end{itemize}
\end{definition}

\noindent
The nodes of $\tpre{\tree{T}}{\node{L}}$ are the nodes of $\tree{T}$, with however the strict descendants of nodes in $\node{L}$ removed.
It is straightforward to verify that $\tpre{\tree{T}}{\node{L}}$ is a finite ordered tree if $\tree{T}$ is, and that $\tpre{\tree{T}}{\node{L}} \preceq \tree{T}$.
While $\node{L}$ can be thought of as specifying nodes in $\tree{T}$ that should be converted into leaves in $\tpre{\tree{T}}{\node{L}}$, this intuition is only partially accurate, since
it is not the case that every $\node{l} \in \node{L}$ is a node in $\tpre{\tree{T}}{\node{L}}$.  In particular, if $\node{l}$ has a strict ancestor in $\node{L}$ this would cause the removal of $\node{l}$ from $\tpre{\tree{T}}{\node{L}}$.
However, if $\node{l} \in \node{L}$ has no strict ancestors in $\node{L}$ then it is indeed a leaf in $\tpre{\tree{T}}{\node{L}}$.

Functions may be defined inductively on ordered trees.

\begin{definition}[Inductive tree functions]~\label{def:node-function}
    Let $\tree{T} = (\node{N},\node{r},p,cs)$ be an ordered tree and $V$ a set; call $\node{N} \to V$ the set of \emph{tree functions from $\tree{T}$ to $V$.}
    \begin{enumerate}
        \item \label{def:node-function-inductively-generated}
              Tree function $f$ is \emph{inductively generated from} $g \in \node{N} \times V^* \rightarrow V$ iff for all $\node{n} \in \node{N}$, $f(\node{n}) = g(\node{n}, f(cs(\node{n})))$.

        \item \label{def:node-function-inductive-update}
              Let $f$ be inductively generated from $g$.  
              Then the \emph{inductive update}, $$f \iupd{\node{n}'}{v},$$ of $f$ at $\node{n}'$ by $v$ is the tree function inductively generated by $g[(\node{n}', \cdot):= v]$, where $g[(\node{n}', \cdot) := v](\node{n}, \vec{w}) = v$ if $\node{n} = \node{n}'$ and $g(\node{n}, \vec{w})$ otherwise.
    \end{enumerate}
\end{definition}

\noindent
Intuitively, a function $f$ is inductively generated from $g$ if $g$ computes the ``single steps" in the recursive definition of $f$.  That is, $f$ uses $g$ to compute the value associated with any $\node{n}$ based on the results $f$ returns for the children of $\node{n}$.  
Operation $f \iupd{\node{n}'}{v}$ then specifies a means for altering the values inductively generated for $f$:  not only does $f \iupd{\node{n}'}{v}$ change the value returned for $\node{n}$, but it also potentially changes the values for ancestors of $\node{n}$ as well.  Inductive updating can be generalized to $f \iupd{\node{n}_1 \cdots \node{n}_j}{v_1 \cdots v_j}$ in the obvious manner, where it is assumed that $\node{n}_1 \cdots \node{n}_j$ is duplicate-free.

We close this section with a lemma about inductively updated functions.  It in essence asserts that updates involving only ancestors of a node $\node{n}$ do not affect the value associated with $\node{n}$.

\begin{lemma}[Inductive update correspondence]\label{lem:inductive-update-correspondence}
    Let $\tree{T} = (\node{N}, \node{r}, p, cs)$ be a finite ordered tree, let $f \in \node{N} \rightarrow V$ be inductively generated from $g$, let $\vec{\node{n}} \in \node{N}^*$ be duplicate-free, and let $\vec{v} \in V^*$ be such that $|\vec{\node{n}}| = |\vec{v}|$.  Then for every $\node{n}' \in \node{N}$ such that $D(\node{n}') \cap \set(\vec{\node{n}}) = \emptyset$, $f\iupd{\vec{\node{n}}}{\vec{v}}(\node{n}') = f(\node{n}')$.
\end{lemma}

\begin{proof}
Fix $T, f, g, \vec{\node{n}} = \node{n}_1 \cdots \node{n}_j$ and $\vec{v} = v_1 \cdots v_j$ as specified in the statement of the lemma.  
The proof proceeds by tree induction.  So let $\node{n}' \in \node{N}$ be such that $D(\node{n}') \cap \{\node{n}_1, \ldots \node{n}_j\} = \emptyset$.
From the definition of $D(\cdot)$ we know that $\node{n}' \in D(\node{n}')$ and that each $\node{n}'' \in c(\node{n}')$ satisfies $D(\node{n}'') \subseteq D(\node{n}')$.  Therefore $\node{n}' \not\in \{\node{n}_1, \ldots, \node{n}_j\}$ and $D(\node{n}'') \cap \{\node{n}_1, \ldots, \node{n}_j\} = \emptyset$ for every child $\node{n}''$ of $\node{n}'$.   
The induction hypothesis guarantees that for every $\node{n}'' \in c(\node{n}')$, $f \iupd{\vec{\node{n}}}{\vec{v}}(\node{n}'') = f(\node{n}'')$, and thus that $f \iupd{\vec{\node{n}}}{\vec{v}}(cs(\node{n}')) = f(cs(\node{n}'))$.  We now reason as follows.
    \begin{align*}
    & f \iupd{\vec{\node{n}}}{\vec{v}}(\node{n}')
    &&
    \\
    &= g[\vec{\node{n}}:=\vec{v}](\node{n}', f\iupd{\vec{\node{n}}}{\vec{v}}(cs(\node{n}')))
    && \text{Definition of $f\iupd{\vec{\node{n}}}{\vec{v}}(\node{n}')$}
    \\
    &= g[\vec{\node{n}}:=\vec{v}](\node{n}', f(cs(\node{n}')))
    && \text{Induction hypothesis}
    \\
    &= g(\node{n}', f(cs(\node{n}')))
    && \text{Definition of $g[\vec{\node{n}} := \vec{v}]$}
    \\
    &= f(\node{n}')
    && \text{$f$ inductively generated}
    \end{align*}
\qedhere   
\end{proof}

\section{Support orderings and fixpoints}\label{sec:support-orderings}
A key contribution of this paper is the formulation of a novel characterization of least fixpoints for monotonic functions over subset lattices.  This characterization may be seen as constructive in a precise sense, and relies on the notion of \emph{support ordering}.

\begin{definition}[Support ordering]\label{def:support-ordering}
    Let $S$ be a set, let $(2^S, \subseteq, \bigcup, \bigcap)$ be the subset lattice generated by $S$, and let $f \in 2^S \rightarrow 2^S$ be a monotonic function over this lattice.  Then $(X, \prec)$ is a \emph{support ordering} for $f$ iff the following hold.
    \begin{enumerate}
        \item 
        $X \subseteq S$ and ${\prec} \subseteq X \times X $ is a binary relation on $X$.
        \item 
        For all $x \in X$, $x \in f(\preimg{{\prec}}{x})$.\footnote{Recall that $\preimg{{\prec}}{x} = \{x' \in X \mid x' \prec x\}$.}
    \end{enumerate}
\end{definition}
We call a support ordering $(X, \prec)$ for monotonic $f$ \emph{well-founded} if $\prec$ is well-founded and
$X$ is \emph{\mbox{(well-)}supported} for $f$ iff there is a (well-founded) binary relation ${\prec} \subseteq X \times X$ such that $(X, \prec)$ is a support ordering for $f$.

Using support orderings, we can give \emph{constructive} accounts of both the least and greatest fixpoints of monotonic functions over subset lattices. We call the characterization constructive in the case of $\mu f$ because, to show that $s \in \mu f$, it suffices to construct a well-founded support ordering $(X, \prec)$ for $f$ such that $s \in X$.  We also establish that $s \in \nu f$ exactly when there is a (not necessarily well-founded) support ordering $(X, \prec)$ such that $s \in X$.

The support-ordering characterization of the least fixpoint $\mu f$ of monotonic function $f$ over a given subset lattice relies on the following lemma, which asserts that the union of a collection of well-supported subsets of $S$ is also well-supported.

\begin{lemma}[Unions of well-supported sets]\label{lem:well-supported-sets}
    Let $S$ be a set, and let $f \in 2^S \rightarrow 2^S$ be monotonic over subset lattice $(2^S, \subseteq, \bigcup, \bigcap)$.  Also let $\mathcal{W} \subseteq 2^S$ be a set of well-supported sets for $f$.  Then $\bigcup\mathcal{W}$ is well-supported for $f$.
\end{lemma}

\remove{
\begin{proofsketch}
To prove that $\bigcup \mathcal{W}$ is well-supported for $f$ we note that there is an ordinal $\beta$ that is in bijective correspondence with $\mathcal{W}$.  We then define ordinal-indexed sequences, $X_\alpha$ and $\prec_\alpha$ ($\alpha < \beta$), with the property that $(X_\alpha, \prec_\alpha)$ is a support ordering for $f$, $\prec_\alpha$ is well-founded, and $\bigcup \mathcal{W} = \bigcup_{\alpha < \beta} X_\alpha$.  We subsequently show that ${\prec} = \bigcup_{\alpha < \beta} {\prec_\alpha}$ is well-founded and that $(\bigcup \mathcal{W}, \prec)$ is a support ordering for $f$, thereby establishing that $\bigcup \mathcal{W}$ is well-supported for $f$.  Details may be found in the appendix. \qedhere
\end{proofsketch}
}

\begin{proof}
This proof uses constructions over the ordinals, for which we use the von Neumann definition~\cite{ciesielski1997set}:  a well-ordered set $\alpha$ is a von Neumann ordinal iff it contains all ordinals preceding $\alpha$.  The well-ordering on von Neumann ordinals is often written $<$ and has the property that $\alpha < \beta$ iff $\alpha \in \beta$, which in turn is true iff $\alpha \subsetneq \beta$.  We also recall that in this paper, we assume the Axiom of Choice~\cite{jech2008axiom}.

To prove that $\bigcup \mathcal{W}$ is well-supported for $f$ we note that there is an ordinal $\beta$ that is in bijective correspondence with $\mathcal{W}$.  We then define ordinal-indexed sequences, $X_\alpha$ and $\prec_\alpha$ ($\alpha < \beta$), with the property that $(X_\alpha, \prec_\alpha)$ is a support ordering for $f$, $\prec_\alpha$ is well-founded, and $\bigcup \mathcal{W} = \bigcup_{\alpha < \beta} X_\alpha$.  We subsequently show that ${\prec} = \bigcup_{\alpha < \beta} {\prec_\alpha}$ is well-founded and that $(\bigcup \mathcal{W}, \prec)$ is a support ordering for $f$, thereby establishing that $\bigcup \mathcal{W}$ is well-supported for $f$.

To this end, let $|\mathcal{W}| = \beta$ be the cardinality of $\mathcal{W}$ (i.e.\ $\beta$ is the least ordinal in bijective correspondence to $\mathcal{W}$), and let $h \in \beta \rightarrow \mathcal{W}$ be a bijection.  It follows that $\mathcal{W} = \{h(\alpha) \mid \alpha < \beta\}$.  Also let $o \in \beta \rightarrow 2^{S \times S}$ be such that for any $\alpha < \beta$, $o(\alpha) \subseteq h(\alpha) \times h(\alpha)$ is a well-founded binary relation with the property that $(h(\alpha), o(\alpha))$ is a support ordering for $f$.

We now define the following ordinal-indexed sequences $X_\alpha$, $X'_\alpha$ and $\prec_\alpha$, $\alpha < \beta$, of subsets of $S$ and binary relations on $X_\alpha$, respectively, as follows, using transfinite recursion.
\begin{align*}
    X_\alpha
    &= \bigl(\bigcup_{\alpha' < \alpha} X_{\alpha'}\bigr) \cup h(\alpha)
    \\
    X'_{\alpha}
    &= h(\alpha) \setminus \bigl( \bigcup_{\alpha' < \alpha} X_{\alpha'} \bigr)
    \\
    \prec_\alpha
    &=
    \bigl(\bigcup_{\alpha' < \alpha} \prec_{\alpha'}\bigr) \cup
    \bigl\{(x',x) \in o(\alpha) \mid x \in X'_\alpha \bigr\}
\end{align*}
Note that $X_\alpha = (\bigcup_{\alpha' < \alpha} X_{\alpha'}) \cup X'_\alpha$ and that $(\bigcup_{\alpha' < \alpha} X_{\alpha'}) \cap X'_\alpha = \emptyset$.
Based on these definitions, it is easy to see that if $\alpha' < \alpha$ then $X_{\alpha'} \subseteq X_\alpha$ and ${\prec_{\alpha'}} \subseteq {\prec_{\alpha}}$.
We now prove two properties of $X_\alpha$ and $\prec_\alpha$ that will be used in what follows.

\begin{enumerate}[left = \parindent, label = P\arabic*., ref = P\arabic*]
    \item\label{item:support-ordering}
          For all $\alpha < \beta$, $(X_\alpha, \prec_\alpha)$ is a support ordering over $f$.

    \item\label{item:well-founded}
          For all $\alpha < \beta$, $\prec_\alpha$ is well-founded.
\end{enumerate}
To prove \ref{item:support-ordering} we use transfinite induction.  So fix $\alpha < \beta$; the induction hypothesis states that for any $\alpha' < \alpha$, $(X_{\alpha'}, \prec_{\alpha'})$ is a support ordering over $f$.  Now consider $x \in X_\alpha$; we must show that $x \in f(\preimg{\prec_\alpha}{x})$.  There are two cases.  In the first, $x \in \bigcup_{\alpha' < \alpha} X_{\alpha'}$, which means $x \in X_{\alpha'}$ for some $\alpha' < \alpha$.  In this case the induction hypothesis guarantees that $(X_{\alpha'}, \prec_{\alpha'})$ is a support ordering, meaning $x \in f(\preimg{\prec_{\alpha'}}{x})$.  Since ${\prec_{\alpha'}} \subseteq {\prec_{\alpha}}$ it follows that
  \[
      {\preimg{{\prec_{\alpha'}}}{x}} \subseteq {\preimg{{\prec_{\alpha}}}{x}},
  \]
  and since $f$ is monotonic and $x \in f(\preimg{\prec_{\alpha'}}{x})$, $x \in f(\preimg{\prec_{\alpha}}{x})$.  In the second case, $x \in X'_\alpha$.  Here it is easy to see that
  \[
      \preimg{{\prec_\alpha}}{x} = \{x' \mid (x',x) \in o(\alpha)\},
  \]
  and since $(h(\alpha), o(\alpha))$ is a support ordering, it immediately follows that $x \in f(\{x' \mid (x',x) \in o(\alpha)\}) = f(\preimg{\prec_\alpha}{x})$.  \ref{item:support-ordering} is thus proved.

To prove \ref{item:well-founded} we again use transfinite induction.  So fix $\alpha < \beta$.  The induction hypothesis states that for all $\alpha' < \alpha$, $\prec_{\alpha'}$ is well-founded; we must show that $\prec_\alpha$ is as well.  So consider a descending chain $C = \cdots \prec_\alpha x_2 \prec_\alpha x_1$; it suffices to show that $C$ must be finite.  There are three cases to consider.

    \begin{description}
        \item[$C$ is a chain in $\bigcup_{\alpha' < \alpha} \prec_{\alpha'}$.]
              In this case $x_1 \in \bigcup_{\alpha' < \alpha} X_{\alpha'}$, meaning there is an $\alpha' < \alpha$ such that $x_1 \in X_{\alpha'}$.  Since $\alpha_1 < \alpha_2$ implies ${\prec_{\alpha_1}} \subseteq {\prec_{\alpha_2}}$, it follows that each $x_i \in X_{\alpha'}$, each $x_{i+1} \prec_{\alpha'} x_i$, and that $C$ is thus a descending chain in $\prec_{\alpha'}$.  Since the induction hypothesis guarantees that $\prec_{\alpha'}$ is well-founded, $C$ must be finite.

        \item[$C$ is a chain in $o(\alpha)$.]
              In this case, since $o(\alpha)$ is well-founded, $C$ must be finite.

        \item[$C$ is a mixture of $\bigcup_{\alpha' < \alpha} \prec_{\alpha'}$ and $o(\alpha)$.]
              In this case, from the definition of $C$ and $\prec_\alpha$ it follows that $C$ can be split into two pieces:
              \begin{enumerate}
                  \item
                        an initial segment $x_{i} \prec_\alpha \cdots \prec_{\alpha} x_1$, where $i \geq 1$, $x_i \in \bigcup_{\alpha' < \alpha} X_{\alpha'}$, and for all $i > j \geq 1$, $x_{j} \in X'_\alpha$ and $(x_{j+1}, x_j) \in o(\alpha)$; and
                  \item
                        a segment $\cdots \prec_\alpha x_{i+1} \prec_\alpha x_i$, where for all $j \geq i$, $(x_{j+1}, x_j) \in \bigcup_{\alpha' < \alpha} \prec_{\alpha'}$.
              \end{enumerate}
              The previous arguments establish that each of these sub-chains must be finite, and thus $C$ is finite as well.
    \end{description}
    To finish the proof of the lemma, we note that the following hold, using arguments given above.

    \begin{itemize}
        \item $(\bigcup_{\alpha < \beta} X_\alpha, \bigcup_{\alpha < \beta} \prec_\alpha)$ is a support ordering.

        \item $\bigcup_{\alpha < \beta} \prec_\alpha$ is well-founded.

        \item $\bigcup\mathcal{W} = \bigcup_{\alpha < \beta} X_\alpha$.
    \end{itemize}
    From the definitions it therefore follows that $\bigcup \mathcal{W}$ is well-supported.\qedhere
\end{proof}

We now have the following.

\begin{theorem}\label{thm:well-supported}
Let $S$ be a set, and let $f \in 2^S \rightarrow 2^S$ be a monotonic function over the subset lattice $(2^S, \subseteq, \bigcup, \bigcap)$.
\begin{enumerate}
    \item\label{subthm:well-supported-1}
    For all $X \subseteq S$, if $X$ is well-supported for $f$ then $X \subseteq \mu f$.
    \item\label{subthm:well-supported-2}
    Let $\mathcal{X} = \{X \subseteq S \mid X \textnormal{ is well-supported for } f\}$.  Then $f(\bigcup\mathcal{X}) = \bigcup\mathcal{X}$.
\end{enumerate}
\end{theorem}

\begin{proof}
To prove the statement~\ref{subthm:well-supported-1}, suppose $X \subseteq S$ is well-supported for $f$, and let ${\prec} \subseteq X \times X$ be a well-founded relation such that $(X, {\prec})$ is a support ordering for $f$.
Also recall that $\mu f = \bigcap \{Y \subseteq S \mid f(Y) \subseteq Y\}$.  To show that $X \subseteq \mu f$ it suffices to show that $X \subseteq Y$ for all $Y$ such that $f(Y) \subseteq Y$.  So fix such a $Y$; we prove that for all $x \in X$, $x \in Y$ using well-founded induction on $\prec$.  So fix $x \in X$.  The induction hypothesis states that for all $x' \prec x$, $x' \in Y$.  By definition of support ordering we know that $x \in f(\preimg{{\prec}}{x})$; the induction hypothesis also guarantees that $\preimg{{\prec}}{x} \subseteq Y$.  Since $f$ is monotonic, $f(\preimg{{\prec}}{x}) \subseteq f(Y)$, and thus we have
$
    x \in f(\preimg{{\prec}}{x}) \subseteq f(Y) \subseteq Y.
$
Hence $x \in Y$.

As for statement~\ref{subthm:well-supported-2},
Lemma~\ref{lem:well-supported-sets} guarantees that $\bigcup\mathcal{X}$ is well-supported; let $\prec$ be the well-founded relation over $\bigcup \mathcal{X}$ such that $(\bigcup \mathcal{X}, \prec)$ is a support ordering for $f$.  It suffices to show that $f(\bigcup\mathcal{X}) \subseteq \bigcup\mathcal{X}$ and $\bigcup\mathcal{X} \subseteq f(\bigcup\mathcal{X})$.  For the former, by way of contradiction assume $x$ is such that $x \in f(\bigcup\mathcal{X})$ and $x \not\in \bigcup\mathcal{X}$.  Now consider the relation $\prec'$ on $(\bigcup\mathcal{X}) \cup \{x\}$ given by:  ${\prec'} = {\prec} \cup \{(x', x) \mid x' \in \bigcup\mathcal{X}\}$.  It is easy to see that $\prec'$ is well-founded, and that $((\bigcup\mathcal{X}) \cup \{x\}, \prec')$ is a support ordering for $f$.  This implies $(\bigcup\mathcal{X}) \cup \{x\} \subseteq \bigcup \mathcal{X}$, which contradicts the assumption that $x \not\in \bigcup\mathcal{X}$.
To see that $\bigcup\mathcal{X} \subseteq f(\bigcup\mathcal{X})$, note that, since $(\bigcup \mathcal{X}, \prec)$ is a support ordering, $x \in f(\preimg{{\prec}}{x})$ and $\preimg{{\prec}}{x} \subseteq \bigcup \mathcal{X}$ for all $x \in \bigcup\mathcal{X}$.  Since $f$ is monotonic we have
$
x \in f(\preimg{{\prec}}{x}) \subseteq f(\bigcup \mathcal{X}) 
$
for all $x \in \bigcup \mathcal{X}$.\qedhere
\end{proof}

The following corollary is immediate.
\begin{corollary}\label{cor:least-fixpoint}
    Let $f$ be a monotonic function over the subset lattice generated by $S$.  Then
    $
        \mu f = \bigcup \{X \in 2^S \mid X \textnormal{ is well-supported for $f$} \}.
    $
\end{corollary}
\begin{proof}
    Let $\mathcal{X} = \bigcup \{X \mid X \textnormal{ is well-supported for } f\}$.  Theorem~\ref{thm:well-supported} guarantees that $\bigcup\mathcal{X} \subseteq \mu f$ and that $\bigcup\mathcal{X}$ is a fixpoint of $f$.  It is immediate that $\bigcup\mathcal{X} = \mu f$.\qedhere
\end{proof}

This corollary may be seen as providing a \emph{constructive} characterization of $\mu f$ in the following sense:  to establish that $x \in \mu f$ it suffices to construct a well-supported set $X$ for $f$ such that $x \in X$.

Support orderings can also be used to characterize $\nu f$, the greatest fixpoint of monotonic function $f$.  This characterization relies on the following observations.
\begin{theorem}\label{thm:supported}
Let $S$ be a set, let $f \in 2^S \to 2^S$ be a monotonic function over the subset lattice $(2^S, \subseteq, \bigcup, \bigcap)$, and let $X \subseteq S$. Then $X$ is supported for $f$ if, and only if, $X \subseteq f(X)$.
\end{theorem}
\begin{proof}
    Let $X \subseteq S$ and $f$ be given. We prove both implications separately.
    \begin{itemize}
        \item[$\Rightarrow$] Assume $X$ is supported for $f$, so there is a support ordering $(X, \prec)$ for $f$. To show that $X \subseteq \nu f$, fix $x \in X$; we must show that $x \in f(X)$. Since $(X, \prec)$ is a support ordering for $f$, it follows that $x \in f(\preimg{{\prec}}{x})$.  Since $f$ is monotonic and $\preimg{{\prec}}{x} \subseteq X$, we have $x \in f(X)$.
        \item[$\Leftarrow$] Assume $X \subseteq f(X)$. To prove this result we must come up with ${\prec} \subseteq X \times X$ such that $(X, \prec)$ is a support ordering for $f$.  Consider ${\prec} = U_X = X \times X$.  Since for every $x \in X, x \in f(X)$, the desired result is immediate.\qedhere
    \end{itemize}
\end{proof}

This theorem establishes that any $X \subseteq S$ is a post-fixpoint for monotonic $f$ (i.e.\/ $X \subseteq f(X)$) if and only if $X$ is supported for $f$.  We also know from Lemma~\ref{lem:tarski-knaster} that
$
    \nu f = \bigcup \{X \subseteq S \mid X \subseteq f(X) \}.
$
Thus we have the following.

\begin{corollary}\label{cor:greatest-fixpoint}
    Let $f$ be a monotonic function over the subset lattice generated by $S$.  Then
    $
        \nu f = \bigcup \{ X \in 2^S \mid X \textnormal{ is supported for $f$}\}.
    $
\end{corollary}

This section closes with definitions and results on support orderings that we use later in this paper.  In what follows we fix a set $S$ and the associated complete lattice $(2^S, \subseteq, \bigcup, \bigcap)$.
The first lemma establishes that any extension of a support ordering is also a support ordering.

\begin{lemma}[Extensions of support orderings]\label{lem:support-ordering-extension}
Let $f \in 2^S \rightarrow 2^S$ be monotonic, let $(X, \prec)$ be a support ordering for $f$, and let ${\prec'} \subseteq S \times S$ be an extension of $\prec$.  Then $(X, \prec')$ is also a support ordering for $f$.
\end{lemma}
\begin{proof}
Follows from monotonicity of $f$ and the fact that $\preimg{{\prec}}{s} \subseteq \preimg{{\prec'}}{s}$.\qedhere
\end{proof}

\noindent
The next lemma establishes that unions of support orderings are also support orderings.
\begin{lemma}[Unions of support orderings]\label{lem:unions-of-support-orderings}
Let $f \in 2^S \rightarrow 2^S$ be monotonic, and let $\mathcal{X}$ be a set of support orderings for $f$.  Then $(S_\mathcal{X},{\prec}_\mathcal{X})$ is a support ordering for $f$, where
\begin{align*}
S_\mathcal{X}   &= \bigcup_{(S,{\prec}) \in \mathcal{X}} S \\
{\prec}_\mathcal{X} &= \bigcup_{(S,{\prec}) \in \mathcal{X}} {\prec}.
\end{align*}
\end{lemma}
\begin{proof}
It suffices to show that for every $s \in S_\mathcal{X}$, $s \in f(\preimg{{\prec_{\mathcal{X}}}}{s})$.  So fix such and $s$.  Since $s \in S_\mathcal{X}$ there is $(S,{\prec}) \in \mathcal{X}$ such that $s \in S$, and as $(S,{\prec})$ is a support ordering we know that $s \in f(\preimg{{\prec}}{s})$.  That $s \in f(\preimg{{\prec_\mathcal{X}}}{s})$ follows immediately from the fact that $f$ is monotonic and $\preimg{{\prec}}{x} \subseteq \preimg{{\prec_\mathcal{X}}}{x}$.
\qedhere
\end{proof}

\noindent
This result should be contrasted with Lemma~\ref{lem:well-supported-sets}.  On the one hand, this lemma asserts a property about all sets of support orderings, whereas that lemma only refers to sets of well-supported sets, i.e.\/ sets of sets having well-founded support orderings.  On the other hand, this lemma makes no guarantees about the properties of support ordering $(S_\mathcal{X}, {\prec}_\mathcal{X})$ \emph{vis \`a vis} the orderings $(S,{\prec})$.  In particular, if all the $(S,{\prec}) \in \mathcal{X}$ are well-founded, it does not follow that ${\prec}_\mathcal{X}$ is well-founded.
Lemma~\ref{lem:well-supported-sets} on the other does guarantee that a well-founded ordering over $S_{\mathcal{X}}$ does exist if each $(S,\prec) \in \mathcal{X}$ is well-founded.

It can also be shown that well-founded support orderings can be extended to well-orderings.

\begin{lemma}[Well-orderings for well-supported sets]
Let $f \in 2^S \rightarrow 2^S$ be monotonic, and let $(X,\prec)$ be a well-founded support ordering for $f$.  Then there is a well-ordering ${\prec'} \subseteq X \times X$ extending $\prec$ such that $(X, \prec')$ is a support ordering for $f$.
\end{lemma}
\begin{proof}
Follows from Lemmas~\ref{lem:well-ordering-extension} and~\ref{lem:support-ordering-extension}. \qedhere
\end{proof}

The next result is a corollary of earlier lemmas.

\begin{corollary}[Support orderings for fixpoints]\label{cor:support-fixpoints}
Let $f \in 2^S \rightarrow 2^S$ be monotonic.
\begin{enumerate}
    \item 
    $(\nu f, U_{\nu f})$ is a support ordering for $f$.\footnote{Recall that $U_{\nu f} = \nu f \times \nu f$.}
    \item
    There is a well-ordering ${\prec} \subseteq \mu f \times \mu f$ such that $(\mu f, \prec)$ is a support ordering for $f$.
\end{enumerate}
\end{corollary}

For technical convenience in what follows we introduce the notions of \emph{$\sigma$-compatible} and \emph{$\sigma$-maximal} support orderings for monotonic $f$ and $\sigma \in \{\mu,\nu\}$.

\begin{definition}[Compatible, maximal support orderings]
Let $f \in 2^S \rightarrow 2^S$ be monotonic, let $\sigma \in \{\mu,\nu\}$, and let $(X, \prec)$ be a support ordering for $f$.  
\begin{enumerate}
    \item $(X, \prec)$ is \emph{$\sigma$-compatible for $f$} iff either $\sigma = \nu$, or $\sigma = \mu$ and $\prec$ is well-founded.
    
    \item $(X, \prec)$ is \emph{$\sigma$-maximal for $f$} iff $X = \sigma f$ and one of the following holds.
    \begin{enumerate}
        \item $\sigma = \nu$ and ${\prec} = U_X$.
        \item $\sigma = \mu$ and $\prec$ is a well-ordering over $X \times X$.
    \end{enumerate}
\end{enumerate}
\end{definition}

\noindent
Corollary~\ref{cor:support-fixpoints} ensures that for any monotonic $f \in 2^S \rightarrow 2^S$ and $\sigma \in \{\mu,\nu\}$ there is a $\sigma$-maximal support ordering for $f$.  When $\sigma = \nu$ this $\sigma$-maximal support ordering is unique, although this uniqueness property in general fails to hold for $\sigma = \mu$; while the fixpoint is unique, the associated well-ordering may not be.

\section{The Propositional Modal Mu-Calculus}\label{sec:mu-calculus}

This section defines the syntax and semantics of the modal mu-calculus and also establishes properties of the logic that will be used later in the paper.

\subsection{Labeled transition systems}

Labeled transition systems are intended to model the behavior of discrete systems.  Define a \emph{sort} $\Sigma$ to be the set of atomic actions that a system can perform.

\begin{definition}[Labeled transition system]\label{def:lts}
    A \emph{labeled transition system} (LTS) of sort $\Sigma$ is a pair $\lts{S}$, where $\states{S}$ is a set of \emph{states} and ${\xrightarrow{}} \subseteq \states{S} \times \Sigma \times \states{S}$ is the \emph{transition relation}.  We write $s \xrightarrow{a} s'$ when $(s, a, s') \in {\xrightarrow{}}$ and $s \xrightarrow{a}$ when $s \xrightarrow{a} s'$ for some $s' \in \states{S}$.  If $K \subseteq \Sigma$ then we write $s \xrightarrow{K} s'$ iff $s \xrightarrow{a} s'$ for some $a \in K$ and $s \xrightarrow{K}$ if $s \xrightarrow{K} s'$ for some $s'$.  If there is no $s'$ such that $s \xrightarrow{a} s'$ / $s \xrightarrow{K} s'$ then we denote this as $s \centernot{\xrightarrow{a}}$ / $s \centernot{\xrightarrow{K}}$.
\end{definition}

\noindent
An LTS $(\states{S}, \xrightarrow{})$ of sort $\Sigma$ represents a system whose state space is $\states{S}$; the presence of transition $s \xrightarrow{a} s'$ indicates that when the system is in state $s$, it can engage in an atomic execution step labeled by $a$ and evolve to state $s'$.
We now introduce two notions of \emph{predecessors} of sets of states in an LTS.

\begin{definition}[Predecessor Sets]\label{def:pre}
    Let $\lts{S}$ be an LTS of sort $\Sigma$, with $S \subseteq \Bstates{S}$ and $K \subseteq \Sigma$.  Then we define the following.
    \begin{enumerate}
        \item $\pre_{\dia{K}}(S) = \{s \in \Bstates{S} \mid \exists s' \in S \colon s \xrightarrow{K} s' \}$
        \item $\pre_{[K]}(S) = \{s \in \Bstates{S} \mid \forall s' \in \Bstates{S} \colon s \xrightarrow{K} s' \implies  s' \in S \}$
    \end{enumerate}
\end{definition}

\noindent
If state $s \in \pre_{\dia{K}}(S)$ then one of its outgoing transitions is labeled by $K$ and leads to a state in $S$, while $s \in \pre_{[K]}(S)$ holds iff every outgoing transition from $s$ labeled by $K$ leads to $S$.  Note that if $s \centernot{\xrightarrow{K}}$ then $s \in \pre_{[K]}(S)$ but $s \not\in \pre_{\dia{K}}(S)$. It immediately follows from the definitions that the operators satisfy the following properties.
\begin{lemma}\label{lem:pre}
    Let $\lts{S}$ be an LTS of sort $\Sigma$, with $K \subseteq \Sigma$ and $S_1, S_2 \subseteq \states{S}$.
    If $S_1 \subseteq S_2$ then $\pre_{[K]}(S_1) \subseteq \pre_{[K]}(S_2)$ and $\pre_{\langle K \rangle}(S_1) \subseteq \pre_{\langle K \rangle}(S_2)$.
\end{lemma}







\subsection{Propositional modal mu-calculus}\label{subsec:propositional-modal-mu-calculus}

The propositional modal mu-calculus, which we often just call the mu-calculus, is a logic for describing properties of states in labeled transition systems.  The version of the logic considered here matches the one in~\cite{BS1992}, which slightly extends~\cite{Koz1983} by allowing sets of labels in the modalities.  We first define the set of formulas of the mu-calculus, then the \emph{well-formed} formulas.  The latter will be the object of study in this paper.

\begin{definition}[Mu-calculus formulas]
    Let $\Sigma$ be a sort and $\Var$ a countably infinite set of \emph{propositional variables}.
    Then
    formulas of the propositional modal mu-calculus over $\Sigma$ and $\Var$ are given by the following grammar
    $$
        \Phi ::= Z
        \mid \lnot\Phi'
        \mid \Phi_1 \land \Phi_2
        \mid [K] \Phi'
        \mid \nu Z.\Phi'
    $$
    where $K \subseteq \Sigma$ and $Z \in \Var$.
\end{definition}

We assume the usual definitions of subformula, etc.
To define the well-formed mu-calculus formulas, we first review the notions of free, bound and positive variables.  Occurrences of $Z$ in $\nu Z.\Phi'$ are said to be \emph{bound}; an occurrence of a variable in a formula that is not bound within the formula is called \emph{free}. A variable $Z$ is free within a formula if it has at least one free occurrence in the formula, and is called \emph{non-free} otherwise.  (So a variable may be non-free in a formula if either all its occurrences are bound, or if the variable has no occurrences at all in the formula.)  A variable $Z$ is \emph{positive} in $\Phi$ if every free occurrence of $Z$ in $\Phi$ occurs inside the scope of an even number of negations.  We can now define the well-formed mu-calculus formulas as follows.

\begin{definition}[Well-formed mu-calculus formulas]~\label{def:mu-calculus-syntax}
A mu-calculus formula over $\Sigma$ and $\Var$ is \emph{well-formed} if each of its subformulas of form $\nu Z.\Phi$ satisfies:  $Z$ is positive in $\Phi$.

We use $\muformsSV$ for the set of well-formed mu-calculus formulas over $\Sigma$ and $\Var$, or just $\muforms$ when $\Sigma$ and $\Var$ are obvious from context.
\end{definition}

We denote substitution for free variables in the usual fashion:  if $Z_1 \cdots Z_n \in \Var^*$ is duplicate-free and $\Phi_1 \cdots \Phi_n \in \muforms^*$ then
we write $\Phi[Z_1 \cdots Z_n := \Phi_1 \cdots \Phi_n]$ for the simultaneous capture-free substitution of each $Z_i$ by $\Phi_i$ in $\Phi$.
We also use the following the standard derived operators.
\begin{align*}
\Phi_1 \lor \Phi_2 
&= \lnot (\lnot \Phi_1 \land \lnot \Phi_2)   \\
\dia{K} \Phi 
&= \lnot [K] \lnot\Phi\\
\mu Z. \Phi 
&= \lnot \nu Z . \lnot \Phi[Z := \lnot Z]\\
\true 
&= \nu Z . Z\\
\false 
&= \lnot \true
\end{align*}
In the definiton of $\mu Z.\Phi$, note that if $Z$ is positive in $\Phi$ then it is also positive in $\lnot \Phi[Z := \lnot Z]$.
Following standard convention, we refer to $\land$ and $\lor$, $[K]$ and $\dia{K}$, and $\nu$ and $\mu$ as \emph{duals}. Formulas extended with these dual operators are in \emph{positive normal form} iff all negation symbols directly apply to variable symbols. It is well-known that every formula can be rewritten to positive normal form when the duals are included in the logic.
We refer to formulas of form $\nu Z.\Phi$ or $\mu Z.\Phi$ as \emph{fixpoint formulas} and write $\sigma Z.\Phi$ for a generic such formula (so $\sigma$ may be either $\nu$ or $\mu$).

To define the semantics of mu-calculus formulas we use \emph{valuations} to handle free propositional variables.
\begin{definition}[Valuations]~\label{def:valuation}
Let $\T = \lts{S}$ be an LTS and $\Var$ a countably infinite set of variables.
Then a \emph{valuation for $\Var$ over $\T$} is a function $\V \in \Var \rightarrow 2^{\states{S}}$.
\end{definition}

\noindent
Since a valuation $\V$ is a function, standard operations on functions such as $\V[Z_1 \cdots Z_n := S_1 \cdots S_n]$, where $Z_1 \cdots Z_n \in \Var^*$ is duplicate-free and $S_1 \cdots S_n \in (2^\states{S})^*$, are applicable.

The semantics of the mu-calculus is now defined as follows.
\begin{definition}[Mu-calculus semantics]~\label{def:mu-calculus-semantics}
    Let $\T = \lts{S}$ be an LTS of sort $\Sigma$
    and $\V \in \Var \rightarrow 2^\states{S}$ a valuation.
    Then
    the semantic function $ \semTV{\Phi} \subseteq \mathcal{S}$, where $\Phi \in \muformsSV$, is defined as follows.
    \begin{align*}
        \semTV{Z}                   & = \V(Z)                                                                         \\
        \semTV{\lnot \Phi}          & = \mathcal{S} \setdiff \semTV{\Phi}                                             \\
        \semTV{\Phi_1 \land \Phi_2} & = \semTV{\Phi_1} \cap \semTV{\Phi_2}                                            \\
        \semTV{[K]\Phi}             & = \pre_{[K]}\left(\semTV{\Phi}\right)                                           \\
        \semTV{\nu Z . \Phi}        & = \bigcup \{ S \subseteq \states{S} \mid S \subseteq \semT{\Phi}{\V[Z := S]} \}
    \end{align*}
    If $s \in \semTV{\Phi}$ then we say that $s$ \emph{satisfies} $\Phi$ in the context of $\T$ and $\V$.
\end{definition}


\noindent
For the dual operators one may derive the following semantic equivalences.
\begin{align*}
    \semTV{\Phi_1 \lor \Phi_2} & = \semTV{\Phi_1} \cup \semTV{\Phi_2}                                            \\
    \semTV{\dia{K} \Phi}       & = \pre_{\dia{K}}\left( \semTV{\Phi} \right)                                     \\
    \semTV{\mu Z . \Phi}       & = \bigcap \{ S \subseteq \states{S} \mid \semT{\Phi}{\V[Z := S]} \subseteq S \}
\end{align*}

We frequently wish to view formulas as functions of their free variables.  The next definition introduces this concept at both the syntactic and semantic level.

\begin{definition}[Formula functions]\label{def:formula-functions}
Let $Z \in \Var$ be a variable and $\Phi \in \muformsSV$ be a formula.
\begin{enumerate}
    \item \label{subdef:synactic-formula-function}
        The \emph{syntactic function}, $\synf{Z}{\Phi} \in \muformsSV \to \muformsSV$, for $Z$ and $\Phi$ is defined as:
        $$
        (\synf{Z}{\Phi})(\Phi') = \Phi[Z:=\Phi'].
        $$
    \item \label{subdef:semantic-formula-function}
        Let $\T = \lts{S}$ be an LTS of sort $\Sigma$, and $\V \in \Var \to \states{S}$ a valuation over $\T$.
        Then the \emph{semantic function}, $\semfTV{Z}{\Phi} \in 2^{\states{S}} \to 2^{\states{S}}$, for $Z$ and $\Phi$ is defined as:
        $$
        \semfTV{Z}{\Phi}(S) = \semT{\Phi}{\V[Z:=S]}.
        $$
\end{enumerate}
\end{definition}

We now state a well-known monotonicity result for formulas in which $Z$ is positive.

\begin{lemma}[Mu-calculus monotonicity]~\label{lem:mu-calculus-semantic-monotonicity}
    Fix $\T = \lts{S}$ and $\V$, and let $\Phi \in \muformsSV$ be such that $Z \in \Var$ is positive in $\Phi$.  Then $\semfZTV{\Phi} \in 2^\states{S} \to 2^\states{S}$ is monotonic over the subset lattice for $\states{S}$.
\end{lemma}
\begin{proof}
We must prove that for all $S_1 \subseteq S_2 \subseteq \states{S}$, 
\[
\semfZTV{\Phi}(S_1) 
= \semT{\Phi}{\V[Z := S_1]} 
\subseteq \semT{\Phi}{\V[Z := S_2]}
= \semfZTV{\Phi}(S_2).
\]
The proof proceeds by induction on $\Phi$. \qedhere
\end{proof}

It turns out that the semantics of well-formed formulas $\nu Z.\Phi$ and $\mu Z.\Phi$ with respect to $\T = \lts{S}$ can be characterized as the greatest and least fixpoints, respectively, of $\semfT{Z}{\Phi}{\V}$ over the subset lattice induced $\states{S}$.  In particular, Lemma~\ref{lem:mu-calculus-semantic-monotonicity} guarantees the monotonicity of $\semfT{Z}{\Phi}{\V}$ over this lattice; Lemma~\ref{lem:tarski-knaster} then implies the characterization.  The next lemma formalizes this insight.

\begin{lemma}[Fixpoint characterizations of formula functions]\label{lem:fixpoint-characterizations}
    Fix $\T$ and $\V$, let $Z \in \Var$, and let $\Phi \in \muformsSV$ be such that $Z \in \Var$ is positive in $\Phi$.  Then $\nu \left( \semfZTV{\Phi} \right) = \semTV{\nu Z.\Phi}$ and $\mu \left(\semfZTV{\Phi} \right) = \semTV{\mu Z.\Phi}$.
\end{lemma}
\begin{proof}
Fix $\T, \V, Z$ and $\Phi$ so that $Z$ is positive in $\Phi$.  We prove the $\nu$ case; the $\mu$ case is left to the reader.  Since $Z$ is positive in $\Phi$ Lemma~\ref{lem:mu-calculus-semantic-monotonicity} guarantees that $\semfTV{Z}{\Phi} \in 2^{\states{S}} \to 2^{\states{S}}$ is monotonic over the subset lattice generated by $\states{S}$, and thus $\nu\left(\semfTV{Z}{\Phi}\right) \subseteq \states{S}$ exists.  We reason as follows.
\begin{align*}
\nu\left(\semfTV{Z}{\Phi}\right)
    &= \bigcup \{S \subseteq \states{S} \mid S \subseteq \left(\semfTV{Z}{\Phi}\right)(S)\}
    && \text{Lemma~\ref{lem:tarski-knaster}}
\\
    &= \bigcup \{S \subseteq \states{S} \mid S \subseteq \semT{\Phi}{\V[Z := S]}\}
    && \text{Definition of $\semfTV{Z}{\Phi}$}
\\
    &= \semTV{\nu Z.\Phi}
    && \text{Definition of $\semTV{\nu Z.\Phi}$}
\end{align*}
\qedhere
\end{proof}

In the remainder of this section, we establish some identities on mu-calculus formulas that we will use later in this paper.
The first result establishes a correspondence between substitution and valuation updates.
\begin{lemma}[Substitution and valuations]\label{lem:substitution}
    Fix $\T$ and $\V$,
    let $\Phi, \Phi_1, \ldots, \Phi_n \in \muformsSV$ and let $Z_1 \cdots Z_n \in \Var^*$ be duplicate-free.  Then
    \[
        \semTV{\Phi[Z_1 \cdots Z_n := \Phi_1 \cdots \Phi_n]}
        =
        \semT{\Phi}{\V[Z_1 \cdots Z_n := \semTV{\Phi_1} \cdots \semTV{\Phi_n}]}
    \]
\end{lemma}
\begin{proof}
For notational conciseness we write $\vec{Z}$ for $Z_1 \cdots Z_n$, $\vec{\Phi}$ for $\Phi_1 \cdots \Phi_n$, and $\semTV{\vec{\Phi}}$ for $\semTV{\Phi_1} \cdots \semTV{\Phi_n}$.
The proof proceeds by induction on the structure of $\Phi$. Most cases follow straightforwardly from the induction hypothesis; we consider the cases for formulas that involve variables at the top level.
\begin{itemize}
    \item $\Phi = Z$. 
    If $Z = Z_i$ for some $i$, the following reasoning applies.
    \begin{flalign*}
    & \semTV{Z[\vec{Z} := \vec{\Phi}]} 
    && \\
    &= \semTV{Z_i[\vec{Z} := \vec{\Phi}]}
    && \text{$Z = Z_i$}
    \\
    & = \semTV{\Phi_i}
    && \text{Definition of substitution}
    \\
    & = \left(\V[\vec{Z} := \semTV{{\vec{\Phi}}}\,]\right)(Z_i)
    && \text{Definition of of $\V[\vec{Z} := \semTV{{\vec{\Phi}}}\,]$}
    \\
    & = \semT{Z}{\V[\vec{Z} := \semTV{{\vec{\Phi}}}]}
    && \text{Definition of $\semT{Z_i}{\V[\vec{Z} := \semTV{{\vec{\Phi}}}\,]}$, $Z = Z_i$}
    \end{flalign*}
    If $Z \neq Z_i$ for all $i$, we reason as follows.
    \begin{flalign*}
    & \semTV{Z[\vec{Z} := \vec{\Phi}]}
    && 
    \\
    & = \semTV{Z}
    && \text{$Z \neq Z_i$ for all $i$}
    \\
    & = \V(Z) 
    && \text{Definition of $\semTV{Z}$}
    \\
    & = \left(\V[\vec{Z} := \semTV{{\vec{\Phi}}}\,]\right)(Z)
    && \text{Definition of $\V[\vec{Z} := \semTV{{\vec{\Phi}}}\,]$, $Z \neq Z_i$ for all $i$}
    \\
    & = \semT{Z}{\V[\vec{Z} := \semTV{{\vec{\Phi}}}]}
    && \text{Definition of $\semT{Z}{\V[\vec{Z} := \semTV{{\vec{\Phi}}}]}$}
    \end{flalign*}
    \item $\Phi = \nu Z . \Phi'$. 
    In this case, and without loss of generality due to the fact that substitution is capture-free, we may assume that $Z$ is not free in any of the $\Phi_i$.
    If $Z \neq Z_i$ for all $i$, the following reasoning applies.
    \begin{flalign*}
        & \semTV{(\nu Z . \Phi')[\vec{Z} := \vec{\Phi}]} 
        && 
        \\
        & = \semTV{\nu Z . \left(\Phi'[\vec{Z} := \vec{\Phi}]\right)}
        && \text{$Z \neq Z_i$, $Z$ not free in $\vec{\Phi}$}
        \\
        & = \bigcup\{ S \subseteq \states{S} \mid S \subseteq \semT{\Phi'[\vec{Z} := \vec{\Phi}]}{\V[Z := S]} \} 
        && \text{Definition of $\semTV{\nu Z . \cdots}$}
        \\
        & = \bigcup\{ S \subseteq \states{S} \mid S \subseteq \semT{\Phi'}{(\V[Z := S])[\vec{Z} := \semTV{{\vec{\Phi}}}]} \} 
        && \text{Induction hypothesis}
        \\
        & = \bigcup\{ S \subseteq \states{S} \mid S \subseteq \semT{\Phi'}{(\V[\vec{Z} := \semTV{{\vec{\Phi}}}\,])[Z := S]} \}
        && \text{$Z \neq Z_i$, for all $i$}
        \\
        & = \semT{\nu Z . \Phi'}{\V[\vec{Z} := \semTV{{\vec{\Phi}}}\,]}
        &&
    \end{flalign*}

    If $Z = Z_i$ for some $i$ then from definition of substitution it follows that
    \[
    (\nu Z.\Phi')[\vec{Z} := \vec{\Phi}]
    =
    (\nu Z.\Phi')[\vec{Z}_{\neq i} := \vec{\Phi}_{\neq i}]
    \]
    where $\vec{Z}_{\neq i} = Z_1 \cdots Z_{i-1} Z_{i+1} \cdots Z_n$, and similarly for $\vec{\Phi}_{\neq i}$.
    Using an inductive argument on $|\vec{Z}| = n$ we can assume that 
    \[
    \semTV{(\nu Z.\Phi')[\vec{Z}_{\neq i} := \vec{\Phi}_{\neq i}]}
    =
    \semT{\nu Z.\Phi'}{\V[\vec{Z}_{\neq i} := \semTV{\vec{\Phi}_{\neq i}}]}.
    \]
    From this fact, and the observation that 
    \[
    \left( \V[\vec{Z} := \semTV{\vec{\Phi}}] \right)[Z_i = S]
    =
    \left( \V[\vec{Z}_{\neq i} := \semTV{\vec{\Phi}_{\neq i}}] \right)[Z_i = S]
    \]
    it is easy to establish the desired result.  The details are left to the reader.
    \qedhere
\end{itemize}
\end{proof}


\noindent
The next result is immediate from this lemma and Lemma~\ref{lem:mu-calculus-semantic-monotonicity}

\begin{corollary}[Monotonicity of substitution]\label{cor:monotonicity-of-substitution}
    Fix $\T, \V$ and $Z \in \Var$, and let $\Phi$, $\Phi_1$ and $\Phi_2$ be such that $Z$ is positive in $\Phi$ and $\semTV{\Phi_1} \subseteq \semTV{\Phi_2}$.  Then
    \[
        \semTV{\Phi[ Z:= \Phi_1 ]}
        \subseteq
        \semTV{\Phi[ Z:= \Phi_2 ]}.
    \]
\end{corollary}

The following lemma states that the semantics of a formula is not affected by the values that a valuation assigns to variables that are not free in the formula.

\begin{lemma}[Semantics and non-free variables]\label{lem:substitution-of-bound-variables}
    Let $\Phi$ be a formula in which $Z \in \Var$ is not free. Then for all $\T = \lts{S}, \V$ and $S \subseteq \states{S}$,
    $
        \semTV{\Phi} = \semfZTV{\Phi}(S).
    $
\end{lemma}
\begin{proof}
Observe that $\semfZTV{\Phi}(S) = \semT{\Phi}{\V[Z := S]}$. It follows from a routine induction on the structure of $\Phi$ that $\semTV{\Phi} = \semT{\Phi}{\V[Z := S]}$ if $Z$ does not appear free in $\Phi$. \qedhere
\end{proof}

\noindent
Lemmas~\ref{lem:fixpoint-characterizations} and~\ref{lem:substitution} guarantee that formulas can be \emph{unfolded}.

\begin{lemma}[Fixpoint unfolding]\label{lem:unfolding}
    Fix $\T$ and $\V$, and let $\sigma Z.\Phi$ be a fixpoint formula.  Then $\semTV{\sigma Z.\Phi} = \semTV{\Phi[Z := \sigma Z.\Phi]}$.
\end{lemma}
\begin{proof}
The result is established using the following reasoning.
\begin{flalign*}
    \semTV{\sigma Z.\Phi} 
    & = \sigma \semfTV{Z}{\Phi} 
    && \text{Lemma~\ref{lem:fixpoint-characterizations}}
    \\
    & = \semfTV{Z}{\Phi}(\sigma \semfTV{Z}{\Phi}) 
    && \text{Definition of fixed point}
    \\
    & = \semT{\Phi}{[Z := \sigma \semfTV{Z}{\Phi}]} 
    && \text{Definition~\ref{def:formula-functions}}
    \\
    & = \semT{\Phi}{[Z := \semTV{\sigma Z.\Phi}]}
    && \text{Lemma~\ref{lem:fixpoint-characterizations}}
    \\
    & = \semTV{\Phi[Z := \sigma Z . \Phi} 
    && \text{Lemma~\ref{lem:substitution}}
\end{flalign*}
\qedhere
\end{proof}

\section{Base proof system}\label{sec:base-proof-system}

This section defines the base proof system for the mu-calculus considered in this paper.  It mirrors the ones given in~\cite{BS1992,Bra1991} and is intended to prove that sets of states in a transition system satisfy mu-calculus formulas.  Later in the paper we will extend this proof system in various ways.
In what follows, fix sort $\Sigma$ and countably infinite propositional variable set $\Var$.

\subsection{Definition lists and sequents}
The proof system reasons about \emph{sequents}, which make statements about sets of states satisfying mu-formulas.  Our sequents also involve \emph{definition lists}, which are used in the construction of proofs to control the unfolding of fixpoint formulas.  We define definition lists and sequents below.

\paragraph{Definition lists.}
Definition lists bind fresh variables in $\Var$ to fixpoint formulas.  In a proof setting, such a list records the fixpoint formulas that have been unfolded previously, so that decisions about whether or to unfold again later in the proof can be made.  Here we define definition lists formally and establish basic results about them.

\begin{definition}[Definition lists]\label{def:definition-list}
    A \emph{definition list} $\Delta$ is a finite sequence $(U_1 = \Phi_1) \cdots (U_n = \Phi_n)$, with each $U_i \in \Var$ and $\Phi_i \in \muformsSV$, satisfying the following.
    \begin{enumerate}
        \item
              If $i \neq j$ then $U_i \neq U_j$.
        \item
              For all $1 \leq i,j \leq n$, $U_i$ cannot appear bound anywhere in $\Phi_j$.
        \item
              If $i \leq j$ then $U_j$ cannot appear free in $\Phi_i$.
    \end{enumerate}
    The individual $(U_i = \Phi_i)$ in a definition list are sometimes called \emph{definitions}, with each $U_i$ referred to as a \emph{definitional constant}. We also define $\Delta(U_i) = \Phi_i$ to be the formula associated with $U_i$ in $\Delta$ and $\dom(\Delta) = \{U_1, \ldots, U_n\}$ to be the set of definitional constants in $\Delta$.
\end{definition}

\noindent
A definition list consists of a sequence bindings, or definitions, of form $(U_i = \Phi_i)$.  The constraints ensure that every $U_i$ is unique and not part of any $\sigma$ operator inside $\Phi_i$.  $U_i$ may also appear free in definitions to the right of $(U_i = \Phi_i)$, but not in $\Phi_i$ or in definitions to the left.

Since definition lists are sequences the sequence notations defined in Section~\ref{subsec:sequences}, including $\emptyL$, $\cdot$ and $\preceq$, are applicable.  We now introduce additional definition-list terms and notation.

\begin{definition}[Prefixes / suffixes of definition lists]\label{def:definition-list-prefix}
Let $\Delta = \Delta_1 \cdot (U = \Phi) \cdot \Delta_2$ be a definition list.
\begin{enumerate}
    \item 
    $\Delta_{\prec U} = \Delta_1$ is the longest prefix of $\Delta$ that omits $U$.
    \item 
    $\Delta_{\preceq U} = \Delta_1 \cdot (U = \Phi)$ is the shortest prefix of $\Delta$ that includes $U$.
    \item 
    $\Delta_{\succ U} = \Delta_2$ is the longest suffix of $\Delta$ that omits $U$.
    \item 
    $\Delta_{\succeq U} = (U = \Phi) \cdot \Delta_2$ is the shortest suffix of $\Delta$ that includes $U$.
\end{enumerate}
\end{definition}

\begin{definition}[Compatibility of definition lists]\label{def:definition-list-compatibility}
Let $\Delta_1$ and $\Delta_2$ be definition lists.  Then $\Delta_2$ is \emph{compatible with} $\Delta_1$ iff $\dom(\Delta_1) \cap \dom(\Delta_2) = \emptyset$ and no $U_2 \in \dom(\Delta_2)$ appears in $\Delta_1(U_1)$ for any $U_1 \in \dom(\Delta_1)$.
\end{definition}

\noindent 
We have the following.

\begin{lemma}[Definition-list concatenation]\label{lem:definition-list-concatenation}
    \begin{enumerate}
        \item
        Let $\Delta$ be a definition list such that $\Delta = \Delta_1 \cdot \Delta_2$.  Then $\Delta_1$ and $\Delta_2$ are definition lists, and $\Delta_2$ is compatible with $\Delta_1$.
        \item
        Suppose that $\Delta_1$ and $\Delta_2$ are definition lists such that $\Delta_2$ is compatible with $\Delta_1$.  Then $\Delta_1 \cdot \Delta_2$ is a definition list.
    \end{enumerate}
\end{lemma}
\begin{proof}
    Immediate from Definitions~\ref{def:definition-list} and~\ref{def:definition-list-compatibility}.\qedhere
\end{proof}

\begin{corollary}\label{cor:prefix-suffix-inheritance}
    Let $\Delta$ be a definition list and $U \in \dom(\Delta)$ be a definitional constant.  Then $\Delta_{\prec U}, \Delta_{\preceq U}$, $\Delta_{\succ U}$ and $\Delta_{\succeq U}$ are also definition lists.
\end{corollary}
\begin{proof}
    Immediate from Lemma~\ref{lem:definition-list-concatenation} and Definition~\ref{def:definition-list-prefix}.\qedhere
\end{proof}

Syntactically, definition lists may be seen as \emph{iterated substitutions}.  The intuition is formalized below.

\begin{definition}[Formula expansion by a definition list]\label{def:formula-expansion}
    Let $\Delta$ be a definition list and $\Phi \in \muformsSV$.  Then $\Phi[\Delta]$, the \emph{expansion} of $\Phi$ with respect to $\Delta$, is defined inductively on $\Delta$ as follows.
    \begin{itemize}
        \item
        $\Phi[\emptyL] = \Phi$
        \item
        $\Phi[\Delta \cdot (U = \Psi)] = \left( \Phi [U := \Psi] \right)[\Delta]$
    \end{itemize}
\end{definition}

\noindent
Note that $\Phi[\Delta]$ contains no occurrences of any elements of $\dom(\Delta)$.

The definition of $\Phi[\Delta]$ recurses from the back of list $\Delta$ to the front.
The following property, which characterizes $\Phi[\Delta]$ for non-empty $\Delta$ in terms of the first definition in $\Delta$, is useful in proofs by induction on $\Delta$.
\begin{lemma}\label{lem:nonempty-formula-expansion}
        Let $\Delta = (U_1 = \Phi_1) \cdot \Delta'$ be a non-empty definition list.  Then for any $\Phi \in \muformsSV$,
        $
            \Phi[\Delta] = \left( \Phi[\Delta'] \right) \,[U_1 := \Phi_1].
        $
\end{lemma}
\remove{
\begin{proofsketch}
    By induction on $\Delta'$. Details may be found in the appendix.\qedhere
\end{proofsketch}
}
\begin{proof}
    Proceeds by induction on $\Delta'$.  There are two cases to consider.
    \begin{itemize}
        \item $\Delta' = \emptyL$.
            Fix $\Phi$.  We have the following.
            \begin{align*}
            \Phi[\Delta]
            &= \Phi[(U_1 = \Phi_1)]
            && \text{$\Delta = (U_1 = \Phi_1) \cdot \Delta', \Delta' = \emptyL$}
            \\
            &= \Phi[\emptyL \cdot (U_1 = \Phi_1)]
            && \text{$\vec{w} = \emptyL \cdot \vec{w}$ for any sequence $\vec{w}$}
            \\
            &= \left( \Phi[\emptyL] \right) [U_1 := \Phi_1]
            && \text{Definition of $\Phi[\emptyL \cdot (U_1 = \Phi_1)]$}
            \\
            &= \left( \Phi[\Delta'] \right) [U_1 := \Phi_1]
            && \text{$\Delta' = \emptyL$}
            \end{align*}
        \item $\Delta' = \Delta'' \cdot (U' = \Phi')$.
            The induction hypothesis guarantees that for any $\Phi$, $\Phi[(U_1 = \Phi_1) \cdot \Delta''] = \left( \Phi[\Delta''] \right)[U_1 := \Phi_1]$.  Now fix $\Phi$.  We reason as follows.
            \begin{align*}
            \Phi[\Delta]
            &= \Phi[(U_1 = \Phi_1) \cdot \Delta']
            && \text{$\Delta = (U_1 = \Phi_1) \cdot \Delta'$}
            \\
            &= \Phi[(U_1 = \Phi_1) \cdot (\Delta'' \cdot (U' = \Phi'))]
            && \text{$\Delta' = \Delta'' \cdot (U' = \Phi')$}
            \\
            &= \Phi[((U_1 = \Phi_1) \cdot \Delta'') \cdot (U' = \Phi'))]
            && \text{Associativity of $\cdot$}
            \\
            &= \left( \Phi[U' := \Phi'] \right) [(U_1 = \Phi_1) \cdot \Delta'']
            && \text{Definition of $\Phi[\cdots (U' = \Phi')]$}
            \\
            &= \left( \left( \Phi[U' := \Phi'] \right)[\Delta''] \right) \,[U_1 := \Phi_1]
            && \text{Induction hypothesis}
            \\
            &= \left( \Phi[\Delta'' \cdot (U' = \Phi'] \right) \,[U_1 := \Phi_1]
            && \text{Definition of $\Phi[\Delta'' \cdot (U' = \Phi')]$}
            \\
            &= \left( \Phi[\Delta'] \right) \,[U_1 := \Phi_1]
            && \text{$\Delta' = \Delta'' \cdot (U' = \Phi')$}
            \end{align*}
            \qedhere
    \end{itemize}
\end{proof}

In a similar vein we define the semantic extension, $\V[\Delta]$, of valuation $\V$ by definition list $\Delta$ as follows.
Essentially, $\V[\Delta]$ updates $\V$ with the semantic interpretation of the equations appearing in $\Delta$. We use this notion later to assign semantics to the sequents labeling the nodes in the proof tree.

\begin{definition}[Valuation extension by a definition list]\label{def:valuation-extension}
    Let $\T = \lts{S}$ be an LTS over $\Sigma$, $\V \in \Var \to 2^\states{S}$ a valuation,
    and $\Delta$ a definition list.  Then $\V[\Delta]$, the \emph{extension} of $\V$ by $\Delta$, is defined inductively on $\Delta$ as follows.
    \begin{enumerate}
        \item
        $\V[\emptyL] = \V$
        \item
        $\V[(U = \Phi) \cdot \Delta] = \bigl(\V \,[\, U := \semTV{\Phi} \;]\,\bigr) \,[\Delta]$
    \end{enumerate}
\end{definition}

In contrast with $\Phi[\Delta]$, the definition of $\V[\Delta]$ recurses from the from front of $\Delta$ to the back.
The next lemma gives a characterization of $\V[\Delta]$ for non-empty $\Delta$ in terms of the last binding contained in $\Delta$.

\begin{lemma}\label{lem:nonempty-valuation-extension}
    Let $\T$ be an LTS over $\Sigma$ and $\Delta = \Delta' \cdot (U = \Phi)$ a non-empty definition list.
    Then for any valuation $\V \in \Var \to 2^\states{S}$, $\V[\Delta] = \left( \V[\Delta'] \right) [U := \semT{\Phi}{\V[\Delta']}]$.
\end{lemma}
\remove{
\begin{proofsketch}
    Fix an arbitrary LTS $\T$ and $\Delta = \Delta' \cdot (U = \Phi)$.
    The proof is by induction on $\Delta'$. Details can be found in the appendix.
\end{proofsketch}
}
\begin{proof}
    Fix an arbitrary LTS $\T$ and $\Delta = \Delta' \cdot (U = \Phi)$.
    We proceed by induction on $\Delta'$.  There are two cases to consider.
    \begin{itemize}
        \item
              $\Delta' = \emptyL$.
              Fix $\V$.  The definition of $\V[\Delta]$ then guarantees the desired result.
        \item
              $\Delta' = (U_1 = \Phi_1) \cdot \Delta''$.
              The induction hypothesis guarantees that for all $\V$,
              $\V[\Delta'' \cdot (U = \Phi)] = \left( \V[\Delta''] \right) [U := \semT{\Phi}{\V[\Delta'']}].$
              We reason as follows.
              \begin{align*}
                  \V[\Delta]
                   & = \V[\Delta' \cdot (U = \Phi)]
                   &                                                                         & \text{$\Delta = \Delta' \cdot (U = \Phi)$}
                  \\
                   & = \V[((U_1 = \Phi_1) \cdot \Delta'') \cdot (U = \Phi)]
                   &                                                                         & \text{$\Delta' = (U_1 = \Phi_1) \cdot \Delta''$}
                  \\
                   & = \V[(U_1 = \Phi_1) \cdot (\Delta'' \cdot (U = \Phi))]
                   &                                                                         & \text{Commutativity of $\cdot$}
                  \\
                   & = \left( \V[ U_1 := \semTV{\Phi_1}] \right) [\Delta'' \cdot (U = \Phi)]
                   &                                                                         & \text{Definition of $\V[(U_1 = \Phi_1) \cdots]$}
                  \\
                   & = \V'[U := \semT{\Phi}{\V'}]
                   &                                                                         & \text{Induction hypothesis;}
                  \\
                   &                                                                         &                                                           & \text{$\V' = \left( \V[\, U_1 := \semTV{\Phi_1}\,] \right) [\Delta'']$}
                  \\
                   & = \V'' [U := \semT{\Phi}{\V''}]
                   &                                                                         & \text{Definition of $\V[(U_1 = \Phi_1) \cdot \Delta'']$;}
                  \\
                   &                                                                         &                                                           & \text{$\V'' = \V[\,(U_1 = \Phi_1) \cdot \Delta''\,]$}
                  \\
                   & = \left( \V[\Delta'] \right)[U := \semT{\Phi}{\V[\Delta']}]
                   &                                                                         & \text{$\Delta' = (U_1 = \Phi_1) \cdot \Delta''$}
              \end{align*}
              \qedhere
    \end{itemize}
\end{proof}

\noindent
In Lemma~\ref{lem:definition-list-correspondence} we establish a correspondence between $\semTV{\Phi[\Delta]}$ and $\semT{\Phi}{\V[\Delta]}$.

\begin{lemma}[Definition-list correspondence]\label{lem:definition-list-correspondence}
    Let $\Phi \in \muformsSV$.  Then for every LTS $\T$ over $\Sigma$, definition list $\Delta$, and valuation $\V$,
    $
    \semTV{\Phi[\Delta]}
    =
    \semT{\Phi}{\V[\Delta]}.
    $
\end{lemma}
\begin{proof}
    Fix $\Phi$ and LTS $\T$ over $\Sigma$.
    The proof proceeds by induction on $\Delta$.  In the base case $\Delta = \emptyL$, and the result is immediate, as $\Phi[\emptyL] = \Phi$ and $\V[\emptyL] = \V$.  Now assume that $\Delta = (U_1 = \Phi_1) \cdot \Delta'$.  The induction hypothesis guarantees that for any valuation $\V$,
    $\semTV{\Phi[\Delta']}
        =
        \semT{\Phi}{\V[\Delta']}.$
    We must prove that for any valuation $\V$,
    $\semTV{\Phi[\Delta]}
        =
        \semT{\Phi}{\V[\Delta]}.$
    So fix $\V$.  We reason as follows.
    \[
        \begin{array}{rcl@{\;\;\;}p{5cm}}
            \semTV{\Phi[\Delta]}
             & =
             & \semTV{\Phi[(U_1 = \Phi_1) \cdot \Delta']}
             & $\Delta = (U_1 = \Phi_1) \cdot \Delta'$
            \\[6pt]
             & =
             & \semTV{(\Phi[\Delta'])[U_1 := \Phi_1]}
             & Lemma~\ref{lem:nonempty-formula-expansion}
            \\[6pt]
             & =
             & \semT{\Phi[\Delta']}{\V[U_1 := \semTV{\Phi_1}]}
             & Lemma~\ref{lem:substitution}
            \\[6pt]
             & =
             & \semT{\Phi}{(\V[U_1 := \semTV{\Phi_1}])\,[\Delta']}
             & Induction hypothesis
            \\[6pt]
             & =
             & \semT{\Phi}{\V[(U_1 = \Phi_1) \cdot \Delta']}
             & Definition of $\V[(U_1 = \Phi_1) \cdot \Delta']$
            \\[6pt]
             & =
             & \semT{\Phi}{\V[\Delta]}
             & $\Delta = (U_1 = \Phi_1) \cdot \Delta'$
        \end{array}
    \]\qedhere
\end{proof}

The next lemma asserts that the semantics of definitional constant $U$ in the definition list $\Delta$ only depends on prefix of $\Delta$ up to and including the definition of $U$; the subsequent definitions in $\Delta$ have no effect.

\begin{lemma}[Definitional-constant semantics]\label{lem:invariance-definitional-constant}
    Let $\T$ be an LTS over $\Sigma$ and $\Delta$ a definition list with $U \in \dom(\Delta)$.
    Then for any valuation $\V \in \Var \to 2^\states{S}$, $\V[\Delta] \left( U \right) = \V[\Delta_{\preceq U}] \left( U \right)$.
\end{lemma}
\remove{
\begin{proofsketch}
    From the definitions of $\Delta_{\preceq U}$ and $\Delta_{\succ U}$ it follows that $\Delta = (\Delta_{\preceq U}) \cdot (\Delta_{\succ U})$.  The proof then proceeds by induction on $\Delta_{\succ U}$, using Lemma~\ref{lem:nonempty-valuation-extension}. The detailed proof can be found in the appendix.
\end{proofsketch}
}
\begin{proof}
    From the definitions of $\Delta_{\preceq U}$ and $\Delta_{\succ U}$ it follows that $\Delta = (\Delta_{\preceq U}) \cdot (\Delta_{\succ U})$.  The proof proceeds by induction on $\Delta_{\succ U}$.  There are two cases to consider.
    \begin{itemize}
        \item
              $\Delta_{\succ U} = \emptyL$.
              In this case $\Delta = \Delta_{\preceq U} = \Delta_{\prec U} \cdot (U = \Phi)$, and the result follows from Lemma~\ref{lem:nonempty-valuation-extension}.
        \item
              $\Delta_{\succ U} = \Delta' \cdot (U' = \Phi')$ for some $U' \not\in \dom(\Delta)$.
              Note that since $\Delta_{\succ U}$ is compatible with $\Delta_{\preceq U}$, so is $\Delta'$, and thus $\Delta_{\preceq U} \cdot \Delta'$ is a definition list.
              The induction hypothesis thus guarantees that
              $\left( \V[\Delta_{\preceq U}] \right) (U) = \left( \V[\Delta_{\preceq U} \cdot \Delta'] \right) (U).$
              We reason as follows.
              \begin{flalign*}
                  & \left( \V[\Delta] \right) (U)
                  &&
                  \\
                  &= \V[ \,\left( \Delta_{\preceq U} \right) \cdot \left( \Delta_{\succ U} \right)\, ] \,(U)
                  && \text{$\Delta = \left( \Delta_{\preceq U} \right) \cdot \left( \Delta_{\succ U} \right)$}
                  \\
                  &= \left( \V[ \,\left( \Delta_{\preceq U} \right) \cdot \Delta' \cdot (U' = \Phi')\, ] \right) \,(U)
                  && \text{$\Delta_{\succ U} = \Delta' \cdot (U' = \Phi')$}
                  \\
                  &= \left( \left(
                      \V [ \,\left( \Delta_{\preceq U} \right) \cdot \Delta'\, ] \right)
                      [U' := \semT{\Phi'}{\V[\Delta']}]\right) \,(U)
                  && \text{Lemma~\ref{lem:nonempty-valuation-extension}}
                  \\
                  &= \left( \V [ \,\left( \Delta_{\preceq U} \right) \cdot \Delta'\, ] \right) \,(U)
                  && \text{$U \neq U'$}
                  \\
                  &= \V[\Delta_{\preceq U}] \,(U)
                  && \text{Induction hypothesis}
              \end{flalign*}
              \qedhere
    \end{itemize}
\end{proof}

\noindent
The next corollaries follow from this lemma.

\begin{corollary}\label{cor:shared-prefix}
    Let $\T$ be an LTS over $\Sigma$, $\Delta$, $\Delta'$ be definition lists, $U \in \dom(\Delta)$ such that $\Delta_{\preceq U} = \Delta'_{\preceq U}$, and $\V$ be a valuation.  Then:
    \begin{enumerate}
        \item\label{cor:shared-prefix-constant-semantics}
              $\left( \V[\Delta] \right) (U) = \left( \V[\Delta'] \right) (U)$.
        \item\label{cor:shared-prefix-constant-formula-semantics}
              $\semT{U}{\V[\Delta]} = \semT{U}{\V[\Delta']}$.
    \end{enumerate}
\end{corollary}

\begin{corollary}\label{cor:prefix-formula-semantics}
    Let $\T$ be an LTS over $\Sigma$ and $\Delta, \Delta_1$ and $\Delta_2$ be definition lists such that $\Delta = \Delta_1 \cdot \Delta_2$, and let $\Phi$ be such that no $U' \in \dom(\Delta_2)$ is free in $\Phi$.  Then for any $\V$, $\semT{\Phi}{\V[\Delta]} = \semT{\Phi}{\V[\Delta_1]}$.
\end{corollary}
\begin{proof}
    By induction on $\Phi$.
\end{proof}

We now show a correspondence between the semantics of a constant $U$ and $\Delta(U)$, where $\Delta$ is a definition list with $U \in \dom(\Delta)$.

\begin{lemma}[$U$ semantic correspondence]\label{lem:U-semantic-correspondence}
    Let $\T$ be and LTS over $\Sigma$, $\Delta$ be a definition list with $U \in \dom(\Delta)$, and $\V$ be a valuation.  Then $\semT{U}{\V[\Delta]} = \semT{\Delta(U)}{\V[\Delta_{\prec U}]} = \semT{\Delta(U)}{\V[\Delta]}$.
\end{lemma}
\begin{proof}
    Fix LTS $\T$, definition list $\Delta$ and $U \in \dom(\Delta)$, and assume $\Delta(U) = \Phi$.
    Let $\V$ be an arbitrary valuation.
    We reason as follows.
    \begin{align*}
        \semT{U}{\V[\Delta]}
         & = \semT{U}{\V[\Delta_{\preceq U}]}
         &                                                                    & \text{Corollary~\ref{cor:shared-prefix}(\ref{cor:shared-prefix-constant-formula-semantics})}
        \\
         & = \semTV{U[\Delta_{\preceq U}]}
         &                                                                    & \text{Lemma~\ref{lem:definition-list-correspondence}}
        \\
         & = \semTV{U[ \,\Delta_{\prec U} \cdot (U = \Phi) \, ]}
         &                                                                    & \Delta_{\preceq U} = \Delta_{\prec U} \cdot (U = \Phi)
        \\
         & = \semTV{\left(U [U := \Phi] \right) \, [ \,\Delta_{\prec U} \, ]}
         &                                                                    & \text{Definition of $U[ \,\Delta_{\prec U} \cdot (U = \Phi) \, ]$}
        \\
         & = \semTV{\Phi [ \,\Delta_{\prec U} \, ]}
         &                                                                    & \text{Definition of substitution}
        \\
         & = \semT{\Phi}{\V[\Delta_{\prec U}]}
         &                                                                    & \text{Lemma~\ref{lem:definition-list-correspondence}}
        \\
         & = \semT{\Delta(U)}{\V[\Delta_{\prec U}]}
         &                                                                    & \Delta(U) = \Phi
        \\
         & = \semT{\Delta(U)}{\V[\Delta]}
         &                                                                    & \text{Corollary~\ref{cor:prefix-formula-semantics}, $\Delta = (\Delta_{\prec U}) \cdot (\Delta_{\succeq U})$,}
        \\
         &
         &                                                                    & \text{no $U' \in \dom(\Delta_{\succeq U})$ free in $\Phi = \Delta(U)$}
    \end{align*}
    \qedhere
\end{proof}

We close our treatment of definition lists by showing the following useful fact about unfolding when $\Delta(U) = \sigma Z.\Phi$ for some $\Phi$.

\begin{lemma}[Definitional-constant unfolding]\label{lem:constant-unfolding}
    Let $\T$ be and LTS over $\Sigma$, $\Delta$ be a definition list with $\Delta(U) = \sigma Z. \Phi$, and $\V$ be a valuation.  Then $\semT{U}{\V[\Delta]} = \semT{\Phi[Z := U]}{\V[\Delta]}$.
\end{lemma}
\begin{proof}
    Fix $\T$, $\Delta$ and $\V$, and let $U$ be such that $\Delta(U) = \sigma Z.\Phi$.  We reason as follows.
    \begin{align*}
    \semT{U}{\V[\Delta]}
    &= \semT{\sigma Z. \Phi}{\V[\Delta_{\prec U}]}
    && \text{Lemma~\ref{lem:U-semantic-correspondence}, $\Delta(U) = \sigma Z. \Phi$}
    \\
    &= \semT{\Phi [Z := \sigma Z.\Phi]}{\V[\Delta_{\prec U}]}
    && \text{Lemma~\ref{lem:unfolding}}
    \\
    &= \semT{\left( \Phi [Z := U] \right) [U := \sigma Z. \Phi]}{\V[\Delta_{\prec U}]}\
    \\
    \multispan4{\hfil Property of substitution; $U$ not free in $\Phi$}
    \\
    &= \semT{\Phi [Z := U]}{\left( \V[\Delta_{\prec U}] \right)[U := \semT{\sigma Z. \Phi}{\Delta_{\prec U}}]}
    && \text{Lemma~\ref{lem:substitution}}
    \\
    &= \semT{\Phi [Z := U]}{\V[\Delta_{\prec U} \cdot (U = \sigma Z. \Phi)]}
    && \text{Lemma~\ref{lem:nonempty-valuation-extension}}
    \\
    &= \semT{\Phi [Z := U]}{\V[\Delta_{\preceq U}]}
    && \text{$\Delta_{\preceq U} = \Delta_{\prec U} \cdot (U = \sigma Z.\Phi)$}
    \\
    &= \semT{\Phi [Z:= U]}{\V[\Delta]}
    && \text{Corollary~\ref{cor:prefix-formula-semantics},}
    \\
    \multispan4{\hfil
       $\Delta = \left( \Delta_{\preceq U} \right) \cdot \left( \Delta_{\succ U} \right)$,
       no $U' \in \dom(\Delta_{\succ U})$ free in $\Phi [Z:= U]$}
    \end{align*}
    \qedhere
\end{proof}

\paragraph{Sequents.}

We can now define the sequents used in the base proof system.

\begin{definition}[Sequents]\label{def:sequent}
Let $\T = \lts{S}$ be an LTS over $\Sigma$ and $\V \in \Var \to 2^\states{S}$ a valuation over $\Var$.
Then a \emph{sequent over $\T$ and $\V$} has form $$S \tnxTVD \Phi,$$ where $S \subseteq \states{S}$, $\Delta$ is a definition list, and $\Phi \in \muformsSV$ has the following properties.
\begin{enumerate}
    \item $\Phi$ is in positive normal form.
    \item Every $U \in \dom(\Delta)$ is positive in $\Phi$.
    \item No $U \in \dom(\Delta)$ is bound in $\Phi$.  
\end{enumerate}
We use $\SeqTV$ for the set of sequents over $\T$ and $\V$.
If $\seq{s} = S \tnxTV{\Delta} \Phi$ is in $\SeqTV$, then we access the components of $\seq{s}$ as follows:
        $\seqst(\seq{s}) = S$,
        $\seqdl(\seq{s}) = \Delta$, and
        $\seqfm(\seq{s}) = \Phi$.
\end{definition}

The intended interpretation of sequent $S \tnxTVD \Phi$ is that it is valid iff every $s \in S$ satisfies $\Phi$, where $\V$ is used to interpret free propositional variables $Z \not\in \dom(\Delta)$ that appear in $\Phi$ and $\Delta$ is used for definitional constants $U \in \dom(\Delta)$ that occur in $\Phi$.  
This notion is formalized as follows.

\begin{definition}[Sequent semantics and validity]\label{def:sequent-semantics}
    Let $\seq{s} = S \tnxTVD \Phi$ be a sequent in $\SeqTV$.
    \begin{enumerate}
    \item
        The \emph{semantics} of $\seq{s}$ is defined as $\semop{\seq{s}} = \semT{\Phi}{\V[\Delta]}$.
    \item
        We define $\seq{s}$ to be \emph{valid} iff $S \subseteq \semop{\seq{s}}$.
    \end{enumerate}
\end{definition}

Sequents of form $S \tnxTVD U$, where $U \in \dom(\Delta)$, play a prominent role in the rest of the paper. The following corollary about such sequents follows directly from our earlier results on definition lists and the semantics of sequents.

\begin{corollary}\label{cor:shared-prefix-sequent-semantics}
    Let $\T$ be an LTS over $\Sigma$, $\Delta$, $\Delta'$ be definition lists, $U \in \dom(\Delta)$ such that $\Delta_{\preceq U} = \Delta'_{\preceq U}$, $\V$ be a valuation, and $\seq{s} = S \tnxTVD U$ and $\seq{s}' = S' \tnxTV{\Delta'} U$ both be sequents.  Then $\semop{\seq{s}} = \semop{\seq{s}'}$.
\end{corollary}

\subsection{Proof rules and tableaux}

We now present the proof rules used in this paper for the mu-calculus formula sequents described in the previous section.  These proof rules come from~\cite{BS1992,Bra1991} and resemble traditional natural-deduction-style proof rules.  Following~\cite{BS1992,Bra1991}, however, we write the conclusion of the proof rule above the premises to emphasize that the conclusion is a ``goal'' and the premises are ``subgoals'' in a goal-directed proof-search strategy.  In what follows we fix sort $\Sigma$, transition system $\T = \lts{S}$ over $\Sigma$, variable set $\Var$ and valuation $\V \in \Var \rightarrow 2^{\states{S}}$.

\begin{definition}[Proof rule]\label{def:proof-rule}
    A \emph{proof rule} has form
    \[
        \proofrule[ \mathit{name}\;\; ]
        {\seq{s}}
        {\seq{s}_1 \;\;\cdots\;\; \seq{s}_n}
        \;\mathit{side\; condition}
    \]
    where $\mathit{name}$ is the name of the rule, $\seq{s}_1 \cdots \seq{s}_n \in (\SeqTV)^*$ 
    is a sequence of sequents called the \emph{premises} of the rule, 
    $\seq{s} \in \SeqTV$ is the \emph{conclusion} of the rule, 
    and the optional $\mathit{side\; condition}$ is a property determining when the rule may be applied.
\end{definition}

Figure~\ref{fig:proof-rules} gives the proof rules from~\cite{Bra1991}, lightly adapted to conform to our notational conventions.  The rules are named for the top-level operators appearing in the formulas of the sequents to which they apply.
Rule $\sigma Z$ is actually short-hand for two rules -- one for each possible value, $\mu$ and $\nu$, of $\sigma$ -- and asserts that a definition involving fresh constant $U$ is added to the end of the definition list of the subgoal.   The Thin rule is needed to ensure completeness of the proof system, and is not required for soundness.  

\begin{figure}
    \[
        \begin{array}{c}
            \proofrule[\land\;\;]
            {S \tnxTVD \Phi_1 \land \Phi_2}
            {S \tnxTVD \Phi_1 \quad S \tnxTVD \Phi_2}

            \qquad
            \proofrule[\lor\;\;]
            {S \tnxTVD \Phi_1 \lor \Phi_2}
            {S_1 \tnxTVD \Phi_1 \quad S_2 \tnxTVD \Phi_2}
            \;S = S_1 \cup S_2
            \\[18pt]

            \proofrule[{[K]}\;\;]
            {S \tnxTVD [K]\Phi}
            {S' \tnxTVD \Phi}
            \;S' = \{s' \in \states{S} \mid \exists s \in S. s \xrightarrow{K} s' \}
            \\[18pt]

            \proofrule[\dia{K}\;\;]
            {S \tnxTVD \langle K \rangle \Phi}
            {f(S) \tnxTVD \Phi}
            \;f \in S \rightarrow \Bstates{S}\ \text{is such that}\ \forall s \in S \colon s \xrightarrow{K} f(s)
            \\[18pt]

            \proofrule[\sigma Z\;\;]
            {S \tnxTVD \sigma Z.\Phi}
            {S \tnxTV{\Delta\cdot(U = \sigma Z.\Phi)} U}
            \;U \ \text{fresh\footnotemark}

            \qquad
            \proofrule[\text{Un}\;\;]
            {S \tnxTVD U}
            {S \tnxTVD \Phi[Z := U]}
            \;\Delta(U) = \sigma Z.\Phi
            \\[18pt]

            \proofrule[\text{Thin}\;\;]
            {S \tnxTVD \Phi}
            {S' \tnxTVD \Phi'}
            \;S \subseteq S'
        \end{array}
    \]
    \caption{Proof rules for mu-formula sequents.}\label{fig:proof-rules}
\end{figure}

The side condition of Rule $\dia{K}$ refers to a function $f$ that is responsible for selecting, for each state $s$ in the goal sequent, a \emph{witness} state $s'$ such that $s \xrightarrow{K} s'$ and $s'$ is in the subgoal sequent.  To apply this rule, the side condition requires the identification of a specific function $f$ that computes these witness states.  We call the function associated in this manner with an application of $\dia{K}$ the \emph{witness function} of the application of the rule.  We further define the set of \emph{rule applications} for a given LTS $\lts{S}$ over $\Sigma$ as follows.
\[
\RuleAppl = 
\{\land, \lor, [K], \mu Z, \nu Z, \textnormal{Un}, \textnormal{Thin}\}
\cup
\{(\dia{K},f) \mid \text{$f$ is a witness function}\}
\]
Note that rule applications are either rule names, if the rule being applied is not $\dia{K}$, or pairs of form $(\dia{K}, f)$ where $f$ is the witness function used in applying the rule.  If $a$ is a rule application then we write $\rn(a)$ for the rule name used in $a$.  Formally, $\rn(a)$ is defined as follows.
\[
\rn(a) =
\begin{cases}
\dia{K} & \text{if $a = (\dia{K},f)$}\\
a       & \text{otherwise}
\end{cases}
\]

Intuitively, proofs are constructed as follows.
\footnotetext{Here $U$ is fresh if has not been previously used anywhere in a proof currently being constructed using these proof rules.} 
Suppose $S \tnxTV{\emptyL} \Phi$ is a sequent we wish to prove.  Based on the form of $\Phi$ we select a rule whose conclusion matches this sequent and apply it, generating the corresponding premises and witness function as required by the rule.  We then recursively build proofs for these premises.  The proof construction process terminates when the validity (or lack thereof) of the current sequent can be immediately established.  The resulting proofs can be viewed as trees, which are called \emph{tableaux} in~\cite{BS1992,Bra1991}, and which we formalize as follows.

\begin{definition}[(Partial) tableaux]\label{def:tableau}
    \begin{enumerate}
    \item\label{subdef:partial-tableau}
        A \emph{partial tableau} has form $\mathbb{T} = \tableauTrl$, where:
        \begin{itemize}
        \item
            $\tree{T} = (\node{N}, \node{r}, p, cs) $ is a finite non-empty ordered tree (cf.\ Definition~\ref{def:ordered-tree}).
        \item
            Partial function $\rho \in \node{N} \pto \RuleAppl$ satisfies:  $\rho(\node{n}) \div$ iff $\node{n}$ is a leaf of $\tree{T}$.
        \item
            $\T = \lts{S}$ is a labeled transition system.
        \item
            $\V \in \Var \to 2^{\states{S}}$ is a valuation.
        \item
            Function $\lambda \in \node{N} \rightarrow \SeqTV$, the \emph{sequent labeling}, satisfies:  if $\rho(\node{n}) \in \RuleAppl$ then
            $\lambda(\node{n})$ and $\lambda(cs(\node{n}))$\footnote{Recall that $cs(\node{n})$ is the sequence of the children of node $\node{n}$ in left-to-right order.} satisfy the form and side condition associated with $\rho(\node{n})$.
        \end{itemize}
        Elements of $\node{N}$ are sometimes called the \emph{proof nodes} of $\mathbb{T}$.

    \item\label{subdef:complete-tableau}
        Partial tableau $\tableauTrl$ is a \emph{complete tableau}, or simply a \emph{tableau},
        iff 
        $\seqdl(\lambda(\node{r})) = \emptyL$ (i.e.\/ the definition list in the root is empty), and
        all leaves $\node{n}$ in $\tree{T}$ are \emph{terminal}, i.e.\/ satisfy one of the following.
        \begin{enumerate}[ref=\theenumi(\alph*)]
        \item \label{subdef:free-leaf}
            $\seqfm(\lambda(\node{n})) = Z$ or $\seqfm(\lambda(\node{n})) = \lnot Z$ for some $Z \in \Var \setminus \dom(\seqdl(\lambda(\node{n})))$
            (in this case $\node{n}$ is called a \emph{free leaf}); or
        \item \label{subdef:diamond-leaf}
            $\seqfm(\lambda(\node{n})) = \dia{K}\ldots$ and there is $s \in \seqst(\lambda(\node{n}))$ such that $s \centernot{\xrightarrow{K}}$
            (in this case $\node{n}$ is called a \emph{diamond leaf}); or
        \item \label{subdef:companion-leaf}
            $\seqfm(\lambda(\node{n})) = U$ for some $U \in \dom(\seqdl(\lambda(\node{n})))$, and there is $\node{m} \in A_s(\node{n})$\footnote{Recall that $A_s(\node{n})$ is the set of strict ancestors of $\node{n}$.} such that
            $\seqfm(\lambda(\node{m})) = U$ and $\seqst(\lambda(\node{n})) \subseteq \seqst(\lambda(\node{m}))$ 
            (in this case $\node{n}$ is called a \emph{$\sigma$-leaf}).%
        \end{enumerate}
        We may also refer to a $\sigma$-leaf as a $\mu$- / $\nu$-leaf if $\Delta(U) = \mu\ldots$ / $\Delta(U) = \nu\ldots$.  
        The deepest node $\node{m}$ making $\node{n}$ a $\sigma$-leaf is the \emph{companion node} of $\node{n}$, with $\node{n}$ then being a \emph{companion leaf} of $\node{m}$.
    \end{enumerate}
\end{definition}

A $\sigma$-leaf in a tableau has a unique companion node based on the Definition~\ref{def:tableau}(\ref{subdef:complete-tableau}), but a node may be a companion node for multiple (or no) companion leaves; all such companion leaves must be in the subtree rooted at the node, however.
The definition of terminal leaf in~\cite{BS1992,Bra1991} also includes an extra case, $\seqst(\lambda(\node{n})) = \emptyset$, in addition to the Conditions \ref{subdef:free-leaf}--\ref{subdef:companion-leaf} in our definition above.  It turns out that this case is unnecessary, as the completeness results in Section~\ref{sec:Completeness} show.

In what follows we adopt the following notational shorthands for proof nodes in partial tableaux.
\begin{notation}[Proof nodes]
Let $\node{n}$ be a proof node in partial tableau $\tableauTrl$.
\begin{itemize}
    \item $\node{n} = S \tnxTVD \Phi$ means $\lambda(\node{n}) = S \tnxTVD \Phi$.
    \item $\semop{\node{n}} = \semop{\lambda(\node{n})} = \semTV{\seqfm(\lambda(\node{n}))}$.
    \item $\node{n}$ is valid iff $\lambda(\node{n})$ is valid.
    \item $\seqst(\node{n}) = \seqst(\lambda(\node{n}))$.
    \item $\seqdl(\node{n}) = \seqdl(\lambda(\node{n}))$.
    \item $\seqfm(\node{n}) = \seqfm(\lambda(\node{n}))$.
\end{itemize}
\end{notation}

Companion nodes and leaves feature prominently later, so we introduce the following notation and remark on a semantic property that they satisfy.

\begin{notation}[Companion nodes and leaves]
    Let $\mathbb{T} = \tableauTrl$ be a partial tableau, with $\tree{T} = (\node{N}, \node{r}, p, cs)$.
    \begin{enumerate}
    \item 
        The set $\cnodesT \subseteq \node{N}$ of companion nodes of $\mathbb{T}$ is given by:
        $$
        \cnodes{\mathbb{T}} = \{ \node{n} \in \node{N} \mid \rho(\node{n}) = \textnormal{Un}\}.
        $$
    \item
        Let $\node{n} \in \cnodes{\mathbb{T}}$.  Then the set $\cleavesT(\node{n}) \subseteq D_s(\node{n})$ of companion leaves of $\node{n}$ is given as follows, where $\cnodes{\mathbb{T},\node{n}} = \cnodesT \cap D_s(\node{n})$ are the companion nodes of $\mathbb{T}$ that are strict descendants of $\node{n}$.
        \begin{align*}
        &\cleavesT(\node{n}) =
        \\
        &\{ \node{n}' \in D_s (\node{n})
           \mid
           c(\node{n}') = \emptyset
           \land
           \seqfm(\node{n}') = \seqfm(\node{n})
           \land
           \seqst(\node{n}') \subseteq \seqst(\node{n})
        \}
        \setminus
        \bigcup_{\node{n}' \in \cnodes{\mathbb{T},\node{n}}} \cleaves{\mathbb{T}}(\node{n}')
        \end{align*}
    \end{enumerate}
\end{notation}
Note that if $\node{n}$ is a companion node then it must be the case that $\seqfm(\node{n}) = U$ for some definitional constant $U$ defined in the definition list $\seqdl(\node{n})$ of $\node{n}$.  Given a companion node $\node{n}$ in $\mathbb{T}$ its associated companion leaves, $\cleavesT(\node{n})$, consist of nodes $\node{n}'$ that:
\begin{enumerate*}[label=(\roman*)]
    \item are leaves ($c(\node{n}') = \emptyset$);
    \item have the same definitional constant as $\node{n}$ for their formula
        ($\seqfm(\node{n}') = \seqfm(\node{n})$);
    \item have their state sets included in the state set of $\node{n}$ 
        ($\seqst(\node{n}') \subseteq \seqst(\node{n})$); and
    \item are not companion leaves of any companion node that is a strict descendant of $\node{n}$ 
        ($\bigcup_{\node{n}' \in \cnodes{\mathbb{T},\node{n}}} \cleavesT(\node{n})$).
\end{enumerate*}

\begin{lemma}[Semantics of companion nodes, leaves]\label{lem:semantic-invariance}
    Let $\mathbb{T} = \tableauTrl$ be a partial tableau.  Then for any proof nodes $\node{n}$ and $\node{n}'$ in $\mathbb{T}$ and definitional constant $U$ such that $\seqfm(\node{n}) = \seqfm(\node{n}') = U$, $\semop{\node{n}} = \semop{\node{n}'}$.
\end{lemma}
\begin{proof}
    Follows from the fact that the side condition of Rule Un, which is the only rule that modifies definition lists, introduces a given definitional constant at most once in a partial tableau, and that the associated definition is added at the end of the definition list in the premise of the rule.  Based on these observations, one can show that $\seqdl(\node{n})_{\preceq U} = \seqdl(\node{n}')_{\preceq U}$.  Lemma~\ref{lem:invariance-definitional-constant} and the definition of $\semop{\node{n}}$ then give the desired result.
    \qedhere
\end{proof}

A tableau is a candidate proof; while the structure of a tableau ensures the correct application of proof rules, this alone does not guarantee that a tableau is a valid proof.  The notion of \emph{successful tableau} fills this gap.

\begin{definition}[Successful leaf / tableau]\label{def:successful-tableau}
    Let $\mathbb{T} = \tableauTrl$ be a tableau.  Then leaf node $\node{n}$ in $\mathbb{T}$ is \emph{successful} iff one of the following hold.
    \begin{enumerate}
        \item\label{def:successful-leaf-Z}
              $\seqfm(\node{n}) = Z$, where $Z \in \Var \setminus \textnormal{dom}(\Delta)$, and $S \subseteq \V(Z)$; or
        \item\label{def:successful-leaf-notZ}
              $\seqfm(\node{n}) = \lnot Z$, where $Z \in \Var \setminus \textnormal{dom}(\Delta)$, and $S \cap \V(Z) = \emptyset$; or
        \item\label{def:successful-leaf-nu}
              $\node{n}$ is a $\nu$-leaf; or
        \item \label{def:successful-tableau-mu}
              $\node{n}$ is a $\mu$-leaf satisfying the condition given below in Definition~\ref{def:successful-mu-leaf}.
    \end{enumerate}
    A tableau is successful iff all its leaves are successful.
    A tableau is \emph{partially successful} iff all its non-$\mu$-leaves are successful.
\end{definition}
If a tableau is partially successful, then establishing that it is (fully) successful only requires showing that each of its $\mu$-leaves is successful, as defined below.  Similarly, a tableau that is not partially successful cannot be successful, regardless of the status of its $\mu$-leaves. In particular, no  diamond leaf is successful, so any tableau containing a  diamond leaf can only be unsuccessful.

The rest of this section is devoted to defining the success of $\mu$-leaves, which is more complicated than the other cases and spans several definitions.  Intuitively, the success of a $\mu$-leaf depends on the well-foundedness of an extended \emph{dependency ordering} involving the companion node of the leaf. We define the dependency ordering in three stages.
The definition of this (extended) dependency ordering essentially corresponds to Bradfield and Stirling's definition of (extended) paths from~\cite{BS1992,Bra1991}; our presentation is intended to simplify proofs. Details of the correspondence can be found in the appendix.

In what follows, fix partial tableau $\mathbb{T}= \tableauTrl$, with $\tree{T} = (\node{N}, r, p, cs)$.
The \emph{local dependency ordering} captures the one-step dependencies between a state in a node in $\mathbb{T}$ and states in the node's children.

\begin{definition}[Local dependency ordering]\label{def:local_dependency_ordering}
    Let $\node{n}, \node{n}' \in \node{N}$ be proof nodes in $\mathbb{T}$, with $\node{n}' \in c(\node{n})$ a child of $\node{n}$. Then $s' <_{\node{n}',\node{n}} s$ iff $s' \in \seqst(\node{n}')$, $s \in \seqst(\node{n})$, and one of the following hold:
    \begin{enumerate}
        \item
              $\rho(\node{n}) = [K]$ and $s \xrightarrow{K} s'$; or
        \item
              $\rho(\node{n}) = (\dia{K},f)$ and $s' = f(s)$; or
        \item
              $\rn(\rho(\node{n})) \not\in \{ [K], \dia{K} \}$ and $s = s'$.
    \end{enumerate}
    Note that since $\node{n}' \in c(\node{n})$, $\node{n}$ is internal and $\rho(\node{n})$ is defined.
\end{definition}

In the second part of the definition of our dependency ordering we extend the local dependency ordering to capture dependencies between states in a node and states in descendants of the node.

\begin{definition}[Dependency ordering]\label{def:dependency_ordering}
    Let $\node{n}, \node{n}'$ be proof nodes in $\mathbb{T}$. Then $s' \lessdot_{\node{n}',\node{n}} s$ iff $s' \in \seqst(\node{n}')$, $s \in \seqst(\node{n})$, and one of the following holds.
    \begin{enumerate}
        \item
        $\node{n} = \node{n}'$ and $s = s'$, or
        \item
        There exist proof node $\node{m}$ and $s'' \in \seqst(\node{m})$ with $s' \lessdot_{\node{n}',\node{m}} s''$ and $s'' <_{\node{m},\node{n}} s$.
    \end{enumerate}
\end{definition}

\noindent
The definition of $\lessdot_{\node{n}',\node{n}}$ is inductive, and may be seen as analogous to the transitive and reflexive closure of $<_{\node{n}',\node{n}}$, modulo the node indices $\node{n}'$ and $\node{n}$ decorating $\lessdot_{\node{n}', \node{n}}$.  It is easy to see that if $s' <_{\node{n}',\node{n}} s$ then $s' \lessdot_{\node{n}',\node{n}} s$, and that if $s' \lessdot_{\node{n}', \node{n}} s$ then $\node{n}' \in D(\node{n})$.\footnote{Recall that $D(\node{n})$ are descendants of $\node{n}$; see Section~\ref{subsec:trees}.}

In the third and final of our definitions, we allow \emph{cycling} through states in companion nodes that are descendants of a given node.

\begin{definition}[Extended dependency ordering]\label{def:extended_path_ordering}
    The \emph{extended dependency ordering}, $<:_{\node{n}',\node{n}}$, and the \emph{companion-node ordering} $<:_{\node{m}}$, are defined mutually recursively as follows.
    \begin{enumerate}
    \item
        Let $\node{m} \in \cnodesT$ be a companion node, with $s, s' \in \seqst(\node{m})$.  Then $s' <:_{\node{m}} s$ iff there is a companion leaf $\node{m}' \in \cleavesT(\node{m})$ with $s' \in \seqst(\node{m}')$ and $s' <:_{\node{m}',\node{m}} s$.
    \item
        Let $\node{n}, \node{n}'$ be proof nodes in $\mathbb{T}$.  Then $s' <:_{\node{n}',\node{n}} s$ iff $s' \in \seqst(\node{n}')$, $s \in \seqst(\node{n})$ and one of the following holds.
            \begin{enumerate}
            \item \label{def:extended_path_ordering-base}
                $s' \lessdot_{\node{n}',\node{n}} s$; or
            \item \label{def:extended_path_ordering-step}
                there exists $\node{m} \in \cnodesT$,
                with $\node{m} \neq \node{n}$ and $\node{m} \neq \node{n}'$,
                and $t, t' \in \seqst(\node{m})$, such that:
                    $s' <:_{\node{n}',\node{m}} t'$,
                    $t' <:_{\node{m}}^{+} t$\footnote{Recall that if $R$ is a binary relation then $R^+$ is the transitive closure of $R$.}, and
                    $t \lessdot_{\node{m},\node{n}} s$.
            \end{enumerate}
    \end{enumerate}
\end{definition}

Intuitively, $s' <:_{\node{n}',\node{n}} s$ is intended to reflect a semantic dependency between state $s'$ in $\node{n}'$ and state $s$ in $\node{n}$.
The first case indicates that this relation holds if applying a sequence of proof rules starting at $\node{n}$ leads to $\node{n}'$, with the rules $[K]$ and $\dia{K}$ inducing state changes in the dependency relation as the rules are applied.  (Of course, the intuition involves soundness assumptions about the proof rules; these are proved later.)  In the second case, the dependency chain can cycle through a companion node if there is a dependency chain from the companion node to one of its companion leaves (this is captured in the relation $<:_{\node{m}}$).  Recall that the states in a companion leaf are also in the leaf's companion node; consequently, if there is a dependency involving a state in the companion node and another state in one of the node's companion leaves, the dependency can be extended with a dependency involving the second state, but starting from the companion node.  This explains the appearance of the transitive-closure operation in this case.

The success criterion for $\mu$-leaves (which is really a condition on the companion nodes for these leaves) in a complete tableau can now be given as follows.

\begin{definition}[Successful $\mu$-leaf]\label{def:successful-mu-leaf}
    Let $\node{n}'$ be a $\mu$-leaf in tableau $\mathbb{T}$ and let $\node{n}$ be the companion node of $\node{n}'$.  Then $\node{n}'$ is successful if and only if $<:_{\node{n}}$ is well-founded.
\end{definition}

Note that this definition implies that either all $\mu$-leaves having the same companion node are successful, or none are.

The remainder of this section is devoted to proving \emph{pseudo-transitivity} properties of the dependency relations $\lessdot_{\node{n}', \node{n}}$ and $<:_{\node{n}', \node{n}}$.  Generally speaking, we would not expect these to be transitive, because e.g.\/ when $s' \lessdot_{\node{n}',\node{n}} s$ holds, $s'$ and $s$ may belong to different sets (namely, the sets of states in their respective proof nodes).  However, if we allow the node labels on the relations to align properly, we do have a property that resembles transitivity.

\begin{lemma}[Pseudo-transitivity of $\lessdot_{\node{n}',\node{n}}$]\label{lem:pseudo-transitivity-of-dependency-ordering}
    Let $\node{n}_1, \node{n}_2$ and $\node{n}_3$ be proof nodes in partial tableau $\mathbb{T}$, and assume $s_1, s_2$ and $s_3$ are such that $s_3 \lessdot_{\node{n}_3, \node{n}_2} s_2$ and $s_2 \lessdot_{\node{n}_2, \node{n}_1} s_1$.
    Then $s_3 \lessdot_{\node{n}_3,\node{n}_1} s_1$.
\end{lemma}

\remove{
\begin{proofsketch}
    Assume $s_3 \lessdot_{\node{n}_3, \node{n}_2} s_2$. The result then follows by induction on the definition of $s_2 \lessdot_{\node{n}_2,\node{n}_1} s_1$. The detailed proof is included in the appendix. \qedhere
\end{proofsketch}
}
\begin{proof}
    Assume that $s_3 \lessdot_{\node{n}_3, \node{n}_2} s_2$.
    The proof proceeds by induction on the definition of $s_2 \lessdot_{\node{n}_2,\node{n}_1} s_1$.  There are two cases consider.
    \begin{itemize}
        \item
        $\node{n}_1 = \node{n}_2$ and $s_1 = s_2$.
        In this case it immediately follows that $s_3 \lessdot_{\node{n}_3, \node{n}_1} s_1$.
        \item
        There exists $\node{m}$ and $s' \in \seqst(\node{m})$ such that $s_2 \lessdot_{\node{n}_2, \node{m}} s'$ and $s' <_{\node{m},\node{n}_1} s_1$.
        In this case the induction hypothesis guarantees that $s_3 \lessdot_{\node{n}_3,\node{m}} s'$.  Since $s' <_{\node{m},\node{n}_1} s_1$, Definition~\ref{def:dependency_ordering} guarantees that $s_3 \lessdot_{\node{n}_3, \node{n}_1} s_1$.
        \qedhere
    \end{itemize}
\end{proof}

We establish a similar pseudo-transitivity property for $<:_{\node{n}', \node{n}}$.
First, the following lemma considers a restricted case of the pseudo-transitivity property for $<:_{\node{n}', \node{n}}$ first, in which $s_2$ and $s_1$ are related by $\lessdot_{\node{n}_2, \node{n}_1}$ rather than $<:_{\node{n}_2, \node{n}_1}$.

\begin{lemma}\label{lem:front_extend_extended_path_ordering}
    Let $\node{n}_1, \node{n}_2$ and $\node{n}_3$ be proof nodes in partial tableau $\mathbb{T}$, such that $s_3 <:_{\node{n}_3, \node{n}_2} s_2$ and $s_2 \lessdot_{\node{n}_2, \node{n}_1} s_1$.
    Then $s_3 <:_{\node{n}_3,\node{n}_1} s_1$.
\end{lemma}
\begin{proof}
    Assume that $s_2 \lessdot_{\node{n}_2, \node{n}_1} s_1$.  The proof proceeds by induction on the definition of $s_3 <:_{\node{n}_3, \node{n}_2} s_2$.  There are two cases to consider.
    \begin{itemize}
        \item
        $s_3 \lessdot_{\node{n}_3,\node{n}_2} s_2$.
        From the pseudo-transitivity of $\lessdot_{\node{n}',\node{n}}$ (Lemma~\ref{lem:pseudo-transitivity-of-dependency-ordering}) it follows that $s_3 \lessdot_{\node{n}_3, \node{n}_1} s_1$,
        and thus
        $s_3 <:_{\node{n}_3,\node{n}_1} s_1$.
        \item
        There exists companion node $\node{m}$,
        with
        $\node{m} \neq \node{n}_2$ and $\node{m} \neq \node{n}_3$,
        and $t, t' \in \seqst(\node{m})$
        such that
        $s_3 <:_{\node{n}_3, \node{m}}t'$,
        $t' <:^+_{\node{m}} t$ and
        $t \lessdot_{\node{m}, \node{n}_2} s_2$.
        Lemma~\ref{lem:pseudo-transitivity-of-dependency-ordering} ensures that $t \lessdot_{\node{m}, \node{n}_1} s_1$,
        and the definition of $<:_{\node{n}_3, \node{n}_1}$ confirms $s_3 <:_{\node{n}_3, \node{n}_1} s_1$.
        \qedhere
    \end{itemize}
\end{proof}
\noindent
We can now prove the pseudo-transitivity of $<:_{\node{n}, \node{n}}$.

\begin{lemma}[Pseudo-transitivity of $<:_{\node{n}', \node{n}}$]\label{lem:join_extended_path_ordering}
Let $\node{n}_1, \node{n}_2$ and $\node{n}_3$ be proof nodes in partial tableau $\mathbb{T}$, and assume $s_3 <:_{\node{n}_3, \node{n}_2} s_2$ and $s_2 <:_{\node{n}_2, \node{n}_1} s_1$.
Then $s_3 <:_{\node{n}_3,\node{n}_1} s_1$.
\end{lemma}
\remove{
\begin{proofsketch}
    Assume that $s_3 <:_{\node{n}_3,\node{n}_2} s_2$.  The result then follows by induction on the definition of $s_2 <:_{\node{n}_2,\node{n}_1} s_1$. The full proof is included in the appendix.\qedhere
\end{proofsketch}
}
\begin{proof}
    Assume that $s_3 <:_{\node{n}_3,\node{n}_2} s_2$.  The proof proceeds by induction on the definition of $s_2 <:_{\node{n}_2,\node{n}_1} s_1$.  There are two cases to consider.
    \begin{itemize}
        \item
        $s_2 \lessdot_{\node{n}_2, \node{n}_1} s_1$.
        The result follows immediately from the pseudo-transitivity of $\lessdot_{\node{n}',\node{n}}$ (Lemma~\ref{lem:front_extend_extended_path_ordering}).
        \item
        There exists companion node $\node{m}$ in $\mathbb{T}$,
        with
        $\node{m} \neq \node{n}_1$ and $\node{m} \neq \node{n}_2$,
        and $t, t' \in \seqst(\node{m})$,
        such that
        $s_2 <:_{\node{n}_2, \node{m}} t'$,
        $t' <:^+_{\node{m}} t$ and
        $t \lessdot_{\node{m}, \node{n}_1} s_1$.
        Since $s_2 <:_{\node{n}_2, \node{m}} t'$, the induction hypothesis guarantees that $s_3 <:_{\node{n}_3, \node{m}} t'$, and
        the definition of $<:_{\node{n}_3, \node{n}_1}$
        then establishes that $s_3 <:_{\node{n}_3, \node{n}_1} s_1$.
        \qedhere
    \end{itemize}
\end{proof}

We now formalize a semantic property, which we call \emph{semantic sufficiency}, enjoyed by the local dependency relation $<_{\node{n}', \node{n}}$ for internal node $\node{n}$.  This property asserts that if for every $s \in \seqst(\node{n})$, and every $s'$ such that that $s' <_{\node{n}', \node{n}} s$, $s'$ belongs to the semantics of $\node{n}'$, then this is sufficient to conclude that $s$ belongs to the semantics of $\node{n}$.  

\begin{lemma}[Semantic sufficiency of $<_{\node{n}', \node{n}}$] \label{lem:semantic-sufficiency-of-<}
    Let $\node{n}$ be an internal proof node in partial tableau $\mathbb{T}$, and let $s \in \seqst(\node{n})$ be such that for all $s'$ and $\node{n}'$ with $s' <_{\node{n}', \node{n}} s$, $s' \in \semop{\node{n}'}$.  Then $s \in \semop{\node{n}}$.
\end{lemma}

\remove{
\begin{proofsketch}
Let $\node{n}$ be an internal node in $\mathbb{T} = \tableauTrl$.  Since $\node{n}$ is internal $\rho(\node{n})$ is defined.  Now assume $s \in \seqst(\node{n})$; the proof proceeds by a case analysis on $\rho(\node{n})$.  A full proof can be found in the appendix. \qedhere
\end{proofsketch}
}
\begin{proof}
    Let $\node{n} = S \tnxTVD \Phi$ be an internal node in $\mathbb{T} = \tableauTrl$.  Since $\node{n}$ is internal, $\rho(\node{n})$ is defined.  Now fix $s \in S$; we proceed by a case analysis on $\rho(\node{n})$.
    \begin{description}
        \item[$\rho(\node{n}) = \land$.]

              In this case $\Phi = \Phi_1 \land \Phi_2$, and $cs(\node{n}) = \node{n}_1\node{n}_2$, where $\node{n}_1 = S \tnxTVD \Phi_1$ and $\node{n}_2 = S \tnxTVD \Phi_2$.  By definition $s' <_{\node{n}', \node{n}} s$ iff $s' = s$ and $\node{n}' = \node{n}_1$ or $\node{n}' = \node{n}_2$.  We reason as follows.
              \begin{flalign*}
                  &\text{For all $s', \node{n}'$ such that $s' <_{\node{n}', \node{n}} s$, $s' \in \semop{\node{n}'}$}\span\span
                  \\
                  &\text{iff}\;\;\; s \in \semop{\node{n}_1} \;\text{and}\; s \in \semop{\node{n}_2}
                  && \text{Definition of $<_{\node{n}', \node{n}}$ when $\rho(\node{n}) = \land$}
                  \\
                  &\text{iff}\;\;\; s \in \semT{\Phi_1}{\V[\Delta]} \;\text{and}\; s \in \semT{\Phi_2}{\V[\Delta]}
                  && \text{Definition of $\semTV{\node{n}_i}$, $i=1,2$}
                  \\
                  & \text{iff}\;\;\; s \in \semT{\Phi_1 \land \Phi_2}{\V[\Delta]}
                  && \text{Definition of $\semT{\Phi_1 \land \Phi_2}{\V[\Delta]}$}
                  \\
                  & \text{iff}\;\;\; s \in \semT{\Phi}{\V[\Delta]}
                  && \text{$\Phi = \Phi_1 \land \Phi_2$}
                  \\
                  & \text{iff}\;\;\; s \in \semop{\node{n}}
                  && \text{Definition of $\semop{\node{n}}$}
              \end{flalign*}

        \item[$\rho(\node{n}) = \lor$.]
              In this case $\Phi = \Phi_1 \lor \Phi_2$, and $cs(\node{n}) = \node{n}_1\node{n}_2$, where $\node{n}_1 = S_1 \tnxTVD \Phi_1$, $\node{n}_2 = S_2 \tnxTVD \Phi_2$ and $S = S_1 \cup S_2$.  By definition $s' <_{\node{n}', \node{n}} s$ iff $s' = s$ and either $\node{n}' = \node{n}_1$, provided  $s \in S_1$, or $\node{n}' = \node{n}_2$, provided $s \in S_2$.  Since either $s \in S_1$ or $s \in S_2$, it follows that  $s <_{\node{n}_1, \node{n}} s$ or $s <_{\node{n}_2, \node{n}} s$ (or both, if $s \in S_1 \cap S_2$).  We reason as follows.
              \begin{flalign*}
                  &\text{For all $s', \node{n}'$ such that $s' <_{\node{n}', \node{n}} s$, $s' \in \semop{\node{n}'}$}\span\span
                  \\
                  &\text{implies}\;\;\; s \in \semop{\node{n}_1} \;\text{or}\; s \in \semop{\node{n}_2}
                  && \text{Definition of $<_{\node{n}', \node{n}}$ when $\rho(\node{n}) = \lor$}
                  \\
                  &\text{iff}\;\;\; s \in \semT{\Phi_1}{\V[\Delta]} \;\text{or}\; s \in \semT{\Phi_2}{\V[\Delta]}
                  && \text{Definition of $\semop{\node{n}_i}$, $i=1,2$}
                  \\
                  & \text{iff}\;\;\; s \in \semT{\Phi_1 \lor \Phi_2}{\V[\Delta]}
                  && \text{Definition of $\semT{\Phi_1 \lor \Phi_2}{\V[\Delta]}$}
                  \\
                  & \text{iff}\;\;\; s \in \semT{\Phi}{\V[\Delta]}
                  && \text{$\Phi = \Phi_1 \lor \Phi_2$}
                  \\
                  & \text{iff}\;\;\; s \in \semop{\node{n}}
                  && \text{Definition of $\semop{\node{n}}$}
              \end{flalign*}

        \item[$\rho(\node{n}) = {[K]}$.]
              In this case $\Phi = [K] \Phi'$, and $cs(\node{n}) = \node{n}''$, where $\node{n}'' = S'' \tnxTVD \Phi'$ and $S'' = \{s'' \in S \mid \exists s \in S. s \xrightarrow{K} s'' \}$.  By definition $s' <_{\node{n}', \node{n}} s$ iff $\node{n}' = \node{n}''$ and $s \xrightarrow{K} s'$.  We reason as follows.
              \begin{flalign*}
                  &\text{For all $s', \node{n}'$ such that $s' <_{\node{n}', \node{n}} s$, $s' \in \semop{\node{n}'}$}\span\span
                  \\
                  &\text{iff}\;\;\; \forall s'. s \xrightarrow{K} s' \implies s' \in \semop{\node{n}''}
                  && \text{Definition of $<_{\node{n}', \node{n}}$ when $\rho(\node{n}) = [K]$}
                  \\
                  & \text{iff}\;\;\; \forall s'. s \xrightarrow{K} s' \implies s' \in \semT{\Phi'}{\V[\Delta]}
                  && \text{Definition of $\semTV{\node{n}''}$}
                  \\
                  & \text{iff}\;\;\; s \in \semT{[K] \Phi'}{\V[\Delta]}
                  && \text{Definition of $\semT{[K] \Phi'}{\V[\Delta]}$}
                  \\
                  & \text{iff}\;\;\; s \in \semT{\Phi}{\V[\Delta]}
                  && \text{$\Phi = [K]\Phi'$}
                  \\
                  & \text{iff}\;\;\; s \in \semop{\node{n}}
                  && \text{Definition of $\semop{\node{n}}$}
              \end{flalign*}

        \item[$\rho(\node{n}) = (\dia{K},f)$.]
              In this case $\Phi = \dia{K} \Phi'$, and $cs(\node{n}) = \node{n}''$, where $\node{n}'' = S'' \tnxTVD \Phi'$ and witness function $f \in S \rightarrow \states{S}$ is such that $S'' = f(S)$ and for all $s \in S, s \xrightarrow{K} f(s)$.  By definition $s' <_{\node{n}', \node{n}} s$ iff $\node{n}' = \node{n}''$ and $s' = f(s)$.  We reason as follows.
              \begin{flalign*}
                  &\text{For all $s', \node{n}'$ such that $s' <_{\node{n}', \node{n}} s$, $s' \in \semop{\node{n}'}$}\span\span
                  \\
                  &\text{iff}\;\;\; f(s) \in \semop{\node{n}''}
                  && \text{Definition of $<_{\node{n}', \node{n}}$ when $\rho(\node{n}) = \dia{K}$}
                  \\
                  & \text{iff}\;\;\; s \xrightarrow{K} f(s) \;\text{and}\; f(s) \in \semT{\Phi'}{\V[\Delta]}
                  && \text{Property of $f$, definition of $\semop{\node{n}''}$}
                  \\
                  & \text{implies}\;\;\; s \in \semT{\dia{K}\Phi'}{\V[\Delta]}
                  && \text{Definition of $\semT{\dia{K}\Phi'}{\V[\Delta]}$}
                  \\
                  & \text{iff}\;\;\; s \in \semT{\Phi}{\V[\Delta]}
                  && \text{$\Phi = \dia{K}\Phi'$}
                  \\
                  & \text{iff}\;\;\; s \in \semop{\node{n}}
                  && \text{Definition of $\semop{\node{n}}$}
              \end{flalign*}

        \item[$\rho(\node{n}) = \sigma Z$.]
              In this case $\Phi = \sigma Z. \Phi'$, and $cs(\node{n}) = \node{n}''$, where $\node{n}'' = S \tnxTV{\Delta'} U$, $U$ is fresh, and $\Delta' = \Delta \cdot (U = \sigma Z. \Phi')$.  By definition $s' <_{\node{n}', \node{n}} s$ iff $\node{n}' = \node{n}''$ and $s = s'$.  We reason as follows.
              \begin{flalign*}
                  &\text{For all $s', \node{n}'$ such that $s' <_{\node{n}', \node{n}} s$, $s' \in \semop{\node{n}'}$}\span\span
                  \\
                  & \text{iff}\;\;\; s \in \semop{\node{n}''}
                  && \text{Definition of $<_{\node{n}', \node{n}}$ when $\rho(\node{n}) = \sigma Z$}
                  \\
                  & \text{iff}\;\;\; s \in \semT{U}{V[\Delta']}
                  && \text{Definition of $\semop{\node{n}''}$}
                  \\
                  & \text{iff}\;\;\; s \in \semTV{ U [ \Delta'] }
                  && \text{Lemma~\ref{lem:definition-list-correspondence}}
                  \\
                  & \text{iff}\;\;\; s \in \semTV{ U [ \Delta \cdot (U := \sigma Z . \Phi') ] }
                  && \text{$\Delta' = \Delta \cdot (U = \sigma Z . \Phi') $}
                  \\
                  & \text{iff}\;\;\; s \in \semTV{ \left(U [U := \sigma Z . \Phi'] \right) [ \Delta ] }
                  && \text{Definition of $U [ \Delta \cdot (U := \sigma Z . \Phi') ]$}
                  \\
                  & \text{iff}\;\;\; s \in \semTV{\left( \sigma Z . \Phi'\right) [\Delta]}
                  \\
                  \multispan4{\hfil Definition of substitution, application of  $\rho(\node{n}) = \sigma Z.$ ensures $U$ fresh}
                  \\
                  & \text{iff}\;\;\; s \in \semT{\sigma Z . \Phi'}{\V[\Delta]}
                  && \text{Lemma~\ref{lem:definition-list-correspondence}}
                  \\
                  & \text{iff}\;\;\; s \in \semTV{\Phi}
                  && \text{$\Phi =\sigma Z . \Phi'$}
                  \\
                  & \text{iff}\;\;\; s \in \semop{\node{n}}
                  && \text{Definition of $\semop{\node{n}}$}
              \end{flalign*}

        \item[$\rho(\node{n}) = \textnormal{Un}$.]
              In this case $\Phi = U$, with $\Delta(U) = \sigma Z. \Phi'$, and $cs(\node{n}) = \node{n}''$, where $\node{n}'' = S \tnxTVD \Phi'[ Z:=U ]$.  By definition $s' <_{\node{n}', \node{n}} s$ iff $\node{n}' = \node{n}''$ and $s = s'$.  We reason as follows.
              \begin{flalign*}
                  &\text{For all $s', \node{n}'$ such that $s' <_{\node{n}', \node{n}} s$, $s' \in \semop{\node{n}'}$}\span\span
                  \\
                  & \text{iff}\;\;\; s \in \semop{\node{n}''}
                  && \text{Definition of $<_{\node{n}', \node{n}}$ when $\rho(\node{n}) = \text{Un}$}
                  \\
                  & \text{iff}\;\;\; s \in \semT{ \Phi' [Z:=U]}{\V[\Delta]}
                  && \text{Definition of $\semop{\node{n}''}$}
                  \\
                  & \text{iff}\;\;\; s \in \semT{U}{\V[\Delta]}
                  && \text{Lemma~\ref{lem:constant-unfolding}}
                  \\
                  & \text{iff}\;\;\; s \in \semT{\Phi}{\V[\Delta]}
                  && \text{$\Phi =U$}
                  \\
                  & \text{iff}\;\;\; s \in \semop{\node{n}}
                  && \text{Definition of $\semop{\node{n}}$}
              \end{flalign*}

        \item[$\rho(\node{n}) = \textnormal{Thin}$.]
              In this case $cs(\node{n}) = \node{n}''$, where $\node{n}'' = S' \tnxTVD \Phi$ and $S \subseteq S'$.  By definition $s' <_{\node{n}', \node{n}} s$ iff $\node{n}' = \node{n}''$ and $s = s'$.  We prove the following bi-implication, which implies the desired result.
              \begin{flalign*}
                  &\text{For all $s', \node{n}'$ such that $s' <_{\node{n}', \node{n}} s$, $s' \in \semop{\node{n}'}$}\span\span
                  \\
                  & \text{iff}\;\;\; s \in \semop{\node{n}''}
                  && \text{Definition of $<_{\node{n}', \node{n}}$ when $\rho(\node{n}) = \text{Thin}$}
                  \\
                  & \text{iff}\;\;\; s \in \semT{ \Phi }{\V[\Delta]}
                  && \text{Definition of $\semop{\node{n}''}$}
                  \\
                  & \text{iff}\;\;\; s \in \semTV{\node{n}}
                  && \text{Definition of $\semop{\node{n}}$}
              \end{flalign*}%
              \qedhere
    \end{description}
\end{proof}

\section{Soundness using support orderings}\label{sec:Soundness-via-support-orderings}

In this section we prove soundness of the proof system in the previous section by showing that for any successful tableau whose root is labeled by sequent $\seq{s} = S \tnxTV{\emptyL} \Phi$, $\seq{s}$ must be valid.  
Our proof relies on establishing that the transitive closure, $<:_{\node{m}}^+$, of the companion-node dependency relation (Definition~\ref{def:extended_path_ordering}) used to define the success of $\mu$-leaves, is a $\sigma$-compatible support ordering for the semantic functions associated with fixpoint nodes $\node{m}$.
Our reliance on support orderings for soundness stands in contrast to Bradfield's and Stirling's soundness proof for essentially this proof system~\cite{BS1992,Bra1991}, which relies on infinitary logic and the introduction of (infinite) ordinal-unfoldings of fixpoint formulas in particular.
Since our ultimate goal in this paper is to reason about timed extensions of the modal mu-calculus, we have opted for a different proof strategy that is based more on semantic rather than syntactic reasoning.  We also note that our use of support orderings is likely to enable the study of other success criteria besides $<:_{\node{m}}$, which is especially interesting for adaptations of this proof system to other settings, such as ones in which formulas are defined equationally or in which less aggressive use is made of formula unfolding than is the case here.

Our proof of soundness proceeds in four steps.
\begin{enumerate}
    \item 
    We show (Section~\ref{subsec:local-soundness}) that tableau rules are \emph{locally sound}, i.e., that when the child nodes of a proof node are valid, then the node itself is also valid.
    \item 
    We wish to be able to reason using tree induction about the meanings of proof nodes and the formulas in those nodes.  This reasoning is sometimes impeded because, due to unfolding, many nodes have the same formulas in them.  To address this problem, we show how syntactically distinct \emph{node formulas} can be constructed in a semantics-preserving fashion for nodes in a proof tree based on the structure of the proof tree. This material is in Section~\ref{subsec:node-formulas}.
    \item 
    We then prove that for companion nodes $\node{m}$ in a tableau, the ordering $<:_{\node{m}}^+$ is a support ordering for the semantic function associated with its node formula. This is done in Section~\ref{subsec:support-ordering-proof}.
    \item 
    Finally, in Section~\ref{subsec:soundness} we combine the previous three results to obtain soundness for the proof system:  if there is a successful tableau whose root $\seq{s}$ is such that $\seqdl(\seq{s}) = \emptyL$ then $\seq{s}$ is valid.
\end{enumerate}


\subsection{Local soundness}\label{subsec:local-soundness}

We call a proof system like ours \emph{locally sound} if for every internal node $\node{n}$ in any partial tableau, the validity of all the children of $\node{n}$ implies the validity of $\node{n}$.  This may be proven as follows.

\begin{lemma}[Local soundness]\label{lem:local-soundness}
    Let $\node{n}$ be an internal proof node in partial tableau $\mathbb{T}$.  Then $\node{n}$ is valid if all its children are valid.
\end{lemma}

\begin{proof}
    Let $\mathbb{T} = \tableauTrl$, with $\tree{T} = (\node{N}, \node{r}, p, cs)$, be a partial tableau with internal node $\node{n}$, and assume that for each $\node{n}' \in c(\node{n})$, node $\node{n}'$ is valid.
    To establish that $\node{n}$ is valid, we must show that $\seqst(\node{n}) \subseteq \semop{\node{n}}$. To do so, we fix $s \in \seqst(\node{n})$ and show that $s \in \semop{\node{n}}$.  In support of this, consider arbitrary $s', \node{n}'$ such that $s' <_{\node{n}',\node{n}} s$ in $\mathbb{T}$.  By definition of $<_{\node{n}',\node{n}}$ it follows that $\node{n}' \in c(\node{n})$ and that $s' \in \seqst(\node{n}')$.  Moreover, since $\node{n}'$ is valid it follows that $s' \in \semop{\node{n}'}$; since this holds for all such $s'$ and $\node{n}'$, Lemma~\ref{lem:semantic-sufficiency-of-<} ensures that $s \in \semop{\node{n}}$.
    \qedhere
\end{proof}

\subsection{Node formulas}\label{subsec:node-formulas}

We now show how to associate a formula $P(\node{n})$
with every node $\node{n}$ in a tableau so that the structure of $P(\node{n})$ is based on the structure of the subtableau rooted at $\node{n}$ and $P(\node{n})$ is disentangled from the definition list of $\node{n}$. We then show that these formulas are semantically equivalent to the formulas embedded in the nodes' sequents in a precise sense.
Since our definition is inductive on the structure of the tableau rooted at $\node{n}$, this facilitates proofs over the semantics of formulas using tree induction.

In the remainder of this section we fix sort $\Sigma$, labeled transition system $\T = \lts{S}$ over $\Sigma$, valuation $\V \in \Var \to 2^{\states{S}}$, and tableau $\mathbb{T} = \tableauTrl$, with $\tree{T} = (\node{N}, \node{r}, p, cs)$.
We also recall the definition of $\cnodesT$ --- the companion nodes of $\mathbb{T}$ --- and fix the definitions of the following sets.
\begin{align*}
\mathbb{U}  &= \bigcup_{\node{n} \in \node{N}} \dom(\seqdl(\node{n}))
\\
\cnodes{\tree{T}'}    &= \node{N}' \cap \cnodesT \text{ for subtree $\tree{T}' = (\node{N}', \ldots)$ of $\tree{T}$}
\\
\cnodes{\tree{T}'}(U) &= \{ \node{n} \in \cnodes{\tree{T}'} \mid \seqfm(\node{n}) = U \} \text{ for subtree $\tree{T}' = (\node{N}', \ldots)$ of $\tree{T}$ and $U \in \mathbb{U}$}
\end{align*}
Set $\mathbb{U}$ contains all the definitional constants appearing in $\mathbb{T}$, while
$\cnodes{\tree{T}'}$ contains the companion nodes of $\mathbb{T}$ in subtree $\tree{T}'$ of $\tree{T}$.  Set $\cnodes{\tree{T}'}(U)$ consists of the companion nodes of $\mathbb{T}$ in subtree $\tree{T}'$ whose formula is $U \in \mathbb{U}$. For subtree $\tree{T}_{\node{n}}$ rooted at node $\node{n}$, note that $\companions{\tree{T}_{\node{n}}} \subseteq \cnodes{\tree{T}} = \cnodesT$. Also, if $\rho(\node{n}) = \text{Un}$ with $c(\node{n}) = \node{n}'$, then $\companions{\tree{T}_{\node{n}}} = \companions{\tree{T}_{\node{n}'}} \cup \{ \node{n} \}$.  
Note that Corollary~\ref{cor:shared-prefix-sequent-semantics} and the definition of $\semop{\node{n}}$ guarantee that for all $U \in \mathbb{U}$ and $\node{n}, \node{n}' \in \cnodesT(U)$, $\semop{\node{n}} = \semop{\node{n}'}$; we write $\sem{U}{}{\mathbb{T}}$ for this common value associated with $U$.
We now define $P(-)$ as follows.

\begin{definition}[Node formulas]~\label{def:node-formulas}
For each companion node $\node{m} \in \cnodesT$ let $Z_\node{m}$ be a unique fresh variable, with $\Var_{\mathbb{T}} = \{ Z_{\node{m}} \mid \node{m} \in \cnodesT \}$ the set of all such variables.  Then for $\node{n} \in \node{N}$ formula $P(\node{n})$ is defined inductively as follows.
\begin{enumerate}
\item\label{subdef:node-formulas-free-leaf}
    If $\node{n}$ is a free leaf (cf.\/ Definition~\ref{def:tableau}(\ref{subdef:free-leaf})) 
    then $P(\node{n}) = \seqfm(\node{n})$.
\item
    If $\node{n}$ is a $\dia{K}$-leaf 
    then $P(\node{n}) = (\seqfm(\node{n}))[\Delta]$.
\item
    If $\node{n}$ is a $\sigma$-leaf with companion node $\node{m}$
    then $P(\node{n}) = Z_\node{m}$.
\item
    If $\rho(\node{n}) = \land$ and $cs(\node{n}) = \node{n}_1\node{n}_2$ 
    then $P(\node{n}) = P(\node{n}_1) \land P(\node{n}_2)$.
\item
    If $\rho(\node{n}) = \lor$ and $cs(\node{n}) = \node{n}_1\node{n}_2$
    then $P(\node{n}) = P(\node{n}_1) \lor P(\node{n}_2)$.
\item
    If $\rho(\node{n}) = [K]$ and $cs(\node{n}) = \node{n}'$ 
    then $P(\node{n}) = [K] (P(\node{n}'))$.
\item
    If $\rho(\node{n}) = (\dia{K},f)$ and $cs(\node{n}) = \node{n}'$
    then $P(\node{n}) = \dia{K} (P(\node{n}'))$.
\item
    If $\rho(\node{n}) = \sigma Z$ and $cs(\node{n}) = \node{n}'$
    then $P(\node{n}) = P(\node{n}')$.
\item
    If $\rho(\node{n}) = \textnormal{Thin}$ and $cs(\node{n}) = \node{n}'$
    then $P(\node{n}) = P(\node{n}')$.
\item\label{subdef:node-formulas-un}
    If $\rho(\node{n}) = \textnormal{Un}$, $\node{n} = S \tnxTVD U$, $\Delta(U) = \sigma Z.\Phi$ and $cs(\node{n}) = \node{n}'$
    then $P(\node{n}) = \sigma Z_\node{n}. \left( P(\node{n}') \right)$.
\end{enumerate}
\end{definition}
When $\node{n}$ is a free leaf (case~\ref{subdef:node-formulas-free-leaf}) $\seqfm(\node{n})$ contains no definitional constants, and thus $(\seqfm(\node{n}))[\Delta] = \seqfm(\node{n}) = P(\node{n})$. Also, when $\rho(\node{n}) = \text{Un}$ (case~\ref{subdef:node-formulas-un}), $\node{n} \in \cnodesT$ is a companion node in $\mathbb{T}$. Thus $P$ associates a syntactically distinct formula to each companion node in $\mathbb{T}$.

Intuitively, $P(\node{n})$ can be seen as the formula whose ``parse tree'' is the sub-tableau of $\mathbb{T}$ rooted at $\node{n}$.  The construction works bottom-up from the leaves that are descendants of $\node{n}$, using the proof rule labeling each internal node to recursively construct formulas from those associated with the node's children.  Each companion node is converted into a $\sigma$-formula, with a freshly generated bound variable that $P(-)$ ensures is assigned to each companion leaf of the companion node.
It is easy to see that $P(-)$ contains no instances of any $U \in \mathbb{U}$.  Moreover, $P(-)$ is an \emph{inductively generated} node function over $\tree{T}$, in the sense of Definition~\ref{def:node-function}(\ref{def:node-function-inductively-generated}).  In particular, the function $g \in \node{N} \times (\muformsSV)^* \to \muformsSV$ used to generate $P(-)$ is defined based on the inductive case associated with $\node{n}$ in Definition~\ref{def:node-formulas}.  For example, $g(\node{n}, \Phi_1\Phi_2) = \Phi_1 \lor \Phi_2$ if $\rho(\node{n}) = \lor$.

We now turn to establishing a semantic equivalence between $\semT{P(\node{n})}{\V'}$ for certain $\V'$ and $\semop{\node{n}}$ in $\mathbb{T}$ by first defining the following notion of \emph{valuation consistency with $\mathbb{T}$}.

\begin{definition}[Valuation consistency]\label{def:consistency}
    Let $\Var_{\mathbb{T}}$ be as given in Definition~\ref{def:node-formulas}, and let $\V$ be the valuation in tableau $\mathbb{T}$.
    Then
    valuation $\V'$ is \emph{consistent with tableau $\mathbb{T}$} iff
    \begin{itemize}
    \item 
        for every $U \in \mathbb{U}$ and $\node{m} \in \companions{\tree{T}}(U)$, $\V'(Z_{\node{m}}) = \sem{U}{}{\mathbb{T}}$, and
    \item 
        for every variable $X \in \Var \setminus \Var_{\mathbb{T}}$, $\V'(X) = \V(X)$.
    \end{itemize}
\end{definition}

\noindent
Intuitively, $\V'$ is consistent with $\mathbb{T}$ iff it assigns the semantics of the associated definitional constant to every fresh variable used in the definition of $P(-)$, and to all other variables it assigns the same value as valuation $\V$ in $\mathbb{T}$.
The following result is immediate from the definitions.

\begin{lemma}\label{lem:consistency-property}
Let $\V'$ be a valuation consistent with $\mathbb{T}$.
Then for all $\node{n} \in \node{N}$ such that $\seqfm(\node{n}) = \Phi$ and $\seqdl(\node{n}) = \Delta$,
$
\semop{\node{n}} = \semT{\Phi[\Delta]}{\V'}.
$
\end{lemma}
\begin{proof}
Suppose $\node{n} = S \tnxTVD \Phi$; we must show that $\semop{\node{n}} = \semT{\Phi[\Delta]}{\V'}$.  By definition, $\semop{\node{n}} = \semT{\Phi}{\V[\Delta]}$.  Lemma~\ref{lem:definition-list-correspondence} then guarantees that $\semT{\Phi}{\V[\Delta]} = \semTV{\Phi[\Delta]}$, and as $\V$ and $\V'$ only disagree on definitional constants in $\mathbb{T}$, and thus are not free in $\Phi[\Delta]$, we have that $\semop{\node{n}} = \semT{\Phi[\Delta]}{\V'}$. \qedhere
\end{proof}

We now prove that $\semop{\node{n}} = \semT{P(\node{n})}{\V'}$ for proof node $\node{n}$ in $\mathbb{T}$ and $\V'$ consistent with $\mathbb{T}$.  This fact establishes that $P(\node{n})$ summarizes all relevant information about the semantics of $\node{n}$, modulo the connection made by $\V'$ between definitional constants in $\node{n}$ and the associated free variables introduced by $P(-)$.  The proof is split across two lemmas; we first consider the special case when $\node{n} \in \cnodesT$ is a companion node, and then use this result to prove the general case.  

\begin{lemma}[Companion-node formulas and semantics]\label{lem:companion-node-formulas-and-semantics}
Let $\V'$ be a consistent valuation for tableau $\mathbb{T}$.  Then for every $\node{m} \in \cnodesT$, $\semT{P(\node{m})}{\V'} = \semop{\node{m}}$.
\end{lemma}
\remove{
\begin{proofsketch} 
For any valuation $\V'$ we call a syntactic transformation of $\Phi$ to $\Gamma$ \emph{semantics-preserving} for $\V'$ iff $\semT{\Phi}{\V'} = \semT{\Gamma}{\V'}$.  
Now let $\V'$ be consistent with $\mathbb{T}$.
We actually prove the following stronger result:
\begin{quote}
    for any $\node{m} \in \cnodesT$ with $\node{m} = S \tnxTVD U$ and $\Delta(U) = \sigma Z.\Phi$, there is a semantics-preserving transformation of $P(\node{m})$ to $(\sigma Z.\Phi)[\Delta]$ for $\V'$.
\end{quote}
The following reasoning then gives the desired result.
\begin{flalign*}
&\semT{P(\node{m})}{\V'}
\\
&{=}\; \semT{(\sigma Z.\Phi)[\Delta]}{\V'}
&& \text{Semantics-preserving transformation for $\V'$}
\\
&{=}\; \semT{(\sigma Z.\Phi)[\Delta]}{\V}
&& \text{$\V'$ consistent for $\mathbb{T}$, no $Z' \in \Var_{\mathbb{T}}$ free in $(\sigma Z.\Phi)[\Delta]$}
\\
&{=}\; \semT{\sigma Z.\Phi}{\V[\Delta]}
&& \text{Lemma~\ref{lem:definition-list-correspondence}}
\\
&{=}\; \semT{U}{\V[\Delta]}
&& \text{$\Delta(U) = \sigma Z.\Phi$, Lemma~\ref{lem:U-semantic-correspondence}}
\\
&{=}\; \semop{\node{m}}
&& \text{Definition of $\semop{\node{m}}$}
\end{flalign*}
The detailed proof, containing the precise semantics-preserving transformation, is included in the appendix.
\qedhere
\end{proofsketch}
}

\begin{proof}
For any valuation $\V'$ we call a syntactic transformation of $\Phi$ to $\Gamma$ \emph{semantics-preserving} for $\V'$ iff $\semT{\Phi}{\V'} = \semT{\Gamma}{\V'}$.  
Now let $\V'$ be consistent with $\mathbb{T}$.
We actually prove the following stronger result:
\begin{quote}
    for any $\node{m} \in \cnodesT$ with $\node{m} = S \tnxTVD U$ and $\Delta(U) = \sigma Z.\Phi$, there is a semantics-preserving transformation of $P(\node{m})$ to $(\sigma Z.\Phi)[\Delta]$ for $\V'$.
\end{quote}
The following reasoning then gives the desired result.
\begin{flalign*}
&\semT{P(\node{m})}{\V'}
\\
&{=}\; \semT{(\sigma Z.\Phi)[\Delta]}{\V'}
&& \text{Semantics-preserving transformation for $\V'$}
\\
&{=}\; \semT{(\sigma Z.\Phi)[\Delta]}{\V}
&& \text{$\V'$ consistent for $\mathbb{T}$, no $Z' \in \Var_{\mathbb{T}}$ free in $(\sigma Z.\Phi)[\Delta]$}
\\
&{=}\; \semT{\sigma Z.\Phi}{\V[\Delta]}
&& \text{Lemma~\ref{lem:definition-list-correspondence}}
\\
&{=}\; \semT{U}{\V[\Delta]}
&& \text{$\Delta(U) = \sigma Z.\Phi$, Lemma~\ref{lem:U-semantic-correspondence}}
\\
&{=}\; \semop{\node{m}}
&& \text{Definition of $\semop{\node{m}}$}
\end{flalign*}

\noindent
The proof therefore reduces to showing how to transform $P(\node{m})$, where companion node $\node{m} \in \cnodesT$ is such that $\node{m} = S \tnxTVD U$ and $\Delta(U) = \sigma Z.\Phi$, to $(\sigma Z.\Phi)[\Delta]$ in a way that is semantics-preserving for $\V'$.  So fix $\node{m} \in \cnodesT$.
The proof proceeds by strong induction on  $|\companions{\tree{T}_{\node{m}}}| \geq 1$.
The induction hypothesis guarantees that for any companion node $\node{m}' = S' \tnxTV{\Delta'} U' \in \companions{\tree{T}_{\node{m}}} \setminus \{ \node{m} \}$, there is a semantics-preserving transformation of $P(\node{m}')$ to $(\Delta'(U'))[\Delta']$ for $\V'$, since in this case $|\companions{\tree{T}_{\node{m}'}}| < |\companions{\tree{T}_{\node{m}}}|$.  We also remark on the following.
\begin{enumerate}
    \item
    As $\node{m} \in \cnodes{T}$, $cs(\node{m}) = \node{n}$ is such that $\node{n} = S \tnxTVD \Phi[Z := U]$.
    \item
    By definition, $P(\node{m}) = \sigma Z_{\node{m}} . P(\node{n})$.
\end{enumerate}

Our transformation from $P(\node{m})$ to $(\sigma Z.\Phi)[\Delta]$ uses \emph{inductive updates} (cf.\/ Definition~\ref{def:node-function-inductive-update}) of $P$. That is, we will transform $P(\node{m})$ by using inductive updates to change the values returned by $P$ for some of the descendants of $\node{m}$.  This will have the effect of rewriting $P(\node{m})$ to $(\sigma Z.\Phi)[\Delta]$ to $\node{m}$.
To define these inductive updates we will use a tree prefix, $\tree{T}'$, of the subtree $\tree{T}_{\node{n}}$ of $\tree{T}$ rooted at $\node{n}$, the sole child of $\node{m}$ (cf.\/ Definition~\ref{def:subtree}).  Let $\node{F} = \{ \node{f} \in D(\node{n}) \mid \rho(\node{f}) \in \{ \sigma Z, \text{Un} \}\}$ be the companion and fixpoint nodes in $\tree{T}_{\node{n}}$. Then \[
\tree{T}' 
= \tpre{\tree{T}_{\node{n}}}{\node{F}} 
= (\node{N}', \node{n}, p', cs')
\]
(cf.\/ Definition~\ref{def:tree-prefix-generation}) is the tree prefix of the subtree $\tree{T}_{\node{n}}$ obtained by converting the nearest descendants of $\node{n}$ that are companion or fixpoint nodes of $\mathbb{T}$ into leaves.  We note that $\tree{T}'$ contains no internal nodes that are companion nodes or fixpoint nodes of $\mathbb{T}$.  It is also straightforward to show that for each $\node{n}' \in \node{N}'$, $\seqfm(\node{n}')$ is a subformula of $\Phi[Z := U]$ and $\seqdl(\node{n}') = \Delta$.  Finally, we remark on a property involving inductive updates of $P$ on nodes in $\node{N}'$; this property allows us to replace $P(\node{n}')$, where $\node{n}' \in \node{N}'$, by a semantically equivalent formula $\Gamma$ in an inductive update of $P$ without changing the semantics of any formulas generated by the updated function.
\begin{quote}
    \textbf{(IU)}
    Let $\node{N}'' \subseteq \node{N}'$, with $\vec{\node{n}}'' = \node{n}''_1 \cdots \node{n}''_j$ an ordering of $\node{N}''$.  Also assume $\vec{\Gamma} = \Gamma_1 \cdots \Gamma_j \in (\muformsSV)^*$ satisfies $\semT{P(\node{n}''_i)}{\V'} = \semT{\Gamma_i}{\V'}$ all $i$.  Then for all $\node{n}' \in \node{N}'$,
    \[
        \semT{P(\node{n}')}{\V'} 
        = \semT{P \iupd{\vec{\node{n}}''}{\vec{\Gamma}}(\node{n}')}{\V'}.
    \]
\end{quote}
Property (IU) follows from the definition of $P \iupd{\vec{\node{n}}''}{\vec{\Gamma}}$ via a simple tree induction on $\tree{T}'$.

We now show how to transform $P(\node{m})$ to $(\sigma Z.\Phi)[\Delta]$ in a semantics-preserving fashion for $\V'$.  The key transformation relies on inductively updating $P$ via the leaves, $\node{L} \subseteq \node{N}'$, of $\tree{T}'$.  Each such leaf falls into one of three disjoint sets.
\begin{align*}
\node{L}_\div  &= \{\node{n}' \in \node{N}' \mid \rho(\node{n}')\div \}
\\
\node{L}_{\textnormal{Un}}  &= \{\node{n}' \in \node{N}' \mid \rho(\node{n}') = \text{Un} \}
\\
\node{L}_\sigma  &= \{\node{n}' \in \node{N}' \mid \rho(\node{n}') = \sigma Z \}
\end{align*}
$\node{L}_\div$ consists of leaves in $\tree{T}'$ that are also leaves in $\tree{T}_{\node{n}}$, and hence in $\tree{T}$, while $\node{L}_{\textnormal{Un}}$ contains the leaves of $\tree{T}'$ that are companion nodes in $\node{T}_{\node{n}}$.
$\node{L}_\sigma$ is the set of the leaves in $\tree{T}'$ whose formulas involve fixpoint formulas.  
Note $\node{L} = \node{L}_\div \cup \node{L}_{\textnormal{Un}} \cup \node{L}_\sigma$.

Now let $\vec{\node{n}}' = \node{n}'_1 \cdots \node{n}'_j$ be an ordering of $\node{L}$, and define $\vec{\Phi}' = \Phi'_1 \cdots \Phi'_j$ by:
\[
\Phi'_i =
\begin{cases}
Z_{\node{m}}
    & \text{if $\seqfm(\node{n}'_i) = U$}
    \\
(\seqfm(\node{n}'_i))[\Delta]
    & \text{otherwise.}
\end{cases}
\]
That is, $\Phi'_i$ is defined to be $Z_{\node{m}}$, the fresh variable associated with $\node{m}$, if the formula in $\node{n}_i$ is $U$, and thus is either a companion leaf of $\node{m}$ in $\mathbb{T}$ (and thus in $\node{L}_\div$) or a companion node for $U$ that is a strict descendant of $\node{m}$ (and thus in $\node{L}_{\textnormal{Un}})$.  Otherwise, $\Phi'_i$ is the formula in $\node{n}'_i$, instantiated by $\Delta$.
Now define
\[
P' = P\iupd{\vec{\node{n}}'}{\vec{\Phi}'}.
\]
We will show the transformations of $P(\node{m})$ to $P'(\node{m})$, and of $P(\node{n})$ to $P'(\node{n})$, are semantics-preserving for $\V'$ by showing that for each $\node{n}'_i \in \node{L}$, $\semT{P(\node{n}'_i)}{\V'} = \semT{\Phi'_i}{\V'}$; Property (IU) then gives the desired result, namely, $\semT{P(\node{n})}{\V'} = \semT{P'(\node{n})}{\V'}$ and thus $\semT{P(\node{m})}{\V'} = \semT{P'(\node{m})}{\V'}$.  The argument proceeds via a case analysis on $\node{n}'_i$.
\begin{description}
\item[$\seqfm(\node{n}'_i) = U$.]
    In this case $\Phi'_i = Z_\node{m}$, and either $\node{n}'_i \in \node{L}_\div$ or $\node{n}'_i \in \node{L}_{\textnormal{Un}}$.  If $\node{n}'_i \in \node{L}_\div$ then $\node{n}'_i$ is a companion leaf of $\node{m}$ in $\mathbb{T}$, and by definition $P(\node{n}'_i) = Z_\node{m} = \Phi'_i$.  Now assume $\node{n}'_i \in \node{L}_{\textnormal{Un}}$ is a companion node in $\mathbb{T}$, meaning $\node{n}'_i \in \companions{\tree{T}_{\node{m}}} \setminus \{ \node{m} \}$. As $\seqdl(\node{n}'_i) = \Delta$ the induction hypothesis guarantees that $\semT{P(\node{n}'_i)}{\V'} = \semT{(\Delta(U))[\Delta]}{\V'} = \semT{(\sigma Z.\Phi)[\Delta]}{\V'}$.  We now reason as follows.
    \begin{flalign*}
        &\semT{P(\node{n}'_i)}{\V'}
    \\
        &{=}\; \semT{(\sigma Z.\Phi)[\Delta]}{\V'}
        && \text{Induction hypothesis}
    \\
        &{=}\; \semT{(\sigma Z.\Phi)[\Delta]}{\V}
        && \text{$\V'$ consistent with $\mathbb{T}$, no $Z' \in \Var_{\mathbb{T}}$ free in $(\sigma Z.\Phi)[\Delta]$}
    \\
        &{=}\; \semT{\sigma Z.\Phi}{\V[\Delta]}
        && \text{Lemma~\ref{lem:definition-list-correspondence}}
    \\
        &{=}\; \semT{U}{\V[\Delta]}
        && \text{$\Delta(U) = \sigma Z.\Phi$, Lemma~\ref{lem:U-semantic-correspondence}}
    \\
        &{=}\; \sem{U}{}{\mathbb{T}}
        && \text{Definition of $\sem{U}{}{\mathbb{T}}$}
    \\
        &{=}\; \V'(Z_{\node{m}})
        && \text{$\V'$ consistent with $\mathbb{T}$}
    \\
        &{=}\; \semT{\Phi'_i}{\V'}
        && \text{$\Phi'_i = Z_{\node{m}}$, so $\semT{\Phi'_i}{\V'} = \semT{Z_{\node{m}}}{\V'} = \V'(Z_{\node{m}})$}
    \end{flalign*}
\item[$\node{n}'_i \in \node{L}_\div, \seqfm(\node{n}'_i) \neq U$.]
    There are two subcases to consider. In the first $\seqfm(\node{n}'_i) \not\in \mathbb{U}$, meaning either $\node{n}'_i$ is a free leaf, in which case we have argued above that $P(\node{n}'_i) = (\seqfm(\node{n}'_i))[\Delta]$,
    or $\node{n}'_i$  is a $\dia{K}$-leaf in $\tree{T}$, in which case by definition $P(\node{n}'_i) = (\seqfm(\node{n}'_i))[\Delta]$.  Regardless, $P(\node{n}'_{i}) = (\seqfm(\node{n}'_{i}))[\Delta] = \Phi'_i$, and the result is immediate.  
    In the second subcase $\seqfm(\node{n}'_{i}) = U' \in \mathbb{U}$ for some $U' \neq U$.  In this case $\node{n}'_{i}$ is a companion leaf of some ancestor companion node $\node{m}'$ of $\node{m}$, and 
    $P(\node{n}'_i) = Z_{\node{m}'}$.  We reason as follows.
        \begin{align*}
        \semT{P(\node{n}'_{i})}{\V'}
        &= \V'( Z_{\node{m}'} )
        && \text{$\semT{P(\node{n}'_{i})}{\V'} = \semT{Z_{\node{m}'}}{\V'} = \V'(Z_{\node{m}'_i})$}
        \\
        &= \sem{U'}{}{\mathbb{T}}
        && \text{Consistency of $\V'$ for $\mathbb{T}$}
        \\
        &= \semop{\node{n}'_{i}}
        && \text{Definition of $\sem{U'}{}{\mathbb{T}}$}
        \\
        &= \semT{U'[\Delta]}{\V'}
        && \text{Lemma~\ref{lem:consistency-property}, $\V'$ consistent with $\mathbb{T}$}
        \\
        &= \semT{\Phi'_{i}}{\V'}
        && \text{$\Phi'_{i} = (\seqfm(\node{n}'_{i}))[\Delta] = U'[\Delta]$}
        \end{align*}

\item[$\node{n}'_i \in \node{L}_{\textnormal{Un}}, \seqfm(\node{n}'_i) \neq U$.]
    In this case $\seqfm(\node{n}'_i) = U' \in \mathbb{U}$ 
    for some $U' \neq U$, 
    and $\node{n}'_i \in \companions{\tree{T}_{\node{m}}} \setminus \{ \node{m} \}$.  
    Hence $\semT{P(\node{n}'_i)}{\V'} = \semT{(\Delta(U')[\Delta]}{\V'}$ according to the induction hypothesis.
    Since $\semT{(\Delta(U')[\Delta]}{\V'} = \semT{U'[\Delta]}{\V'} = \semT{\Phi'_i}{\V'}$, $\semT{P(\node{n}'_i)}{\V'} = \semT{\Phi'_i}{\V'}$.
\item[$\node{n}'_i \in \node{L}_\sigma$.]
    In this case $\seqfm(\node{n}'_{i}) = \sigma' Z'.\Phi'$ for some $\sigma', Z'$ and $\Phi'$.  Also,
    $cs(\node{n}'_{i}) = \node{n}''_{i}$ in $\tree{T}$ is such that $\node{n}''_{i} \in \companions{\tree{T}}$, with 
    \[
    \seqdl(\node{n}''_{i})
    = \Delta' 
    = \Delta \cdot (U' = \sigma'Z'.\Phi')
    \]
    for some $U' \not\in \dom(\Delta)$ and $\seqfm(\node{n}''_{i}) = U'$.
    Since $\node{n}''_{i} \in \companions{\tree{T}_{\node{m}}} \setminus \{ \node{m} \}$, the induction hypothesis implies
    $\semT{P(\node{n}''_{i})}{\V'} 
    = \semT{(\Delta'(U'))[\Delta']}{\V'}
    = \semT{(\sigma'Z'.\Phi')[\Delta']}{\V'}$.
    We prove that $\semT{P(\node{n}'_i)}{\V'} = \semT{\Phi'_i}{\V'}$ as follows.
        %
    \begin{align*}
    &\semT{P(\node{n}'_{i})}{\V'}
    \\
    &= \semT{P(\node{n}''_{i})}{\V'}
    && \text{Definition of $P(\node{n}'_{i})$ when $\seqfm(\node{n'_i}) = \sigma' Z'.\Phi'$}
    \\
    &= \semT{(\sigma'Z'.\Phi')[\Delta']}{\V'}
    && \text{Induction hypothesis}
    \\
    &= \semT{(\sigma'Z'.\Phi')[\Delta]}{\V'}
    && \text{$U'$ not free in $\sigma'Z'.\Phi'$}
    \\
    &= \semT{(\seqfm(\node{n}'_i))[\Delta]}{\V'}
    && \text{$\seqfm(\node{n}'_i) = \sigma'Z'.\Phi'$}
    \\
    &= \semT{\Phi'_i}{\V'}
    && \text{$\Phi'_i = (\seqfm(\node{n}'_i))[\Delta]$}
    \end{align*}
\end{description}

\noindent
The final part of the semantics-preserving transformation of $P(\node{m})$ to $(\sigma Z.\Phi)[\Delta]$ for $\V'$ involves the following steps.
\begin{itemize}
\item 
    We show that $\Gamma = \sigma Z_{\node{m}}.(\Phi[Z := Z_{\node{m}}])$ is such that $P'(\node{m}) = \Gamma[\Delta]$,
    and thus $\semT{P'(\node{m})}{\V'} = \semT{\Gamma[\Delta]}{\V'}$.
\item
    Since $Z_{\node{m}}$ is not free in $\Phi$, we also know that for any valuation $\V''$, $\semT{\Gamma}{\V''} = \semT{\sigma Z.\Phi}{\V''}$.
    In particular, $\semT{\Gamma}{\V'[\Delta]} = \semT{(\sigma Z.\Phi)}{\V'[\Delta]}$, whence $\semT{\Gamma[\Delta]}{\V'} = \semT{(\sigma Z.\Phi)[\Delta]}{\V'}$ and we have the following,
    \[
    \semT{P(\node{m})}{\V'} 
    = \semT{P'(\node{m})}{\V'}
    = \semT{\Gamma[\Delta]}{\V'}
    = \semT{(\sigma Z.\Phi)[\Delta]}{\V'},
    \]
    which is what is to be proved.
\end{itemize}

To finish the proof, we first show that $P'(\node{n}) = \Gamma'[\Delta]$, where $\Gamma' = \Phi[Z := Z_{\node{m}}]$.  In support of this, consider.
\[
    \Phi''_i =
    \begin{cases}
        Z_{\node{m}}
        & \text{if $\seqfm(\node{n}'_i) = U$}
    \\
        \seqfm(\node{n}'_i)
        & \text{otherwise.}
    \end{cases}
\]
It is immediate that $\Phi'_i = \Phi''_i[\Delta]$ for all $1 \leq i \leq j$.  Now consider
\[
P'' = P\iupd{\vec{\node{n}}'}{\vec{\Phi}''}.
\]
If $\node{n}' \in \node{L}$ then either $\seqfm(\node{n}') = U$ and $P''(\node{n}') = Z_{\node{m}}$, or $P''(\node{n}') = \seqfm(\node{n}')$.
Based on this observation and the definition of $P''$, it follows that $P''(\node{n}') = (\seqfm(\node{n}'))[U := Z_{\node{m}}]$ for any $\node{n}' \in \node{N}'$.  In particular, $P''(\node{n}) = (\seqfm(\node{n}))[U := Z_{\node{m}}] = (\Phi[Z:=U])[U := Z_{\node{m}}] = \Phi[Z := Z_{\node{m}}] = \Gamma'$.
It also follows from the properties of $P''$ and $\Phi''_i$ that for any $\node{n}' \in \node{N}'$, $P'(\node{n}') = (P''(\node{n}'))[\Delta]$.  Thus $\Gamma'[\Delta] = (P''(\node{n}))[\Delta] = P'(\node{n}) = \Gamma_1'$.

The final step of the proof is the observation that $\Gamma = P''(\node{m})$ and that $\Gamma_2[\Delta] = (P''(\node{m}))[\Delta] = P'(\node{m})$.
\qedhere

\end{proof}

\noindent
The next lemma extends the previous one, which focused only on companion nodes, to all nodes.

\begin{lemma}[Node formulas and node semantics]\label{lem:node-formulas-and-node-semantics}
Let $\V'$ be a consistent valuation for tableau $\mathbb{T}$.  Then for every $\node{n} \in \node{N}$, $\semT{P(\node{n})}{\V'} = \semop{\node{n}}$.
\end{lemma}

\remove{
\begin{proofsketch}
    The proof proceeds by induction on $\tree{T}$, the tree embedded in $\mathbb{T}$. The case where $\rho(\node{n}) = \text{Un}$ follows from Lemma~\ref{lem:companion-node-formulas-and-semantics}.
    \seeappendix
\end{proofsketch}
}
\begin{proof}
Let valuation $\V'$ be consistent with $\mathbb{T}$.  The proof is by induction on $\tree{T}$.  So fix node $\node{n}$ in $\tree{T}$.  The induction hypothesis asserts that for all $\node{n}' \in c(\node{n})$, $\semop{\node{n}'} = \semT{P(\node{n}')}{\V'}$.  The proof now proceeds by an analysis of $\rho(\node{n})$.
\begin{description}
\item[$\rho(\node{n}) \div$.]
    In this case $\node{n}$ is a leaf.  We proceed by an analysis on the form of $\node{n}$.
    \begin{itemize}
    \item
        $\node{n}$ is a free leaf or $\dia{K}$-leaf.  Let $\Phi = \seqfm(\node{n})$; we reason as follows.
        \begin{align*}
        \semT{P(\node{n})}{\V'}
        &= \semT{\Phi[\Delta]}{\V'}
        && \text{Definition of $P(\node{n})$}
        \\
        &= \semT{\Phi[\Delta]}{\V}
        && \text{Consistency of $\V'$, no $Z_\node{m}$ can appear in $\Phi[\Delta]$}
        \\
        &= \semT{\Phi}{\V[\Delta]}
        && \text{Lemma~\ref{lem:definition-list-correspondence}}
        \\
        &= \semop{\node{n}}
        && \text{Definition of $\semop{\node{n}}$}
        \end{align*}
    \item
    $\node{n} $ is a $\sigma$-leaf.  Let $\seqfm(\node{n}) = U$, and let $\node{m}$ be the companion node of $\node{n}$.  Then $P(\node{n}) = Z_\node{m}$, where $Z_\node{m} \in \Var_{\mathbb{T}}$ is the fresh variable associated $\node{m}$.  We reason as follows.
        \begin{align*}
        \semT{P(\node{n})}{\V'}
        &= \semT{Z_\node{m}}{\V'}
        && \text{Definition of $P(\node{n})$}
        \\
        &= \V'( Z_\node{m})
        && \text{Definition of $\semTV{Z_\node{m}}$}
        \\
        &= \sem{U}{}{\mathbb{T}}
        && \text{Consistency of $\V'$ for $\mathbb{T}$}
        \\
        &= \semop{\node{n}}
        && \text{Definition of $\sem{U}{}{\mathbb{T}}$}
        \end{align*}
    \end{itemize}

\item[$\rho(\node{n}) = \land$.]
    In this case we know that $\node{n} = S \tnxTVD \Phi_1 \land \Phi_2$ and that $cs(\node{n}) = \node{n}_1\node{n}_2$, where each $\node{n}_i = S \tnxTVD \Phi_i$.  The induction hypothesis guarantees that $\semT{P(\node{n}_i)}{\V'} = \semop{\node{n}_i} = \semT{\Phi_i}{\V[\Delta]}$.  We reason as follows.
    \begin{align*}
    \semT{P(\node{n})}{\V'}
    &= \semT{P(\node{n}_1) \land P(\node{n}_2)}{\V'}
    && \text{Definition of $P(\node{n})$}
    \\
    &= \semT{P(\node{n}_1)}{\V'} \cap \semT{P(\node{n}_2)}{\V'}
    && \text{Semantics of $\land$}
    \\
    &= \semop{\node{n}_1} \cap \semop{\node{n}_2}
    && \text{Induction hypothesis (twice)}
    \\
    &= \semT{\Phi_1}{\V[\Delta]} \cap \semT{\Phi_2}{\V[\Delta]}
    && \text{Definition of $\semop{\node{n}_i}$}
    \\
    &= \semT{\Phi_1 \land \Phi_2}{\V[\Delta]}
    && \text{Semantics of $\land$}
    \\
    &= \semop{\node{n}}
    && \text{Definition of $\semop{\node{n}}$}
    \end{align*}
\item[$\rho(\node{n}) = \lor$.]
    In this case we know that $\node{n} = S \tnxTVD \Phi_1 \lor \Phi_2$ and that $cs(\node{n}) = \node{n}_1\node{n}_2$, where each $\node{n}_i = S_i \tnxTVD \Phi_i$ and $S = S_1 \cup S_2$.  The induction hypothesis guarantees that $\semT{P(\node{n}_i)}{\V'} = \semop{\node{n}_i} = \semT{\Phi_i}{\V[\Delta]}$.  We reason as follows.
    \begin{align*}
    \semT{P(\node{n})}{\V'}
    &= \semT{P(\node{n}_1) \lor P(\node{n}_2)}{\V'}
    && \text{Definition of $P(\node{n})$}
    \\
    &= \semT{P(\node{n}_1)}{\V'} \cup \semT{P(\node{n}_2)}{\V'}
    && \text{Semantics of $\lor$}
    \\
    &= \semop{\node{n}_1} \cup \semop{\node{n}_2}
    && \text{Induction hypothesis (twice)}
    \\
    &= \semT{\Phi_1}{\V[\Delta]} \cup \semT{\Phi_2}{\V[\Delta]}
    && \text{Definition of $\semop{\node{n}_i}$}
    \\
    &= \semT{\Phi_1 \lor \Phi_2}{\V[\Delta]}
    && \text{Semantics of $\lor$}
    \\
    &= \semop{\node{n}}
    && \text{Definition of $\semop{\node{n}}$}
    \end{align*}
\item[$\rho(\node{n}) = {[K]}$.]
    In this case we know that $\node{n} = S \tnxTVD [K] \Phi'$ and that $cs(\node{n}) = \node{n}'$, where $\node{n}' = S' \tnxTVD \Phi'$ and $S' = \{ s' \mid \exists s \in S \colon s \xrightarrow{K} s' \}$.  The induction hypothesis guarantees that $\semT{P(\node{n}')}{\V'} = \semTV{\node{n}'} = \semT{\Phi'}{\V[\Delta]}$.  We reason as follows.
    \begin{align*}
    \semT{P(\node{n})}{\V'}
    &= \semT{[K] (P(\node{n}'))}{\V'}
    && \text{Definition of $P(\node{n})$}
    \\
    &= \pre_{[K]} (\semT{P(\node{n}')}{\V'})
    && \text{Semantics of $[K]$}
    \\
    &= \pre_{[K]} (\semop{\node{n}'})
    && \text{Induction hypothesis}
    \\
    &= \pre_{[K]} (\semT{\Phi'}{\V[\Delta]})
    && \text{Definition of $\semop{\node{n}'}$}
    \\
    &= \semT{[K] \Phi'}{\V[\Delta]}
    && \text{Semantics of $[K]$}
    \\
    &= \semop{\node{n}}
    && \text{Definition of $\semop{\node{n}}$}
    \end{align*}
\item[$\rho(\node{n}) = (\dia{K},f)$.]
    In this case we know that $\node{n} = S \tnxTVD \dia{K} \Phi'$ and that $cs(\node{n}) = \node{n}'$, where $\node{n}' = f(S) \tnxTVD \Phi'$.  The induction hypothesis guarantees that $\semT{P(\node{n}')}{\V'} = \semop{\node{n}'} = \semT{\Phi'}{\V[\Delta]}$.  We reason as follows.
    \begin{align*}
    \semT{P(\node{n})}{\V'}
    &= \semT{\dia{K} P(\node{n}')}{\V'}
    && \text{Definition of $P(-)$}
    \\
    &= \pre_{\dia{K}} (\semT{P(\node{n}')}{\V'})
    && \text{Semantics of $\langle K \rangle$}
    \\
    &= \pre_{\dia{K}} (\semop{\node{n}'})
    && \text{Induction hypothesis}
    \\
    &= \pre_{\dia{K}} (\semT{\Phi'}{\V[\Delta]})
    && \text{Definition of $\semT{\node{n}'}{\V[\Delta]}$}
    \\
    &= \semT{\dia{K} \Phi'}{\V[\Delta]}
    && \text{Semantics of $\dia{K}$}
    \\
    &= \semop{\node{n}}
    && \text{Definition of $\semop{\node{n}}$}
    \end{align*}
\item[$\rho(\node{n}) = \sigma Z$.]
    In this case we know that $\node{n} = S \tnxTVD \sigma Z.\Phi'$ and that $cs(\node{n}) = \node{n}'$, where $\node{n}' = S \tnxTV{\Delta'} U$ for some fresh definitional constant $U$ and $\Delta' = \Delta \cdot (U = \sigma Z.\Phi')$.  The induction hypothesis guarantees that $\semT{P(\node{n}')}{\V'} = \semop{\node{n}'} = \semT{U}{\V[\Delta']}$.  We reason as follows.
    \begin{align*}
    \semT{P(\node{n})}{\V'}
    &= \semT{P(\node{n}')}{\V'}
    && \text{Definition of $P(\node{n})$}
    \\
    &= \semop{\node{n}'}
    && \text{Induction hypothesis}
    \\
    &= \semT{U}{\V[\Delta']}
    && \text{Definition of $\semop{\node{n}'}$}
    \\
    &= \semTV{U[\Delta']}
    && \text{Lemma~\ref{lem:definition-list-correspondence}}
    \\
    &= \semTV{(\sigma Z.\Phi')[\Delta]}
    && \text{Definition of $U[\Delta']$, $\Delta' = \Delta \cdot (U = \sigma Z.\Phi')$}
    \\
    &= \semT{\sigma Z. \Phi'}{\V[\Delta]}
    && \text{Lemma~\ref{lem:definition-list-correspondence}}
    \\
    &= \semop{\node{n}}
    && \text{Definition of $\semop{\node{n}}$}
    \end{align*}
\item[$\rho(\node{n}) = \textnormal{Un}$.]
    Follows immediately from Lemma~\ref{lem:companion-node-formulas-and-semantics}.
\item[$\rho(\node{n}) = \textnormal{Thin}$.]
    In this case we know that $\node{n} = S \tnxTVD \Phi$ and that $cs(\node{n}) = \node{n}' $, where $\node{n}' = S' \tnxTVD \Phi$ for some $S \subseteq S'$.  The induction hypothesis guarantees that $\semT{P(\node{n}')}{\V'} = \semop{\node{n}'} = \semT{\Phi}{\V[\Delta]}$.  We reason as follows.
    \begin{align*}
    \semT{P(\node{n})}{\V'}
    &= \semT{P(\node{n}')}{\V'}
    && \text{Definition of $P(-)$}
    \\
    &= \semop{\node{n}'}
    && \text{Induction hypothesis}
    \\
    &= \semT{\Phi}{\V[\Delta]}
    && \text{Definition of $\semop{\node{n}'}$}
    \\
    &= \semop{\node{n}}
    && \text{Definition of $\semop{\node{n}}$}
    \end{align*}
\end{description}\qedhere
\end{proof}

\noindent
The final corollary asserts that when the definition list in a proof node is empty, there is no need to make special provision for consistent valuations.

\begin{corollary}\label{cor:node-formulas-vs-node-semantics}
Let $\node{n} \in \node{N}$ be such that $\seqdl(\node{n}) = \emptyL$.  Then $\semop{\node{n}} = \semTV{P(\node{n})}$.
\end{corollary}
\begin{proof}
Fix $\node{n} \in \node{N}$ such that $\seqdl(\node{n}) = \emptyL$.  Based on Lemma~\ref{lem:node-formulas-and-node-semantics} we know that for any valuation $\V'$ that is consistent with $\mathbb{T}$, $\semT{P(\node{n})}{\V'} = \semop{\node{n}}$.  It may also be seen that no $Z' \in \Var_{\mathbb{T}}$ can be free in $P(\node{n})$, and since $\V'$ is consistent with $\mathbb{T}$ we have that $\semT{P(\node{n})}{\V'} = \semTV{P(\node{n})}$.  Consequently, $\semop{\node{n}} = \semTV{P(\node{n})}$.  \qedhere
\end{proof}

\subsection{Support Orderings for Companion Nodes}\label{subsec:support-ordering-proof}

As the next step in our soundness proof, we establish that for all companion nodes $\node{n}$ in the tableau, $(\seqst(\node{n}),<:_{\node{n}}^+)$, where $<:_{\node{n}}^+$ is the transitive closure of the extended dependency ordering on $\node{n}$, is a support ordering for a semantic function derived from $P(\node{n})$.  This fact is central in the proof of soundness, as it establishes a key linkage between the tableau-based ordering $<:_{\node{n}}^+$ and the semantic notion of support ordering.

In order to prove this result about $<:_\node{n}^+$ we first introduce a derived dependency relation, which we call the \emph{support dependency ordering} (notation $\leq:_{\node{m},\node{n}}$).
This ordering is based on the extended dependency ordering $<:_{\node{m},\node{n}}$, but it allows dependencies based on cycling through node $\node{n}$ first, in case $\node{n}$ is a companion node.
Specifically, the support dependency ordering captures exactly the dependencies guaranteeing that $s$ is in the semantics of $\node{n}$ if for every $s'$, $\node{m}$ with $s' \leq :_{\node{m},\node{n}} s$, state $s'$ is in the semantics of $\node{m}$.
If $\node{n}$ is a node in which the unfolding rule has been applied, to show that $s$ is in the semantics of $\node{n}$, we may require to first show that $s'$ is in the semantics of $\node{n}$.
This is not captured by the relation $<:_{\node{m},\node{n}}$, which does not take the dependencies within node $\node{n}$ into account.

\begin{definition}[Support dependency ordering]\label{def:support-dependency-ordering}
Let $\node{m},\node{n} \in \node{N}$ be proof nodes in $\mathbb{T}$.
The \emph{support dependency ordering}, $\leq :_{\node{m},\node{n}}$ is defined as follows:
\[
\leq :_{\node{m},\node{n}} = 
\begin{cases}
  <:_{\node{m},\node{n}} \rcomp <:_\node{n}^* & \text{if $\rho(\node{n}) = \textnormal{Un}$} \\
  <:_{\node{m},\node{n}} & \text{otherwise}
\end{cases}
\]
\end{definition}

We now remark on some properties of $\leq:_{\node{m},\node{n}}$ that will be used below.
We first note that ${\leq:_{\node{m},\node{n}}}$ extends ${<:_{\node{m},\node{n}}}$ (as well as $<_{\node{m},\node{n}}$ and $\lessdot_{\node{m},\node{n}}$, since $<:_{\node{m},\node{n}}$ extends both of these relations):  for all $s, s'$, if $s' <:_{\node{m},\node{n}} s$ then $s' \leq:_{\node{m},\node{n}} s$.
Also, if $\node{n}$ is a companion node (i.e.\/ $\rho(\node{n}) = \text{Un}$) then the transitivity of $<:_{\node{n}}^*$ guarantees that ${\leq:_{\node{m},\node{n}}} = (\leq:_{\node{m},\node{n}} \rcomp <:_{\node{n}}^*)$, as in this case
$$
{\leq:_{\node{m},\node{n}}}
=
(<:_{\node{m},\node{n}} \rcomp <:_\node{n}^*)
=
(<:_{\node{m},\node{n}} \rcomp <:_\node{n}^* \rcomp <:_\node{n}^*)
=
(\leq:_{\node{m},\node{n}} \rcomp <:_\node{n}^*).
$$
From the definition of $<:_{\node{m},\node{n}}$ (cf.\/ Definition~\ref{def:extended_path_ordering}) we have that if $\node{m} = \node{n}$ then $<:_{\node{m},\node{n}} = I_{\seqst(\node{n})}$ is the identity relation over $\seqst(\node{n})$.  From this fact we can make the following observations.  First, if $\node{m} = \node{n}$ and $\node{n}$ is not a companion node, then ${\leq:_{\node{m},\node{n}}} = {<:_{\node{m},\node{n}}}$ is the identity relation over $\seqst(\node{n})$.
Second, if $\node{m} = \node{n}$ and $\node{n}$ is a companion node then $\leq:_{\node{m},\node{n}}$ is ${<:_{\node{n}}^*}$, the reflexive and transitive
closure of the companion node ordering for $\node{n}$.
.
The next lemma establishes a technical property, derived from the definition of $<:_{\node{m},\node{n}}$, that is satisfied by $\leq:_{\node{m},\node{n}}$. It is used later to show that $\leq:_{\node{m},\node{n}}$ obeys a pseudo-transitivity law.

\begin{lemma}[Characterization of $\leq:_{\node{m},\node{n}}$]
\label{lem:support-dependency-ordering-characterization}
Let $\node{n}_1, \node{n}_2$ and $\node{n}_3$ be proof nodes in $\mathbb{T}$, with $\node{n}_2$ a companion node and $\node{n}_3 \neq \node{n}_2$, and let $s_1, s_2, s_2'$ and $s_3$ satisfy the following.
\begin{enumerate}
    \item $s_2 \leq:_{\node{n}_2,\node{n}_1} s_1$
    \item $s_2' <:_{\node{n}_2}^+ s_2$
    \item $s_3 \leq:_{\node{n}_3,\node{n}_2} s_2'$
\end{enumerate}
Then $s_3 \leq:_{\node{n}_3,\node{n}_1} s_1$.
\end{lemma}
\remove{
\begin{proofsketch}
Follows from the definitions of $\leq:_{\node{m},\node{n}}$ and $<:_{\node{m},\node{n}}$ and the pseudo-transitivity (Lemma~\ref{lem:join_extended_path_ordering}) of $<:_{\node{m},\node{n}}$.
The detailed proof is included in the appendix \qedhere
\end{proofsketch}
}
\begin{proof}
Fix nodes $\node{n}_1, \node{n}_2$ and $\node{n}_3$ and states $s_1, s_2, s_2'$ and $s_3$ satisfying the conditions in the statement of the lemma.  We must show that $s_3 \leq:_{\node{n}_3,\node{n}_1} s_1$.  There are two cases to consider.
\begin{description}
\item[$\node{n}_2 = \node{n}_1$.]
    In this case $\node{n}_1 = \node{n}_2$ is a companion node, and from the observations above we know that ${\leq:_{\node{n}_2,\node{n}_1}} = {<:_{\node{n}_1}^*} = {<:_{\node{n}_2}^*}$.  Thus 
    $s_2' <:_{\node{n}_2}^+ s_2 <:_{\node{n}_2}^* s_1$, meaning $s_2' <:_{\node{n}_2}^* s_1$.  As $s_3 \leq:_{\node{n}_3,\node{n}_2} s_2' <:_{\node{n}_2}^* s_1$ and ${\leq:_{\node{n}_3,\node{n}_2}} = {\leq:_{\node{n}_3,\node{n}_2}} \rcomp {<:_{\node{n}_2}^*}$, we can conclude that $s_3 \leq:_{\node{n}_3,\node{n}_2} s_1$ and thus, since $\node{n}_2 = \node{n}_1$, $s_3 \leq:_{\node{n}_3,\node{n}_1} s_1$.
\item[$\node{n}_2 \neq \node{n}_1$.]
    In this case there exists $s_1'$ such that $s_2 <:_{\node{n}_2,\node{n}_1} s_1'$ and $s_1' <:_{\node{n}_1}^* s_1$.  It suffices to establish that $s_3 <:_{\node{n}_3,\node{n}_1} s_1'$, as the definition of $\leq:_{\node{n}_3,\node{n}_1}$ then guarantees that $s_3 \leq:_{\node{n}_3,\node{n}_1} s_1$.
    The proof proceeds by induction on the definition of $<:_{\node{n}_2,\node{n}_1}$.
    In the base case, $s_2 \lessdot_{\node{n}_2,\node{n}_1} s_1'$.  Since $s_3 \leq:_{\node{n}_3,\node{n}_2} s_2'$ there must exist $s_2''$ such that $s_3 <:_{\node{n}_3,\node{n}_2} s_2''$ and $s_2'' <:_{\node{n}_2}^* s_2'$.  Since $s_2' <:_{\node{n}_2}^+ s_2$ we have that $s_2'' <:_{\node{n}_2}^+ s_2$, and the definition of $<:_{\node{n}_3,\node{n}_1}$ gives $s_3 <:_{\node{n}_3,\node{n}_1} s_1'$.

    In the induction step, there exist node $\node{n}_1'$ such that $\node{n}_1' \neq \node{n}_1$, $\node{n}_1' \neq \node{n}_2$, and states $t_1, t_1'$ such that $t_1 \lessdot_{\node{n}_1', \node{n}_1} s_1'$, $t_1' <:_{\node{n}_1'} t_1$, and $s_2 <:_{\node{n}_2,\node{n}_1'} t_1'$.  The induction hypothesis guarantees that $s_3 <:_{\node{n}_3,\node{n}_1'} t_1'$.  The definition of $<:_{\node{n}_3,\node{n}_1}$ now guarantees that $s_3 <:_{\node{n}_3,\node{n}_1} s_1'$.
    \qedhere
\end{description}
\end{proof}

\noindent
Relation $\leq:_{\node{m},\node{n}}$ also enjoys a pseudo-transitivity property.
\begin{lemma}[Pseudo-transitivity of $\leq:_{\node{m},\node{n}}$]
\label{lem:pseudo-transitivity-of-support-dependency-ordering}
    Let $\node{n}_1, \node{n}_2$ and $\node{n}_3$ be proof nodes in partial tableau $\mathbb{T}$, and assume $s_1, s_2$ and $s_3$ are such that $s_3 \leq:_{\node{n}_3, \node{n}_2} s_2$ and $s_2 \leq:_{\node{n}_2, \node{n}_1} s_1$.
    Then $s_3 \leq:_{\node{n}_3,\node{n}_1} s_1$.
\end{lemma}
\remove{
\begin{proofsketch}
Follows from the pseudo-transitivity of $<:_{\node{m},\node{n}}$ (Lemma~\ref{lem:join_extended_path_ordering}) and the preceding observations.
The detailed proof is included in the appendix.
\qedhere
\end{proofsketch}
}
\begin{proof}
Suppose that $\node{n}_1, \node{n}_2$ and $\node{n}_3$, and $s_1, s_2$ and $s_3$, are such that $s_3 \leq:_{\node{n}_3,\node{n}_2} s_2$ and $s_2 \leq:_{\node{n}_2,\node{n}_1} s_1$.  We must show that $s_3 \leq:_{\node{n}_3,\node{n}_1} s_1$.  There are two cases to consider.
\begin{description}
\item[$s_2 <:_{\node{n}_2,\node{n}_1} s_1$.]
    We consider two sub-cases.  
    In the first, $s_3 <:_{\node{n}_3,\node{n}_2} s_2$; the pseudo-transitivity of $<:_{\node{m},\node{n}}$ immediately implies that $s_3 <:_{\node{n}_3,\node{n}_1} s_1$, so $s_3 \leq:_{\node{n}_3,\node{n}_1} s_1$.
    In the second sub-case, $s_3 \centernot{<:}_{\node{n}_3,\node{n}_2} s_2$.  
    As $s_3 \leq:_{\node{n}_3,\node{n}_2} s_2$ it therefore must be the case that $\node{n}_2$ is a companion node and that
    $$s_3 \,(<:_{\node{n}_3, \node{n}_2} \rcomp <:_{\node{n}_2}^*))\, s_2,$$
    meaning
    that there exists $s_2'$ such that 
    \begin{align*}
        s_3 &<:_{\node{n}_3, \node{n}_2} s_2' \\
        s_2' &<:_{\node{n}_2}^* s_2.
    \end{align*}
    If $\node{n}_2 = \node{n}_1$ then $<:_{\node{n}_2,\node{n}_1}$ is the identity relation, and thus $s_2 = s_1$.  This fact, and the fact that $\node{n}_2 = \node{n}_1$, ensures that $s_3 \leq_{\node{n}_3,\node{n}_1} s_1$.
    If $\node{n}_2 \neq \node{n}_1$ and $\node{n}_2 = \node{n}_3$, it also follows that ${<:_{\node{n}_3,\node{n}_2}}$ is the identity relation, and thus $s_3 = s_2'$.  We again have that $s_3 \leq:_{\node{n}_3,\node{n}_1} s_1$.
    Finally, if $\node{n}_2 \neq \node{n}_1$ and $\node{n}_2 \neq \node{n}_3$ then Lemma~\ref{lem:support-dependency-ordering-characterization} gives the desired result.
\item[$s_2 \centernot{<:}_{\node{n}_2,\node{n}_1} s_1$.]
    In this case, it must be that $\node{n}_1$ is a companion node and that
    there exists $s_1'$ such that $s_2 <:_{\node{n_2},\node{n_1}} s_1'$ and $s_1' <:_{\node{n}_1}^* s_1$.  From the previous case we know that $s_3 \leq:_{\node{n}_3,\node{n}_1} s_1'$; the the fact that $\leq:_{\node{n}_3,\node{n}_1} = (\leq:_{\node{n}_3,\node{n}_1} \rcomp <:_{\node{n}_1}^*)$ immediately guarantees that $s_3 \leq:_{\node{n}_3,\node{n}_1} s_1$.
    \qedhere
\end{description}
\end{proof}

In the remainder of this section we wish to establish that for any companion node $\node{n}$ in successful tableau $\mathbb{T}$, $<:_{\node{n}}^+$ is a support ordering for a function derived from $P(\node{n})$.  In order to define this function we must deal with the free variables embedded in $P(\node{n})$.  In particular, if $\node{m}$ is a companion node that is a strict ancestor of $\node{n}$ then variable $Z_{\node{m}}$ may appear free in $P(\node{n})$; this would be the case if any of the companion leaves of $\node{m}$ are also descendants of $\node{n}$. To accommodate these free variables in $P(\node{n})$ we will define a modification of valuation $\V$ that assigns sets of states to these variables based on the $<:_{\node{m}', \node{n}}$ relation, where $\node{m}'$ is a companion leaf of $\node{m}$ that is also a descendant of $\node{n}$. 

\begin{definition}[Influence extensions of valuations]\label{def:support-extension-of-valuation}
Let $\mathbb{T} = \tableauTrl$ be a tableau, with $\node{m}_1 \cdots \node{m}_k$ an ordering on the companion nodes $\cnodes{\mathbb{T}}$ of $\mathbb{T}$.
Also let $\node{n}$ be a node in $\mathbb{T}$, with $S = \seqst(\node{n})$ the states in $\node{n}$.  We define the following.
\begin{enumerate}
\item 
    $\cleaves{\node{m}_i,\node{n}} = \cleaves{\mathbb{T}}(\node{m}_i) \cap D(\node{n})$ is the set of companion leaves of $\node{m}_i$ that are also descendants of $\node{n}$.
\item
    The set of states in companion leaves of $\node{m}_i$ that influence state $s$ in $\node{n}$ is given as follows.
    \[
    S_{\node{n},s,\node{m}_i} = \bigcup_{\node{m}' \in \cleaves{\node{m}_i, \node{n}}} \preimg{(\leq:_{\node{m}',\node{n}})}{s}
    \]
    We also define $$S_{\node{n},\node{m}_i} = \bigcup_{s \in S} S_{\node{n},s,\node{m}_i}$$ to be the set of states in companion leaves of $\node{m}_i$ that influence $\node{n}$.
\item
    The \emph{influence extension} of $\V$ for state $s$ in node $\node{n}$ is defined as
    \[
    \V_{\node{n},s} = \V[Z_{\node{m}_1} \cdots Z_{\node{m}_k} = S_{\node{n},s,\node{m}_1} \cdots S_{\node{n},s,\node{m}_k}].
    \]
    Similarly
    \[
    \V_\node{n} = \V[Z_{\node{m}_1} \cdots Z_{\node{m}_k} = S_{\node{n},\node{m}_1} \cdots S_{\node{n},\node{m}_k}]
    \]
    is the influence extension of $\V$ for node $\node{n}$.
\end{enumerate}
\end{definition}


Intuitively, $S_{\node{n},s,\node{m}_i}$ contains all the states in the companion leaves of $\node{m}_i$ at or below node $\node{n}$ that influence the determination that state $s$ belongs in node $\node{n}$.
Note these definitions also set $\V_{\node{n},s}(Z_{\node{m}_i}) = \emptyset$ in case node $\node{m}_i$ has no companion leaves that are descendants of $\node{n}$; when this happens $Z_{\node{m}_i}$ cannot appear free in $P(\node{n})$.  Also note that $Z_{\node{m}_i}$ does not appear free in $P(\node{n})$ if $\node{m}_i$ is a descendant of $\node{n}$, as $P(\node{m}_i) = \sigma Z_{\node{m}_i}.\Phi'$ for some $\Phi'$ is a subformula of $P(\node{n})$ and contains all occurrences of $Z_{\node{m}_i}$ in $P(\node{n})$. In both cases the value assigned to $Z_{\node{m}_i}$ by $\V_{\node{n},s}$ does not affect the semantics of $P(\node{n})$.

We now state a technical but useful lemma about dependency extensions.
\begin{lemma}[Monotonicity of extensions]
\label{lem:monotonicity-of-dependency-extensions}
Let $\mathbb{T} = \tableauTrl$ be a tableau with nodes $\node{n}$ and $\node{n}'$ and states $s$ and $s'$ such that $s' <_{\node{n}',\node{n}} s$.  Then:
\begin{enumerate}
\item 
    for all $Z \in \Var_{\mathbb{T}}, \V_{\node{n}',s'}(Z) \subseteq \V_{\node{n},s}(Z)$, and 
\item
    for all $Z \in \Var \setminus \Var_{\mathbb{T}}, \V_{\node{n}',s'}(Z) = \V_{\node{n},s}(Z)$.
\end{enumerate}
\end{lemma}
\begin{proof}
Follows from the definition of $\V_{\node{n},s}$ and the fact that the definition of $<_{\node{n}',\node{n}}$ ensures that $\node{n}' \in c(\node{n})$ and thus $D(\node{n}') \subseteq D(\node{n})$.  Consequently $\cleaves{\node{m}_i,\node{n}'} \subseteq \cleaves{\node{m}_i,\node{n}}$ for all $\node{m}_i \in \cnodes{\mathbb{T}}$, and $s' <_{\node{n}',\node{n}} s$ guarantees that $S_{\node{n}',s',\node{m}_i} \subseteq S_{\node{n},s,\node{m}_i}$. 
\end{proof}

The next corollary is an immediate consequence of this lemma.
\begin{corollary}
\label{cor:monotonicity-of-dependency-extensions}
Let $\mathbb{T} = \tableauTrl$ be a tableau, with $\node{n}$ and $\node{n}'$ and states $s$ and $s'$ such that $s' <_{\node{n}',\node{n}} s$.  Then $\semT{P(\node{n})}{\V_{\node{n}',s'}} \subseteq \semT{P(\node{n})}{\V_{\node{n}',s}}$.
\end{corollary}
\begin{proof}
Follows from Lemma~\ref{lem:monotonicity-of-dependency-extensions} and the fact that every occurrence of any $Z_{\node{m}} \in \Var_{\mathbb{T}}$ in $P(\node{n}')$ must be positive.
\end{proof}

We now state and prove the main lemma of this section, which is that for any companion node $\node{n}$ in a successful tableau, $<:_{\node{n}}^+$ is a support ordering for a semantic function derived from $P(\node{n})$.

\begin{lemma}[$<:_{\node{n}}^+$ is a support ordering]
\label{lem:support-ordering-for-companion-nodes}
Let $\mathbb{T} = \tableauTrl$ be a successful tableau, with $\node{n} \in \cnodes{\mathbb{T}}$ a companion node of $\mathbb{T}$ and $\node{n}'$ the child of $\node{n}$ in $\tree{T}$.  
Also let $S = \seqst(\node{n})$.  
Then $(S, <:_{\node{n}}^+)$ is a support ordering for
$\semfT{Z_\node{n}}{P(\node{n}')}{\V_{\node{n}}}$.
\end{lemma}
\begin{proof}
Fix successful tableau $\mathbb{T} = \tableauTrl$, with $\tree{T} = (\node{N},\node{r},p,cs)$, and let $\node{n} \in \cnodes{\mathbb{T}}$ be a companion node of $\mathbb{T}$ with $S = \seqst(\node{n})$.  We prove the following stronger result.
\begin{quote}
For every $\node{m} \in D(\node{n})$ and $s \in S$ statements \ref{stmt:necessity} and \ref{stmt:support} hold.
\begin{enumerate}[left=\parindent, label=S\arabic*., ref=S\arabic*]
\item\label{stmt:necessity}
    For all $x$ such that $x \leq:_{\node{m},\node{n}} s$, 
    $x \in \semT{P(\node{m})}{\V_{\node{m},x}}$.
\item\label{stmt:support}
    If $\node{m} \in \cnodes{\mathbb{T}}$,
    $\node{m}' = cs(\node{m})$
    and
    $x$ satisfies $x \leq:_{\node{m},\node{n}} s$
    then 
    $(S_x, <:_{\node{m},x})$ is a support ordering for $\semfT{Z_{\node{m}}}{P(\node{m}')}{\V_{\node{m},x}}$, where
    $S_x = \preimg{(<:_{\node{m}}^*)}{x}$
    and
    ${<:_{\node{m},x}} = \restrict{(<:_{\node{m}}^+)}{S_x}$.
\end{enumerate}
\end{quote}
To see that proving \ref{stmt:necessity} and~\ref{stmt:support} for all $\node{m} \in D(\node{n})$ and $s \in S$ establishes the lemma, first note that as $\node{n} \in \cnodes{\mathbb{T}}$ and $\node{n} \in D(\node{n})$, such a proof would imply that for all $x \in S$, $(S_x, <:_{\node{n},x})$ is a support ordering for $\semfT{Z_{\node{n}}}{P(\node{n}')}{\V_{\node{n},x}}$. 
We also note that for all $x \in S$, 
$\V_{\node{n},x}(Z) \subseteq \V_{\node{n}}(Z)$ if $Z \in \Var_{\mathbb{T}}$ and
$\V_{\node{n},x}(Z) = \V_{\node{n}}(Z)$ otherwise;
these imply that for all $x \in S$ and $S' \subseteq \states{S}$, $\semfT{Z_{\node{n}}}{P(\node{n}')}{\V_{\node{n},x}}(S') \subseteq \semfT{Z_{\node{n}}}{P(\node{n}')}{\V_{\node{n}}}(S')$.  
Therefore, for all $x \in S$ $(S_x, <:_{\node{n}.x})$ is also a support ordering for $\semfT{Z_{\node{n}}}{P(\node{n}')}{\V_{\node{n}}}$.  
It is easy to see that $\bigcup_{x \in S} S_x = S$ and also that $\bigcup_{x \in S} {<:_{\node{n},x}} = {<:_{\node{n}}^+}$; 
Lemma~\ref{lem:unions-of-support-orderings} then guarantees that $(S, <:_{\node{n}}^+)$ is a support ordering for $\semfT{Z_{\node{n}}}{P(\node{n}')}{\V_{\node{n}}}$, which is what the lemma asserts.

We prove the stronger result by tree induction on $\tree{T}_{\node{n}}$, the subtree rooted at $\node{n}$ in $\tree{T}$.
(Recall that, per Definition~\ref{def:subtree}, the set of nodes in $\tree{T}_{\node{n}}$ is $D(\node{n})$.)
So pick $\node{m} \in D(\node{n})$; the induction hypothesis states that for all $s \in S$, \ref{stmt:necessity} and \ref{stmt:support} hold for all $\node{m}' \in c(\node{m})$.  We must show that for all $s \in S$, \ref{stmt:necessity} and ~\ref{stmt:support} hold for $\node{m}$.  The proof proceeds by a case analysis on $\rho(\node{m})$.  We first consider the cases in which $\rho(\node{m}) \neq \text{Un}$, meaning $\node{m} \not\in \cnodes{\mathbb{T}}$.  In each of these cases \ref{stmt:support} vacuously holds for all $s \in S$, so all that needs to be proved is \ref{stmt:necessity} for all $s \in S$.  So fix $s \in S$; the case analysis for $\rho(\node{n}) \neq \text{Un}$ is as follows.
\begin{description}
\item[$\rho(\node{m}) \div$.]
    In this case $\rho(\node{m})$ is undefined, and $\node{m}$ must be a leaf.  To prove \ref{stmt:necessity} holds for $s$ we first note that since $\mathbb{T}$ is successful it follows that $\node{m}$ is a successful leaf, meaning that either $\node{m}$ is a free leaf or a $\sigma$-leaf (successful tableaux cannot contain any diamond leaves, so this leaf type need not be considered).
    There are two subcases to consider.
    In the first, $\node{m}$ is a free leaf.
    In this case, by definition $P(\node{m}) = \seqfm(\node{m})$, where $\seqfm(\node{m})$ is either $Z$ or $\lnot Z$ for some $Z \in \Var \setminus \left(\Var_{\mathbb{T}} \cup \dom(\seqdl(\node{m}))\right)$.  
    For any such $Z$ and $x \in \seqst(\node{m})$ Lemma~\ref{lem:monotonicity-of-dependency-extensions} ensures that $\V_{\node{m},x}(Z) = \V(Z)$, whence
    \[
    \semT{P(\node{m})}{\V_{\node{m},x}} 
    = \semTV{\seqfm(\node{m})} 
    = \sem{\node{m}}{}{},
    \]
    and since $\node{m}$ is successful, we have that $x \in \semT{P(\node{m})}{\V_{\node{m},x}}$ for all $x \in \seqst(\node{m})$, and thus for all $x$ such that $x \leq:_{\node{m},\node{n}} s$.
    
    In the second subcase, $\node{m}$ is a $\sigma$-leaf, meaning it is a companion leaf for some $\node{m}_i \in \cnodes{\mathbb{T}}$.  Now fix $x$ such that $x \leq:_{\node{m},\node{n}} s$; we must show that $x \in \semT{P(\node{m})}{\V_{\node{m},x}}$. 
    In this case $P(\node{m}) = Z_{\node{m}_i}$, and $\semT{P(\node{m})}{\V_{\node{m},x}} = \V_{\node{m},x}(Z_{\node{m}_i}) = S_{\node{m},x,\node{m}_i}$. Since $\cleaves{\node{m}_i,\node{m}} = \{\node{m}\}$ in this case, $S_{\node{m},x,\node{m}_i} = \preimg{(\leq:_{\node{m},\node{m}})}{x} = \{x\}$.  As $x \in \{x\}$ the desired result holds.
\item[$\rho(\node{m}) = \land$.]
    In this case $\node{m} = S' \tnxTVD \Phi_1 \land \Phi_2$ for some $\Phi_1$ and $\Phi_2$, $cs(\node{m}) = \node{m}'_1 \node{m}'_2$, and $\node{m}'_i = S' \tnxTVD \Phi_i$ for $i = 1,2$. The induction hypothesis ensures that for all $s' \in S$, \ref{stmt:necessity} holds for each $\node{m}'_i$; we must show that \ref{stmt:necessity} holds for $\node{m}$ and $s$.
    To this end, let $x$ be such that $x \leq:_{\node{m},\node{n}} s$; we must show that $x \in \semT{P(\node{m})}{\V_{\node{m},x}}$.  Note that $x <_{\node{m}'_i, \node{m}} x$ for $i = 1,2$; the pseudo-transitivity of $\leq:_{\node{m},\node{n}}$ guarantees that $x \leq:_{\node{m}'_i,\node{n}} s$, and the induction hypothesis then ensures that $x \in \semT{P(\node{m}'_i)}{\V_{\node{m}'_i, x}}$ for $i = 1,2$.  Corollary~\ref{cor:monotonicity-of-dependency-extensions} guarantees that $x \in \semT{P(\node{m}'_i)}{\V_{\node{m},x}}$ for $i = 1,2$, and the semantics of $\land$ then ensures that $x \in \semT{P(\node{m}_1) \land P(\node{m}_2)}{\V_{\node{m},x}} = \semT{P(\node{m})}{\V_{\node{m},x}}$, thereby establishing \ref{stmt:necessity}.
\item[$\rho(\node{m}) = \lor$.]
    In this case $\node{m} = S' \tnxTVD \Phi_1 \lor \Phi_2$ for some $\Phi_1$ and $\Phi_2$, $cs(\node{m}) = \node{m}'_1 \node{m}'_2$, and $\node{m}'_i = S_i' \tnxTVD \Phi_i$ for $i = 1,2$, with $S' = S'_1 \cup S'_2$. The induction hypothesis ensures that for all $s' \in S$, \ref{stmt:necessity} holds for each $\node{m}'_i$; we must show that \ref{stmt:necessity} holds for $\node{m}$ and $s$.
    To this end, let $x$ be such that $x \leq:_{\node{m},\node{n}} s$; we must show that $x \in \semT{P(\node{m})}{\V_{\node{m},x}}$.  Note that $x <_{\node{m}'_i, \node{m}} x$ for at least one of $i = 1$ or $i = 2$; the pseudo-transitivity of $\leq:_{\node{m},\node{n}}$ guarantees that $x \leq:_{\node{m}'_i,\node{n}} s$, and the induction hypothesis then ensures that $x \in \semT{P(\node{m}'_i)}{\V_{\node{m}'_i, x}}$ for at least one of $i = 1$ or $i = 2$.  Corollary~\ref{cor:monotonicity-of-dependency-extensions} guarantees that $x \in \semT{P(\node{m}'_i)}{\V_{\node{m},x}}$ for at least one of $i = 1$ or $i = 2$, and the semantics of $\lor$ then ensures that $x \in \semT{P(\node{m}_1) \lor P(\node{m}_2)}{\V_{\node{m},x}} = \semT{P(\node{m})}{\V_{\node{m},x}}$, thereby establishing \ref{stmt:necessity}.
\item[$\rho(\node{m}) = [K{]}$.]
    In this case $\node{m} = S' \tnxTVD [K] \Phi$ for some $\Phi$, $cs(\node{m}) = \node{m}'$, and $\node{m}' = S'' \tnxTVD \Phi$, where $S'' = \{s'' \in \states{S} \mid \exists s' \in S'. s' \xrightarrow{K} s'' \}$. The induction hypothesis ensures that for all $s' \in S$, \ref{stmt:necessity} holds for $\node{m}'$; we must show that \ref{stmt:necessity} holds for $\node{m}$ and $s$.
    To this end, let $x$ be such that $x \leq:_{\node{m},\node{n}} s$; we must show that $x \in \semT{P(\node{m})}{\V_{\node{m},x}}$.  Note that $x' <_{\node{m}', \node{m}} x$ for each $x'$ such that $x \xrightarrow{K} x'$; the pseudo-transitivity of $\leq:_{\node{m},\node{n}}$ guarantees that each such $x'$ satisfies $x' \leq:_{\node{m}',\node{n}} s$, and the induction hypothesis then ensures that each such $x' \in \semT{P(\node{m}')}{\V_{\node{m}', x'}}$. Corollary~\ref{cor:monotonicity-of-dependency-extensions} guarantees that each such $x' \in \semT{P(\node{m}')}{\V_{\node{m},x}}$, and the semantics of $[K]$ then ensures that $x \in \semT{[K] P(\node{m}')}{\V_{\node{m},x}} = \semT{P(\node{m})}{\V_{\node{m},x}}$, thereby establishing \ref{stmt:necessity}.
\item[$\rho(\node{m}) = (\dia{K},f)$.]
    In this case $\node{m} = S' \tnxTVD \dia{K} \Phi$ for some $\Phi$, $cs(\node{m}) = \node{m}'$, $\node{m}' = f(S') \tnxTVD \Phi$, and $f \in S' \to \states{S}$ has the property that $s'' \xrightarrow{K} f(s'')$ for all $s'' \in S'$.
    The induction hypothesis ensures that for all $s' \in S$, \ref{stmt:necessity} holds for $\node{m}'$; we must show that \ref{stmt:necessity} holds for $\node{m}$ and $s$.
    To this end, let $x$ be such that $x \leq:_{\node{m},\node{n}} s$; we must show that $x \in \semT{P(\node{m})}{\V_{\node{m},x}}$.  Note that $f(x) <_{\node{m}', \node{m}} x$; the pseudo-transitivity of $\leq:_{\node{m},\node{n}}$ guarantees that $f(x)$ satisfies $f(x) \leq:_{\node{m}',\node{n}} s$, and the induction hypothesis then ensures that $f(x) \in \semT{P(\node{m}')}{\V_{\node{m}', f(x)}}$. Corollary~\ref{cor:monotonicity-of-dependency-extensions} guarantees that $f(x) \in \semT{P(\node{m}')}{\V_{\node{m},f(x)}}$, and the semantics of $\dia{K}$ then ensures that $x \in \semT{\dia{K} P(\node{m}')}{\V_{\node{m},x}} = \semT{P(\node{m})}{\V_{\node{m},x}}$, thereby establishing \ref{stmt:necessity}.
\item[$\rho(\node{m}) = \sigma Z$.]
    In this case $\node{m} = S' \tnxTVD \sigma Z.\Phi$ for some $\Phi$, $cs(\node{m}) = \node{m}'$, and $\node{m}' = S' \tnxTV{\Delta'} U$, where $\Delta' = \Delta \cdot (U = \sigma Z.\Phi)$.
    The induction hypothesis ensures that for all $s' \in S$, \ref{stmt:necessity} holds for $\node{m}'$; we must show that \ref{stmt:necessity} holds for $\node{m}$ and $s$.
    To this end, let $x$ be such that $x \leq:_{\node{m},\node{n}} s$; we must show that $x \in \semT{P(\node{m})}{\V_{\node{m},x}}$.  Note that $x <_{\node{m}', \node{m}} x$; the pseudo-transitivity of $\leq:_{\node{m},\node{n}}$ guarantees that $x$ satisfies $x \leq:_{\node{m}',\node{n}} s$, and the induction hypothesis then ensures that $x \in \semT{P(\node{m}')}{\V_{\node{m}', x}}$. Corollary~\ref{cor:monotonicity-of-dependency-extensions} guarantees that $x \in \semT{P(\node{m}')}{\V_{\node{m},x}}$, and the semantics of $U$ and $\sigma Z.\Phi$, and Lemma~\ref{lem:companion-node-formulas-and-semantics}, then ensure that $x \in \semT{P(\node{m})}{\V_{\node{m},x}}$, thereby establishing \ref{stmt:necessity}.
\item[$\rho(\node{m}) = \text{Thin}$.]
    In this case $\node{m} = S' \tnxTVD \Phi$ for some $\sigma, Z$ and $\Phi$, $cs(\node{m}) = \node{m}'$, and $\node{m}' = S'' \tnxTVD \Phi$, where $S' \subseteq S''$.
    The induction hypothesis ensures that for all $s' \in S$, \ref{stmt:necessity} holds for $\node{m}'$; we must show that \ref{stmt:necessity} holds for $\node{m}$ and $s$.
    To this end, let $x$ be such that $x \leq:_{\node{m},\node{n}} s$; we must show that $x \in \semT{P(\node{m})}{\V_{\node{m},x}}$.  Note that $x <_{\node{m}', \node{m}} x$; the pseudo-transitivity of $\leq:_{\node{m},\node{n}}$ guarantees that $x$ satisfies $x \leq:_{\node{m}',\node{n}} s$, and the induction hypothesis then ensures that $x \in \semT{P(\node{m}')}{\V_{\node{m}', x}}$. Corollary~\ref{cor:monotonicity-of-dependency-extensions} guarantees that $x \in \semT{P(\node{m}')}{\V_{\node{m},x}}$, and as $P(\node{m}) = P(\node{m}')$ in this case, \ref{stmt:necessity} holds.
\end{description}
The final case to be considered is $\rho(\node{m}) = \text{Un}$; in this case, $\node{m} \in \cnodes{\mathbb{T}}$, meaning
$\node{m} = X \tnxTVD U$ where
$U \in \dom(\Delta)$,
$\Delta(U) = \sigma Z.\Phi$,
$cs(\node{m}) = \node{m}'$ and
$\node{m}' = X \tnxTVD \Phi[Z := U]$.  
The induction hypothesis guarantees that \ref{stmt:necessity} and \ref{stmt:support} hold for all $s \in S$ for $\node{m}'$; we must show \ref{stmt:necessity} and \ref{stmt:support} hold for all $s \in S$ for $\node{m}$.
So fix $s \in S$.
We consider \ref{stmt:support} first.  
Let $x \leq:_{\node{m},\node{n}} s$, and define $f_{\node{m},x} = \semfT{Z_\node{m}}{P(\node{m}')}{\V_{\node{m},x}}$.  
We must show that $(S_x, <:_{\node{m},x})$ is a support ordering for $f_{\node{m},x}$.  
Following Definition~\ref{def:support-ordering} it suffices to prove that for every $x' \in S_x, x' \in f_{\node{m},x}(\preimg{(<:_{\node{m},x})}{x'})$.  
So fix $x' \in S_x$.  
By definition of $S_x$ this means that $x' <:_{\node{m}}^* x$.
Since $x' <_{\node{m}',\node{m}} x'$, it follows that $x' \leq:_{\node{m}',\node{m}} x'$ and, due to the pseudo-transitivity Lemma~\ref{lem:pseudo-transitivity-of-support-dependency-ordering}, that $x' \leq:_{\node{m}',\node{n}} s$.  
From the induction hypothesis, we know that \ref{stmt:necessity} holds for $\node{m}'$ and $x'$, meaning $x' \in \semT{P(\node{m}')}{\V_{\node{m}',x'}}$.
To complete this part of the proof it suffices to establish that
$\semT{P(\node{m}')}{\V_{\node{m}',x'}} \subseteq f_{\node{m},x}(\preimg{({<:_{\node{m},x}})}{x'})$.
We begin by noting that since $x' <_{\node{m}',\node{m}} x'$ and all occurrences of any $Z \in \Var_{\mathbb{T}}$ in $P(\node{m}')$ are positive, Lemma~\ref{lem:monotonicity-of-dependency-extensions}
ensures that $\semT{P(\node{m}')}{\V_{\node{m}',x'}} \subseteq \semT{P(\node{m}')}{\V_{\node{m},x'}}$.
It therefore suffices to show that $\semT{P(\node{m}')}{\V_{\node{m},x'}} \subseteq f_{\node{m},x}(\preimg{({<:_{\node{m},x}})}{x'})$.
From the definition of $f_{\node{m},x}$ we have that
\[
f_{\node{m},x}(\preimg{({<:_{\node{m},x}})}{x'})
=
\semT{P(\node{m}')}{\V_{\node{m},x}[Z_\node{m} := \preimg{(<:_{\node{m},x})}{x'}]}.
\]
Because every $Z \in \Var_{\mathbb{T}}$ appearing in $P(\node{m}')$ appears only positively, to establish that $\semT{P(\node{m}')}{\V_{\node{m},x'}} \subseteq f_{\node{m},x}(\preimg{({<:_{\node{m},x}})}{x'})$
it suffices to show the following.
\begin{enumerate}
\item\label{prop:eq}
    For all $Z \in \Var \setminus \Var_{\mathbb{T}}$, 
    $\V_{\node{m},x'}(Z) = \left(\V_{\node{m},x}[Z_\node{m} := \preimg{(<:_{\node{m},x})}{x'}] \right) (Z)$.
\item\label{prop:subset}
    For all $Z \in \Var_{\mathbb{T}}$,
    $\V_{\node{m},x'}(Z) \subseteq \left(\V_{\node{m},x}[Z_\node{m} := \preimg{(<:_{\node{m},x})}{x'}]\right) (Z)$.
\end{enumerate}
Property~\ref{prop:eq} follows immediately from the fact that for all $Z \in \Var \setminus \Var_{\mathbb{T}}$,
\[
\V(Z)
= \V_{\node{m},x'}(Z)
= \V_{\node{m},x}(Z)
= \left(\V_{\node{m},x}[Z_\node{m} := \preimg{(<:_{\node{m},x})}{x'}]\right) (Z).
\]
To establish Property~\ref{prop:subset}, fix $Z \in \Var_{\mathbb{T}}$.  There are two sub-cases to consider.  In the first, $Z = Z_{\node{m}}$.  In this case 
\[
\left(\V_{\node{m},x}[Z_\node{m} := \preimg{(<:_{\node{m},x})}{x'}]\right) (Z)
= \preimg{(<:_{\node{m},x})}{x'}
= \preimg{(<:_\node{m}^+)}{x'};
\]
we must show that $\V_{\node{m},x'} (Z) \subseteq \preimg{(<:_\node{m}^+)}{x'}$, i.e.\/ that $x'' <:_\node{m}^+ x'$ for all $x'' \in \V_{\node{m},x'} (Z)$.  So fix $x'' \in \V_{\node{m},x'} (Z)$.
From Definition~\ref{def:support-extension-of-valuation}, 
\[
\V_{\node{m},x'} (Z) = \bigcup_{\node{m}'' \in \cleaves{\mathbb{T}}(\node{m})} \preimg{(\leq:_{\node{m}'',\node{m}})}{x'}.
\]
Thus, there must exist leaf $\node{m}'' \in \cleaves{\mathbb{T}}(\node{m})$ such that $x'' \leq:_{\node{m}'',\node{m}} x'$. 
From the definition of $\leq:_{\node{m}'',\node{m}}$ it follows that there is $y$ such that $x'' <:_{\node{m}'',\node{m}} y$ and $y <:_{\node{m}}^* x'$; since $\node{m}''$ is a companion leaf of $\node{m}$ we also have that $x'' <:_{\node{m}} y$, whence $x'' <:_{\node{m}}^+ x'$.

In the second subcase, $Z \in \Var_{\mathbb{T}}$ but $Z \neq Z_{\node{m}}$.  In this case, we know that $\left(\V_{\node{m},x}[Z_\node{m} := \preimg{(<:_{\node{m},x})}{x'}]\right) (Z) = \V_{\node{m},x} (Z)$; we must therefore show that $\V_{\node{m},x'}(Z) \subseteq \V_{\node{m},x}(Z)$.  Since $Z \in \Var_{\mathbb{T}}$ it follows that there is a companion node $\node{k} \in \cnodes{\mathbb{T}}$ such that $Z = Z_{\node{k}}$ and the following hold.
\begin{align*}
\V_{\node{m},x'} (Z)
&= \bigcup_{\node{k}' \in \cleaves{\node{k},\node{m}}} \preimg{(\leq:_{\node{k}',\node{m}})}{x'}
\\
\V_{\node{m},x} (Z)
&= \bigcup_{\node{k}' \in \cleaves{\node{k},\node{m}}} \preimg{(\leq:_{\node{k}',\node{m}})}{x}
\end{align*}
Now assume $x'' \in \V_{\node{m},x'} (Z)$; we must show $x'' \in \V_{\node{m},x} (Z)$.  Since $x'' \in \V_{\node{m},x'} (Z)$ there must exist leaf node $\node{k}' \in \cleaves{\node{k},\node{m}}$ such that $x'' \leq:_{\node{k}', \node{m}} x'$; if we can show that $x'' \leq:_{\node{k}',\node{m}} x$ then we will have established that $x'' \in \V_{\node{m},x} (Z)$.  Recall that $x' <:_{\node{m}}^* x$; this implies that $x' \leq:_{\node{m},\node{m}} x$, and the pseudo-transitivity of $\leq:_{\node{k}',\node{m}}$ then ensures that $x'' \leq:_{\node{k}',\node{m}} x$.

To finish the proof we now need to show that statement~\ref{stmt:necessity} holds for companion node $\node{m}$ and $s \in S$.  So fix $x$ such that $x \leq:_{\node{m},\node{n}} s$; we must show that $x \in \semT{P(\node{m})}{\V_{\node{m},x}}$.  We know from the definition of $P$ that $P(\node{m}) = \sigma Z_\node{m}.P(\node{m}')$.  
If $\sigma = \nu$ then as statement~\ref{stmt:support} holds we know that $(S_x, <:_{\node{m},x})$ is a support ordering for $\semfT{Z_{\node{m}}}{P(\node{m}')}{\V_{\node{m},x}}$.  Therefore, $S_x$ is supported for 
$\semfT{Z_{\node{m}}}{P(\node{m}')}{\V_{\node{m},x}}$, and Corollary~\ref{cor:greatest-fixpoint} ensures that $S_x \subseteq \nu \left(\semfT{Z_{\node{m}}}{P(\node{m}')}{\V_{\node{m},x}}\right)$.  As $x \in S_x$ and 
\[\nu \left(\semfT{Z_{\node{m}}}{P(\node{m}')}{\V_{\node{m},x}}\right) = \semT{P(\node{m})}{\V_{\node{m},x}},
\]
the result follows.  
Now assume that $\sigma = \mu$.  In this case, since $<:_{\node{m}}$ is well-founded it follows from Lemma~\ref{lem:transitive-closure-well-founded} that $<:_{\node{m}}^+$ is as well, as is $<:_{\node{m},x}$.  Therefore $S_x$ is well-supported for $\semfT{Z_{\node{m}}}{P(\node{m}')}{\V_{\node{m},x}}$, and Corollary~\ref{cor:least-fixpoint} guarantees that $S_x \subseteq \mu \left(\semfT{Z_{\node{m}}}{P(\node{m}')}{\V_{\node{m},x}}\right)$.  As $x \in S_x$ and $\mu \left(\semfT{Z_{\node{m}}}{P(\node{m}')}{\V_{\node{m},x}}\right) = \semT{P(\node{m})}{\V_{\node{m},x}}$, the result follows.
\qedhere
\end{proof}

\noindent
The next corollary specializes the previous lemma to the case of \emph{top-level companion nodes} in successful tableau $\mathbb{T}$.  A companion node $\node{n} \in \cnodes{\mathbb{T}}$ is top-level iff $A_s(\node{n}) \cap \cnodes{\mathbb{T}} = \emptyset$; recalling that $A_s(\node{n})$ is the set of strict ancestors of $\node{n}$ in $\mathbb{T}$, a companion node is top-level iff it has no strict ancestors that are companion nodes.  It is straightforward to see that if $\node{n}$ is a top-level companion node, then $\node{n} = S \tnxTV{(U = \sigma Z.\Phi)} U$ for some $\sigma$, $Z$, $\Phi$ and $U$; note that the definition list of $\node{n}$ contains only the single element $(U = \sigma Z.\Phi)$.  We have the following.

\begin{corollary}[Support orderings for top-level companion nodes]
\label{cor:support-orderings-for-top-level-companion-nodes}
Let $\mathbb{T} = \tableauTrl$ be a successful tableau, with $\node{n} \in \cnodes{\mathbb{T}}$ a top-level companion of $\mathbb{T}$ and $\node{n}'$ the child of $\node{n}$ in $\tree{T}$.  Also let $S = \seqst(\node{n})$.  Then $(S, <:_{\node{n}}^+)$ is a support ordering for $\semfTV{Z_{\node{n}}}{P(\node{n}')}$.
\end{corollary}

\begin{proof}
Follows from the fact that since $\node{n}$ is top-level, the only variable in $\Var_{\mathbb{T}}$ that can appear free in $P(\node{n}')$ is $Z_\node{n}$.  Lemma~\ref{cor:monotonicity-of-dependency-extensions} thus guarantees that for any $S'$, $\semfT{Z_\node{n}}{P(\node{n}')}{\V_{\node{n}}}(S') = \semfTV{Z_\node{n}}{P(\node{n}')}(S')$. Lemma~\ref{lem:support-ordering-for-companion-nodes} then establishes the corollary.
\end{proof}

\subsection{Soundness}\label{subsec:soundness}

We now prove that our proof system is sound by establishing that the root sequent of every successful tableau is valid.

\begin{theorem}[Soundness of mu-calculus proof system]\label{thm:soundness}
Fix LTS $(\states{S},\to)$ of sort $\Sigma$ and valuation $\V$, and let $\mathbb{T} = \tableauTrl$ be a successful tableau for sequent $\seq{s} \in \Seq{\T}{\Var}$, where $\seqdl(\seq{s}) = \emptyL$.  Then $\seq{s}$ is valid.
\end{theorem}
\begin{proof}
Let $\tree{T} = (\node{N}, \node{r}, p, cs)$ be the tree component of $\mathbb{T}$, and define $\node{L} \subseteq \node{N}$ as follows.
\[
\node{L} = 
\{ \node{n} \in \node{N} \mid \seqdl(\node{n}) = \emptyL \land \rho(\node{n}) = \sigma Z\}
\]
Now consider the tree prefix $\tpre{\tree{T}}{\node{L}}$ of $\tree{T}$ (cf.\/ Definition~\ref{def:tree-prefix}).  It can be seen that $\tpre{\tree{T}}{\node{L}} = (\node{N'},\node{r},p',cs')$ is such that $\node{N}'$ contains precisely the nodes of $\tree{T}$ for which $\seqdl(\node{n}) = \emptyL$.  Moreover, each leaf $\node{n}$ of $\tpre{\tree{T}}{\node{L}}$ is either a leaf of $\tree{T}$ or has the property that $\node{n} = S \tnxTV{\emptyL} \sigma Z.\Phi$ for some $S$ and $\Phi$ and that the child $\node{n}'$ of $\node{n}$ in $\tree{T}$ is such that $\node{n}' = S \tnxTV{(U = \sigma Z.\Phi)} U$ is a top-level companion node.  We will show that each leaf $\node{n}$ of $\tpre{\tree{T}}{\node{L}}$ is valid; this fact, and Lemma~\ref{lem:local-soundness}, can be used as the basis for a simple inductive argument on $\tpre{\tree{T}}{\node{L}}$ to establish that every node in $\tpre{\tree{T}}{\node{L}}$ is valid, including root node $\node{r}$, whose sequent label is $\seq{s}$.  The result follows from the definition of node validity, which says that a node in a tableau is valid iff its sequent label is valid.

So fix leaf $\node{n}$ in $\tpre{\tree{T}}{\node{L}}$.  There are two cases to consider.
In the first, $\node{n}$ is also a leaf in $\tree{T}$.  In this case, since $\mathbb{T}$ is successful, $\node{n}$ is successful, and therefore valid.  
In the second case, $\node{n} = S \tnxTV{\emptyL} \sigma Z.\Phi$ and has a single child $\node{n}' = S \tnxTV{(U = \sigma Z.\Phi)} U$ that is a top-level companion node.
Let $\node{n}''$ be the child of $\node{n}'$.
Corollary~\ref{cor:support-orderings-for-top-level-companion-nodes} guarantees that $(S, <:_{\node{n}'}^+)$ is a support ordering for $\semfTV{Z_{\node{n}'}}{P(\node{n}'')}$.
We will now show that $S \subseteq \semTV{P(\node{n}')}$.
There are two sub-cases to consider.
In the first, $\sigma = \nu$.
It follows from the definitions that $S$ is supported for $\semfTV{Z_{\node{n}'}}{P(\node{n}'')}$ and thus by Corollary~\ref{cor:greatest-fixpoint}, $S \subseteq \nu (\semfTV{Z_{\node{n}'}}{P(\node{n}'')}) = \semTV{P(\node{n}')}$.
In the second, $\sigma = \mu$.  Since $\mathbb{T}$ is successful $<:_{\node{n}'}$ is well-founded, meaning that $<:_{\node{n}'}^+$ is also well-founded.  Thus $S$ is well-supported for $\semfTV{Z_{\node{n}'}}{P(\node{n}'')}$ and thus by Corollary~\ref{cor:least-fixpoint}, $S \subseteq \mu (\semfTV{Z_{\node{n}'}}{P(\node{n}'')}) = \semTV{P(\node{n}')}$.
Since
$P(\node{n}) = P(\node{n}')$, $S \subseteq \semTV{P(\node{n})}$.  Also note that since $\node{n}'$ is a top-level companion node, $P(\node{n})$ can contain no free occurrences of any $Z' \in \Var_{\mathbb{T}}$, meaning that for any $\V'$ consistent with $\mathbb{T}$, $\semTV{P(\node{n})} = \semT{P(\node{n})}{\V'}$.  Consequently, Lemma~\ref{lem:companion-node-formulas-and-semantics} implies that $\semTV{P(\node{n})} = \sem{\node{n}}{}{}$, and thus $S \subseteq \sem{\node{n}}{}{}$, whence $\node{n}$ is valid.
\qedhere
\end{proof}

\section{Completeness}\label{sec:Completeness}

This section now establishes the completeness of our proof system.  Call a tableau $\tableauTrl$, where $\tree{T} = (\node{N}, \node{r}, p, cs)$, \emph{successful for} sequent $S \tnxTVD \Phi$ iff it is successful and $\node{r} = S \tnxTVD \Phi$.  We show that for any $\T$, $\V$, $S$ and $\Phi$, if $S \tnxTV{\emptyL} \Phi$ is valid then there is a successful tableau for $S \tnxTV{\emptyL} \Phi$.

The completeness results in this section rely heavily on tableau manipulations; in particular, several proofs define constructions for merging multiple successful tableaux into a single successful tableau.  These  constructions in turn rely on variations of well-founded induction over support orderings for the semantic functions used to give meaning to fixpoint formulas, and become subtle in the setting of mutually recursive fixpoints.  To clarify and simplify these arguments, the first subsection below  introduces relevant notions from general fixpoint theory in the setting of mutual recursion.  These results are then used later in this section to define the tableau constructions we need to establish completeness.

\subsection{Mutual recursion and fixpoints}
\label{subsec:mutual-recursion}

Mu-calculus formulas of form $\sigma Z.(\cdots \sigma' Z'.( \cdots Z \cdots ) \cdots)$ are said to be \emph{mutually recursive}, because of the semantics of the outer fixpoint formula, $\sigma Z. \cdots$, depends on the semantics of the inner fixpoint formula, $\sigma' Z'. \cdots$, which in turn (because $Z$ is free in its body) depends on the semantics of the outer fixpoint.  If $\sigma \neq \sigma'$ then these mutually recursive fixpoints are also said to be \emph{alternating}.  Alternating fixpoints present challenges when reasoning about completeness; in support of the constructions to come, in this section we develop a theory of mutually recursive fixpoints in the general setting of recursive functions over complete lattices. In particular, we show how to define mutually recursive fixpoints in terms of binary functions, and we define a property of binary relations, which we call \emph{quotient well-foundedness}, that can be applied to support orderings for mutually recursive fixpoints in order to support a form of well-founded induction. 

\paragraph{Mutual recursion.}

Let $S$ be a set, with $(2^S, \subseteq, \bigcup, \bigcap)$ the subset lattice over $S$ (cf.\/ Definition~\ref{def:subset-lattice}).  To define mutually recursive fixpoints in this setting we use \emph{binary} monotonic functions over $2^S$ defined as follows.

\begin{definition}[Monotonic binary functions]
Binary function $f \in 2^S \times 2^S \rightarrow 2^S$ is \emph{monotonic} iff for all $X_1, X_2, Y_1, Y_2 \in 2^S$, if $X_1 \subseteq X_2$ and $Y_1 \subseteq Y_2$ then $f(X_1, Y_1) \subseteq f(X_2, Y_2)$.
\end{definition}

\noindent
Binary functions are monotonic when they are monotonic in each argument individually.

We now define the following operations on binary functions towards our goal of defining mutually recursive fixpoints.

\begin{definition}[Binary-function operations]\label{def:binary-function-operations}
Let $f \in 2^S \times 2^S \rightarrow 2^S$ be monotonic.
\begin{enumerate}
    \item\label{it:fix-parm}
    Let $X, Y \subseteq S$.  Then functions $f_{(X, \cdot)}, f_{(\cdot, Y)} \in 2^S \rightarrow 2^S$ are defined by:
    \[
    f_{(X, \cdot)}(Y) = f_{(\cdot, Y)}(X) = f(X,Y)
    \]
    Since $f_{(X, \cdot)}$ and $f_{(\cdot, Y)}$ are monotonic when $f$ is, we may further define $f_{(\cdot, \sigma)}, f_{(\sigma, \cdot)} \in 2^S \rightarrow 2^S$, where $\sigma \in \{\mu, \nu\}$, as follows.
    \begin{align*}
    f_{(\cdot,\sigma)}(X) &= \sigma f_{(X, \cdot)}\\
    f_{(\sigma,\cdot)}(Y) &= \sigma f_{(\cdot, Y)}
    \end{align*}
    \item\label{it:sigma-composition}
    Suppose further that $g \in 2^S \times 2^S \rightarrow 2^S$ is binary and monotonic over $2^S$ and that $\sigma \in \{\mu,\nu\}$.  Then function $(f [\sigma] g) \in 2^S \rightarrow 2^S$ is defined as follows.
    \[
    (f [\sigma] g)(X) = f(X, g_{(\cdot, \sigma)}(X))
    \]
\end{enumerate}
\end{definition}

To understand the above definitions, first note that if $f$ is binary and monotonic then $f_{(X,\cdot)} \in 2^S \to 2^S$ is the unary function obtained by holding the first argument of $f$ fixed at $X$, leaving only the second argument to vary. Similarly, $f_{(\cdot,Y)} \in 2^S \to 2^S$ is the unary function obtained by holding the second argument of $f$ fixed at $Y$. The monotonicity of $f$ guarantees the monotonicity of $f_{(X, \cdot)}$ and $f_{(\cdot, Y)}$, and thus fixpoints $\sigma f_{(X,\cdot)}$ and $\sigma f_{(\cdot,Y)}$ are well defined for all $X$ and $Y$ and $\sigma \in \{\mu,\nu\}$.  This fact ensures that unary function $f_{(\cdot, \sigma)}$ and $f_{(\sigma, \cdot)}$ are well-defined; in the first case, given argument $X$, $f_{(\cdot, \sigma)}(X)$ returns the result of computing the $\sigma$-fixpoint of $f$ when the first argument of $f$ is held at $X$.  The second case is similar.  It is straightforward to establish that $f_{(\cdot,\sigma)}$ and $f_{(\sigma, \cdot)}$ are also monotonic for all $\sigma \in \{\mu,\nu\}$.  Finally, for each $\sigma \in \{\mu,\nu\}$ the operation $[\sigma]$ defines a composition operation that converts binary monotonic functions $f$ and $g$ into a unary monotonic function with the following behavior.  Given input $X \subseteq S$ the composition function applies $f$ to $X$ and the result of computing the $\sigma$-fixpoint of $g$ with its first argument held at $X$.  To understand the importance of this function, consider the following notional pair of mutually recursive equations, where $f$ and $g$ are binary monotonic functions.
\begin{align*}
X &\stackrel{\sigma}{=} f(X,Y)\\
Y &\stackrel{\sigma'}{=} g(X,Y)
\end{align*}
Here $\sigma, \sigma' \in \{\mu,\nu\}$; the intention of these equations is to define $X$ and $Y$ as the mutually recursive $\sigma$ and $\sigma'$ fixpoints of $f$ and $g$, with the first equation dominating the second one.  In the usual definitions of such equation systems, $X$ is defined to be $\sigma (f [\sigma'] g)$, i.e. the $\sigma$-fixpoint of the function $f [\sigma'] g$, while $Y$ is taken to be $g_{(\cdot, \sigma')}(X)$; see, e.g.,~\cite{Mad1997}.

The next lemma highlights the role that the $f [\sigma] g$ construct plays in the semantics of the mu-calculus formulas that involve nested fixpoints.  The statement relies on the notion of a \emph{maximal fixpoint subformula}.  Briefly, if $\Phi$ is a formula then $\sigma Z.\Gamma$ is a maximal fixpoint subformula of $\Phi$ iff it is a subformula of $\Phi$ and is not a proper subformula of another fixpoint subformula of $\Phi$.

\begin{lemma}[Nested fixpoint semantics]\label{lem:nested-fixpoint-semantics}
Let $\sigma Z.\Phi \in \muformsSV$ be a formula, let $\sigma' Z'.\Gamma$ be a maximal fixpoint subformula of $\Phi$, and let $\Phi'$ and $W \in \Var$ be such that $W$ is fresh and $\Phi = \Phi'[W:=\sigma'Z'.\Gamma]$.  Then $\semfTV{Z}{\Phi} = f [\sigma'] g$, where $f(X,Y) = \semT{\Phi'}{\V[Z, W := X,Y]}$ and $g(X,Y) = \semT{\Gamma}{\V[Z, Z' := X, Y]}$.
\end{lemma}
\remove{
\begin{proofsketch}
Follows from Lemma~\ref{lem:substitution} and the definition of $f [\sigma'] g$.
The detailed proof is included in the appendix.
\qedhere
\end{proofsketch}
}
\begin{proof}
Fix $\sigma Z.\Phi$, $\sigma'Z'.\Gamma$, $W$, $f$ and $g$ as stated.  We must prove that for all $X$, $\semfTV{Z}{\Phi}(X) = (f [\sigma'] g)(X)$. We reason as follows.
\begin{align*}
\semfTV{Z}{\Phi}(X)
    &= \semT{\Phi}{\V[Z := X]}
    && \text{Definition of $\semfTV{Z}{\Phi}$}
\\
    &= \semT{(\Phi'[W := \sigma' Z'.\Gamma])}{\V[Z := X]}
    && \text{$\Phi = \Phi'[W:=\sigma' Z'.\Gamma]$}
\\
    &= \semT{\Phi'}{\V[Z, W := X, \semT{\sigma'Z'.\Gamma}{\V[Z := X]}]}
    && \text{Lemma~\ref{lem:substitution}, $W$ fresh}
\\
    &= \semT{\Phi'}{\V[Z,W := X, g_{(\cdot, \sigma')}(X)]}
    && \text{$g_{(\cdot,\sigma')} = \semT{\sigma'Z'.\Gamma}{\V[Z := X]}$}
\\
    &= f(X, g_{(\cdot, \sigma')}(X))
    && \text{Definition of $f$}
\\
    &= (f [\sigma'] g)(X)
    && \text{Definition of $f [\sigma'] g$}
\end{align*}
\qedhere
\end{proof}

\noindent
Note this result implies that $\semTV{\sigma Z.\Phi} = \sigma (f [\sigma'] g)$, where $\sigma Z . \Phi$, $f$ and $g$ are defined as in the lemma.

\paragraph{Quotient well-foundedness and well-orderings.}
Our goal in this part of the paper is to characterize a support ordering for $g$ in terms of a given a support ordering for $f [\sigma] g$.  For unary functions, support orderings may be either well-founded or not, and this property is sufficient to characterize both the greatest and least fixponts of these functions.  For function $f [\sigma] g$, where $f$ and $g$ are binary and have mutually recursive fixpoints, an intermediate notion, which we call \emph{quotient well-foundedness}, is important, especially when the mutually recursive fixpoints are of different types (i.e.\/ one is least while the other is greatest). 

\label{subsec:qwf-and-wo}
\begin{definition}[Quotient well-founded (qwf) / well-ordering (qwo)]\label{def:qwf}
Let $S$ be a set and $R \subseteq S \times S$ a binary relation over $S$, and let $(Q_R, \sqsubseteq)$ be the quotient of $R$ (cf.\/ Definition~\ref{def:relation-quotient}), with ${\sqsubset} = {\sqsubseteq^-}$ the irreflexive core of $\sqsubseteq$.
\begin{enumerate}
\item\label{subdef:qwf}
$R$ is \emph{quotient well-founded} (qwf) iff $\sqsubset$ is well-founded over $Q_R$.
\item\label{subdef:qwo}
$R$ is a \emph{quotient well-ordering} (qwo) iff $\sqsubset$ is a well-ordering over $Q_R$.
\end{enumerate}
\end{definition}

That is, $R$ is quotient well-founded iff the irreflexive core of the partial order induced by $R$ over its equivalence classes is well-founded.  Note that $R$ can be qwf without being well-founded; when this is the case the non-well-foundedness of $R$ can be seen as due solely to non-well-foundedness within its equivalence classes.  It is also easy to see that if $R$ is well-founded then it is qwf as well; in this case each $Q \in Q_R$ has form $\{s\}$ for some $s \in S$.  Also note that the universal relation $U_S = S \times S$ over $S$ is trivially qwf, as its quotient has one equivalence class, namely, $S$.

It turns out that if a relation is total and quotient well-founded, then it is also a quotient well-ordering.

\begin{lemma}[Total qwf relations are qwos]\label{lem:total-qwo}
Suppose that $R \subseteq S \times S$ is total and qwf. Then $R$ is a qwo. 
\end{lemma}
\begin{proof}
Let $R \subseteq S \times S$ be total and qwf, with $(Q_R, \sqsubseteq)$ the quotient of $R$.  We must show that ${\sqsubset} = {\sqsubseteq^-}$ is a well-ordering, i.e.\/ is well-founded and total, over $Q_R$.  Well-foundedness of $\sqsubset$ is immediate from the fact that $R$ is qwf.  We now must show that $\sqsubset$ is total, i.e.\/ is irreflexive and transitive and has the property that for any $Q_1, Q_2 \in Q_R$ such that $Q_1 \neq Q_2$, either $Q_1 \sqsubset Q_2$ or $Q_2 \sqsubset Q_1$.  Irreflexivity and transitivity are immediate from the fact that $\sqsubset = \sqsubseteq^-$ is the irreflexive core of partial order $\sqsubseteq$.  Now suppose $Q_1, Q_2 \in Q_R$ are such that $Q_1 \neq Q_2$; we must show that either $Q_1 \sqsubset Q_2$ or $Q_2 \sqsubset Q_1$.  From the definition of $Q_R$ it follows that there are $s_1, s_2 \in S$ such that $Q_1 = [s_1]_R$ and $Q_2 = [s_2]_R$.  Moreover, since $Q_1 \neq Q_2$ it must be that $s_1 \not\sim_R s_2$, and since $R$ is total we have that either $s_1 \inr{R} s_2$ and $s_2 \not\inr{R} s_1$, whence $Q_1 = [s_1]_R \sqsubset [s_2]_R = Q_2$, or $s_2 \inr{R} s_1$ and $s_1 \not\inr{R} s_2$, whence $Q_2 \sqsubset Q_1$.  As $\sqsubset$ is a well-ordering on $Q_R$, $R$ is by definition a qwo.
\qedhere
\end{proof}

The next result establishes the existence of so-called \emph{pseudo-minimum} elements in subsets drawn from qwos.

\begin{definition}[Pseudo-minimum elements]\label{def:minimal}
Let $R \subseteq S \times S$ be a binary relation over $S$, and let $X \subseteq S$.
Then $x \in X$ is \emph{$R$-pseudo-minimum for $X$} iff $x$ is an $R$-lower bound for $X$.
\end{definition}

\noindent
An element $x \in X$ is an $R$-pseudo-minimum for $X$ iff $x$ is an $R$-lower bound for $X$.  This does not imply that $x$ is an $R$-minimum, or even $R$-minimal, since even though $x$ is a $R$-pseudo-minimum there may exist $x' \in X$ such that $x' \inr{R} x$. In this case it must hold that $x \sim_R x'$, however.

The next lemma states a pseudo-minimum result for qwo relations that are total.  (It should be noted that a relation can be a qwo and still not itself be total.)

\begin{lemma}[Pseudo-minimum elements and quotient well-orderings]\label{lem:qwo-pseudo-minimum}
Suppose qwo $R \subseteq S \times S$ is total.  Then every non-empty subset $X \subseteq S$ contains an $R$-pseudo-minimum element.
\end{lemma}
\begin{proof}
Fix total qwo $R \subseteq S \times S$, and let $X \subseteq S$ be non-empty.  We must exhibit an $R$-pseudo-minimum element $x \in X$.  To this end, consider the quotient $(Q_R, \sqsubseteq)$ of $R$, and let ${\sqsubset} = {\sqsubseteq^-}$ be the irreflexive core of $\sqsubseteq$.  Note that $\sqsubset$ is a well-ordering over $Q_R$.  Now consider $Q_X \subseteq Q_R$ defined by $Q_X = \{\,[x]_R \mid x \in S\}$.  It follows that there is a $Q \in Q_X$ that is a $\sqsubset$-minimum for $Q_X$, and that $Q = [x]_R$ for some $x \in X$.  Since $R$ is total it is transitive, meaning $R^* = R^=$. It follows that $x \inr{R} x'$ for all $x' \in [x]_R$ such that $x \neq x'$; it is also the case that $x \inr{R} x'$ for all $x'$ such that $[x]_R \sqsubset [x']_R$.  These facts imply that $x \inr{R} x'$ for all $x' \in X$. \qedhere
\end{proof}

We now establish a relationship between support orderings for $f [\sigma] g$ and $g$.  We first define a notion of compatibility for such support orderings.

\begin{definition}[Local consistency of support orderings] \label{def:consistent-support-orderings}
Let $f, g \in 2^S \times 2^S \rightarrow 2^S$ be monotonic, and let $\sigma_1, \sigma_2 \in \{\mu,\nu\}$.  Further let $(X, \prec)$ be a $\sigma_1$-compatible support ordering for $f [\sigma_2] g$, with $Y_x = \sigma_2 g_{(\preimg{{\prec}}{x}, \cdot)}$ for $x \in X$ and $Y = \sigma_2\, g_{(\preimg{{\prec}}{X}, \cdot)}$.  Then $\sigma_2$-compatible support ordering $(Y,\prec')$ for $g_{(\preimg{{\prec}}{X}, \cdot)}$ is \emph{locally consistent} with $(X, \prec)$ iff for all $x \in X$, $(Y_x, \restrict{(\prec')}{Y_x})$ is a $\sigma_2$-compatible support ordering for $g_{(\preimg{\prec}{x}, \cdot)}$ and for all $y \in Y_x, \preimg{{\prec'}}{y} \subseteq Y_x$.
\end{definition}

Intuitively, $(Y, \prec')$ is locally consistent with $(X, \prec)$ if $\prec'$ not only supports the fact that $Y$ is the $\sigma_2$-fixpoint for $g_{(\preimg{{\prec}}{X})}$, but via the restriction of $\prec'$ to $Y_x$, it also provides localized support for the fact that $Y_x$ is the $\sigma_2$-fixpoint for $g_{(\preimg{{\prec}}{x}, \cdot)}$, for each $x \in X$.  In addition, for any $x \in X$ and $y \in Y_x$ every element in the support set $\preimg{{\prec'}}{y}$ with respect to $\prec'$ must also be an element of $Y_x$.  Note that this last aspect of the definition ensures that for any $x$ and $y \in Y_x$, 
$$
\preimg{{\prec'}}{y} = \preimg{(\restrict{(\prec')}{Y_x})}{y}.
$$

The next lemma establishes that, for given support orderings of a specific type for $f [\sigma_2] g$, consistent support orderings exist for $g$.

\begin{lemma}[From composite to local support orderings]\label{lem:fg-support}
Let $f, g \in 2^S \times 2^S \rightarrow 2^S$ be monotonic and $\sigma_1, \sigma_2 \in \{\mu,\nu\}$, with $X = \sigma_1 (f [\sigma_2] g)$.  Also let $(X,\prec)$ be a $\sigma_1$-compatible, total qwf{\,}
support ordering for $f[\sigma_2]g$ and $Y = \sigma_2 g_{(\preimg{{\prec}}{X}, \cdot)}$.  Then there is a $\sigma_2$-compatible, total qwf support ordering $(Y,\prec')$ for $g_{(\preimg{{\prec}}{X}, \cdot)}$ that is locally consistent with $(X,\prec)$.
\end{lemma}

\begin{proof}
Fix monotonic $f, g \in 2^S \times 2^S \rightarrow 2^S$ and $\sigma_1, \sigma_2 \in \{\mu,\nu\}$, let $X = \sigma_1 (f [\sigma_2] g)$, and let ${\prec} \subseteq X \times X$ be such that $(X, \prec)$ is a $\sigma_1$-compatible, total qwf support ordering for $f [\sigma_2] g$.  Also fix $Y = \sigma_2 g_{(\preimg{{\prec}}{X}, \cdot)}$.  We must construct a $\sigma_2$-compatible, total qwf ${\prec'} \subseteq Y \times Y$ such that $(Y, \prec')$ is a support ordering for $g_{(\preimg{{\prec}}{X}, \cdot)}$ that is locally consistent with $(X, \prec)$.

Let $(Q_\prec, \sqsubseteq)$ be the quotient of $(S,\prec)$ as given in Definition~\ref{def:relation-quotient}, and let ${\sqsubset} = {\sqsubseteq^-}$ be the irreflexive core of $\sqsubseteq$.  For notational convenience, if $Z \subseteq S$ then we define
$$
g_Z = g_{(\preimg{{\prec}}{Z}, \cdot)}.
$$
Since $\prec$ is total and qwf it follows that $\sqsubset$ is a well-ordering on $Q_\prec$.  For any $Q \in Q_\prec$ define
$$
Y_Q = \sigma_2 \,g_Q,
$$
and let $(Y_Q, \prec'_Q)$
be a $\sigma_2$-maximal support ordering for $g_Q$.  Since ${\prec'_Q} \subseteq Y_Q \times Y_Q$ is $\sigma_2$-maximal, we have that $\prec'_Q$ is a well-ordering if $\sigma_2 = \mu$, and $Y_Q \times Y_Q$ if $\sigma_2 = \nu$. 
In either case it is easy to see that $\prec'_Q$ is total and qwf. 
Moreover, since $Q \subseteq X$ it follows that $\preimg{{\prec}}{Q} \subseteq \preimg{{\prec}}{X}$, and this means that for all $Z \subseteq S$, $g_Q(Z) \subseteq g_X(Z)$.
Consequently $(Y_Q, \prec'_Q)$ is also a $\sigma_2$-compatible support ordering for $g_X$, as for all $y \in Y_Q$, $y \in g_Q(\preimg{{\prec'_Q}}{y}) \subseteq g_X(\preimg{{\prec'_Q}}{y})$.
We now define the following using well-founded induction on ${\sqsubset} = P({\prec})^-$.
\begin{align*}
Y'_{\sqsubset Q} &= \bigcup_{Q' \sqsubset Q} Y'_{Q'}
\\
Y'_{Q}  &= Y_Q \cup Y'_{\sqsubset Q}
\\
Y''_{Q} &= Y_Q \setminus Y'_{\sqsubset Q}
\\
\prec''_{Q} &= 
\left( \bigcup_{Q' \sqsubset Q} \prec''_{Q'} \right)
\cup \left( Y'_{\sqsubset Q} \times Y''_Q \right)
\cup \left(\restrict{\prec'_Q}{Y''_Q}\right)
\end{align*}

\noindent
An inductive argument further establishes that for each $Q$, $(Y'_Q, \prec''_Q)$ is a $\sigma_2$-compatible support ordering for $g_Q$ and that $\prec''_Q$ is total and qwf.

Now consider $Y' = \bigcup_{Q \in Q_\prec} Y'_Q$ and ${\prec''} = \bigcup_{Q \in Q_\prec} \prec''_Q$.  It is straightforward to show that $(Y', \prec'')$ is a $\sigma_2$-compatible, total, qwf support ordering for $g_X$ since each $(Y_Q, \prec'_Q)$ is.  Also note that $Y' \subseteq Y$.  To finish the construction of $(Y, \prec')$, take $Y'' = Y \setminus Y'$, and let $(Y, \prec''')$ be a maximal $\sigma_2$-compatible support ordering for $g_X$.  Note that $\prec'''$ is qwf, and well-founded if $\sigma_2 = \mu$ and $Y \times Y$ if $\sigma_2 = \nu$.  Now define the following.
\begin{align*}
\prec'  &= 
{\prec''}
\cup \left(Y' \times Y''\right)
\cup \left(\restrict{\prec'''}{Y''}\right)
\end{align*}
From the reasoning above it follows that $(Y, \prec')$ is a $\sigma_2$-compatible support ordering for $g_X$, and that $\prec'$ is total and qwf.

To complete the proof we must show that $(Y, \prec')$ is locally consistent with $(X, \prec)$.  To this end, fix $x \in X$ and define ${\prec'_x} = \restrict{(\prec')}{Y_x}$.  We must show that $(Y_x, \prec'_x)$ is a $\sigma_2$-compatible support ordering for $g_x = g_{\{x\}}$.  Recall that $[x] \in Q_\prec$ is the equivalence class containing $x$.  We begin by noting that since $\prec$ is total and qwf,
\[
\preimg{{\prec}}{x} = [x] \cup \left(\bigcup_{Q \sqsubset [x]} Q\right).
\]
Also, $\sigma_2 (g_x) = \sigma_2 (g_{[x]}) = Y_{[x]}$.  The definition of $\prec'$ further guarantees that $\prec'_x = \prec''_{Q_x}$.  As we know that $(Y_Q, \prec''_Q)$ is a $\sigma_2$-compatible support ordering for $g_Q$ for all $Q \in Q_\prec$, the desired result holds.
\qedhere
\end{proof}

\subsection{Tableau normal form}

Later in this section we use constructions on tableaux to prove completeness of our proof system.  The tableaux we work with have a restricted form, which we call \emph{tableau normal form} (TNF).  TNF is defined as follows.

\begin{definition}[Tableau normal form (TNF)]\label{def:tableau-normal-forms}
Let $\mathbb{T} = \tableauTrl$ be a tableau, with $\tree{T} = (\node{N}, \node{r}, p, cs)$. 
\begin{enumerate}
    \item \label{def:thinning-restricted}
    $\mathbb{T}$ is \emph{thinning-restricted} iff $\rho(\node{r}) \neq \textnormal{Thin}$ and for all $\node{n} \neq \node{r}$, $\rho(\node{n}) = \sigma Z.$ iff $\rho(p(\node{n})) = \textnormal{Thin}$.
    \item \label{def:unfolding-limited}
    $\mathbb{T}$ is \emph{unfolding-limited} iff for each definitional constant $U$ appearing in $\mathbb{T}$ there is exactly one node $\node{n}_U$ such that $\seqfm(\node{n}_U) = U$ and $\rho(\node{n}_U) = \textnormal{Un}$.
    \item \label{def:irredundant}
    $\mathbb{T}$ is \emph{irredundant} iff for each node $\node{n}$ such that $\rho(\node{n}) = \lor$ and $cs(\node{n}) = \node{n}_1\node{n}_2$, $\seqst(\node{n}_1) \cap \seqst(\node{n}_2) = \emptyset$.
    \item \label{def:TNF}
    $\mathbb{T}$ is in \emph{tableau normal form} (TNF) iff it is thinning-restricted, unfolding-limited and irredundant.
\end{enumerate}
\end{definition}
Intuitively, $\mathbb{T}$ is 
thinning-restricted if Thin is not applied to the root node, every application of the $\sigma Z.$ rule to a non-root node is immediately preceded by a single instance of Thin, and there are otherwise no other applications of Thin.  It is unfolding-limited if each definitional constant is unfolded exactly once, and irredundant if for each $\lor$-node, the state sets of the node's children are disjoint (i.e. no state can appear in both children, meaning there can be no redundant reasoning about states in the $\lor$-node).  The tableau is in TNF if it is thinning-restricted, unfolding-limited and irredundant.  

In what follows we will on occasion build new (successful) TNF tableaux out of existing (successful) TNF tableaux that are \emph{structurally equivalent}.

\begin{definition}[Structural tableau equivalence]\label{def:structurally-equivalent-tableaux}
Let $\tree{T}_1 = (\node{N}_1, \node{r}_1, p_1, cs_1)$ and $\tree{T}_2 = (\node{N}_2, \node{r}_2, p_2, cs_2)$ be finite ordered trees.
\begin{enumerate}
    \item \label{subdef:tree-isomorphism}
      Bijection $\iota \in \node{N}_1 \to \node{N_2}$ is a \emph{tree isomorphism from $\tree{T}_1$ to $\tree{T}_2$} iff it satisfies:
      \begin{enumerate}
        \item
          $\iota(\node{r}_1) = \node{r}_2$;
        \item
          for all $\node{n}_1 \in \node{N}_1$ $\iota(p_1(\node{n}_1)) = p_2(\iota(\node{n}_1))$; and
        \item
          for all $\node{n}_1 \in \node{N}_1$ $\iota(cs_1(\node{n}_1)) = cs_2(\iota(\node{n}_1))$.
      \end{enumerate}
    \item \label{subdef:isomorphic-trees}
      $\tree{T}_1$ and $\tree{T}_2$ are \emph{isomorphic} iff there exists a tree isomorphism $\iota$ from $\tree{T}_1$ to $\tree{T}_2$.  In this case we call $\iota$ a \emph{witnessing tree isomorphism from $\tree{T}_1$ to $\tree{T}_2$}.
    \item\label{subdef:structurally-equivalent-tableaus}
      Fix LTS $\T$, and let 
      \begin{align*}
        \mathbb{T}_1 &= (\tree{T}_1, \rho_1, \T, \V_1, \lambda_1)\\
        \mathbb{T}_2 &= (\tree{T}_2, \rho_2, \T, \V_2, \lambda_2)
      \end{align*} 
      be tableaux.  Then $\mathbb{T}_1$ and $\mathbb{T}_2$ are \emph{structurally equivalent} iff $\tree{T}_1$ and $\tree{T}_2$ are isomorphic, with witnessing tree isomorphism $\iota$ from $\tree{T}_1$ to $\tree{T}_2$, and the following hold for all $\node{n}_1 \in \node{N}_1$.
      \begin{enumerate}
          \item
            $\rn(\rho_1(\node{n}_1)) = \rn(\rho_2(\iota(\node{n}_1)))$.
          \item
            $\seqfm(\lambda_1(\node{n}_1)) = \seqfm(\lambda_2(\iota(\node{n}_1)))$
          \item
            $\seqdl(\lambda_1(\node{n}_1)) = \seqdl(\lambda_2(\iota(\node{n}_1)))$
      \end{enumerate}
      In this case we refer to $\iota$ as a \emph{structural tableau morphism} from $\mathbb{T}_1$ to $\mathbb{T}_2$.
\end{enumerate}
\end{definition}

\noindent
The definitions of tree isomorphism and isomorphic trees are standard.  Two tableaux are structurally equivalent if they are ``almost isomorphic", in the standard sense.  Specifically, their trees must be isomorphic, and isomorphic nodes in the two trees must have the same proof rule applied to them, although they may have different witness functions if the rule involved is $\dia{K}$.  Sequents labeling isomorphic nodes may differ in their valuations, and the set of states mentioned in the sequents, although the formulas and definition lists must be the same.  Intuitively, structurally equivalent tableaux may be seen as employing the same reasoning, but on slightly different, albeit similar, sequents.

Call a tableau \emph{diamond-leaf-free} if it contains no diamond leaves; recall that that any successful tableau must be diamond-leaf-free.  The next lemma establishes that if two diamond-leaf-free TNF tableaux have root sequents $\seq{s}_1$ and $\seq{s}_2$ such that $\seqfm(\seq{s}_1) = \seqfm(\seq{s}_2)$ and $\seqdl(\seq{s}_1) = \seqdl(\seq{s}_2) = \emptyL$, then the tableaux must be structurally equivalent.  It relies on an assumption that we make throughout this section:  that definitional constants as introduced in the $\sigma Z.$ rule are generated uniformly.  That is, if the sequent labeling a node has form $S \tnxTV{\Delta} \sigma Z.\Phi$ and rule $\sigma Z.$ is applied, then a given definitional constant $U$ depending only on $\sigma Z.\Phi$ and $\Delta$ is introduced, with the child sequent $S \tnxTV{\Delta \cdot (U = \sigma Z.\Phi)} U$ being generated.

\begin{lemma}[Structural equivalence of TNF tableaux]\label{lem:structural-equivalence-of-TNF-tableaux}
Let $\seq{s}_1 = S_1 \tnxT{\V_1}{\emptyL} \Phi$ and $\seq{s}_2 = S_2 \tnxT{\V_2}{\emptyL} \Phi$ be sequents, with $\mathbb{T}_1$ and $\mathbb{T}_2$ diamond-leaf-free TNF tableaux for $\seq{s}_1$ and $\seq{s}_2$, respectively.  Then $\mathbb{T}_1$ and $\mathbb{T}_2$ are structurally equivalent.
\end{lemma}
\begin{proof}
Fix sequents $\seq{s}_1 = S_1 \tnxT{\V_1}{\emptyL} \Phi$ and $\seq{s}_2 = S_2 \tnxT{\V_2}{\emptyL} \Phi$, and let $\mathbb{T}_1$ and $\mathbb{T}_2$ be diamond-leaf-free TNF tableaux for $\seq{s}_1$ and $\seq{s}_2$, respectively, such that $\mathbb{T}_i = (\tree{T}_i,\rho_i, \T, \V_i, \lambda_i)$, and $\tree{T}_i = (\node{N}_i,\node{r}_i,p_i,cs_i)$, for $i = 1,2$.

It suffices to construct a structural tableau morphism from $\mathbb{T}_1$ to $\mathbb{T}_2$, so that $\iota(\node{r}_1) = \node{r}_2$, and the following properties are satisfied for all $\node{n}_1 \in \node{N}_1$.
\begin{enumerate}[left=\parindent, label=S\arabic*., ref=S\arabic*]
    \item \label{stm:related-parents} $\iota(p_1(\node{n}_1)) = p_2(\iota(\node{n}_1))$;
    \item \label{stm:related-children} $\iota(cs_1(\node{n}_1)) = cs_2(\iota(\node{n}_1))$; 
    \item \label{stm:equal-rule} $\rn(\rho_1(\node{n}_1)) = \rn(\rho_2(\iota(\node{n}_1)))$;
    \item \label{stm:equal-formula} $\seqfm(\lambda_1(\node{n}_1)) = \seqfm(\lambda_2(\iota(\node{n}_1)))$; and
    \item \label{stm:equal-definition-list} $\seqdl(\lambda_1(\node{n}_1)) = \seqdl(\lambda_2(\iota(\node{n}_1)))$.
\end{enumerate}

The definition of $\iota$ is given in a co-inductive fashion (i.e., ``from the root down''). Effectively, the definition is such that when two nodes are in bijective correspondence, they have the same formula and definition list, i.e., they satisfy statements~\ref{stm:equal-formula} and~\ref{stm:equal-definition-list}.
When two nodes are in bijective correspondence and have the same formula and definition list, we show that necessarily the same proof rule is applied to both, i.e., the nodes satisfy statement~\ref{stm:equal-rule}, and the bijective correspondence can be extended to their children in a way that statement~\ref{stm:related-children} is satisfied, and furthermore, the children satisfy statements~\ref{stm:related-parents},~\ref{stm:equal-formula} and~\ref{stm:equal-definition-list}.

The construction of $\iota$ is as follows.
We begin by taking $\iota(\node{r}_1) = \node{r}_2$. We observe that $\lambda_1(\node{r}_1) = \seq{s}_1 = S_1 \tnxT{\V_1}{\emptyL} \Phi$ and $\lambda_2(\node{r}_2) = \seq{s}_2 = S_2 \tnxT{\V_2}{\emptyL} \Phi$, hence $\seqfm(\lambda_1(\node{r}_1)) = \Phi = \seqfm(\lambda_2(\node{r}_2)) = \seqfm(\lambda_2(\iota(\node{r}_1)))$ and $\seqdl(\lambda_1(\node{r}_1)) = \emptyL = \seqdl(\lambda_2(\node{r}_2)) = \seqdl(\lambda_2(\iota(\node{r}_1)))$, so statements~\ref{stm:equal-formula} and~\ref{stm:equal-definition-list} clearly hold of $\node{r}_1$. Furthermore, $p_1(\node{r}_1)$ and $p_2(\node{r}_2)$ are both undefined, so statement~\ref{stm:related-parents} trivially holds.

Now, fix $\node{n}_1$ and $\node{n}_2$ such that $\iota(\node{n}_1) = \node{n}_2$, and assume they satisfy statements~\ref{stm:related-parents},~\ref{stm:equal-formula} and~\ref{stm:equal-definition-list}.
We distinguish cases based on $\rn(\rho_1(\node{n}_1))$, the name of the rule applied to $\node{n}_1$.
\begin{itemize}
    \item $\rn(\rho_1(\node{n}_1))\div$. So, $\node{n}_1$ is a leaf. As $\mathbb{T}_1$ is diamond leaf free, $\seqfm(\lambda_1(\node{n}_1)) \in \{ Z, \neg Z, U \}$ for $Z \in \Var \setminus \dom(\seqdl(\lambda_1(\node{n}_1)))$, $U \in \seqdl(\lambda_1(\node{n}_1))$; as  $\seqfm(\lambda_1(\node{n}_1)) =  \seqfm(\lambda_2(\node{n}_2))$ and $\seqdl(\lambda_1(\node{n}_1)) =  \seqdl(\lambda_2(\node{n}_2))$, also $\rho_2(\node{n}_2))\div$, hence statement~\ref{stm:equal-rule} is satisfied.
    Since $\node{n}_1$ and $\node{n}_2$ do not have a rule applied, neither node has children, and statement~\ref{stm:related-children} is satisfied trivially.
    
    \item $\rn(\rho_1(\node{n}_1)) \in \{ \wedge, \vee, [K], \dia{K} \}$. Then,
    it follows directly from the fact that $\mathbb{T}_1$ and $\mathbb{T}_2$ are in TNF, and Definition~\ref{def:tableau-normal-forms}(\ref{def:thinning-restricted}), that $\seqfm(\lambda_1(\node{n}_1))$ and $\seqfm(\lambda_2(\node{n}_2))$ are not of the form $\sigma Z . \Phi$, and therefore, $\rn(\rho_2(\node{n}_2)) \neq \textnormal{Thin}$. Therefore, $\rn(\rho_2(\node{n}_2))$ must be dictated by the shape of $\seqfm(\lambda_2(\node{n}_2))$, and as $\seqfm(\lambda_1(\node{n}_1)) = \seqfm(\lambda_2(\node{n}_2))$ it immediately follows that $\rn(\rho_2(\node{n}_2)) = \rn(\rho_1(\node{n}_1))$, and statement~\ref{stm:equal-rule} is satisfied.
    
    Since $\rn(\rho_2(\node{n}_2)) = \rn(\rho_1(\node{n}_1))$, they have the same number of children, and we define $\iota$ such that $\iota(cs(\node{n}_1)) = cs(\node{n}_2) = cs(\iota(\node{n}_1))$. So statement~\ref{stm:related-children} is satisfied immediately. It is furthermore easy to see that these children satisfy statements~\ref{stm:related-parents},~\ref{stm:equal-formula} and~\ref{stm:equal-definition-list}.
    
    \item $\rn(\rho_1(\node{n}_1)) = \textnormal{Thin}$. Then, as $\mathbb{T}_1$ is in TNF, it follows directly from Definition~\ref{def:tableau-normal-forms}(\ref{def:thinning-restricted}) that $\seqfm(\lambda_1(\node{n}_1)) = \seqfm(\lambda_2(\node{n}_2)) = \sigma Z . \Phi$, for some $Z$ and $\Phi$, and that $\rn(\rho_1(p_1(\node{n}_1))) \neq \textnormal{Thin}$.
    According to our coinductive hypothesis, $\iota(p_1(\node{n}_1)) = p_2(\node{n}_2)$, hence $\rn(\rho_2(p_2(\node{n}_2))) = \rn(\rho_1(p_1(\node{n}_1))) \neq \textnormal{Thin}$; as $\seqfm(\lambda_2(\node{n}_2)) = \sigma Z . \Phi$ it must therefore be the case that $\rn(\rho_2(\node{n}_2)) = \rn(\rho_1(\node{n}_1)) = \textnormal{Thin}$, as rule $\sigma Z$, the other rule that could potentially be applied, is not allowed due to the restriction to TNF.
    
    Since $\rn(\rho_2(\node{n}_2)) = \rn(\rho_1(\node{n}_1))$, they have the same number of children, and we define $\iota$ such that $\iota(cs(\node{n}_1)) = cs(\node{n}_2) = cs(\iota(\node{n}_1))$. So statement~\ref{stm:related-children} is satisfied immediately. It is furthermore easy to see that these children satisfy statements~\ref{stm:related-parents},~\ref{stm:equal-formula} and~\ref{stm:equal-definition-list}.
    
    \item $\rn(\rho_1(\node{n}_1)) = \sigma Z$. Then, as $\mathbb{T}_1$ is in TNF, it follows directly from Definition~\ref{def:tableau-normal-forms}(\ref{def:thinning-restricted}) that $\seqfm(\lambda_1(\node{n}_1)) = \seqfm(\lambda_2(\node{n}_2)) = \sigma Z . \Phi$, for some $Z$ and $\Phi$, and that $\rn(\rho_1(p_1(\node{n}_1))) = \textnormal{Thin}$. According to our coinductive hypothesis, $p_1(\node{n}_1) = p_2(\node{n}_2)$, hence $\rn(\rho_2(p_2(\node{n}_2))) = \rn(\rho_1(p_1(\node{n}_1))) = \textnormal{Thin}$; as $\seqfm(\lambda_2(\node{n}_2)) = \sigma Z . \Phi$ it must therefore be the case that $\rn(\rho_2(\node{n}_2)) = \rn(\rho_1(\node{n}_1)) = \sigma Z$.
    
    Since $\rn(\rho_2(\node{n}_2)) = \rn(\rho_1(\node{n}_1))$, they have the same number of children, and we define $\iota$ such that $\iota(cs(\node{n}_1)) = cs(\node{n}_2) = cs(\iota(\node{n}_1))$. So statement~\ref{stm:related-children} is satisfied immediately. It is furthermore easy to see that these children satisfy statements~\ref{stm:related-parents},~\ref{stm:equal-formula} and~\ref{stm:equal-definition-list}. Note that in particular statements~\ref{stm:equal-formula} and~\ref{stm:equal-definition-list} require that the definitional constants are generated uniformly, so the same definitional constant is introduced in the application of the $\sigma Z$ rule in both tableaux.\qedhere
\end{itemize}
\end{proof}

\remove{
\begin{proofsketch}
It suffices to give a structural tableau morphism from $\mathbb{T}_1$ to $\mathbb{T}_2$.  This can be done co-inductively using $\tree{T}_1$ and $\tree{T}_2$, the trees embedded in $\mathbb{T}_1$ and $\mathbb{T}_2$.  The limitations imposed by TNF on the use of the Thin and Un rules ensure the desired similarities in sequents labeling isomorphic tree nodes, while the diamond-leaf-free property ensures that all leaves must be free leaves, i.e.\/ of form $Z$ or $\lnot Z$ for some $Z$ free in $\Phi$, or $\sigma$-leaves, i.e.\/ of form $U$ for some $U$ defined in the definition list of the leaf.
\qedhere
\end{proofsketch}
}

\subsection{Completeness via tableau constructions}

We now turn to proving lemmas about the existence of successful TNF tableaux for different classes of sequents. The first establishes the existence of such tableaux for valid sequents whose formulas contain no fixpoint subformulas. 

\begin{lemma}[Fixpoint-free completeness]\label{lem:fixpoint-free-completeness}
Let $\T, \V, \Phi$ and $S$ be such that $\Phi$ is fixpoint-free and $S \subseteq \semTV{\Phi}$.  Then there is a successful TNF tableau for $S \tnxTV{\emptyL} \Phi$.
\end{lemma}
\begin{proof}
Let $\T = \lts{S}$ be an LTS of sort $\Sigma$ and $\V$ be a valuation.  The proof proceeds by structural induction on $\Phi$; the induction hypothesis states that for any subformula $\Phi'$ of $\Phi$ and $S'$ such that $S' \subseteq \semTV{\Phi'}$, $S' \tnxTV{\emptyL} \Phi'$ has a successful TNF tableau.  The argument involves a case analysis on the form of $\Phi$.  Most cases are routine and left to the reader.  We consider here the cases for $\lor$ and $\dia{K}$.

Assume $\Phi = \Phi_1 \lor \Phi_2$; let $S_1 = S \cap \semTV{\Phi_1}$ and $S_2 = S \setminus S_1$.  It is easy to see that $S_1 \subseteq \semTV{\Phi_1}$ and $S_2 \subseteq \semTV{\Phi_2}$; hence the induction hypothesis guarantees that both $S_1 \tnxTV{\emptyL} \Phi_1$ and $S_2 \tnxTV{\emptyL} \Phi_2$ have successful TNF tableaux.  Without loss of generality, assume that these tableaux have disjoint sets of proof nodes.  We obtain a successful TNF tableau for $S \tnxTV{\emptyL} \Phi$ by creating a fresh tree node labeled by this sequent and having as its left child the root of the successful TNF tableau for $S_1 \tnxTV{\emptyL} \Phi_1$ and as its right child the root of the successful TNF tableau for $S_2 \tnxTV{\emptyL} \Phi_2$.  The proof rule applied to the new node is $\lor$.  It is easy to establish that this tableau is successful and TNF.

Now assume $\Phi = \dia{K} \Phi'$, and let $f \in S \to \states{S}$ be a function such that for every $s \in S$, $s \xrightarrow{K} f(s)$ and $f(s) \in \semTV{\Phi'}$.
Such an $f$ must exist, as $S \subseteq \semTV{\dia{K}\Phi'}$ and thus for every $s \in S$ there is an $s' \in \states{S}$ such that $s \xrightarrow{K} s'$ and $s' \in \semTV{\Phi'}$.
Since $f(S) \subseteq \semTV{\Phi'}$, the induction hypothesis guarantees the existence of a successful TNF tableau for $f(S) \tnxTV{\emptyL} \Phi'$.  We now construct a successful TNF tableau for $S \tnxT{\V'}{\emptyL} \Phi$ as follows.  Create a fresh tree node labeled by $S \tnxTV{\emptyL} \Phi$, and let its only child be the root node for the successful TNF tableau for $f(S) \tnxTV{\emptyL} \Phi'$.  The rule application associated with the new node is $(\dia{K},f)$. The new tableau is clearly successful and TNF. 
\qedhere
\end{proof}

We now prove the existence of successful TNF tableaux for different classes of sequents involving fixpoint formulas.  Before doing so, we first define the notion of \emph{compliance} between a tableau and a support ordering.

\begin{definition}[Tableau compliance with $\prec$]\label{def:tableau-compliance}
Let $\T, \V, Z, \Phi, \sigma$ and $S$ be such that $S = \semTV{\sigma Z.\Phi}$.  Also let $(S, \prec)$ be a $\sigma$-compatible support ordering for $\semfTV{Z}{\Phi}$.  Then
TNF
tableau $\tableauTrl$ with root node $\node{r} = S \tnxTV{\emptyL} \sigma Z.\Phi$ and $\node{r'} = cs(\node{r}) = S \tnxTV{(U = \sigma Z.\Phi)} U$ is \emph{compliant} with $\prec$ iff whenever $s' <:_{\node{r'}} s$, $s' \prec s$.
\end{definition}

Intuitively, tableau $\mathbb{T}$ for sequent $S \tnxTV{\emptyL} \sigma Z.\Phi$ is compliant with support ordering $\prec$ iff every extended dependency from $s \in S$ to a state in a companion leaf of $\node{r}'$ is also a semantic dependency as reflected in $\prec$.
Note that since the root node of $\mathbb{T}$ is a fixpoint node and $\mathbb{T}$ is in TNF, $\rho(\node{r}) = \sigma Z.$ and thus $cs(\node{r}) = \node{r'} = S \tnxTV{U = \sigma Z.\Phi} U$ for some definitional constant $U$. Also, since $\mathbb{T}$ is TNF $\rho(\node{r}') = \text{Un}$, and $\node{r'}$ is the only node in which unfolding is applied to definitional constant $U$.

The next lemma continues the sequence of results in this section on the existence of successful TNF tableaux for valid sequents.  In this case, formulas have form $\sigma Z.\Phi$, where $\Phi$ contains no fixpoint subformulas, and have specific $\sigma$-compatible support orderings, and the result shows how successful TNF tableaux that are compliant with the given support ordering can be constructed.

\begin{lemma}[Single-fixpoint completeness]\label{lem:single-fixpoint-completeness}
Fix $\T$, and let $\Phi, Z, \V, \sigma$ and $S$ be such that $\Phi$ is fixpoint-free and $S = \semTV{\sigma Z.\Phi}$.  Also let $(S, \prec)$ be a $\sigma$-compatible, total, qwf support ordering for $\semfTV{Z}{\Phi}$.  Then $S \tnxTV{\emptyL} \sigma Z.\Phi$ has a successful TNF tableau compliant with $(S, \prec)$.
\end{lemma}
\remove{
\begin{proofsketch}
Fix $\T= \lts{\states{S}}$ of sort $\Sigma$, and let $\Phi, Z, \V, \sigma$ and $S$ be such that $\Phi$ is fixpoint-free and $S = \semTV{\sigma Z.\Phi}$.  
Also let $(S, \prec)$ be a $\sigma$-compatible, total, qwf support ordering for $f = \semfTV{Z}{\Phi}$.  
Since $s \in f(\preimg{{\prec}}{s})$ for all $s \in S$ it follows from Lemma~\ref{lem:fixpoint-free-completeness} that for each such $s$ there is a successful TNF tableau $\mathbb{T}_s$ for sequent $\{ s \} \tnxT{\V_s}{\emptyL} \Phi$, where $\V_s = \V [Z := \preimg{{\prec}}{s}]$.  
The proof then shows how these tableaux may be merged into a single successful TNF tableau for $S \tnxT{\V_S}{\emptyL} \Phi$, where $\V_S = \V[Z := \preimg{{\prec}}{S}]$, with the property that extended dependencies from the root to leaves labeled by $Z$ are also semantic-support relationships according to $\prec$.  
It is then shown how to convert this tableau into a successful TNF tableau for $S \tnxTV{\emptyL} \sigma Z.\Phi$ that is compliant with $(S, \prec)$.  \seeappendix \qedhere
\end{proofsketch}
}
\begin{proof}
Fix $\T = \lts{\states{S}}$ of sort $\Sigma$, and let $\Phi, Z, \V, \sigma$ and $S$ be such that $\Phi$ is fixpoint-free and $S = \semTV{\sigma Z.\Phi}$.  Also let $(S, \prec)$ be a $\sigma$-compatible, total, qwf support ordering for $f = \semfTV{Z}{\Phi}$.
We must construct a successful TNF tableau for sequent $S \tnxTV{\emptyL} \sigma Z.\Phi$ that is compliant with $(S,\prec)$. The proof consists of the following steps.
\begin{enumerate}
    \item\label{it:step-single-state-tableau}
    For each $s \in S$ we use Lemma~\ref{lem:fixpoint-free-completeness} to establish the existence of a successful TNF tableau for sequent $\{s\} \tnxT{\V_s}{\emptyL} \Phi$, where $\V_s = \V[Z := \preimg{{\prec}}{s}]$.
    \item\label{it:step-full-set-tableau}
    We then construct a successful TNF tableau for sequent $S \tnxT{\V_S}{\emptyL} \Phi$, where $\V_S = \V[Z := \preimg{{\prec}}{S}]$, from the individual tableaux for the $s \in S$.
    \item\label{it:step-fixpoint-tableau}
    We convert the tableau for $S \tnxT{\V_s}{\emptyL} \Phi$ into a successful TNF tableau for $S \tnxTV{\emptyL} \sigma Z.\Phi$ that is compliant with $\prec$.
\end{enumerate}

We begin by noting that $S = \sigma f$, and that since $(S,\prec)$ is a support ordering for $f$, it is the case that $s \in f(\preimg{{\prec}}{s})$ for every $s \in S$.  From the definition of $f$ and $\V_s$ it therefore follows that for each $s \in S$, $s \in \semT{\Phi}{\V_s}$.
Now, let $(Q_{\prec}, \sqsubseteq)$ be the quotient of $(S, {\prec})$ as given in Definition~\ref{def:relation-quotient}, with ${\sqsubset} = {\sqsubseteq^-}$ the irreflexive core of $\sqsubseteq$.  (Recall that each equivalence class $Q \in Q_\prec$ is such that that $Q \subseteq S$.)   Since $\prec$ is total and qwf Lemma~\ref{lem:total-qwo} guarantees that it is a qwo, which means that $\sqsubset$ is a well-ordering (i.e.\/ is total and well-founded) over $Q_\prec$. If $s \in S$, then we write $[s]$ for the unique $Q \in Q_\prec$ such that $s \in Q$ (i.e.\/ $[s]$ is the equivalence class of $s$ with respect to the equivalence on $S$ induced by $\prec$).  It is easy to see the following.
\[
\preimg{{\prec}}{s} = 
\bigcup_{\{s' \mid s' \prec s\}} [s']
\]
Note that if $s \prec s$ then $[s] \subseteq {\preimg{\prec}{s}}$.

\paragraph{Step~\ref{it:step-single-state-tableau} of proof outline:  construct tableau for $\{s\} \tnxT{\V_s}{\emptyL} \Phi$, where $s \in S$.}
For any $s \in S$ we have that $s \in \semT{\Phi}{\V_s}$, meaning
$\{s\} \subseteq \semTV{\Phi}$ is valid. Since $\Phi$ is fixpoint-free Lemma~\ref{lem:fixpoint-free-completeness} guarantees the existence of a successful TNF tableau 
\begin{align*}
\mathbb{T}_s 
&= \tableau{\tree{T}_s}{\rho_s}{\T}{\V_s}{\lambda_s}, \text{\rm where}
\\
\tree{T}_s
&= (\node{N}_s,\node{r}_s,p_s,cs_s)
\end{align*}
for $\{s\} \tnxT{\V_s}{\emptyL} \Phi$.
We now remark on some properties of $\mathbb{T}_s$.
\begin{enumerate}
    \item\label{it:obs-dependencies}
    Suppose $s'$ and $\node{n}' \in \node{N}_s$ are such that $\seqfm(\node{n}') = Z$ (so $\node{n}'$ is a leaf whose formula is $Z$, the variable bound by $\sigma$ in $\sigma Z.\Phi$) and $s' <:_{\node{n}',\node{r}_s} s$.  Then $s' \prec s$, since $s' \in \V_s(Z) = \preimg{{\prec}}{s}$.
    \item\label{it:obs-isomorphism}
    Let $s' \in S$, and let $\mathbb{T}_{s'}$ be the successful TNF tableau for $\{s'\} \tnxT{\V_{s'}}{\emptyL} \Phi$.  Lemma~\ref{lem:structural-equivalence-of-TNF-tableaux} and the fact that successful tableaux must be diamond-leaf-free guarantees that $\mathbb{T}_s$ and $\mathbb{T}_{s'}$ are structurally equivalent.
\end{enumerate}
Observation \ref{it:obs-dependencies} establishes a relationship between dependencies involving the single state in the root node of $\mathbb{T}_s$ and states in the leaves involving $Z$.
Observation \ref{it:obs-isomorphism} guarantees that the successful TNF tableaux due to Lemma~\ref{lem:fixpoint-free-completeness} are structurally equivalent, and hence isomorphic as trees, and satisfying the property that isomorphic nodes in these trees share the same formulas, definition lists ($\emptyL$ in this case) and rule names.
In what follows we use $\tree{T} = (\node{N}, \node{r}, p, cs)$, $\seqfm(\node{n})$ and $\rn(\node{n})$ for these common structures and write $\mathbb{T}_s = (\tree{T}, \rho_s, \T, \V_s, \lambda_s)$ for $s \in S$, noting that for all $s, s' \in S$, $\rn(\rho_s(\node{n})) = \rn(\rho_{s'}(\node{n})) = \rn(\node{n})$.  

\paragraph{Step~\ref{it:step-full-set-tableau} of proof outline:  construct tableau for $S \tnxT{\V_S}{\emptyL} \Phi$.}
We now construct a successful TNF tableau for $S \tnxT{\V_S}{\emptyL} \Phi$ satisfying the following:  if $s, s'$ and $\node{n}'$ are such that $\seqfm(\node{n}') = Z$ and $s' <:_{\node{n}',\node{r}} s$, then $s' \prec s$.  There are two cases to consider.
In the first case, $S = \emptyset$.  In this case, ${\prec} = \preimg{{\prec}}{S} = \emptyset$, and $\emptyset \tnxT{\V_S}{\emptyL} \Phi$ is valid and therefore, by Lemma~\ref{lem:fixpoint-free-completeness}, has a successful TNF tableau.  Define $\mathbb{T}_S$ to be this tableau.  Note that since $S = \emptyset$ $T_S$ vacuously satisfies the property involving $<:$.

In the second case, $S \neq \emptyset$; we will construct $\mathbb{T}_S = (\tree{T}, \rho_S, \T, \V_S, \lambda_S)$ that is structurally equivalent to each $\mathbb{T}_s$ for $s \in S$.
The intuition behind the construction is to ``merge" the individual tableaux $\mathbb{T}_s$ for the $s \in S$ by assigning to each node $\node{n}$ in $\mathbb{T}_S$ the set of states obtained by appropriately combining all the sets of states each individual tableau $\mathbb{T}_s$ assigns to the node.  Care must be taken with nodes involving the $\lor$ and $\dia{K}$ proof rules.

Since $\tree{T}$ $\rho$ is already given, completing the construction of $\mathbb{T}_S$ only requires that we define $\rho_S$ and $\lambda_S$, which we do so that the following invariants hold for each $\node{n} \in \node{N}$.
\begin{invariants}
  \item\label{inv:rule} 
    If $\rho_S(\node{n})$ is defined, then $\rn(\rho_S(\node{n})) = \rn(\node{n})$ and the sequents assigned by $\lambda_S$ to $\node{n}$ and its children are consistent with $\rho_S(\node{n})$.
  \item\label{inv:formula}
    $\seqfm(\lambda_S(\node{n})) = \seqfm(\node{n})$
  \item\label{inv:state-set} 
    $\seqst(\lambda_S(\node{n})) \subseteq \bigcup_{s \in S} \seqst(\lambda_s(\node{n}))$
\end{invariants}
The definitions of $\rho_S$ and $\lambda_S$ are given in a co-inductive fashion (i.e.\/ ``from the root down").  We begin by taking $\lambda_S(\node{r}) = S \tnxT{\V_S}{\emptyL} \Phi$ for root node $\node{r}$; invariants~\ref{inv:formula} and \ref{inv:state-set} clearly hold of $\lambda_S(\node{r})$.  For the co-inductive step we assume that $\node{n}$ satisfies \ref{inv:formula} and \ref{inv:state-set} and define $\rho_S(\node{n}')$ and $\lambda(\node{n}')$ for each child $\node{n}'$ of $\node{n}$ so that \ref{inv:rule} holds of $\node{n}$ and \ref{inv:formula} and \ref{inv:state-set} hold of each of the $\node{n}'$.
This is done below based on $\rn(\node{n})$, the name of the rule applied to $\node{n}$.  Note that because each $\mathbb{T}_s$ is in TNF and $\Phi$ is fixpoint-free there can be no $\node{n} \in \node{N}$ such that $\rn(\node{n}) \in \{\textnormal{Thin}, \sigma Z, \text{Un}\}$. Thus the case analysis below need not consider these possibilities.  In what follows we let $S_{\node{n}} = \seqst(\lambda_S(\node{n}))$ be the set of states in the sequent labeling $\node{n}$.
\begin{description}
    \item[$\rn(\node{n}) \bottom$.]
    In this case \ref{inv:rule} holds vacuously.  Note also that $\node{n}$ must be a leaf and therefore has no children.
    
    \item[$\rn(\node{n}) = \land$.]
    In this case $cs(\node{n}) = \node{n}_1\node{n}_2$ and $\seqfm(\node{n}) = \seqfm(\node{n}_1) \land \seqfm(\node{n}_2)$.  
    Define 
    $\rho_S(\node{n}) = \land$
    and
    $\lambda_S(\node{n}_1) = S_\node{n} \tnxT{\V_S}{\emptyL} \seqfm(\node{n}_1)$ and $\lambda_S(\node{n}_2) = S_\node{n} \tnxT{\V_S}{\emptyL} \seqfm(\node{n}_2)$.
    Invariant \ref{inv:rule} clearly holds for $\node{n}$, while \ref{inv:formula} and \ref{inv:state-set} each hold for $\node{n}_1$ and $\node{n}_2$ since they do for $\node{n}$ and $\rho_S(\node{n}) = \land$.
    
    \item[$\rn(\node{n}) = \lor$.]
    In this case $cs(\node{n}) = \node{n}_1\node{n}_2$ and $\seqfm(\node{n}) = \seqfm(\node{n}_1) \lor \seqfm(\node{n}_2)$.
    Take $\rho_S(\node{n}) = \lor$.
    We now construct $S_1$ and $S_2$, the sets of states in $\node{n}_1$ and $\node{n}_2$, as follows.
    For $s \in S_\node{n}$ define 
    \[
    I_s = \{s' \in S \mid s \in \seqst(\lambda_{s'}(\node{n}))\}.
    \]
    Intuitively, $I_s$ consists of all states $s' \in S$ whose tableaux $\mathbb{T}_{s'}$ contain state $s$ in node $\node{n}$.  This set must be non-empty since $s \in S_\node{n}$ and Property~\ref{inv:state-set} holds of $\node{n}$.  Since $\prec$ is a qwo, Lemma~\ref{lem:qwo-pseudo-minimum} guarantees that $I_s$ has at least one pseudo-minimum element $s'$:  $s' \in I_s$ has the property that $s' \prec s''$ for all $s'' \in I_s$.
    Select $s'$ to be one of these pseudo-minimum elements.
    Since $\mathbb{T}_{s'}$ is successful and TNF it must be the case that either $s \in \seqst(\lambda_{s'}(\node{n}_1))$ or $s \in \seqst(\lambda_{s'}(\node{n}_2))$, but not both.  Now define the following.
    \begin{align*}
        S_{1,s}
        &=
        \begin{cases}
            \{s\} 
            & \text{if $s \in \seqst(\lambda_{s'}(\node{n}_1))$}
            \\
            \emptyset
            & \text{otherwise}
        \end{cases}
        \\
        S_{2,s}
        &= \{s\} \setminus S_{1,s}
        \\
        S_1
        &= \bigcup_{s \in S} S_{1,s}
        \\
        S_2
        &= \bigcup_{s \in S} S_{2,s}
    \end{align*}
    For any $s \in S_{\node{n}}$, since either $s \in S_{1,s}$ or $s \in S_{2,s}$, but not both, it follows that either $s \in S_1$ or $s \in S_2$, but not both.
    Therefore, $S_1 \cup S_2 = S_{\node{n}}$ and $S_1 \cap S_2 = \emptyset$.
    It can also be seen that if $s \in S_{i,s}$ then 
    $s \in \seqst(\lambda_{s'}(\node{n}_i))$ 
    and thus $s \in \bigcup_{s'' \in S} \seqst(\lambda_{s''}(\node{n}_i))$.
    Now define
    \begin{align*}
    \lambda_S(\node{n}_1) &= S_1 \tnxT{\V_S}{\emptyL} \seqfm(\node{n}_1)
    \\
    \lambda_S(\node{n}_2) &= S_2 \tnxT{\V_S}{\emptyL} \seqfm(\node{n}_2).
    \end{align*}
    Based on the definitions of $S_1$ and $S_2$ invariant~\ref{inv:rule} certainly holds for $\node{n}$, as do \ref{inv:formula} and \ref{inv:state-set} for each of $\node{n}_1$ and $\node{n}_2$.
    
    \item[$\rn(\node{n}) = [K{]}$.]
    In this case $cs(\node{n}) = \node{n}'$ and $\seqfm(\node{n}) = [K] \seqfm(\node{n}')$.  
    Set $\rho_S(\node{n}) = [K]$.
    Now let 
    \[
    S' = \{ s' \mid \exists s \in S_{\node{n}} \colon s \xrightarrow{K} s'\},
    \]
    and define $\lambda_S(\node{n}') = S' \tnxT{\V_S}{\emptyL} \seqfm(\node{n}')$.  Invariant~\ref{inv:rule} holds for $\node{n}$, while \ref{inv:formula} and \ref{inv:state-set} hold for $\node{n}'$ based on the fact that these hold by assumption for $\node{n}$.
    
    \item[$\rn(\node{n}) = \dia{K}$.]
    In this case $cs(\node{n}) = \node{n}'$ and $\seqfm(\node{n}) = \dia{K}\seqfm(\node{n}')$. 
    To define $\rho_S(\node{n})$ we first
    construct a witness function $f_\node{n} \in S_{\node{n}} \rightarrow \states{S}$ such that $s \xrightarrow{K} f_{\node{n}}(s)$ for all $s \in S_\node{n}$ and such that $f_\node{n}(S_\node{n}) \subseteq \bigcup_{s \in S}\seqst(\lambda_s(\node{n}'))$.
    This function will then be used to define the sequent labeling $\node{n}'$.
    So fix $s \in S_\node{n}$; we construct $f_{\node{n}}(s)$ based on the tableaux $\mathbb{T}_{s'}$ whose sequent for $\node{n}$ contains $s$.
    To this end, define
    \[
    I_s = \{s' \in S \mid s \in \seqst(\lambda_{s'}(\node{n})) \}.
    \]
    Intuitively, $I_s \subseteq S$ contains all states $s'$ whose tableau $\mathbb{T}_{s'}$ contains state $s$ in $\node{n}$.
    Clearly $I_s$ is non-empty and thus contains a pseudo-minimum element $s'$ (Lemma~\ref{lem:qwo-pseudo-minimum}).
    Now consider $\rho_{s'}(\node{n})$; it has form $(\dia{K},f_{\node{n},s'})$, where $f_{\node{n},s'}$ is the witness function for $\node{n}$ in tableau $\mathbb{T}_{s'}$.
    This means that $f_{\node{n}, s'} \in \seqst(\lambda_{s'}(\node{n})) \rightarrow \states{S}$ is such that $\seqst(\lambda_{s'}(\node{n}')) = f_{\node{n},s'}(\seqst(\lambda_{s'}(\node{n})))$.
    We now define $f_{\node{n}}(s) = f_{\node{n},s'}(s)$, $\rho_S(\node{n}) = (\dia{K}, f_{\node{n}})$, and $\lambda_S(\node{n}') = f_\node{n}(S_\node{n}) \tnxT{\V_S}{\emptyL} \seqfm(\node{n}')$.  It can be seen that invariant~\ref{inv:rule} holds of $\node{n}$ and that \ref{inv:formula} and \ref{inv:state-set} hold of $\node{n}'$.
\end{description}
This construction ensures that Properties~\ref{inv:rule}--\ref{inv:state-set} hold for all $\node{n}$.

To establish that $\mathbb{T}_S$ is successful we must show that every leaf in $\mathbb{T}_S$ is successful (cf.\/ Definition~\ref{def:successful-tableau}), 
which amounts to showing that for each such leaf $\node{n}$, 
$\seqst(\lambda_S(\node{n})) \subseteq \semT{\seqfm(\node{n})}{\V_S}$.  
Since $\Phi$ contains no fixpoint subformulas there are two cases to consider.
\begin{description}
    \item[$\seqfm(\node{n}) = Z$.]
        In this case the formula labeling $\node{n}$ is the bound variable $Z$ in $\sigma Z.\Phi$.  Since for each $s \in S$ $\mathbb{T}_{s}$ is successful and $\preimg{{\prec}}{s} \subseteq \preimg{{\prec}}{S}$, we have 
        $$
        \seqst(\lambda_{s}(\node{n})) 
        \subseteq \semT{\seqfm(\node{n})}{\V_s} 
        = \V_s(Z) 
        = \preimg{{\prec}}{s}.
        $$
        As invariant~\ref{inv:state-set} ensures that $\seqst(\lambda_S(\node{n})) \subseteq \bigcup_{s \in S} \seqst(\lambda_{s}(\node{n}))$, it follows that 
        $$
        \seqst(\lambda_S(\node{n})) 
        \subseteq \bigcup_{s \in S} \seqst(\lambda_{s}(\node{n}))
        = \bigcup_{s \in S} \preimg{{\prec}}{s}
        = \preimg{{\prec}}{S}
        = \V_S(Z) 
        = \semT{\seqfm(\node{n})}{\V_S}.
        $$
        Leaf $\node{n}$ is therefore successful.
    \item[$\seqfm(\node{n}) \in \{Y, \lnot Y\}$ for some $Y \neq Z$ free in $\sigma Z.\Phi$.]
        The argument is very similar to the previous case, the only difference being that for any $s \in S$, $\V_s(Y) = \V_S(Y)$ and thus $\semT{\seqfm(\node{n})}{\V_s} = \semT{\seqfm(\node{n})}{\V_S}$.
\end{description}

That $\mathbb{T}_S$ is in TNF follows from the fact that it is successful (and thus diamond-leaf-free) and that each $\mathbb{T}_s$ is successful and TNF, and from the definitions of $\rho_S$ and $\lambda_S$.

We now establish that for all leaves 
$\node{n}$ in $\mathbb{T}_S$ such that $\seqfm(\node{n}) = Z$, 
and $s_\node{n}, s_\node{r} \in S$ 
such that $s_\node{n} <:_{\node{n}, \node{r}} s_\node{r}$ in $\mathbb{T}_S$,  
$s_\node{n} \prec s_\node{r}$.  
We begin by noting that if $\node{n} = \node{r}$ then $\Phi = Z$ and $s_\node{n} <:_{\node{n}, \node{r}} s_\node{r}$ iff $s_\node{n} = s_\node{r}$.
In this case, if $\sigma = \nu$ then $S = \semTV{\nu Z.Z} = \states{S}$ and the result holds because $(S, \prec)$ is a support ordering for $f$, and thus must be reflexive since $f$ is the identity function.  If instead $\sigma = \mu$ then $S = \semTV{\mu Z.Z} = \emptyset$ and the result is vacuously true.  

Now assume that $\node{n} \neq \node{r}$. 
We start by remarking on a property that holds of all $\node{n}_1, \node{n}_2, s_1$ and $s_2$ such that $s_2 <_{\node{n}_2, \node{n}_1} s_1$ in $\mathbb{T}_S$:  for every $s \in S$ such that $s_1 \in \seqst(\lambda_s(\node{n}_1))$, either $s_2 \in \seqst(\lambda_s(\node{n}_2))$, or there exists $s' \prec s$ such  that $s_2 \in \seqst(\lambda_{s'}(\node{n}_2))$.  
That is, when $s_2 <_{\node{n}_2, \node{n}_1} s_1$ in $\mathbb{T}_S$,
meaning $\node{n}_2$ is a child of $\node{n}_1$ and $s_2$ is a direct dependent of $s_1$ in $\mathbb{T}_S$,
and $\mathbb{T}_s$ is a tableau containing $s_1$ in $\node{n}_1$,
then $s_2$ is also contained in $\node{n}_2$ of either $\mathbb{T}_s$ or $\mathbb{T}_{s'}$ for some $s' \prec s$.  This is easily observed based on the definition of $<_{\node{n}_1, \node{n}_2}$, as well as the construction of $\mathbb{T}_S$ above and its use of pseudo-minimal states in the $\lor$ and $\dia{K}$ cases.  
A simple inductive argument lifts this result to the case when $\node{n}_1 \neq \node{n}_2$ and $s_1, s_2$ are such that $s_2 <:_{\node{n}_2, \node{n}_1} s_1$ in $\mathbb{T}_S$:  for all $s \in S$ such that $s_1 \in \seqst(\lambda_s(\node{n}_1))$, either $s_2 \in \seqst(\lambda_s(\node{n}_2))$ or there exists $s' \prec s$ such that $s_2 \in \seqst(\lambda_{s'}(\node{n}_2))$.  
Now consider $\node{n}$ and $\node{r}$ as given above, and assume $s_\node{n} <:_{\node{n},\node{r}} s_\node{r}$.  
It follows from the definition of $\mathbb{T}_s$ that if $s_{\node{r}} \in \seqst(\lambda_s(\node{r}))$ then $s = s_\node{r}$, since $\seqst(\lambda_s(\node{r})) = \{s\}$.  
There are now two cases to consider.
In the first, $s_\node{n} \in \seqst(\lambda_{s}(\node{n}))$,
whence $s_\node{n} \in \V_{s}(Z)$ and $s_\node{n} \prec s = s_\node{r}$.
In the second, there is an $s' \in S$ such that $s' \prec s_\node{r}$ and $s_\node{n} \in \seqst(\lambda_{s'}(\node{n}))$.  
In this case $s_\node{n} \in \V_{s'}(Z) = \preimg{{\prec}}{s'}$, meaning that $s_\node{n} \prec s'$.  Since $\prec$ is total, and hence transitive, we have $s_\node{n} \prec s' \prec s_\node{r}$ and thus $s_\node{n} \prec s_\node{r}$.

\paragraph{Step~\ref{it:step-fixpoint-tableau} of proof outline:  construct tableau for $S \tnxT{\V}{\emptyL} \sigma Z.\Phi$.}
To complete the proof, we convert $\mathbb{T}_S$ into a tableau $\mathbb{T}_\sigma = (\tree{T}_\sigma, \rho_\sigma, \T, \V, \lambda_\sigma)$ for sequent $S \tnxTV{\emptyL} \sigma Z.\Phi$ as follows.  We create two fresh tree nodes $\node{r}_1, \node{r}_2 \not\in \node{N}$ and add these into $\tree{T}_\sigma$ along with all the nodes of $\tree{T}$.  The root of $\tree{T}_\sigma$ is taken to be $\node{r}_1$; the parent of $\node{r}_2$ is then $\node{r}_1$, while the parent of $\node{r}$, the original root of $\tree{T}$, is $\node{r}_2$.  The other nodes of $\tree{T}$ retain their parents and sibling structure from $\tree{T}$.  We also define $\rho_\sigma(\node{r}_1) = \sigma Z$ and $\rho_\sigma(\node{r}_2) = \text{Un}$; for all nodes $\node{n}$ in $\tree{T}$, $\rho_\sigma(\node{n}) = \rho_S(\node{n})$.  If $\seq{s} = S \tnxTVD \Phi'$ and $Z \in \Var$ then take
\[
    \seq{s}[Z := \Gamma] = S \tnxTV{\Delta} \Phi'[Z := \Gamma]l
\]
We now define $\lambda_\sigma$ as follows.
\[
\lambda_\sigma(\node{n}) =
\begin{cases}
    S \tnxTV{\emptyL} \sigma Z.\Phi     & \text{if $\node{n} = \node{r}_1$} \\[6pt]
    S \tnxTV{(U = \sigma Z.\Phi)} U         & \text{if $\node{n} = \node{r}_2, U$ fresh} \\[6pt]
    \lambda_S(\node{n})[Z:=U]
                                            & \text{otherwise}
\end{cases}
\]

To finish the proof we must establish that $\mathbb{T}_\sigma$ is a successful TNF tableau compliant with $(S,\prec)$.
That $\mathbb{T}_\sigma$ is a tableau follows from the fact that $\mathbb{T}_S$ is a successful tableau.
In particular, sequents $\lambda_\sigma(\node{r}_1), \lambda_\sigma(\node{r}_2)$ and $\lambda_\sigma(\node{r})$ represent valid applications of rules $\rho_\sigma(\node{r}_1)$ and $\rho_\sigma(\node{r}_2)$.
Moreover, since $\lambda_\sigma(\node{n}) = S_\node{n} \tnxTV{(U = \sigma Z.\Phi)} \Phi_n[Z:=U]$, where $\lambda_S(\node{n}) = S_\node{n} \tnxTV{\emptyL} \Phi_n$, it can be seen that for each non-leaf $\node{n}$, $\lambda_\sigma(\node{n})$ and $\lambda_\sigma(cs(\node{n}))$ are consistent with  rule application $\rho_\sigma(\node{n})$, based on the structure of $\mathbb{T}_S$.
Finally, consider $\sigma$-leaf $\node{n}$ in $\mathbb{T}_\sigma$; that is, $\lambda_\sigma(\node{n}) = S_\node{n} \tnxTV{(U = \sigma Z.\Phi)} U$.  Since $\mathbb{T}_S$ is successful we have that $\node{n}$ is successful and that $\lambda_S(\node{n}) = S_\node{n} \tnxTV{\emptyL} Z$, meaning that $S_\node{n} \subseteq \V_S(Z) = \preimg{{\prec}}{S} \subseteq S$.  Since $\node{r}_2$ is the only node with an application of rule Un, it must be the companion node of $\node{n}$. Since $\seqst(\lambda_\sigma(\node{r}_2)) = S$, 
$\node{n}$ is terminal, and $\mathbb{T}_\sigma$ is indeed a tableau.

The fact that $\mathbb{T}_S$ is TNF and that definitional constant $U$ is unfolded only once in $\mathbb{T}_\sigma$ guarantees that $\mathbb{T}_\sigma$ is TNF.

We now argue that $\mathbb{T}_\sigma$ is compliant with $(S, \prec)$.  To this end, suppose that $s_1, s_2 \in S$ are such that $s_2 <:_{\node{r}_2} s_1$ in $\mathbb{T}_\sigma$; we must show that $s_2 \prec s_1$.
This follows immediately from the fact that $s_2 <:_{\node{r}_2} s_1$ in $\mathbb{T}_\sigma$ iff there is a leaf $\node{n}$ such that $\seqfm(\lambda_S(\node{n})) = Z$, $s_2 \in \seqst(\lambda_S(\node{n}))$, and $s_2 <:_{\node{n},\node{r}} s_1$ in $\mathbb{T}_S$.  Previous arguments then establish that $s_2 \prec s_1$.

To establish that $\mathbb{T}_\sigma$ is successful we must show that each of its leaves is successful.  Suppose $\node{n}$ is a non-$\sigma$-leaf leaf; that is, $\seqfm(\lambda_\sigma(\node{n})) \neq U$. In this case $\lambda_\sigma(\node{n}) = \lambda_S(\node{n})$, and the success of $\mathbb{T}_S$ and all its leaves guarantees the success of this leaf.  Now suppose that $\node{n}$ is a $\sigma$-leaf, meaning that $\seqfm(\lambda_\sigma(\node{n})) = U$.  If $\sigma = \nu$ then this leaf is successful.  If $\sigma = \mu$ then we note that for $(S, \prec)$ to be a support ordering for $S = \mu f$, $\prec$ must be well-founded.  Compliance of $\mathbb{T}_\sigma$ with $\prec$ guarantees that $<:_{\node{r}_2}$ is also well-founded, and thus $\node{n}$ is successful in this case also.
This completes the proof.
\qedhere
\end{proof}

\noindent
As an immediate corollary, we have the following.

\begin{corollary}\label{cor:single-fixpoint-completeness}
Fix $\T$, and let $\Phi, Z, \V, \sigma$ and $S$ be such that $\Phi$ is fixpoint-free and $S = \semTV{\sigma Z.\Phi}$.  Then $S \tnxTV{\emptyL} \sigma Z.\Phi$ has a successful tableau.
\end{corollary}
\begin{proof}
Follows from Lemma~\ref{lem:single-fixpoint-completeness} and the fact that every $\sigma$-maximal support ordering $(S, \prec)$ for $\semfTV{Z}{\Phi}$ is total and qwf.
\qedhere
\end{proof}

We now state and prove a generalization of Lemma~\ref{lem:single-fixpoint-completeness} in which the body of the fixpoint formula is allowed also to have fixpoint subformulas.

\begin{lemma}[Fixpoint completeness]\label{lem:fixpoint-completeness}
Fix $\T$, and let $\Phi, Z, \V, \sigma$ and $S$ be such that $S = \semTV{\sigma Z.\Phi}$.  Also let $(S, \prec)$ be a $\sigma$-compatible, total, qwf support ordering for $\semfTV{Z}{\Phi}$.  Then $S \tnxTV{\emptyL} \sigma Z.\Phi$ has a successful TNF tableau compliant with $(S, \prec)$.
\end{lemma}
\remove{
\begin{proofsketch}
Fix $\T = \lts{\states{S}}$ of sort $\Sigma$, and let $\Phi, Z, \V, \sigma$ and $S$ be such that $S = \semTV{\sigma Z.\Phi}$.  Also let $(S, \prec)$ be a $\sigma$-compatible, total, qwf support ordering for $\semfTV{Z}{\Phi}$.  The proof proceeds by strong induction on the number of fixpoint subformulas in $\Phi$.  If $\Phi$ contains no fixpoint subformulas then the result follows from Lemma~\ref{lem:single-fixpoint-completeness}.   
If $\Phi$ does contain fixpoint subformulas, then we select a maximal such formula of form $\sigma' Z'.\Gamma$ and use Lemma~\ref{lem:nested-fixpoint-semantics} to decompose $\sigma Z.\Phi$ into $\sigma Z.\Phi'$ and $\sigma'Z'.\Gamma$, with associated semantic functions $f$ and $g$, such that $\semTV{\sigma Z.\Phi} = \sigma (f [\sigma'] g)$.  
From this we use Lemma~\ref{lem:fg-support} to obtain appropriate $\sigma'$-compatible, total, qwf support orderings for $g_{(\cdot, \sigma')}(Q) $ from $(S,\prec)$ for equivalence classes $Q \subseteq S$ in the quotient of $(S, \prec)$. 
The induction hypothesis then guarantees successful TNF tableaux involving $\sigma Z.\Phi'$ and $\sigma'Z'.\Gamma$; we then give constructions for combining these tableaux into a successful TNF tableau for $S \tnxTV{\emptyL} \sigma Z.\Phi$ that is compliant with $(S,\prec)$.  The detailed proof is included in the appendix. \qedhere
\end{proofsketch}
}
\begin{proof}
Fix $\T = \lts{\states{S}}$ of sort $\Sigma$.  We prove the following: for all $\Phi$, and $Z, \V, \sigma$ and $S$ with
$S = \semTV{\sigma Z.\Phi}$,
and $\sigma$-compatible, total qwf support ordering $(S,\prec)$ for $\semfTV{Z}{\Phi}$, $S \tnxTV{\emptyL} \sigma Z.\Phi$ has a successful TNF tableau $\mathbb{T}_\Phi$ that is compliant with $(S,\prec)$.
To simplify notation we use the following abbreviations.
\begin{align*}
f_\Phi      &= \semfTV{Z}{\Phi}
\\
\V_X        &= \V[Z := \preimg{{\prec}}{X}]
\end{align*}
Note that $\V_S = \V[Z := \preimg{{\prec}}{S}]$.  When $s \in S$ we also write $\V_s$ in lieu of $V_{\{ s\}}$.

The proof proceeds by strong induction on the number of fixpoint subformulas of $\Phi$. There are two cases to consider.  In the first case, $\Phi$ contains no fixpoint formulas. Lemma~\ref{lem:single-fixpoint-completeness} immediately gives the desired result.

In the second case, $\Phi$ contains at least one fixpoint subformula.  The outline of the proof in this case is as follows.
\begin{enumerate}
    \item\label{it:step-decompose} 
    We decompose $\Phi$ into $\Phi'$, which uses a new free variable $W$, and $\sigma' Z'.\Gamma$ in such a way that $\Phi = \Phi'[W:=\sigma' Z'.
    \Gamma]$.
    \item\label{it:step-outer-tableau}
    We inductively construct a successful TNF tableau $\mathbb{T}_{\Phi'}$ for $S \tnxT{\V'}{\emptyL} \sigma Z.\Phi'$ that is compliant with $(S,\prec)$ where:
    \begin{align*}
        S'  &= \semT{\sigma'Z'.\Gamma}{\V_S} \\
        \V' &= \V[W:=S'].
    \end{align*}
    ($S'$ may be seen as the semantic content of $\sigma'Z'.\Gamma$ relevant for $\semTV{\sigma Z.\Phi}$.)
    \item\label{it:step-inner-tableau}
    We construct a successful TNF tableau $\mathbb{T}_\Gamma$ satisfying a compliance-related property for $S' \tnxT{\V_S}{\emptyL} \sigma'Z'.\Gamma$ by merging inductively constructed tableaux involving subsets of $S'$.
    \item\label{it:step-tableau-composition}
    We show how to compose $\mathbb{T}_\Phi$ and $\mathbb{T}_\Gamma$ to yield a successful TNF tableau for $S \tnxTV{\emptyL} \sigma Z.\Phi$ that is compliant with $(S,\prec)$.
\end{enumerate}
We now work through each of these proof steps.

\paragraph{Step~\ref{it:step-decompose} of proof outline:  decompose $\Phi$.}

Let $\sigma' Z'.\Gamma$ be a maximal fixpoint subformula in $\Phi$ as defined previously in this section.  Also let $W \in \Var$ be a fresh propositional variable, and define $\Phi'$ so that it contains exactly one instance of $W$ and so that
\[
\Phi = \Phi'[W := \sigma' Z'.\Gamma]
\]
($\Phi'$ is obtained by replacing one maximal instance of $\sigma' Z'.\Gamma$ in $\Phi$ by $W$.)
Note that $\Phi'$ and $\Gamma$ contain strictly fewer fixpoint subformulas than $\Phi$.  Lemma~\ref{lem:nested-fixpoint-semantics} may now be applied to conclude that $f_\Phi = f [\sigma'] g$, where
$f, g \in 2^\states{S} \times 2^\states{S} \rightarrow 2^S$ are defined as follows.
\begin{align*}
f(X,Y) &= \semT{\Phi'}{\V[Z, W := X, Y]}\\
g(X,Y)  &= \semT{\Gamma}{\V[Z, Z' := X, Y]}
\end{align*}
It is the case that $(S,\prec)$ is a $\sigma$-compatible, total qwf support ordering for $f [\sigma'] g$ since it is for $f_\Phi$ and $f_\Phi = f[\sigma'] g$.  

\paragraph{Step~\ref{it:step-outer-tableau} of proof outline: construct tableau for $S \tnxT{\V'}{\emptyL} \sigma Z.\Phi'$.}
Since $\Phi'$ has strictly fewer fixpoint subformulas than $\Phi$, we wish to apply the induction hypothesis to infer the existence of successful TNF tableau 
$\mathbb{T}_{\Phi'}$
for $S \tnxT{\V'}{\emptyL} \sigma Z.\Phi'$ that is compliant with $(S,\prec)$.  
To do this it suffices to confirm that $S = \semT{\sigma Z.\Phi'}{\V'} = \sigma f_{\Phi'}$, where $f_{\Phi'} = \semfT{Z}{\Phi'}{\V'}$,
and that $(S,\prec)$ is a support ordering for $f_{\Phi'}$ (it is already $\sigma$-compatible, total and qwf).

We begin by showing that $(S,\prec)$ is a support ordering for $f_{\Phi'}$; to do so 
we must establish that for every $s \in S$, $s \in f_{\Phi'}(\preimg{{\prec}}{s})$.  
It suffices to show that for every $s \in S$, 
$f_{\Phi}(\preimg{{\prec}}{s}) \subseteq f_{\Phi'}(\preimg{{\prec}}{s})$, as the fact that $(S,\prec)$ is a support ordering for $f_{\Phi}$ guarantees that $s \in f_{\Phi}(\preimg{{\prec}}{s}) \subseteq f_{\Phi'}(\preimg{{\prec}}{s})$. So fix $s \in S$; we reason as follows.
\begin{align*}
f_{\Phi} (\preimg{{\prec}}{s})
&= \semT{\Phi}{\V_s}
&& \text{Definition of $f_{\Phi}$}
\\
&= \semT{\Phi'[W:=\sigma'Z'.\Gamma']}{\V_s}
&& \Phi = \Phi'[W:=\sigma'Z'.\Gamma']
\\
&= \semT{\Phi'}{\V_s[W := \semT{\sigma'Z'.\Gamma}{\V_s}]}
&& \text{Lemma~\ref{lem:substitution}}
\\
&\subseteq \semT{\Phi'}{\V_s[W := \semT{\sigma'Z'.\Gamma}{\V_S}]}
&& \text{$\V_s(Y) \subseteq V_S(Y)$ all $Y$; see below}
\\
&= \semT{\Phi'}{\V_s[W := S']}
&& \text{Definition of $S'$} 
\\
&= \semT{\Phi'}{\V'_s}
&& \text{Definition of $\V'_s$}
\\
&= f_{\Phi'}(\preimg{{\prec}}{s})
&& \text{Definition of $f_{\Phi'}$}
\end{align*}
To see that $\V_s(Y) \subseteq V_S(Y)$ holds for all $Y$, note that $\V_s(Y) = \V_S(Y)$ when $Y \neq Z$ and
\[
\V_s(Z) = \preimg{{\prec}}{s} \subseteq \preimg{{\prec}}{S} = \V_S(Z).
\]
Monotonicity then guarantees that
$
\semT{\sigma'Z'.\Gamma}{\V_s} \subseteq \semT{\sigma'Z'.\Gamma}{\V_S},
$
justifying this step of argument. 

To prove that $S = \semT{\sigma Z.\Phi'}{\V'}$ note that since we just established that $(S,\prec)$ is a $\sigma$-compatible support ordering for $f_{\Phi'}$, Theorem~\ref{thm:well-supported} and Corollary~\ref{cor:support-fixpoints} guarantee that $S \subseteq S''$, where $S'' = \sigma f_{\Phi'} = \semT{\sigma Z.\Phi'}{\V'}$.
It remains to show that $S = S''$.
Suppose to the contrary that $S \neq S''$.  Since we already know that $S \subseteq S''$ it must be the case that $S \subsetneq S''$.  
Based on Theorem~\ref{thm:well-supported} and Corollary~\ref{cor:support-fixpoints} it follows that there must exist a relation ${\prec'} \subseteq S'' \times S''$ such that $(S'',\prec')$ is a $\sigma$-compatible support ordering for $f_{\Phi'}$. 
However, this yields a contradiction, because we can then construct a relation ${\prec''} \subseteq S'' \times S''$ such that $(S'', \prec'')$ is a $\sigma$-compatible support ordering for $f_{\Phi}$.  As $S = \sigma f_\Phi$ this would imply that $S'' \subseteq S$, and thus that $S = S''$.  The construction of $\prec''$ is as follows.
\[
{\prec''}
=
{\prec}
\;\cup\;
\left((\preimg{{\prec}}{S}) \times (S'' \setminus S)\right)
\;\cup\;
\{ (s_1, s_2) \mid s_2 \in S'' \setminus S \land s_1 \prec' s_2 \}
\]
Intuitively, $s_1 \prec'' s_2$ when one of the following hold.
\begin{itemize}
    \item $s_1 \prec s_2$ (in this case, $s_1, s_2 \in S$); or
    \item $s_1 \in \preimg{{\prec}}{S}$ (so $s_1 \in S$ also) and $s_2 \in S''$ but $s_2 \not\in S$; or
    \item $s_1 \prec' s_2$ and $s_1, s_2 \not\in S$
\end{itemize}
Based on this definition the following can easily be established.
\[
\preimg{{{\prec}''}}{s} =
\begin{cases}
    \preimg{{\prec}}{s} & \text{if $s \in S$} \\
    \preimg{{\prec}}{S} \cup \preimg{{\prec'}}{s}    & \text{if $s \in S''\setminus S$}
\end{cases}
\]
Note that if $\prec$ and $\prec'$ are well-founded then so is $\prec''$, so $(S, \prec'')$ must be $\sigma$-compatible.
We now show that $(S'', \prec'')$ is a support ordering for $f_{\Phi}$ by establishing that for every $s \in S''$, $s \in f_{\Phi}(\preimg{{\prec''}}{s})$.
There are two cases to consider.
When $s \in S$ this is immediate from the fact that $\preimg{{\prec''}}{s} = \preimg{{\prec}}{s}$; since $(S,\prec)$ is a support-ordering for $f_{\Phi}$ this means $s \in f_\Phi(\preimg{{\prec}}{s}) = f_\Phi(\preimg{{\prec''}}{s})$.
Now suppose $s \in S'' \setminus S$; then $\preimg{{\prec''}}{s} = \preimg{{\prec}}{S} \cup \preimg{{\prec'}}{s}$.  This implies that $\preimg{{\prec}}{S} \subseteq \preimg{{\prec''}}{s}$.  Now recall that $S' = \semT{\sigma' Z'.\Gamma}{\V_S}$, and that $\V_S = \V[Z := \preimg{{\prec}}{S}]$.  It follows that $S' \subseteq \semT{\sigma' Z'.\Gamma}{\V''}$, where $\V'' = \V[Z := \preimg{{\prec''}}{s}]$, and from this we may reason as follows.
\begin{align*}
f_{\Phi'} (\preimg{{\prec''}}{s}) 
&= \semT{\Phi'}{\V'[Z := \preimg{{\prec''}}{s}]}
\\
&= \semT{\Phi'}{\V[Z, W:= \preimg{{\prec''}}{s}], S']}
\\
&\subseteq \semT{\Phi'}{\V[Z, W := \preimg{{\prec''}}{s}, \semT{\sigma Z'.\Phi'}{\V[Z := \preimg{{\prec''}}{s}]}]}
\\
&= \semT{\Phi'[W := \sigma Z'.\Phi']}{\V[Z := \preimg{{\prec''}}{s}]}
\\
&= \semT{\Phi}{\V[Z := \preimg{{\prec''}}{s}]}
\\
&= f_\Phi(\preimg{{\prec''}}{s}).
\end{align*}
Based on this and the fact that $s \in f_{\Phi'} (\preimg{{\prec''}}{s})$, we have that
$s \in f_{\Phi} (\preimg{{\prec''}}{s})$ as well, and $(S'',\prec'')$ is a support ordering for $f_\Phi$ in addition to $f_{\Phi'}$.  We have arrived at the contradiction we set out to establish, and it must be the case that $S = S''$.  Thus $S = \semT{\sigma Z.\Phi'}{\V'}$ and $(S, \prec)$ is a $\sigma$-compatible support ordering for $f_{\Phi'} = \semfT{Z}{\Phi'}{\V'}$.
We may now apply the induction hypothesis to infer the existence of successful TNF tableau
\[
\mathbb{T}_{\Phi'} = (\tree{T}_{\Phi'}, \rho_{\Phi'}, \T, \V', \lambda_{\Phi'}),
\]
where $\tree{T}_{\Phi'} = (\node{N}_{\Phi'}, \node{r}_{\Phi'}, p_{\Phi'}, cs_{\Phi'})$, such that  $\lambda_{\Phi'}(\node{r}_{\Phi'}) = S \tnxT{\V'}{\emptyL} \sigma Z.\Phi'$ and $\mathbb{T}_{\Phi'}$ is compliant with $(S, \prec)$.
It is also easy to see that $\mathbb{T}_{\Phi'}$ contains exactly one successful leaf $\node{n}_W$ such that $\seqfm(\lambda_{\Phi'}(\node{n}_W)) = W$.

\paragraph{Step~\ref{it:step-inner-tableau} of proof outline:  construct tableau for $S' \tnxT{\V_S}{\emptyL} \sigma'Z'.\Gamma$.}

Since $\Gamma$ contains strictly fewer fixpoint subformulas than $\Phi$, the induction hypothesis also guarantees the existence of certain successful TNF tableaux involving $\sigma' Z'.\Gamma$.
We remark on these tableaux and show how to use them to construct a successful TNF tableau $\mathbb{T}_\Gamma$ for sequent $S' \tnxT{\V_S}{\emptyL} \sigma'Z'.\Gamma$ 
and also ensuring a key property, formalized below as Property~\ref{goal:dependencies}, involving extended dependencies between $s \in S$ and states in $\semT{\sigma'Z'.\Gamma}{\V_s}$.
In what follows, for any $X \subseteq S$ define
\begin{align*}
g_X &= g_{(\preimg{{\prec}}{X},\cdot)}
\\
S'_X &= \sigma' g_X
\end{align*}
That is, $g_X \in 2^{\states{S}} \to 2^{\states{S}}$ computes the semantics of $\Gamma$, which has both $Z$ and $Z'$ free, with the semantics of $Z$ fixed to be $\preimg{{\prec}}{X}$ and $Z'$ interpreted as the input provided to $g_X$.  $S'_X$ is then the $\sigma'$ fixpoint of $g_X$.  It can easily be seen that for any $X \subseteq S$, $\preimg{{\prec}}{X} \subseteq \preimg{{\prec}}{S}$ and thus
\[
S'_X = \sigma' g_X \subseteq \sigma' g_S = \semT{\sigma'Z'.\Gamma}{\V_S} = S'.
\]

In order to apply the induction hypothesis to generate successful TNF compliant tableaux for sequents involving $\sigma'Z'.\Gamma$ we also need $\sigma'$-compatible, total qwf support orderings for the values of $g_X$ we wish to consider.  
We obtain these from Lemma~\ref{lem:fg-support} as follows.  
Recall that $(S,\prec)$ is a $\sigma$-compatible total qwf support ordering for $f_\Phi$, and that $f_\Phi = f [\sigma'] g$.  
The lemma guarantees the existence of a $\sigma'$-compatible, total qwf support ordering $(S', \prec')$ for $g_S$
that is locally consistent\footnote{See Definition~\ref{def:consistent-support-orderings} for the meaning of local consistency.} with $(S,\prec)$.

Now let partial order $(Q_\prec, \sqsubseteq)$ be the quotient of $(S,\prec)$ (cf.\/ Definition~\ref{def:relation-quotient}), with $\sqsubseteq^-$ the irreflexive core of $\sqsubseteq$. 
Totality of $\prec$ guarantees that if $Q_1 \sqsubseteq Q_2$ then $\preimg{{\prec}}{Q_1} \subseteq \preimg{{\prec}}{Q_2}$.
If $x \in S$ then we write $[x] \in Q_\prec$ for the equivalence class of $x$.  Since $\prec$ is total it follows that for all $x, x'$, if $x' \in [x]$ then also $x \in [x']$ and $\preimg{{\prec}}{x} = \preimg{{\prec}}{x'}$. 

The construction we present below for $\mathbb{T}_\Gamma$ proceeds in three steps.
\begin{itemize}
    \item
    For each $Q \in Q_\prec$ we inductively construct a successful TNF tableau $\mathbb{T}_{\Gamma,Q}$ for sequent $S'_Q \tnxT{\V_Q}{\emptyL} \sigma' Z'.\Gamma$ that is compliant with a subrelation of $\prec'$.
    \item
    We then merge the individual $\mathbb{T}_{\Gamma,Q}$ to form a successful TNF tableau $\mathbb{T}'_\Gamma$ compliant with $\prec'$ whose root sequent contains as its state set the union of all the individual root-sequent state sets of the $\mathbb{T}_{\Gamma,Q}$.
    \item
    We perform a final operation to obtain $\mathbb{T}_\Gamma$.
\end{itemize}

\textit{Constructing $\mathbb{T}_{\Gamma,Q}$.}
We begin by noting that since $(S',\prec')$ is locally consistent with $(S,\prec)$ and $\preimg{{\prec}}{x} = \preimg{{\prec}}{x'}$ if $x \in [x']$ it follows that for any $Q \in Q_\prec$, $(S'_Q, \prec'_Q)$, where ${\prec'_Q} = \restrict{{\prec'}}{S'_Q}$, is a $\sigma'$-compatible, total qwf support ordering for $g_Q$.  Thus, based on the induction hypothesis there exists, for each $Q \in Q_\prec$, a successful TNF tableau 
\[
\mathbb{T}_{\Gamma,Q} = (\tree{T}_{\Gamma,Q}, \rho_{\Gamma,Q}, \T, \V_Q, \lambda_{\Gamma,Q})
\]
for $S'_Q \tnxT{\V_Q}{\emptyL} \sigma'Z'.\Gamma$,
where $\tree{T}_{\Gamma,Q} = (\node{N}_{\Gamma,Q}, \node{r}_{\Gamma,Q}, p_{\Gamma,Q}, cs_{\Gamma,Q})$, that is compliant with $\prec'_Q$.

We now note that Lemma~\ref{lem:structural-equivalence-of-TNF-tableaux} guarantees that for all $Q, Q' \in Q_\prec$, $\mathbb{T}_{\Gamma,Q}$ and $\mathbb{T}_{\Gamma,Q'}$ are structurally equivalent.
This means:
\begin{enumerate}
    \item $\tree{T}_{\Gamma,Q}$ and $\tree{T}_{\Gamma,Q'}$ are isomorphic; and
    \item For isomorphic nodes $\node{n}$ and $\node{n}'$ in 
        $\tree{T}_{\Gamma,Q}$ and $\tree{T}_{\Gamma,Q'}$, respectively, $\rn(\rho_{\Gamma,Q}(\node{n})) = \rn(\rho_{\Gamma,Q'}(\node{n}'))$,
        $\seqfm(\lambda_{\Gamma,Q}(\node{n})) = \seqfm(\lambda_{\Gamma,Q'}(\node{n}'))$, and
        $\seqdl(\lambda_{\Gamma,Q}(\node{n})) = \seqdl(\lambda_{\Gamma,Q}(\node{n}'))$.
\end{enumerate}
In other words, the only differences between these tableaux are the state sets appearing in the sequents at each tree node and the witness functions used in rule applications involving $\dia{K}$.
We consequently assume in what follows that there is a single common tree $\tree{T}_\Gamma = (\node{N}_\Gamma, \node{r}_\Gamma, p_\Gamma, cs_\Gamma)$.
We also introduce the following functions with respect to $\node{N}$ that return the common elements in the sequents and rule applications labeling the tree nodes.
\begin{align*}
    \seqfm_\Gamma(\node{n})    &- \text{the formula labeling $\node{n}$ in all the $\mathbb{T}_{\Gamma,Q}$} 
    \\
    \seqdl_\Gamma(\node{n})    &- \text{the definition list labeling $\node{n}$ in all the $\mathbb{T}_{\Gamma,Q}$}
    \\
    \rn_\Gamma(\node{n})       &- \text{the rule name in the rule application for $\node{n}$ in all the $\mathbb{T}_{\Gamma,Q}$}
\end{align*}

%
\textit{Constructing $\mathbb{T}'_\Gamma$.}
When $Q_\prec \neq \emptyset$ we can merge the non-empty set of tableaux $\{ \mathbb{T}_{\Gamma,Q} \mid Q \in Q_\prec\}$ into a single successful tableau $\mathbb{T}'_\Gamma = (\tree{T}_\Gamma, \rho'_\Gamma, \T, \V_{S}, \lambda'_\Gamma)$, sharing 
the same tree $\tree{T}_\Gamma$ and functions $\rho_\Gamma$, $\seqfm_\Gamma$ and $\seqdl_\Gamma$ as the $\mathbb{T}_{\Gamma,Q}$,
with the following properties.
\begin{enumerate}[left = \parindent, label = G\arabic*., ref = G\arabic*]
    \item\label{goal:root}
    $\lambda'_\Gamma(\node{r}_\Gamma) = \left(\bigcup_{Q \in Q_{\prec}} S'_Q \right) \tnxT{\V_S}{\emptyL} \sigma' Z'.\Gamma$.
    \item\label{goal:dependencies}
    For all $x \in S$, $y \in S'_{[x]}$, and $x' <:_{\node{n}',\node{r}_\Gamma} y$ in $\mathbb{T}'_\Gamma$ such that $\seqfm_\Gamma(\node{n}') = Z$, $x' \prec x$.
\end{enumerate}
Property~\ref{goal:dependencies} is important later in the proof, and we comment on it briefly here.
The property asserts that if $y$ is a state in the root of $\mathbb{T}'_{\Gamma,[x]}$ for some $x \in S$, and if there is an extended dependency between $y$ and some state $x'$ in a $Z$-leaf ($Z$ being the bound variable in $\sigma Z.\Phi$), then $x' \prec x$.  In other words, $y$ can only depend on states in $Z$-leaves that (semantically) support $x$.

As $\tree{T}_\Gamma, \T$ and $\V_S$ are already defined, completing the construction of $\mathbb{T}'_\Gamma$ only requires that we define $\rho'_\Gamma$ and $\lambda'_\Gamma$, which we do so that the following invariants hold for each $\node{n} \in \node{N}$.
\begin{invariants}
    \item
      If $\rho'_\Gamma(\node{n})$ is defined then the sequents assigned by $\lambda'_\Gamma$ to $\node{n}$ and its children are consistent with rule application $\rho'_\Gamma(\node{n})$.
    \item
      $\seqfm(\lambda'_\Gamma(\node{n})) = \seqfm(\node{n})$
    \item
      $\seqst(\lambda'_\Gamma(\node{n})) \subseteq \bigcup_{Q \in Q_\prec} \seqst(\lambda_{\Gamma,Q}(\node{n}))$
    \item\label{inv:definition-list}
      $\seqdl(\lambda'_\Gamma(\node{n})) = \seqdl_\Gamma(\node{n})$
\end{invariants}
These invariants are suitably updated versions of the invariants appearing in the proof of Lemma~\ref{lem:single-fixpoint-completeness}.

The definitions of $\rho'_\Gamma$ and $\lambda'_\Gamma$ are given in a co-inductive fashion (i.e.\/ ``from the root down" rather than ``the leaves up").
More specifically, the construction first assigns a sequent to $\node{r}_\Gamma$, the root of $\tree{T}_\Gamma$ that ensures that Property~\ref{goal:root} holds.  It immediately follows that Properties~\ref{inv:formula}, \ref{inv:state-set} and \ref{inv:definition-list} also hold for the root. Then for any non-leaf node $\node{n}$ whose sequent satisfies \ref{inv:formula}, \ref{inv:state-set} and \ref{inv:definition-list}, the co-induction step defines sequents for each child of $\node{n}$ so that these sequents each satisfy \ref{inv:formula}, \ref{inv:state-set} and \ref{inv:definition-list} and so that Property~\ref{inv:rule} holds for $\node{n}$.
Property~\ref{goal:dependencies} will be proved later.

We begin by defining $\lambda'_\Gamma(\node{r}_\Gamma) =  \left(\bigcup_{Q \in Q_{\prec}} S'_Q \right) \tnxT{\V_S}{\emptyL} \sigma' Z'.\Gamma$.  Property~\ref{goal:root} is immediate, as is the fact that \ref{inv:formula}, \ref{inv:state-set} and \ref{inv:definition-list} hold for $\node{r_\Gamma}$.

For the co-inductive step, suppose $\lambda'_\Gamma(\node{n}) = S_\node{n} \tnxT{\V_S}{\Delta_{\node{n}}} \Gamma_{\node{n}}$ and that invariants \ref{inv:formula}, \ref{inv:state-set} and \ref{inv:definition-list} all hold for $\lambda'_\Gamma(\node{n})$.  We define $\lambda'_\Gamma(\node{n}')$ for each child $\node{n}'$ of $\node{n}$, and $\rho'_{\Gamma}(\node{n})$, the rule application for $\node{n}$, so as to ensure that invariant \ref{inv:rule} holds for $\node{n}$ and that \ref{inv:formula}, \ref{inv:state-set} and \ref{inv:definition-list} are established for each child $\node{n}'$.  The constructions in many cases closely match those found in Lemma~\ref{lem:single-fixpoint-completeness}.
\begin{description}

    \item[$\rn_\Gamma(\node{n})\div$, or $\rn_\Gamma(\node{n}) \in \{\land,\lor, {[K]},\dia{K}\}$.]
    The constructions in this case mirror those in Lemma~\ref{lem:single-fixpoint-completeness} for $\lambda_S$; the only difference is that the definition lists in the child sequents are inherited from the parent, rather than always being $\emptyL$.  It is straightforward to see that invariant~\ref{inv:rule} holds for $\node{n}$ while \ref{inv:formula}, \ref{inv:state-set} and \ref{inv:definition-list} hold for the children of $\node{n}$.
    
    \item[$\rn_\Gamma(\node{n}) = \sigma Z''$.]
    In this case $\seqfm_\Gamma(\node{n}) = \sigma''Z''.\Gamma'$ for some $\sigma''$, $Z''$ and $\Gamma'$; $cs(\node{n}) = \node{n}'$;
    $\seqdl_\Gamma(\node{n}') = \Delta_{\node{n}'} = \Delta_{\node{n}} \cdot (U' = \seqfm_\Gamma(\node{n}))$ for some $U' \not\in \dom(\Delta_{\node{n}})$;
    and $\seqfm_\Gamma(\node{n}') = U'$.  Define $\rho'_{\Gamma}(\node{n}) = \sigma Z''$ and $\lambda'_\Gamma(\node{n}') = S_{\node{n}} \tnxT{\V_S}{\Delta_{\node{n}'}} U'$.  
    It is easy to establish that invariant~\ref{inv:rule} holds for $\node{n}$ and that \ref{inv:formula}, \ref{inv:state-set} and \ref{inv:definition-list} hold for $\node{n}'$.
    
    \item[$\rn_\Gamma(\node{n}) = \textnormal{Un}$.]
    In this case $cs(\node{n}) = \node{n}'$, 
    $\seqdl_\Gamma(\node{n}') = \Delta_\node{n}$
    and $\seqfm_\Gamma(\node{n}) = U'$ for some $U' \in \dom(\Delta_{\node{n}})$.  Let $\Delta_{\node{n}}(U') = \sigma''Z''.\Gamma'$; then $\seqfm_\Gamma(\node{n}') = \Gamma'[Z'' := U']$.  
    Define 
    $\rho'_\Gamma(\node{n}) = \text{Un}$ and
    $\lambda'_\Gamma(\node{n}') = S_{\node{n}} \tnxT{\V_S}{\Delta_{\node{n}}}
    \seqfm_\Gamma(\node{n}')$. 
    It is easy to establish that invariant~\ref{inv:rule} holds for $\node{n}$ and that \ref{inv:formula}, \ref{inv:state-set} and \ref{inv:definition-list} hold for $\node{n}'$.
    
    \item[$\rn_\Gamma(\node{n}) = \textnormal{Thin}.$]
    In this case $cs(\node{n}) = \node{n}'$,
    $\seqdl_\Gamma(\node{n}') = \Delta_{\node{n}}$ and
    $\seqfm_\Gamma(\node{n}) = \seqfm_\Gamma(\node{n}')$.
    Take $S_{\node{n}'} = \bigcup_{Q \in Q_\prec} \seqst(\lambda_{\Gamma,Q}(\node{n}'))$, and
    define 
    $\rho'_\Gamma(\node{n}) = \text{Thin}$ and
    $\lambda'_\Gamma(\node{n}') = S_{\node{n}'} \tnxT{\V_S}{\Delta_{\node{n}}} \seqfm_\Gamma(\node{n}')$. 
    Invariant~\ref{inv:rule} can be shown to hold for $\node{n}$, while \ref{inv:formula}, \ref{inv:state-set} and \ref{inv:definition-list} hold for $\node{n}'$.  
    (Indeed, a stronger version of \ref{inv:state-set} holds in this case, as $\seqst(\lambda'_\Gamma(\node{n}')) = \bigcup_{Q \in Q_\prec} \seqst_{\Gamma,Q}(\node{n}')$.)
\end{description}

We now argue that $\mathbb{T}'_\Gamma$ is a successful TNF tableau that is compliant with $\prec'$ and that Property~\ref{goal:dependencies} holds (Property~\ref{goal:root} has already been established.).
We first must establish that $\mathbb{T}'_\Gamma$ is indeed a tableau.  Because of invariant~\ref{inv:rule} it suffices to show that the sequent labeling any leaf node is terminal (cf.\/ Definition~\ref{def:tableau}(\ref{subdef:complete-tableau})).  Let $\node{n}$ be a leaf in $\tree{T}_\Gamma$; we note that $\seqfm(\node{n})$ has form either $Z''$ or $\lnot Z''$, where $Z''$ is free in $\sigma'Z'.\Gamma$, or $U'$ for some definitional constant $U' \in \dom(\seqdl(\node{n}))$.  
In the first two cases $\node{n}$ is clearly terminal.  
In the latter case we must argue that $\seqst(\lambda'_\Gamma(\node{n})) \subseteq \seqst(\lambda'_\Gamma(\node{m}))$, where $\node{m}$ is the intended companion node of $\node{n}$ (i.e.\/ the strict ancestor of $\node{n}$ such that $\seqfm_\Gamma(\node{m}) = \seqfm_\Gamma(\node{n}) = U'$).  From the definition of $\mathbb{T}'_\Gamma$ and that fact that each $\mathbb{T}_{\Gamma,Q}$ is a TNF tableau we observe the following.
\begin{enumerate}
    \item 
        $\rho'_\Gamma(\node{m}) = \text{Un}$, and $\node{m}$ is the only internal node whose formula is $U'$.
    \item
        Since $p_\Gamma(\node{m})$ is the parent of $\node{m}$, $\rho'_\Gamma(p_\Gamma(\node{m})) = \sigma Z''.$ for some $Z''$.
    \item
        Either $p_\Gamma(\node{m})$ is the root of $\tree{T}_\Gamma$ (i.e.\/ $p(\node{m}) = \node{r}_\Gamma$), or $p(p(\node{m}))$, the grandparent of $\node{m}$, is defined, and $\rho'_\Gamma(p(p(\node{m}))) = \text{Thin}$.  In either case, from the definition of the construction it can be shown that $\seqst(\lambda_\Gamma(\node{m})) = \bigcup_{Q \in Q_\prec} \seqst(\lambda_{\Gamma,Q}(\node{m}))$.
\end{enumerate}
Because each $\mathbb{T}_{\Gamma,Q}$ is a tableau  it follows that for each $Q \in Q_\prec$, $\node{n}$ is terminal in $\mathbb{T}_{\Gamma,Q}$ and thus $\seqst(\lambda_{\Gamma,Q}(\node{n})) \subseteq \seqst(\lambda_{\Gamma,Q}(\node{m}))$.
Also, since invariant~\ref{inv:state-set} holds of $\node{n}$ in $\mathbb{T}'_\Gamma$ we have that $\seqst(\lambda'_\Gamma(\node{n})) \subseteq \bigcup_{Q \in Q_\prec} \seqst(\lambda_{\Gamma,Q}(\node{n}))$.  We can now reason as follows
\[
\seqst(\lambda'_\Gamma(\node{n}))
\subseteq \bigcup_{Q \in Q_\prec} \seqst(\lambda_{\Gamma, Q}(\node{n}))
\subseteq \bigcup_{Q \in Q_\prec} \seqst(\lambda_{\Gamma,Q}(\node{m}))
=         \seqst(\lambda'_\Gamma(\node{m}))
\]
to see that $\node{n}$ is terminal in $\mathbb{T}'_\Gamma$ and thus $\mathbb{T}'_\Gamma$ is indeed a tableau.

We now need to show that $\mathbb{T}'_\Gamma$ is successful and compliant with $\prec'$, and that Property~\ref{goal:dependencies} holds of the tableau.  
Before doing that, however, we remark on properties of dependency relations in $\mathbb{T}'_\Gamma$ that will be used in the arguments to follow.  
To begin with, the following holds of all $\node{n}_1, \node{n}_2, s_1$ and $s_2$ such that $s_2 <_{\node{n}_2,\node{n}_1} s_1$ in $\mathbb{T}'_\Gamma$:
\begin{quote}
for every $Q \in Q_\prec$ such that $s_1 \in \seqst(\lambda_{\Gamma, Q}(\node{n}_1))$, either $s_2 \in \seqst(\lambda_{\Gamma,Q}(\node{n_2}))$, and thus $s_2 <_{\node{n}_2,\node{n}_1} s_1$ in $\mathbb{T}_{\Gamma,Q}$, or there exists $Q'$ such that $Q' \sqsubset Q$ such that $s_2 \in \seqst(\lambda_{\Gamma,Q'}(\node{n}_2))$.
\end{quote}
This is immediate from the definition of $<_{\node{n}_2, \node{n}_1}$ (Definition~\ref{def:local_dependency_ordering}) and $\mathbb{T}'_\Gamma$:  the need for the $Q'$ case comes from the construction used for rules $\lor$ and $\dia{K}$.%
\footnote{An analogous property is is used in the proof of Lemma~\ref{lem:single-fixpoint-completeness}.}  
An inductive argument based on the definition of $<:$ lifts this result to $\node{n}_1, \node{n}_2, s_1$ and $s_2$ such that $s_2 <:_{\node{n}_2,\node{n}_1} s_2$ in $\mathbb{T}'_\Gamma$:  for all $Q \in Q_\prec$ such that $s_1 \in \seqst(\lambda_{\Gamma,Q}(\node{n}_1))$, either $s_2 \in \seqst(\lambda_{\Gamma,Q}(\node{n}_2))$ and $s_2 <:_{\node{n}_2,\node{n}_1} s_1$ in $\mathbb{T}_{\Gamma,Q}$, or there exists $Q' \sqsubset Q$ such that $s_2 \in \seqst(\lambda_{\Gamma,Q'}(\node{n}_2))$.

We now establish that $\mathbb{T}'_\Gamma$ is successful by showing that every leaf in $\mathbb{T}'_\Gamma$ is successful (cf.\/ Definition~\ref{def:successful-tableau}), 
which amounts to showing that for each leaf $\node{n}$, 
$\seqst(\lambda'_\Gamma(\node{n})) \subseteq \semT{\seqfm_\Gamma(\node{n})}{\V_S}$.  
There are four cases to consider.
\begin{description}
    \item[$\seqfm_\Gamma(\node{n}) = Z$.]
        Analogous to the same case in the proof of Lemma~\ref{lem:single-fixpoint-completeness}.
    \item[$\seqfm_\Gamma(\node{n}) \in \{Z'', \lnot Z''\}$ for some $Z'' \neq Z$ free in $\sigma'Z'.\Gamma$.]
        Analogous to the same case in the proof of Lemma~\ref{lem:single-fixpoint-completeness}.
    \item[$\node{n}$ is a $\nu$-leaf.]
        In this case $\seqfm_\Gamma(\node{n})$ is successful by definition.
    \item[$\node{n}$ is a $\mu$-leaf.]
        Let $\node{m}$ be the companion node of $\node{n}$; we must show that $<:_\node{m}$ is well-founded in $\mathbb{T}'_\Gamma$. 
        Suppose to the contrary that this is not the case, i.e. that there is an infinite descending chain 
        $\cdots <:_\node{m} s_2 <:_\node{m} s_1$ 
        with each $s_i \in \seqst(\lambda'_\Gamma((\node{m}))$.
        From the definition of $<:$ (cf.\/ Definition~\ref{def:extended_path_ordering}) it follows that for all $j > 1$, $s_j <:_{\node{n_j},\node{m}} s_{j-1}$ in $\mathbb{T}'_\Gamma$ for some companion leaf $\node{n}_j$ of $\node{m}$ (note that $\node{n}$ is one of these $\node{n}_j$).
        Since each $\mathbb{T}_{\Gamma,Q}$ is successful we know that $<:_{\node{m}}$ in $\mathbb{T}_{\Gamma,Q}$ is well-founded for any $Q \in Q_\prec$.
        Recall also that $\sqsubset$ is a well-ordering on $Q_\prec$.
        Now consider $\cdots <:_{\node{m}} s_2 <:_{\node{m}} s_1$.  Since $s_1 \in \seqst(\lambda'_\Gamma(\node{m}))$ invariant~\ref{inv:state-set} guarantees that $s_1 \in \seqst(\lambda_{\Gamma,Q}(\node{n}_1))$ for some $Q \in Q_\prec$.
        Now consider the $s_j$, $j > 1$.
        Since each $s_j <:_{\node{n}_j,\node{m}} s_1$ in $\mathbb{T}'_\Gamma$ the preceding argument ensures that 
        either $s_j <:_{\node{n},\node{m}} s_1$ in $\mathbb{T}_{\Gamma,Q}$, and thus $s_j <:_{\node{m}} s_1$ in $\mathbb{T}_{\Gamma,Q}$, 
        or there exists $s' \in S$ such that $[s'] \sqsubset Q$ and $s_j \in \seqst(\lambda_{\Gamma,[s']}(\node{n}_j))$.
        However, $<:_{\node{m}}$ in $\mathbb{T}_{\Gamma,Q}$ is well-founded, so only finitely many of the $s_j$ can satisfy $s_j <:_{\node{m}} s_1$ in $\mathbb{T}_{\Gamma,Q}$; there must be some $j > 1$ such that $s_j \in \seqst(\lambda_{\Gamma,[s']}(\node{n}_2))$ some $[s'] \sqsubseteq Q$.  But then, for $\cdots <:_{\node{m}} s_2 <:_{\node{m}} s_1$ to be an infinite descending chain there must be an infinite descending chain in $\sqsubset$.  As $\sqsubset$ is well-founded, this is a contradiction, and $<:_\node{m}$ must be well-founded in $\mathbb{T}'_\Gamma$, meaning $\node{n}$ is a successful leaf.
\end{description}

To prove compliance of $\mathbb{T}'_\Gamma$ with $\prec'$, assume that $s_2 <:_{\node{r}_\Gamma} s_1$ in $\mathbb{T}'_\Gamma$; we must show that $s_2 \prec' s_1$.  Let $\node{n}$ be a companion leaf of $\node{r}_\Gamma$ such that $s_2 <:_{\node{n},\node{r}_\Gamma} s_1$.
From the arguments above we know that $s_1 \in \seqst(\lambda_{\Gamma,Q}(\node{r}_\Gamma))$ for some $Q \in Q_\prec$; assume further $Q$ is the minimum such element in $Q_\prec$ with respect to $\sqsubset$ (which is guaranteed to exist because $\sqsubset$ is a well-ordering on $Q_\prec$). Also, either $s_2 <:_{\node{n},\node{r}_\Gamma} s_1$ in $\mathbb{T}_{\Gamma,Q}$ or there is $Q' \sqsubset Q$ such that $s_2 \in \seqst(\lambda_{\Gamma,Q'}(\node{n}))$.  
In the former case $s_2 <:_{\node{r}_\Gamma} s_1$, and the fact that $\mathbb{T}_{\Gamma,Q}$ is successful and compliant with ${\prec'_Q} \subseteq {\prec'}$ guarantees that $s_2 \prec' s_1$.  
In the latter case, since $\node{n}$ is a companion leaf of $\node{r}_\Gamma$ we know that $s_2 \in \seqst(\lambda_{\Gamma,Q}(\node{r}_\Gamma))$ and $s_2 \in \seqst(\lambda_{\Gamma,Q'}(\node{r}_\Gamma))$.
Moreover, the fact that $Q$ is minimum ensures that $s_1 \not\in \seqst(\lambda_{\Gamma,Q'}) = S'_{Q'} = \semT{\sigma'Z'.\Gamma}{\V_{Q'}}$.
However, since $\prec'$ is locally consistent with $\prec$ means that $s_1 \not\prec' s_2$, and as $\prec'$ is total it must be that $s_2 \prec' s_1$.

We now prove Property~\ref{goal:dependencies} for $\mathbb{T}'_\Gamma$.  So fix $x \in S$, $y' \in S'_{[x]}$ and $x' <:_{\node{n}',\node{r}_\Gamma} y$ where $\seqfm_\Gamma(\node{n}') = Z$.  Note that $x' \in \preimg{{\prec}}{S}$.  We must show that $x' \prec x$. 
From facts established above we know that either $x' <:_{\node{n}',\node{r}_\Gamma} y$ in $\mathbb{T}_{\Gamma,[x]}$, or there exists a $Q \sqsubset [x]$ such that $x' \in \seqst(\lambda_{\Gamma,Q}(\node{n}'))$.  In the former case the success of $\mathbb{T}'_Q$ guarantees that $x' \in \V_{[x]}(Z) = \preimg{{\prec}}{[x]}$, meaning $x' \prec x$.  In the latter case $x' \in \V_{Q}(Z) = \preimg{{\prec}}{Q}$; since $\prec$ is total and $Q \sqsubset [x]$ it follows that $\preimg{{\prec}}{Q} \subseteq \preimg{{\prec}}{[x]}$, so $x' \in \preimg{{\prec}}{[x]}$ and $x' \prec x$.

\textit{Construction of $\mathbb{T}_\Gamma$.}
We now show how to construct successful TNF $\mathbb{T}_\Gamma$ for sequent $S' \tnxT{\V_S}{\emptyL} \sigma'Z'.\Gamma$ that is compliant with $\prec'$ and satisfies Property~\ref{goal:dependencies}.

We begin by noting that since $(S',\prec')$ is a $\sigma'$-compatible, total qwf support ordering for $g_S$, the induction hypothesis and Lemma~\ref{lem:structural-equivalence-of-TNF-tableaux} guarantee the existence of a successful TNF tableau 
$$
\mathbb{T}_{\Gamma,S} = (\tree{T}, \rho_{\Gamma,S}, \T, \V_S, \lambda_{\Gamma,S})
$$
for sequent $S' \tnxT{\V_S}{\emptyL} \sigma'Z'.\Gamma$ that is compliant with $\prec'$ and structurally equivalent to $\mathbb{T}_{\Gamma,Q}$ for any $Q \in Q_\prec$.  
There are two cases to consider.  In the first case, $S = \emptyset$.  In this case, $\mathbb{T}_{\Gamma,S}$ vacuously satisfies Property~\ref{goal:dependencies}, and we take $\mathbb{T}_\Gamma$ to be $\mathbb{T}_{\Gamma,S}$.

In the second case, $S \neq \emptyset$, and thus $Q_\prec \neq \emptyset$.  In this case it is not guaranteed that $\mathbb{T}_{\Gamma,S}$ satisfies \ref{goal:dependencies}, so we cannot take $\mathbb{T}$ to be $\mathbb{T}_{\Gamma,S}$.  Instead, we build $\mathbb{T}_\Gamma$ using a coinductive definition of $\lambda_\Gamma$ that merges $\mathbb{T}'_\Gamma$ and $\mathbb{T}_{\Gamma,S}$ in a node-by-node fashion, starting with $\node{r}_\Gamma$, the root, so that certain invariants are satisfied.  The invariants in this case are the same as \ref{inv:rule}--\ref{inv:definition-list} from construction of $\mathbb{T}'_\Gamma$, but adapted as follows.
\begin{invariants}
    \item
      If $\rho_\Gamma(\node{n})$ is defined then the sequents assigned by $\lambda_\Gamma$ to $\node{n}$ and its children are consistent with the rule application $\rho_\Gamma(\node{n})$.
    \item
      $\seqfm(\lambda_\Gamma(\node{n})) = \seqfm(\node{n})$
    \item
      $\seqst(\lambda_\Gamma(\node{n})) \subseteq \seqst(\lambda'_\Gamma(\node{n})) \cup \seqst(\lambda_{\Gamma,S}(\node{n}))$
    \item
      $\seqdl(\lambda_\Gamma(\node{n})) = \seqdl_\Gamma(\node{n})$
\end{invariants}
(That is, $\lambda'_\Gamma$ is replaced by $\lambda_\Gamma$, and $\bigcup_{Q \in Q_\prec} \seqst(\lambda_{\Gamma,Q}(\node{n}))$ is replaced by $\seqst(\lambda'_\Gamma(\node{n})) \cup \seqst(\lambda_{\Gamma,S}(\node{n}))$.)
As before, the definition of $\lambda_\Gamma$ begins by assigning a value to $\lambda_\Gamma(\node{r}_\Gamma)$ so that invariants~\ref{inv:formula}--\ref{inv:definition-list} are satisfied.  The coinductive step then assumes that $\node{n}$ satisfies these invariants and defines $\lambda_\Gamma$ for the children of $\node{n}$ and $\rho_\Gamma(\node{n})$ so that \ref{inv:rule} holds for $\node{n}$ and \ref{inv:formula}--\ref{inv:definition-list} hold for each child.

To start the construction, define $\lambda_\Gamma(\node{r}_\Gamma) = S' \tnxT{\V_S}{\emptyL} \sigma'Z'.\Gamma$.  Invariants~\ref{inv:formula}--\ref{inv:definition-list} clearly hold of $\node{r}_\Gamma$. In particular, it should be noted that 
$$
\seqst(\lambda_\Gamma(\node{r}_\Gamma)) = S' = \seqst(\lambda'_\Gamma(\node{r}_\Gamma)) \cup \seqst(\lambda_\Gamma(\node{r}_\Gamma)),
$$
since $\seqst(\lambda'_\Gamma(\node{r}_\Gamma)) = \left(\bigcup_{Q \in Q_{\prec}} S'_Q \right) \subseteq S' = \seqst(\lambda_\Gamma(\node{r}_\Gamma))$.

For the coinductive step, assume that $\node{n}$ is such that $\lambda_\Gamma(\node{n})$ satisfies \ref{inv:formula}--\ref{inv:definition-list}; we must define $\rho_\Gamma(\node{n})$ and $\lambda_\Gamma$ for the children of $\node{n}$ so that \ref{inv:rule} holds for $\lambda_\Gamma(\node{n})$ and \ref{inv:formula}--\ref{inv:definition-list} holds for each child.  As in the definition of $\mathbb{T}'_\Gamma$ the construction proceeds via a case analysis on $\rn_\Gamma(\node{n})$.  The only non-routine cases involve rules $\lor$, $\dia{K}$ and Thin.  We give these below; the other cases are left to the reader.
\begin{description}
    \item[$\rn_\Gamma(\node{n}) = \lor$.]
        In this case $cs(\node{n}) = \node{n}_1\node{n}_2$, 
        $\seqfm_\Gamma(\node{n}) = \seqfm_\Gamma(\node{n}_1) \lor \seqfm_\Gamma(\node{n}_2)$; let $\Delta = \seqdl_\Gamma(\node{n}) = \seqdl_\Gamma(\node{n}_1) = \seqdl_\Gamma(\node{n}_2)$.  
        For each $s \in \seq \seqst(\lambda_\Gamma(\node{n}))$ define the following.
        \begin{align*}
        S_{1,s} &=
            \begin{cases}
                \{s\}       & \text{if $s \in \seqst(\lambda'_\Gamma(\node{n}))$ and 
                              $s \in \seqst(\lambda'_\Gamma(\node{n}_1))$}\\
                \{s\}       & \text{if $s \not\in \seqst(\lambda'_\Gamma(\node{n}))$,
                              $s \in \seqst(\lambda_{\Gamma,S}(\node{n}))$ and         
                              $\seqst(\lambda_{\Gamma,S}(\node{n}_1))$}\\
                \emptyset   & \text{otherwise}
            \end{cases}
            \\
        S_{2,s} &= \{s\} \setminus S_{1,s}
        \end{align*}
        Define $\rho_\Gamma(\node{n}) = \lor$.
        Taking $S_1 = \bigcup_{s \in S} S_{1,s}$ and $S_2 = \bigcup_{s \in S} S_{2,s}$, we set
        \begin{align*}
        \lambda_\Gamma(\node{n}_1)  &= S_1 \tnxT{\V_S}{\Delta} \seqfm_\Gamma(\node{n}_1)\\
        \lambda_\Gamma(\node{n}_2)  &= S_2 \tnxT{\V_S}{\Delta} \seqfm_\Gamma(\node{n}_2).
        \end{align*}
        It is easy to see that invariant~\ref{inv:rule} holds for $\node{n}$, while \ref{inv:formula}--\ref{inv:definition-list} hold for $\node{n}_1$ and $\node{n}_2$.
    \item[$\rho_\Gamma(\node{n}) = \dia{K}$.]
        In this case $cs(\node{n}) = \node{n}'$, 
        $$
        \lambda_\Gamma(\node{n}) = S_\node{n} \tnxT{\V_S}{\Delta} \dia{K}\seqfm_\Gamma(\node{n}')
        $$
        for some $S_\node{n}$ and $\Delta$, and
        $\seqdl_\Gamma(\node{n}') = \Delta$.  
        We must construct a witness
        function $f_{\Gamma,\node{n}} \in S_\node{n} \to \states{S}$ such that $s \xrightarrow{K} f_{\Gamma,\node{n}}(s)$ for each $s \in S_\node{n}$ and such that $f(S_\node{n}) \subseteq \seqst(\lambda'_\Gamma(\node{n}')) \cup \seqst(\lambda_{\Gamma, S}(\node{n}'))$.  
        Since $\mathbb{T}'_\Gamma$ and $\mathbb{T}_{\Gamma,S}$ are successful it follows there are functions $f'_{\Gamma, \node{n}} \in \seqst(\lambda'_\Gamma(\node{n})) \to \states{S}$ and
        $f'_{\Gamma, S, \node{n}} \in \seqst(\lambda_{\Gamma,S}(\node{n})) \to \states{S}$
        such that:
        \begin{itemize}
            \item
                For all $s \in \seqst(\lambda'_\Gamma(\node{n}))$ $s \xrightarrow{K} f_{\Gamma,\node{n}}(s)$;
            \item 
                $\seqst(\lambda'_\Gamma(\node{n}')) = f_{\Gamma,\node{n}}(\seqst(\lambda'_\Gamma(\node{n})))$;
            \item
                For all $s \in \seqst(\lambda_{\Gamma, S}(\node{n}))$ $s \xrightarrow{K} f_{\Gamma,S,\node{n}}(s)$; and
            \item
                $\seqst(\lambda_{\Gamma,S}(\node{n}')) = f_{\Gamma,S,\node{n}}(\seqst(\lambda_{\Gamma,S}(\node{n})))$.
        \end{itemize}
        We define $f_{\Gamma,\node{n}}$ as follows.
        \[
        f_{\Gamma,\node{n}}(s)
        =
        \begin{cases}
            f'_{\Gamma,\node{n}}(s)
                & \text{if $s \in \seqst(\lambda'_\Gamma(\node{n}))$}
            \\
            f_{\Gamma,S, \node{n}}(s)
                & \text{if $s \in \seqst(\lambda_{\Gamma,S}(\node{n})) \setminus \seqst(\lambda'_\Gamma(\node{n}))$}
        \end{cases}
        \]
        Since $\node{n}$ satisfies \ref{inv:state-set} it follows that $S_\node{n} \subseteq \seqst(\lambda'_\Gamma(\node{n})) \cup \seqst(\lambda_{\Gamma,S}(\node{n}))$, and thus $f_{\Gamma,\node{n}}$ is well-defined.
        We now take $\rho_\Gamma(\node{n}) = (\dia{K},f_{\Gamma,\node{n}})$ and
        $$
        \lambda_\Gamma(\node{n}') = f_{\Gamma,\node{n}}(S_\node{n}) \tnxT{\V_S}{\Delta} \seqfm_\Gamma(\node{n}');
        $$
        It is clear that \ref{inv:rule} holds for $\node{n}$ and that \ref{inv:formula}--\ref{inv:definition-list} hold for $\node{n}'$. 
    \item[$\rho_\Gamma(\node{n}) = \text{Thin}$.]
        In this case, since both $\mathbb{T}'_\Gamma$ and $\mathbb{T}_{\Gamma,S}$ are TNF it follows that 
        $$
        \lambda_\Gamma(\node{n}) = S_\node{n} \tnxT{\V_S}{\Delta} \sigma'' Z''.\Gamma'
        $$
        for some $S, \sigma''Z''.\Gamma'$ and $\Delta$, that $\seqfm_\Gamma(\node{n}') = \sigma''Z''.\Gamma'$, and that $\seqdl_\Gamma(\node{n}') = \Delta$.
        Now define $\rho_\Gamma(\node{n}) = \text{Thin}$
        and $\lambda_\Gamma(\node{n}')$ as follows.
        $$
        \lambda_\Gamma(\node{n}') = 
        \seqst(\lambda'_\Gamma(\node{n}')) \cup \seqst(\lambda_{\Gamma,S}(\node{n}')
        \tnxT{\V_S}{\Delta}
        \sigma'' Z''.\Gamma'
        $$
        It is the case that \ref{inv:rule} holds of $\node{n}$ and that \ref{inv:formula}--\ref{inv:definition-list} hold of $\node{n}'$.
\end{description}
It can be shown that $\mathbb{T}_\Gamma$ is successful, TNF and compliant with $\prec'$ by adapting the arguments given above for $\mathbb{T}'_\Gamma$. We now establish that \ref{goal:dependencies} also holds of $\mathbb{T}_\Gamma$.  To this end, fix $x \in S$, $y \in S'_{[x]},$ and $x' <:_{\node{n}',\node{r}_\Gamma} y$, where $\seqfm(\lambda_\Gamma(\node{n}') = Z$.  From the construction of $\mathbb{T}_\Gamma$ it is easy to see in this case that $y \in \seqst(\lambda'_\Gamma(\node{n}'))$ and that $x' <:_{\node{n}',\node{r}_\Gamma} y$ in $\mathbb{T}_\Gamma$ iff $x' <:_{\node{n}',\node{r}_\Gamma} y$ in $\mathbb{T}'_\Gamma$.  Since we have established that \ref{goal:dependencies} holds for $\mathbb{T}'_\Gamma$, the desired result follows.

\paragraph{Step~\ref{it:step-tableau-composition} of proof outline:  construct tableau for $S \tnxTV{\emptyL} \sigma Z.\Phi$.}

To complete the proof we construct tableau $\mathbb{T}_\Phi = (\tree{T}_\Phi, \rho_\Phi, \T, \V, \lambda_\Phi)$, where $\tree{T}_\Phi = (\node{N}_\Phi, \node{r}_\Phi, p_\Phi, cs_\Phi)$, from $\mathbb{T}_{\Phi'}$ and $\mathbb{T}_\Gamma$ and establish that it is successful and compliant with $(S,\prec)$.  Without loss of generality we assume that $\node{N_\Phi'} \cap \node{N}_\Gamma = \emptyset$.  We begin by defining $\tree{T}_\Phi$ as follows, where $\node{n}_W \in \node{N}_{\Phi'}$ is the unique (leaf) node in $\tree{T}_{\Phi'}$ such that $\seqfm(\lambda_{\Phi'}(\node{n}_W)) = W$.
\begin{itemize}
    \item $\node{N}_\Phi = \node{N}_{\Phi'} \cup \node{N}_\Gamma$
    \item $\node{r}_\Phi = \node{r}_{\Phi'}$
    \item (Partial) parent function $p_\Phi \in \node{N}_\Phi \rightarrow \node{N}_\Phi$ is given as follows.
        \[
        p_\Phi(\node{n}) =
        \begin{cases}
            p_{\Phi'}(\node{n}) & \text{if $\node{n} \in \node{N}_{\Phi'}$}\\
            \node{n}_W          & \text{if $\node{n} = \node{r}_\Gamma$}\\
            p_{\Gamma}(\node{n}) & \text{otherwise}
        \end{cases}
        \]
    \item Child-ordering function $cs_\Phi \in \node{N}_\Phi \to (\node{N}_\Phi)^*$ is given as follows.
        \[
        cs_\Phi(\node{n}) =
        \begin{cases}
            \node{r}_\Gamma       & \text{if $\node{n} = \node{n}_W$} \\
            cs_{\Phi'}(\node{n})  & \text{if $\node{n} \in \node{N}_{\Phi'} \setminus \{\node{n}_W\}$}\\
            cs_\Gamma (\node{n})   & \text{otherwise}
        \end{cases}
        \]
\end{itemize}
This construction creates $\tree{T}_\Phi$ by in effect inserting $\tree{T}_\Gamma$ into $\tree{T}_{\Phi'}$ as the single subtree underneath $\node{n}_W$.

To finish the construction of $\mathbb{T}_\Phi$ we need to specify $\rho_\Phi$ and $\lambda_\Phi$.  The first of these is given as follows.
\[
\rho_\Phi(\node{n}) =
\begin{cases}
    \rho_{\Phi'}(\node{n})
    & \text{if $\node{n} \in \node{N}_{\Phi'} \setminus \{\node{n}_W\}$}
    \\
    \text{Thin}
    & \text{if $\node{n} = \node{n}_W$}
    \\
    \rho_\Gamma(\node{n})
    & \text{otherwise}
\end{cases}
\]
In this definition all nodes inherit their rule applications from $\mathbb{T}_{\Phi'}$ and $\mathbb{T}_\Gamma$, with the exception of $\node{n}_W$, which, as a leaf, has no rule in $\mathbb{T}_{\Phi'}$.  In $\mathbb{T}_\Phi$, $\node{n}_W$ is assigned the rule Thin.

To define $\lambda_\Phi$, we first introduce some notation.  Let $\seq{s} = S'' \tnxTVD \Phi''$ be a sequent.  Then $\seq{s}[Z := \Gamma']$ is the sequent obtained by substituting all free instances of $Z$ in both $\Phi''$ and $\Delta$ with $\Gamma'$.  Now suppose that $U \not\in \dom(\Delta)$; the sequent $(U = \Phi''')\cdot \seq{s}$ is defined to be $S'' \tnxTV{(U = \Phi''') \cdot \Delta} \Phi''$.  We may now define $\lambda_\Phi$ as follows.  Assume without loss of generality that $U$ is the definitional constant introduced by the application of Rule $\sigma Z$ to $\lambda_{\Phi'}(\node{r}_{\Phi'})$, and that $U$ does not appear in $\mathbb{T}_\Gamma$.
\[
\lambda_\Phi (\node{n}) =
\begin{cases}
    S'' \tnxTV{\emptyL} \Phi''
    &
    \\
    \multicolumn{2}{l}{
        \qquad\text{
            if $\node{n} \in \node{N}_{\Phi'}$ and
            $S'' \tnxT{\V'}{\emptyL} \Phi'' = \lambda_{\Phi'}(\node{n})[W := \sigma' Z'.\Gamma]$
        }
    }
    \\
    S'' \tnxTV{\Delta'} \Gamma'
    &
    \\
    \multicolumn{2}{l}{
        \qquad\text{
            if $\node{n} \in \node{N}_\Gamma$ and
            $S'' \tnxT{\V_S}{\Delta'} \Gamma' 
            = (U = \sigma Z.\Phi) \cdot ((\lambda_\Gamma(\node{n}) [Z := U])$
        }
    }
\end{cases}
\]
Intuitively, if $\node{n}$ is a node in $\mathbb{T}_{\Phi'}$ then $\lambda_\Phi(\node{n})$ modifies the sequent $\lambda_{\Phi'}(\node{n})$ by replacing all free occurrences of $W$ in the definition list and formula components of the sequent by $\sigma' Z'.\Gamma$, and the occurrence of $\V' = \V[W := S']$ on the turnstile by $\V$.  Likewise, if $\node{n}$ is a node $\node{N}_\Gamma$ then $\lambda_\Phi(\node{n})$ modifies the sequent $\lambda_\Gamma(\node{n})$ by replacing all occurrences of $Z$ by $U$, prepending the definition $(U = \sigma Z.\Phi)$ to the front of the definition list, and replacing the occurrence of $\V_S$ on the turnstile by $\V$.  
To complete the proof of the lemma we must show that $\mathbb{T}_\Phi$ is a successful TNF tableau that is compliant with $(S,\prec)$.
That $\mathbb{T}_\Phi$ is indeed a tableau is immediate from its construction and the fact that $\mathbb{T}_{\Phi'}$ and $\mathbb{T}_\Gamma$ are tableaux:  in particular, every leaf in $\mathbb{T}_\Phi$ corresponds either to a leaf in $\mathbb{T}_{\Phi'}$ or to a leaf in $\mathbb{T}_\Gamma$ and is therefore guaranteed to be terminal.
The TNF property follows similarly.
We now argue that $\mathbb{T}_\Phi$ is successful and compliant with $(S, \prec)$.
We begin by noting that since $\mathbb{T}_{\Phi'}$ and $\mathbb{T}_\Gamma$ are successful, every leaf $\node{n}$ such that $\seqfm(\lambda_\Phi(\node{n})) \neq U$, where $U$ is the definitional constant associated with $\sigma Z.\Phi$, is also successful in $\mathbb{T}_\Phi$.  Proving that $\mathbb{T}_\Phi$ is successful therefore reduces to proving the success of each $U$-leaf.  If we can show compliance of $\mathbb{T}_\Phi$ with $(S,\prec)$, then the success of each $U$-leaf also follows, for in the particular case when $\sigma = \mu$, the well-foundedness of $<:_{\node{r}'}$, where $\node{r}' = cs(\node{r}_\Phi)$ is the unique (due to TNF) $U$-companion node in $\mathbb{T}_\Phi$, follows immediately from the well-foundedness of $\prec$. To this end, suppose that $s, s' \in S$ are such that $s' <:_{\node{r}'} s$; we must show that $s' \prec s$.  Since $s' <:{\node{r}'} s$ holds there must exist a $U$-leaf $\node{n}$ such that $s' \in \seqst(\lambda_\Phi(\node{n}))$ and $s' <:_{\node{n},\node{r}'} s$ in $\mathbb{T}_\Phi$.  There are two cases to consider.
\begin{enumerate}
    \item 
    $\node{n}'$ is a leaf in $\mathbb{T}_{\Phi'}$.  In this case the compliance of $\mathbb{T}_{\Phi'}$ with respect to $(S, \prec)$ guarantees that $s' \prec s$.
    \item
    $\node{n}'$ is a leaf in $\mathbb{T}_\Gamma$.  In this case there must exist $y \in \semT{\sigma'Z'.\Gamma}{\V_{[s]}}$ such that $s' <:_{\node{n},\node{r}_\Gamma} y$ in $\mathbb{T}_\Phi$, and thus $\mathbb{T}_\Gamma$, and $y <:_{\node{r}_\Gamma, \node{r}'} s$.  Since $\mathbb{T}_\Gamma$ satisfies \ref{goal:dependencies} it must follow in this case that $s' \prec s$.
\end{enumerate}
This completes the proof.
\qedhere
\end{proof}

With these lemmas in hand we may now state and prove the completeness theorem.

\begin{theorem}[Completeness]\label{thm:completeness}
Let $\T = \lts{S}$ be an LTS and $\V$ a valuation, and let $S$ and $\Phi$ be such that $S \subseteq \semTV{\Phi}$.  Then $S \tnxTV{\emptyL} \Phi$ has a successful tableau.
\end{theorem}

\begin{proof}
The proof proceeds as follows.  Let $\sigma_1 Z_1 . \Phi_1, \ldots, \sigma_n Z_n . \Phi_n$ be the top-level fixpoint subformulas of $\Phi$, and let $W_1, \ldots, W_n$ be fresh variables.  Define $\Phi'$ to be the fixpoint-free formula containing exactly one occurrence of each $W_i$ and such that
\[
\Phi = \Phi'[W_1, \ldots, W_n := \sigma_1 Z_1 . \Phi_1, \ldots \sigma_n Z_n . \Phi_2].
\]
Also define $S_i = \semTV{\sigma_i Z_i . \Phi_i}$ for each $i = 1, \ldots n$, and let $\V'$ be defined by
\[
\V' = \V[W_1, \ldots W_n := S_1, \ldots S_N]
\]
Lemma~\ref{lem:fixpoint-completeness} guarantees a successful tableau $\mathbb{T}'$ for $S \tnxT{\V'}{\emptyL} \Phi'$.

Also, for each $i = 1, \ldots, n$ define $(S_i, \prec_i)$ to be a $\sigma_i$-maximal support ordering for $\semfTV{Z_i}{\Phi_i}$.  Lemma~\ref{lem:fixpoint-completeness} guarantees a successful tableau $\mathbb{T}_i$ for each sequent $S_i \tnxTV{\emptyL} \sigma_i Z_i . \Phi_i$.  We may now construct a successful tableau $\mathbb{T}$ by making each leaf in $\mathbb{T}'$ whose formula is $W_i$ the parent of the root of tableau $\mathbb{T}_i$, setting the rule for this former leaf node to be Thin and replacing all occurrences of $W_i$ by $\sigma_i Z_i.\Phi_i$.  The resulting tableau is guaranteed to be successful.
\qedhere
\end{proof}

\section{Proof search}\label{sec:proof-search}

From a reasoning perspective, the most difficult aspect of the proof system in this paper is the well-foundedness of the support orderings for nodes labeled with a least fixed-point formula.
In the abstract setting of infinite-state labeled transition systems considered in this paper, unfortunately there is really no alternative.
In many situations, however, this well-foundedness can be inferred from the additional semantic information present in the definition of the transition relations.

One interesting such case involves modifying the terminal criterion for $\mu$-leaves, so that they are only terminal if their set of states is empty.
Intuitively, this means we simply keep unfolding nodes labeled with definitional constants that are associated with a least fixpoint formula until we reach a node with an empty state set.
Of course, with this modification, the unfolding procedure does not necessarily terminate, and hence the proof system is not complete for infinite-state labeled transition systems in general.
However, for instance for the class of infinite-state transition systems induced by timed automata, the resulting tableaux method is complete~\cite{DC2005}.
Furthermore, even in case the modification is not complete, the result is a sound semi-decision procedure in the sense that, if a (finite) proof tree is obtained, the proof tree is successful.

In this section we show how this modified termination condition can also be proven sound using our results.
We first modify Definition~\ref{def:tableau} to change the termination condition.
\begin{definition}[$\nu$-complete tableau]\label{def:weak-tableau}
    Partial tableau $\tableauTrl$ is a \emph{$\nu$-complete tableau},
    if $\seqdl(\node{r}) = \emptyL$, and
    if all leaves $\node{n}$ in $\tree{T}$ are \emph{$\nu$-terminal}, i.e.\/ satisfy at least one of the following.
    \begin{enumerate}[label=(\alph*)]
        \item \label{def:weak-tableau-terminal-proposition}
            $\seqfm(\node{n}) = Z$ or $\seqfm(\node{n}) = \lnot Z$ for $Z \in \Var \setminus \dom(\seqdl(\node{n}))$; or
        \item  $\seqfm(\node{n}) = \dia{K} \ldots$ and there is $s \in \seqst(\node{n})$ such that $s \centernot{\xrightarrow{K}}$; or
        
        \item $\seqfm(\node{n}) = U$ for some $U \in \dom(\seqdl(\node{n}))$, and there is $\node{m} \in A_s(\node{n})$ such that $\seqfm(\node{m}) = U$, and either
           \begin{enumerate}[label=\roman*.]
           \item \label{def:mu-companion-leaf}
                $\seqdl(\node{n})(U) = \mu Z. \Phi$, and $\seqst(\node{n}) = \emptyset$; or
            \item \label{def:nu-companion-leaf}
              $\seqdl(\node{n})(U) = \nu Z. \Phi$, $\seqst(\node{n}) \subseteq \seqst(\node{m})$.
      \end{enumerate}
    \end{enumerate}
\end{definition}
Note that this definition only splits clause~\ref{subdef:companion-leaf} of Definition~\ref{def:tableau} such that the original definition only applies to definitional constants whose right-hand sides are greatest fixed point formulas.
Definitional constants whose right-hand sides are least fixed points are only $\nu$-terminal in case the set of states is empty.
Note that, as the set of states in a $\mu$-leaf is empty in case of a $\nu$-complete tableau, and it has a corresponding companion node in the tableaux, a $\nu$-complete tableaux is a complete tableaux, as stated in the following lemma.

\begin{lemma}\label{lem:weak-tableau-is-tableau}
    Let $\tableauTrl$ be a $\nu$-complete tableau.  Then $\tableauTrl$ is a complete tableau.
\end{lemma}
\begin{proof}
    Immediate from the definitions. \qedhere
\end{proof}

This immediately means that $\nu$-complete tableaux are sound.

\begin{theorem}\label{thm:soundness-weak-tableau}
Fix LTS $(\states{S},\to)$ of sort $\Sigma$ and valuation $\V$. Let $\mathbb{T} = \tableauTrl$ be a successful $\nu$-complete tableau for sequent $\seq{s} \in \Seq{\T}{\Var}$, where $\seqdl(\seq{s}) = \emptyL$.  Then $\seq{s}$ is valid.
\end{theorem}
\begin{proof}
    Follows immediately from Lemma~\ref{lem:weak-tableau-is-tableau} and Theorem~\ref{thm:soundness}.\qedhere
\end{proof}

Hence the proof system remains sound if we modify the termination criterion in such a way that we always keep unfolding $\mu$-nodes until the set of states in the node is empty.

We would like to point out here that the preceding theorem holds in part due to the definition of complete tableaux. Even if a node in the tableaux labeled by a definitional constant satisfies the conditions of a terminal node, the definition allows further unfolding of that node.
Bradfield defines a tableau as ``a proof-tree that is built from applications of proof rules, starting at the root of the tableau, until all leaves of the tree are terminal''~\cite[Definition 3.4]{Bra1991}.
This definition does not allow the continued unfolding that we use in our $\nu$-complete tableaux; as a consequence, the original soundness proof does not carry over immediately to this case.

\section{Timed modal mu-calculus}\label{sec:timed-mu-calculus}

Timed transition systems are infinite-state transition systems used to give semantics to, for example, timed and hybrid automata.  Besides introducting transitions labeled by time, timed transition systems also capture continuous as well as discrete system behavior.
In this section, we extend the proof system described in Section~\ref{sec:base-proof-system} to timed transition systems and properties expressed in a timed extension of the mu-calculus, which enriches the mu-calculus with two modalities for expressing properties of continuous timed behavior.
These operators correspond to timed versions of the well-known until and release operators from Linear-Time Temporal Logic (LTL).
This section represents another illustration of the extensibility of the proof system in this paper: the proofs of soundness and completeness given earlier only need to be extended with cases for the new modal operators to cover the timed setting.

We begin the section by defining timed transition systems and the extension of the mu-calculus to be considered.

\begin{definition}[Timed sort]
Sort $\Sigma$ is a \emph{timed sort} iff $\Rnneg \subseteq \Sigma$, i.e., every non-negative real number is an element of $\Sigma$.  If $\Sigma$ is a timed sort we write $\act(\Sigma) = \Sigma \setminus \Rnneg$ for the set of non-numeric elements in $\Sigma$.  We sometimes refer to elements of $\act(\Sigma)$ as \emph{actions} and $\tdelay \in \Rnneg$ as \emph{time delays}.
\end{definition}

\begin{definition}[Timed transition system~\cite{BCL2011}]\label{def:timed-transition-system}
    LTS $(\states{S}, \to)$ of timed sort $\Sigma$ is a \emph{timed transition system (TTS)} if it satisfies the following conditions.
    \begin{enumerate}
        \item For all states $s \in \states{S}$, $s \xrightarrow{0} s$, i.e., the transition system is \emph{time-reflexive}.
        \item For all states $s,s',s'' \in \states{S}$ and $\tdelay \in \Rnneg$ if $s \xrightarrow{\tdelay} s'$ and $s \xrightarrow{\tdelay} s''$ then $s' = s''$, i.e., the transition system is \emph{time-deterministic}.
        \item For all states $s, s', s'' \in \states{S}$ and $\tdelay, \tdelay' \in \Rnneg$, if $s \xrightarrow{\tdelay} s'$ and $s' \xrightarrow{\tdelay'} s''$ then $s \xrightarrow{\tdelay + \tdelay'} s''$, i.e., the transition system is \emph{time-additive}.
        \item For all states $s, s' \in \states{S}$ and $\tdelay \in \Rnneg$, if $s \xrightarrow{\tdelay} s'$ then for all $\tdelay', \tdelay'' \in \Rnneg$ with $\tdelay = \tdelay' + \tdelay''$, there exists $s'' \in \states{S}$ such that $s \xrightarrow{\tdelay'} s''$ and $s'' \xrightarrow{\tdelay''} s'$, i.e., the transition system is \emph{time-continuous}.
   \end{enumerate}
\end{definition}

In a TTS a transition can either be labeled by an action, i.e., those labels in $\act(\Sigma)$, or a time delay $\tdelay \geq 0$.  Intuitively, if a system is in state $s$ and $s \xrightarrow{\tdelay} s'$ then the system is in state $s'$ after $\tdelay$ units of time have elapsed, assuming no action has occurred in the interim.

The timed mu-calculus of~\cite{FC2014} extends the mu-calculus from Definition~\ref{def:mu-calculus-syntax} with a timed modal operator $\tR{\Phi_1}{\Phi_2}$.  This operator is analogous to the release operator of LTL, in the sense that for $\tR{\Phi_1}{\Phi_2}$ to hold of a state, either $\Phi_2$ must hold along every time instant of the time trajectory emanating from $s$, or there is a point along the trajectory in which $\Phi_1$ holds, which releases the system from having to maintain $\Phi_2$.
These intuitions are formalized below.

\begin{definition}[Timed mu-calculus syntax~\cite{FC2014}]\label{def:timed-mu-calculus-syntax}
Let $\Sigma$ be a timed sort and $\Var$ a countably infinite set of propositional variables. Then formulas of the timed modal mu-calculus over $\Sigma$ and $\Var$ are given by the following grammar
    $$
        \Phi ::= Z
        \mid \lnot\Phi'
        \mid \Phi_1 \land \Phi_2
        \mid [K] \Phi'
        \mid \tR{\Phi_1}{\Phi_2}
        \mid \nu Z. \Phi'
    $$
    where $K \subseteq \Sigma$, $Z \in \Var$, and formulas of the form $\nu Z . \Phi'$, $Z$ are such that $Z$ must be positive in $\Phi'$.
\end{definition}
Besides the standard dualities introduced in Section~\ref{subsec:propositional-modal-mu-calculus}, we also have the following dual for $\forall_{\Phi_1} (\Phi_2)$:
\[
  \tU{\Phi_1}{\Phi_2} = \lnot \tR{\lnot \Phi_1}{\lnot \Phi_2}.
\]
Also in analogy with LTL, $\tU{\Phi_1}{\Phi_2}$ can be seen as an until operator:  specifically, a state satisfies $\tU{\Phi_1}{\Phi_2}$ if $\Phi_2$ is true after some time delay from $s$, and until that time instant $\Phi_1$ must be true.

Before we introduce the semantics of the timed mu-calculus, we fix some notation to deal with time delays and time intervals.
\begin{notation}
    Let $\T = (\states{S}, \to)$ be a TTS of timed sort $\Sigma$, and let $\tdelay \in \Rnneg$.
    \begin{itemize}
        \item
        We write 
        \[
            \tsucc(s,\tdelay) = s' \text{ iff } s \xrightarrow{\tdelay} s'
        \]
        This definition is well-formed because of the time-determinacy of $\xrightarrow{}$.  Note that $\tsucc(s,\tdelay) \div$ iff $s \centernot{\xrightarrow{\tdelay}}$, and that if $\tdelay' \geq \tdelay$ and $\tsucc(s,\tdelay)\div$ then $\tsucc(s,\tdelay') \div$. Furthermore, because of time continuity, if $\tsucc(s,\tdelay) \in \states{S}$, $\tsucc(s,\tdelay') \in \states{S}$ for all $\tdelay' \leq \tdelay$. 
        
        \item
        If $s \in \states{S}$ then we write
        \begin{align*}
            \tdelays(s)
            &= \{ \tdelay \in \Rnneg \mid s \xrightarrow{\tdelay} \}
        \\
            \tsucc(s)
            &= \{ s' \in \states{S} \mid \exists \tdelay \in \Rnneg \colon s \xrightarrow{\tdelay} s' \}
        \\ 
            \tsucc_{<}(s, \tdelay) 
            &= \{s' \in \states{S} \mid \exists \tdelay' \in \Rnneg \colon \tdelay' < \tdelay \land s \xrightarrow{\tdelay'} s' \}
        \\ 
            \tsucc_{\leq}(s, \tdelay) 
            &= \{s' \in \states{S} \mid \exists \tdelay' \in \Rnneg \colon \tdelay' \leq \tdelay \land s \xrightarrow{\tdelay'} s' \}
        \end{align*}
        for all possible delays, the states reachable by any delay, the states reachable by delays less than $\tdelay$, and states reachable by delays less than or equal to $\tdelay$, respectively, from $s$.
    \end{itemize}
\end{notation}

\noindent
For the semantics of the timed mu-calculus, we follow the definition in~\cite{Fon2014}.
\begin{definition}[Timed mu-calculus semantics]\label{def:timed-mu-calculus-semantics}
    Let $\T = (\states{S}, \to)$ be a TTS of sort $\Sigma$ and $\V \in \Var \to 2^{\states{S}}$ a valuation. Then the semantic function $\semTV{\Phi} \subseteq \states{S}$, where $\Phi$ is a timed mu-formula, is defined as in Definition~\ref{def:mu-calculus-semantics}, extended with the following clause.
    \begin{align*}
       & \semTV{\tR{\Phi_1}{\Phi_2}} \\
        &{=}\; \{ s \in \states{S} \mid \forall \tdelay \in \tdelays(s) \colon  \tsucc_<(s,\tdelay) \cap \semTV{\Phi_1} = \emptyset
        \implies \tsucc(s,\tdelay) \in \semTV{\Phi_2} \}
    \end{align*}
\end{definition}

\noindent
Intuitively, $s$ satisfies $\tR{\Phi_1}{\Phi_2}$ if for every possible delay transition from $s$, either the target state satisfies $\Phi_2$, or there is a delay transition of smaller duration whose target state satisfies $\Phi_1$, thereby releasing $s$ of the responsibility of keeping $\Phi_2$ true beyond that delay.
For the dual operator one may derive the following semantic equivalence:
\begin{align*}
    & \semTV{\tU{\Phi_1}{\Phi_2}} \\
    & {=}\; \{ s \in \states{S} \mid \exists \tdelay \in \tdelays(s) \colon
        \tsucc_<(s,\delta) \subseteq \semTV{\Phi_1} \land \tsucc(s,\delta) \in \semTV{\Phi_2} \}.
\end{align*}

\noindent
Based on this characterization, one can see that $\tU{\Phi_1}{\Phi_2}$ captures a notion of until.  Specifically, $s$ satisfies $\tU{\Phi_1}{\Phi_2}$ if there is a delay transition from $s$ leading to a state satisfying $\Phi_2$, and all delay transitions of strictly shorter duration from $s$ lead to states satisfying $\Phi_1$.
%

The notions of \emph{timed definition list} and \emph{timed sequent} are the obvious generalizations of the same notions given in Definitions~\ref{def:definition-list} and~\ref{def:sequent}, as is the notion of the semantics $\sem{\seq{s}}{}{}$ of timed sequent $\seq{s}$, which generalizes Definition~\ref{def:sequent-semantics}.  The sequent extractor functions $\seqst$, $\seqdl$ and $\seqfm$ also carry over to the timed setting.

To obtain proof rules for the timed mu-calculus, we extend those from Figure~\ref{fig:proof-rules} with the two rules named $\forall$ and $\exists$ as shown in Figure~\ref{fig:proof-rules-timed}.

\begin{figure}
    \[
        \begin{array}{c}

            \proofrule[\exists\;\;]
            {S \tnD \tU{\Phi_1}{\Phi_2}}
            {f_<(S) \tnD \Phi_1 \quad f_=(S) \tnD \Phi_2}
            \;f \in S \to \Rnneg
            \\[18pt]

            \proofrule[\forall\;\;]
            {S \tnD \tR{\Phi_1}{\Phi_2}}
            {g_<(S) \tnD \Phi_1 \quad g_=(S) \tnD \Phi_2}
        \end{array}
    \]
    where
    \begin{itemize}
        \item $f$ is such that for all $s \in S$, $f(s) \in \tdelays(s)$, and $f_<(S), f_=(S) \subseteq \states{S}$ are defined as
        \begin{align*}
            f_<(S) &= \bigcup_{s \in S} \tsucc_{<}(s,f(s)) 
            \\
            f_=(S) &= \{ \tsucc(s,f(s)) \mid s \in S \}.
        \end{align*}
        \item $g$ is such that for all $s \in S$ and $\delta \in \tdelays(s)$, 
        $g(s,\tdelay) \leq \tdelay$, and $g_<(S), g_=(S) \subseteq \states{S}$ are defined as
        \begin{align*}
        g_<(S) &= \{\tsucc(s,g(s,\delta)) \mid s \in S \land \delta \in \tdelays(s) \land g(s,\delta) < \delta \}
        \\
        g_=(S) &= \{\tsucc(s,g(s,\delta)) \mid s \in S \land \delta \in \tdelays(s) \land g(s,\delta) = \delta \}
        \end{align*}
    \end{itemize}
    \caption{Proof rules for timed modal operators (extends Figure~\ref{fig:proof-rules}).}
    \label{fig:proof-rules-timed}
\end{figure}

Rule $\exists$ is similar to the rule $\dia{K}$ in that it uses a witness function $f$.  
In this case $f$ is intended to identify, for every $s \in S$,  a \emph{witness delay} $f(s) \in \tdelays(s)$ allowed from $s$ such that the state $\tsucc(s,f(s))$ reached by delaying $f(s)$ time units from $s$ satisfies $\Phi_2$ and also such that all states reached from $s$ using smaller delays, i.e., those states in $\tsucc_{<}(s,f(s))$, must satisfy $\Phi_1$.

Proof rule $\forall$ also uses a witness function $g$ in its side condition, but its functionality is a bit more complicated than the witness function used in rule $\exists$.
In particular, $g$ takes two arguments, a state and a delay, and while  $g$ may be partial, it is required that for any $s \in S$, $g(s,\delta)$ is defined for every $\delta \in \tdelays(s)$. If $\delta \in \tdelays(s)$ and $g(s,\delta) < \delta$ is intended to be a delay shorter than $\delta$ such that the state reached from $s$ via that delay satisfies the ``release'' formula $\Phi_1$.  If no such shorter delay exists, then $\delta(s,\delta) =\delta$, reflecting the fact that the state reached from $s$ after $\delta$ has not been released from the obligation to keep $\Phi_2$ true.  

We sometimes refer to the function $f$ mentioned in the side condition of the $\exists$ rule an \emph{$\exists$-function} and the function $g$ referred to in the side condition of the $\forall$ rule as a \emph{$\forall$-function}.  In what follows we use
\[
\RuleApplT = \RuleAppl 
\cup \{(\exists, f) \mid f \ \text{is an $\exists$-function}\}
\cup \{(\forall, g) \mid g \ \text{is a $\forall$-function}\}
\]
for the set of rule applications for the timed mu-calculus.
The definitions of partial and complete tableaux (Definition~\ref{def:tableau}) carry over immediately to partial and complete \emph{timed} tableaux, with $\RuleApplT$ replacing $\RuleAppl$.
Observe that, in particular, no new terminal nodes need to be added: due to time reflexivity, $s \in \tsucc(s)$ for any state $s$, and thus type-correct $\exists$- and $\forall$-functions can always be given when applying the $\exists$ and $\forall$ rules. (Of course the chosen functions may not necessarily lead to \emph{successful} tableaux.) Likewise, the definitions of successful terminals and successful timed tableaux are the obvious adaptations of the notions given in Definition~\ref{def:successful-tableau}.  Also, if $\node{n}$ is a node in a timed partial tableau then the semantics $\sem{\node{n}}{}{}$ of $\node{n}$ carries over in a straightforward manner.

We now extend the local dependency ordering (Definition~\ref{def:local_dependency_ordering}) to timed tableaux as follows.
\begin{definition}[Timed local dependency ordering]\label{def:local-dependency-ordering-timed}
Let $\node{n}, \node{n}'$ be proof nodes in timed tableau $\mathbb{T}$, with $\node{n}' \in c(\node{n})$ a child of $\node{n}$. Then $s' <_{\node{n}',\node{n}} s$ iff $s' \in \seqst(\node{n}')$, $s \in \seqst(\node{n})$, and one of the following hold:
    \begin{enumerate}
        \item $\rho(\node{n}) = [K]$ and $s \xrightarrow{K} s'$; or
        \item $\rho(\node{n}) = (\dia{K},f)$ and $s' = f(s)$; or
        \item $\rn(\rho(\node{n})) \not \in \{ [K], \langle K \rangle, \forall, \exists \}$ and $s = s'$; or
        \item $\rho(\node{n}) = (\exists, f)$, $cs(\node{n}) = \node{n}_1\node{n}_2$, and either
              \begin{itemize}
                  \item $\node{n}' = \node{n}_1$ and $s' \in \tsucc_{<}(s,f(s))$, or
                  \item $\node{n}' = \node{n}_2$ and $s' = \tsucc(s,f(s))$; or
              \end{itemize}
        \item $\rho(\node{n}) = (\forall,g)$, $cs(\node{n}) = \node{n}_1\node{n}_2$, and either
              \begin{itemize}
                  \item $\node{n}' = \node{n}_1$ and $s' = \tsucc(s, g(s,\tdelay))$ for some $\tdelay \in \tdelays(s)$ with $g(s,\tdelay) < \tdelay$, or
                  \item $\node{n}' = \node{n}_2$ and $s' = \tsucc(s, g(s,\tdelay))$ for some $\tdelay \in \tdelays(s)$ with $g(s,\tdelay) = \tdelay$
              \end{itemize}
    \end{enumerate}
\end{definition}
The dependency ordering $\lessdot_{\node{n}',\node{n}}$ and extended dependency ordering $<:_{\node{n}',\node{n}}$ are adapted to timed tableaux in the obvious way, using the timed local dependency ordering as a basis.

We next establish that the timed local dependency ordering also satisfies the semantic sufficiency property.

\begin{lemma}[Semantic sufficiency of timed $<_{\node{n}', \node{n}}$] \label{lem:semantic-sufficiency-of-timed-<}
    Let $\node{n}$ be an internal proof node in partial timed tableau $\mathbb{T}$, and let $s \in \seqst(\node{n})$ be such that for all $s'$ and $\node{n}'$ with $s' <_{\node{n}', \node{n}} s$, $s' \in \semop{\node{n}'}$.  Then $s \in \semop{\node{n}}$.
\end{lemma}
\remove{
\begin{proofsketch}
    The proof of the cases where $\rho(\node{n}) \in \{\exists, \forall\}$ follows the exact same line of reasoning as the the cases in the proof of Lemma~\ref{lem:semantic-sufficiency-of-<}. \seeappendix
\end{proofsketch}
}
\begin{proof}
    Let $\node{n} = S \tnxTVD \Phi$ be an internal node in $\mathbb{T} = \tableauTrl$.  Since $\node{n}$ is internal, $\rho(\node{n})$ is defined.  Now fix $s \in S$.  The proof proceeds by case analysis on $\rho(\node{n})$.
    All cases other than $(\exists,f)$ and $(\forall,g)$ are identical to the proof of Lemma~\ref{lem:semantic-sufficiency-of-<}. We therefore only show the remaining two cases.

    \begin{itemize}
        \item $\rho(\node{n}) =  (\exists,f)$.
              In this case $\Phi = \tU{\Phi_1}{\Phi_2}$, and $cs(\node{n}) = \node{n}_1\node{n}_2$ where $\node{n}_1 = f_<(S) \tnxTVD \Phi_1$ and $\node{n}_2 = f_=(S) \tnxTVD \Phi_2$.
              By definition, $s' <_{\node{n}',\node{n}} s$ iff $\node{n}' \in \{ \node{n}_1, \node{n}_2 \}$, and $s' = \tsucc(s,f(s))$ if $\node{n}' = \node{n}_2$, and $s' \in \tsucc_{<}(s,f(s))$ otherwise. We reason as follows.
              \begin{flalign*}
                  & \text{For all $s'$, $\node{n}'$ such that $s' <_{\node{n}',\node{n}} s$, $s' \in \semop{\node{n}'}$}\span\span
                  \\
                  & \text{iff}\;\;\; \tsucc(s,f(s)) \in \semop{\node{n}_2}\ \text{and}\ \tsucc_{<}(s,f(s)) \subseteq \semop{\node{n}_1}\span\span
                  \\
                  &
                  && \text{Definitions of $\node{n}_1, \node{n}_2, <_{\node{n}', \node{n}}$ when $\rho(\node{n}) = (\exists,f)$}
                  \\
                  & \text{iff}\;\;\; f(s) \in \tdelays(s), \tsucc(s,f(s)) \in \semT{\Phi_2}{\V[\Delta]},\ \text{and}\ \tsucc_{<}(s,f(s)) \subseteq \semT{\Phi_1}{\V[\Delta]} \span\span
                  \\
                  &
                  && \text{Property of $f$, definition of $\semop{\node{n}_1}$, $\semop{\node{n}_2}$}
                  \\
                  & \text{implies}\;\;\; s \in \semT{\tU{\Phi_1}{\Phi_2}}{\V[\Delta]}
                  && \text{Definition of $\semT{\tU{\Phi_1}{\Phi_2}}{\V[\Delta]}$} \\
                  & \text{iff}\;\;\; s \in \semT{\Phi}{\V[\Delta]}
                  && \text{$\Phi = \tU{\Phi_1}{\Phi_2}$} \\
                  & \text{iff}\;\;\; s \in \semop{\node{n}}
                  && \text{Definition of $\semop{\node{n}}$}
              \end{flalign*}

        \item $\rho(\node{n}) = (\forall, g)$.
              In this case $\Phi = \tR{\Phi_1}{\Phi_2}$, and $cs(\node{n}) = \node{n}_1\node{n}_2$ where $\node{n}_1 = g_<(S) \tnxTVD \Phi_1$ and $\node{n}_2 = g_=(S) \tnxTVD \Phi_2$,.
              By definition, $s' <_{\node{n}',\node{n}} s$ iff $\node{n}' \in \{ \node{n}_1, \node{n}_2 \}$, and $\node{n}' = \node{n}_1$ if $s' = \tsucc(s, g(s,\delta))$ for some $\tdelay \in \tdelays(s)$ with $g(s,\tdelay) < \tdelay$, and $s' = \tsucc(s, g(s,\delta))$ for some $\tdelay \in \tdelays(s)$ with $g(s,\tdelay) = \tdelay$ otherwise.
              We reason as follows.
              \begin{flalign*}
                  & \text{For all $s'$, $\node{n}'$ such that $s' <_{\node{n}',\node{n}} s$, $s' \in \semop{\node{n}'}$}\span\span
                  \\
                  & \text{iff}\;\;\; \text{for all $\tdelay \in \tdelays(s)$, $g(s,\tdelay) = \tdelay$ and $\tsucc(s,g(s,\tdelay)) \in \semop{\node{n}_2}$}\span\span
                  \\
                  & \qquad \text{or $g(s,\tdelay)) < \tdelay$ and $\tsucc(s,g(s,\tdelay)) \in \semop{\node{n}_1}$}\span
                  \\
                  \multispan4{\hfil \text{Definitions of $\node{n}_1, \node{n}_2, <_{\node{n}',\node{n}}$ when $\rho(\node{n}) = (\forall,g)$}}
                  \\
                  & \text{iff}\;\;\; \text{for all $\tdelay \in \tdelays(s)$, $g(s,\tdelay) = \tdelay$ and $\tsucc(s,g(s,\tdelay)) \in \semT{\Phi_2}{\V[\Delta]}$}\span\span \\
                  & \qquad \text{or $g(s,\tdelay)) < \tdelay$ and $\tsucc(s,g(s,\tdelay)) \in \semT{\Phi_1}{\V[\Delta]}$}\span\span
                  \\
                  &
                  && \text{Definition of $\sem{\node{n}_1}{}{}, \sem{\node{n}_2}{}{}$}
                  \\
                  & \text{implies } \text{for all $\tdelay \in \tdelays(s)$, $\tsucc_{<}(s,\tdelay) \cap \semT{\Phi_1}{\V[\Delta]} \neq \emptyset$ or $\tsucc(s, \tdelay) \in \semT{\Phi_2}{\V[\Delta]}$} \span\span
                  \\
                  \multispan4{\hfil \text{Definition of $g$: if $g(s,\tdelay) < \tdelay$, then $\tsucc(s,g(s,\tdelay)) \in \tsucc_{<}(s,g(s,\tdelay))$}}
                  \\
                  & \text{iff}\;\;\; s \in \semT{\tR{\Phi_1}{\Phi_2}}{\V[\Delta]}
                  && \text{Definition of $\semT{\tR{\Phi_1}{\Phi_2}}{\V[\Delta]}$}
                  \\
                  & \text{iff}\;\;\; s \in \semT{\Phi}{\V[\Delta]}
                  && \text{$\Phi = \tR{\Phi_1}{\Phi_2}$}
                  \\
                  & \text{iff}\;\;\; s \in \semop{\node{n}}
                  && \text{Definition of $\semop{\node{n}}$}
              \end{flalign*}
              \qedhere
    \end{itemize}

\end{proof}

\subsection{Soundness}
We now use the semantic sufficiency result from the previous section to generalize the soundness results from Section~\ref{sec:Soundness-via-support-orderings} to timed tableaux.

First, observe that the local soundness result from Lemma~\ref{lem:local-soundness} carries over to the timed setting immediately, since it solely relies on the semantic sufficiency result for $<_{\node{n},\node{n}'}$.

We next show how the node formulas from Definition~\ref{def:node-formulas} generalize to nodes in a timed tableau.

\begin{definition}[Timed node formulas]~\label{def:node-formulas-timed}
    For each companion node $\node{m} \in \companions{\mathbb{T}}$ let $Z_\node{m}$ be a unique fresh variable, with $\Var_{\mathbb{T}} = \{ Z_\node{m} \mid\node{m} \in \companions{\mathbb{T}} \}$ the set of all such variables. Then for node $\node{n} \in \node{N}$ formula $P(\node{n})$ is defined inductively as follows.
    Cases 1-10 are as in Definition~\ref{def:node-formulas}. The cases for the relativized modal operators are as follows.
    \begin{enumerate}
        \setcounter{enumi}{10}
        \item If $\rho(\node{n}) = (\exists,f)$ and $cs(\node{n}) = \node{n}_1\node{n}_2$ then $P(\node{n}) = \exists_{P(\node{n}_1)}(P(\node{n}_2))$.
        \item If $\rho(\node{n}) = (\forall,g)$ and $cs(\node{n}) = \node{n}_1\node{n}_2$ then $P(\node{n}) = \forall_{P(\node{n}_1)}(P(\node{n}_2))$.
    \end{enumerate}
\end{definition}
Valuation consistency (Definition~\ref{def:consistency}) and the associated results in Lemmas~\ref{lem:consistency-property} and~\ref{lem:companion-node-formulas-and-semantics} carry over immediately to the timed setting. We next use these results to link the node formulas and the semantics in timed tableaux.

\begin{lemma}[Timed node formulas and semantics]\label{lem:node-formulas-and-node-semantics-timed}
    Let $\V'$ be a consistent valuation for timed tableau $\mathbb{T}$.  Then for every $\node{n} \in \node{N}$, $\semT{P(\node{n})}{\V'} = \semop{\node{n}}$.
\end{lemma}
\remove{
\begin{proofsketch}
    Analogous to the proof of Lemma~\ref{lem:node-formulas-and-node-semantics}.
\end{proofsketch}
}
\begin{proof}
    Let valuation $\V'$ be consistent with $\mathbb{T}$.
    The proof is by induction on $\tree{T}$.  
    So fix node $\node{n}$ in $\tree{T}$.  
    The induction hypothesis asserts that for all $\node{n}' \in c(\node{n})$, $\semT{P(\node{n}')}{\V'} = \semop{\node{n}'} $.
    The proof now proceeds by an analysis of $\rho(\node{n})$.
    All cases except $\rn(\node{n}) \in \{ \forall, \exists \}$ are identical to those in Lemma~\ref{lem:node-formulas-and-node-semantics}. For the remaining cases, the proofs are as follows.

    \begin{itemize}

        \item $\rho(\node{n}) = (\exists,f)$.
              In this case we know that $\node{n} = S \tnxTVD \tU{\Phi_1}{\Phi_2}$ and that $cs(\node{n}) = \node{n}_1\node{n}_2$, where each $\node{n}_i = S_i \tnxTVD \Phi_i$,
              $S_1 = f_<(S)$, and $S_2 = f_=(S)$.
              The induction hypothesis guarantees that $\semT{P(\node{n}_i)}{\V'} = \semop{\node{n}_i} = \semT{\Phi_i}{\V[\Delta]}$.
              We reason as follows.
              \begin{align*}
                  & \semT{P(\node{n})}{\V'}
                  \\
                   & = \semT{\exists_{P(\node{n}_1)}(P(\node{n}_2))}{\V'}
                   &                                                                                                                                                                   & \text{Definition of $P(-)$}
                  \\
                   & = \{ s \in \states{S} \mid \exists \tdelay \in \tdelays(s) \colon \tsucc_{<}(s,\tdelay) \subseteq \semT{P(\node{n}_1)}{\V'} \land \tsucc(s,\tdelay) \in \semT{P(\node{n}_2)}{\V'} \} \span\span
                  \\
                   &                                                                                                                                                                   &                                         & \text{Semantics of $\exists$}
                  \\
                   & = \{ s \in \states{S} \mid \exists \tdelay \in \tdelays(s) \colon \tsucc_{<}(s,\tdelay) \subseteq \semop{\node{n}_1} \land \tsucc(s,\tdelay) \in \semop{\node{n}_2} \} \span\span
                  \\
                   &                                                                                                                                                                   &                                         & \text{Induction hypothesis (twice)}
                  \\
                   & = \{ s \in \states{S} \mid \exists \tdelay \in \tdelays(s) \colon \tsucc_{<}(s,\tdelay) \subseteq \semT{\Phi_1}{\V[\Delta]} \land \tsucc(s,\tdelay) \in \semT{\Phi_2}{\V[\Delta]} \} \span\span
                  \\
                   &                                                                                                                                                                   &                                         & \text{Definition of $\semop{\node{n}_i}$}
                  \\
                   & = \semT{\tU{\Phi_1}{\Phi_2}}{\V[\Delta]}
                   &                                                                                                                                                                   & \text{Semantics of $\exists$}
                  \\
                   & = \semop{\node{n}}
                   &                                                                                                                                                                   & \text{Definition of $\semop{\node{n}}$}
              \end{align*}

        \item $\rho(\node{n}) = (\forall,g)$.
              In this case we know that $\node{n} = S \tnxTVD \tR{\Phi_1}{\Phi_2}$ and that $cs(\node{n}) = \node{n}_1\node{n}_2$, where each $\node{n}_i = S_i \tnxTVD \Phi_i$, $S_1 = g_<(S)$ and $S_2 = g_=(S)$.
              The induction hypothesis guarantees that $\semT{P(\node{n}_i)}{\V'} = \semop{\node{n}_i} = \semT{\Phi_i}{\V[\Delta]}$.
              We reason as follows.
              \begin{align*}
                  & \semT{P(\node{n})}{\V'} \\
                   &{=}\; \semT{\tR{P(\node{n}_1)}{P(\node{n}_2)}}{\V'}
                   &                                                                                                                                                                             & \text{Definition of $P(-)$}
                  \\
                   &{=}\; \{ s \in \states{S} \mid \forall \tdelay \in \tdelays(s) \colon \tsucc_{<}(s,\tdelay) \cap \semT{P(\node{n}_1)}{\V'} = \emptyset \span\span\\
                   & \hspace{120pt} \implies \tsucc(s,\tdelay) \in \semT{P(\node{n}_2)}{\V'} \} \span\span
                  \\
                   &                                                                                                                                                                             &                                         & \text{Semantics of $\forall$}
                  \\
                   &{=}\; \{ s \in \states{S} \mid \forall \tdelay \in \tdelays(s) \colon \tsucc_{<}(s,\tdelay) \cap \semop{\node{n}_1} = \emptyset \implies \tsucc(s,\tdelay) \in \semop{\node{n}_2} \} \span\span
                  \\
                   &                                                                                                                                                                             &                                         & \text{Induction hypothesis (twice)}
                  \\
                   &{=}\; \{ s \in \states{S} \mid \forall \tdelay \in \tdelays(s) \colon \tsucc_{<}(s,\tdelay) \cap \semT{\Phi_1}{\V[\Delta]} = \emptyset \span\span \\
                   & \hspace{120pt}\implies \tsucc(s,\tdelay) \in \semT{\Phi_2}{\V[\Delta]} \} \span\span
                  \\
                   &                                                                                                                                                                             &                                         & \text{Definition of $\semop{\node{n}_i}$}
                  \\
                   &{=}\; \semT{\tR{\Phi_1}{\Phi_2}}{\V[\Delta]}
                   &                                                                                                                                                                             & \text{Semantics of $\forall$}
                  \\
                   &{=}\; \semop{\node{n}}
                   &                                                                                                                                                                             & \text{Definition of $\semop{\node{n}}$}
              \end{align*}
              \qedhere
    \end{itemize}
\end{proof}
Corollary~\ref{cor:node-formulas-vs-node-semantics}, as well as the definitions of support dependency ordering and influence extensions of valuations (Definitions~\ref{def:support-dependency-ordering} and~\ref{def:support-extension-of-valuation}) and the associated Lemmas~\ref{lem:support-dependency-ordering-characterization}, \ref{lem:pseudo-transitivity-of-support-dependency-ordering} and~\ref{lem:monotonicity-of-dependency-extensions} and Corollary~\ref{cor:monotonicity-of-dependency-extensions} generalize to timed tableaux in the obvious way.

These results now allow us to prove that $<:_\node{n}^+$ is also a support ordering for companion nodes in timed tableaux, generalizing Lemma~\ref{lem:support-ordering-for-companion-nodes}.
\begin{lemma}\label{lem:support-ordering-for-companion-nodes-timed}
    Let $\mathbb{T} = \tableauTrl$ be a successful timed tableau with $\node{n} \in \companions{\mathbb{T}}$ a companion node of $\mathbb{T}$ and $\node{n}'$ the child of $\node{n}$ in $\tree{T}$. Also let $S = \seqst(\node{n})$. Then $(S, <:_{\node{n}}^+)$ is a support ordering for $\semfT{Z_{\node{n}}}{P(\node{n}')}{\V_{\node{n}}}$.
\end{lemma}
\begin{proof}
The proof is analogous to the proof of Lemma~\ref{lem:support-ordering-for-companion-nodes}, inductively proving the following stronger result.
Fix successful timed tableau $\mathbb{T} = \tableauTrl$ with $\tree{T} = (\node{N},\node{r},p,cs)$ and let $\node{n} \in \companions{\mathbb{T}}$ a companion node of $\mathbb{T}$ with $S = \seqst(\node{n})$. We prove that for every $\node{m} \in D(\node{n})$ and $s \in S$ statements 
 \ref{stmt:necessity-timed} and \ref{stmt:support-timed} hold.
\begin{enumerate}[left=\parindent, label=S\arabic*., ref=S\arabic*]
\item\label{stmt:necessity-timed}
    For all $x$ such that $x \leq:_{\node{m},\node{n}} s$, 
    $x \in \semT{P(\node{m})}{\V_{\node{m},x}}$.
\item\label{stmt:support-timed}
    If $\node{m} \in \cnodes{\mathbb{T}}$,
    $\node{m}' = cs(\node{m})$
    and
    $x$ satisfies $x \leq:_{\node{m},\node{n}} s$
    then 
    $(S_x, <:_{\node{m},x})$ is a support ordering for $\semfT{Z_{\node{m}}}{P(\node{m}')}{\V_{\node{m},x}}$, where
    $S_x = \preimg{(<:_{\node{m}}^*)}{x}$
    and
    ${<:_{\node{m},x}} = \restrict{(<:_{\node{m}}^+)}{S_x}$.
\end{enumerate}
The proof proceeds by case analysis on the form of $\rho(\node{m})$; all cases are completely analogous to those in the proof of Lemma~\ref{lem:support-ordering-for-companion-nodes}, except the proofs of statement~\ref{stmt:necessity-timed} in case $\rho(\node{m}) \in \{ (\forall, g), (\exists,f) \}$. The proof for these two cases is as follows.

If $\rho(\node{m}) = (\exists,f)$, $\node{m} = S' \tnxTVD \tU{\Phi_1}{\Phi_2}$ for some $\Phi_1$ and $\Phi_2$, $cs(\node{m}) = \node{m}'_1\node{m}'_2$, $\node{m}'_i = S'_i \tnxTVD \Phi_i$ for $i = 1,2$, $S'_1 = f_<(S')$ and $S'_2 = f_=(S')$. 
The induction hypothesis ensures that \ref{stmt:necessity-timed} holds for each $\node{m}'_i$; we must show that \ref{stmt:necessity-timed} holds for $\node{m}$ and $s$.
To this end, let $x$ be such that $x \leq:_{\node{m},\node{n}} s$; we must show that $x \in \semT{P(\node{m})}{\V_{\node{m},x}}$.
Let $x'' = \tsucc(x,f(x))$, and 
note that $x'' <_{\node{m}_2',\node{m}} x$; the pseudo-transitivity of $\leq:_{\node{m},\node{n}}$ guarantees that $x''$ satisfies $x'' \leq:_{\node{m}'_2,\node{n}} s$, and the induction hypothesis then ensures that $x'' \in \semT{P(\node{m}'_2)}{\V_{\node{m}'_2,x''}}$. Corollary~\ref{cor:monotonicity-of-dependency-extensions} guarantees that $x'' \in \semT{P(\node{m}'_2)}{\V_{\node{m},x}}$.
Next, fix arbitrary $x' \in \tsucc_{<}(x, f(x))$, and note that $x' <_{\node{m}'_1,\node{m}} x$.
Using the same line of argument as before, we find that $x' \in \semT{P(\node{m}'_1)}{\V_{\node{m},x}}$.
So, $\tsucc_{<}(x, f(x)) \subseteq \semT{P(\node{m}'_1)}{\V_{\node{m},x}}$.
Now, the semantics of $\exists$ ensures that $x \in \semT{P(\node{m})}{\V_{\node{m},x}}$.

If $\rho(\node{m}) = (\forall,g)$, $\node{m} = S' \tnxTVD \tR{\Phi_1}{\Phi_2}$ for some $\Phi_1$ and $\Phi_2$, $cs(\node{m}) = \node{m}'_1\node{m}'_2$, $\node{m}'_i = S'_i \tnxTVD \Phi_i$ for $i = 1,2$, and $S_1 = g_<(S)$, $S_2 = g_=(S)$.
The induction hypothesis ensures that \ref{stmt:necessity-timed} holds for each $\node{m}'_i$; we must show that \ref{stmt:necessity-timed} holds for $\node{m}$ and $s$.
To this end, let $x$ be such that $x \leq:_{\node{m},\node{n}} s$; we must show that $x \in \semT{P(\node{m})}{\V_{\node{m},x}}$.
The result follows from the semantics of $\forall$ if we prove for all $\tdelay \in \tdelays(x)$ that either $\tsucc_{<}(x,\tdelay) \cap \semT{P(\node{m}'_1)}{\V_{\node{m},x}} \neq \emptyset$ or $\tsucc(x,\tdelay) \in \semT{P(\node{m}'_2)}{\V_{\node{m},x}}$.
So, fix arbitrary $\tdelay \in \tdelays(x)$.
Suppose $g(x,\tdelay) = \tdelay$, and let $x' = \tsucc(x,g(x,\tdelay))$; then $x' <_{\node{m}'_2,\node{m}} x$. The pseudo-transitivity of $\leq:_{\node{m},\node{n}}$ guarantees that $x' \leq:_{\node{m}'_2,\node{n}} s$, and the induction hypothesis then ensures that $x' \in \semT{P(\node{m}'_2)}{\V_{\node{m}'_2,x'}}$.
Corollary~\ref{cor:monotonicity-of-dependency-extensions} guarantees that $x' \in \semT{P(\node{m}'_2)}{\V_{\node{m},x}}$.
Next, suppose $g(x,\tdelay) = \tdelay' < \tdelay$, and let $x'' = \tsucc(x,g(x,\tdelay'))$. Observe $x'' \in \tsucc_{<}(x,g(x,\tdelay))$, and $x'' <_{\node{m}'_1,\node{m}} x$. The pseudo-transitivity of $\leq:_{\node{m},\node{n}}$ guarantees that $x'' \leq:_{\node{m}'_1,\node{n}} s$, and the induction hypothesis then ensures that $x'' \in \semT{P(\node{m}'_1)}{\V_{\node{m}'_1,x''}}$. Corollary~\ref{cor:monotonicity-of-dependency-extensions} guarantees that $x'' \in \semT{P(\node{m}'_1)}{\V_{\node{m},x}}$. As $x'' \in \tsucc_{<}(x,g(x,\tdelay))$, $\tsucc_{<}(x,g(x,\tdelay)) \cap \semT{P(\node{m}'_1)}{\V_{\node{m},x}} \neq \emptyset$.
Hence it follows from the semantics of $\forall$ that $x \in \semT{P(\node{m})}{\V_{\node{m},x}}$.\qedhere
\end{proof}

Corollary~\ref{cor:support-orderings-for-top-level-companion-nodes} again immediately generalizes to timed tableaux based on the previous lemma. We now conclude our soundness proof.

\begin{theorem}[Soundness of timed mu-calculus proof system\label{thm:soundness-timed-tableau}]
    Fix TTS $(\states{S},\to)$ of timed sort $\Sigma$ and valuation $\V$, and let $\mathbb{T} = \tableauTrl$ be a successful timed tableau for sequent $\seq{s}$, where $\seqdl(\seq{s}) = \emptyL$. Then $\seq{s}$ is valid.
\end{theorem}
\begin{proof}
    Follows the exact same line of reasoning as the proof of Theorem~\ref{thm:soundness}.\qedhere
\end{proof}

Observe that the results in this section only ever involve adding cases for the new operators to definitions, as well as to the proofs that use case distinction on the operators. In all of the cases we considered, the results that need to be added for these new operators are straightforward, and follow the same line of reasoning as the other operators in the mu-calculus.
This illustrates the extensibility of the proof methodology we developed, at least for proving soundness of tableaux when adding new operators to the mu-calculus.

\subsection{Completeness}\label{sec:timed-completeness}
We finally turn our attention to proving completeness of the tableaux construction for the timed mu-calculus.
In particular, we show that the construction used to establish completeness in Section~\ref{sec:Completeness} can be straightforwardly adapted to account for the new modalities introduced by the timed mu-calculus. This, again, illustrates the extensibility of the proofs given in this paper.

We first note that the notion of tableau normal form (TNF) introduced in Definition~\ref{def:tableau-normal-forms} carries over to timed tableaux as well; the proof of the corresponding Lemma~\ref{lem:structural-equivalence-of-TNF-tableaux} only needs to be adapted by including $\forall$ and $\exists$ in case that $\rn(\rho(\node{n})) \in \{\land,\lor,[K],\dia{K}\}$.  The details of that adaptation are routine and left to the reader.
We now establish completeness for timed mu-calculus formulas without fixed points, extending the result in Lemma~\ref{lem:fixpoint-free-completeness}.

\begin{lemma}[Timed fixpoint-free completeness]\label{lem:timed-fixpoint-free-completeness}
Let $\T, \V, \Phi$ and $S$ be such that $\Phi$ is a fixpoint-free timed mu-calculus formula and $S \subseteq \semTV{\Phi}$.  Then there is a successful TNF timed tableau for $S \tnxTV{\emptyL} \Phi$.
\end{lemma}
\begin{proof}
Let $\T = (\states{S}, \to)$ be a TTS of timed sort $\Sigma$, and $\V$ be a valuation.
The proof proceeds by structural induction on $\Phi$; the induction hypothesis states that for any subformula $\Phi'$ of $\Phi$ and $S'$ such that $S' \subseteq \semTV{\Phi'}$, $S' \tnxTV{\emptyL} \Phi'$ has a successful TNF timed tableau.
The proof is completely analogous to that of Lemma~\ref{lem:fixpoint-free-completeness}, and involves a case analysis on the form of $\Phi$.
We here only show the cases for the new operators, $\forall$ and $\exists$.

Assume $\Phi = \exists_{\Phi_1} \Phi_2$;
we first establish that there exists a function $f \in S \to \Rnneg$ such that for all $s \in S$, $f(s) \in \tdelays(s)$,  $\tsucc_{<}(s,f(s)) \subseteq \semTV{\Phi_1}$, and $\tsucc(s,f(s)) \in \semTV{\Phi_2}$.
Fix arbitrary $s \in S$. Then $s \in \semTV{\exists_{\Phi_1} \Phi_2}$, and the definition of $\semTV{\exists_{\Phi_1} \Phi_2}$ guarantees the existence of $\tdelay \in \tdelays(s)$ such that $\tsucc_{<}(s,\tdelay) \subseteq \semTV{\Phi_1}$ and $\tsucc(s,\tdelay) \in \semTV{\Phi_2}$. Let $\tdelay$ be such and set $f(s) = \tdelay$. This immediately satisfies the required conditions.
It follows from this definition of $f$ that $f_<(S) \subseteq \semTV{\Phi_1}$. Furthermore, as for every $s \in S$, $\tsucc(s,f(s)) \in \semTV{\Phi_2}$, $f_=(S) \subseteq \semTV{\Phi_2}$.
Hence, the induction hypothesis guarantees that successful TNF timed tableaux exist for $f_<(S) \tnxTV{\emptyL} \Phi_1$ and $f_=(S) \tnxTV{\emptyL} \Phi_2$.
Without loss of generality, assume that these tableaux have disjoint sets of proof nodes.
We now construct a successful TNF timed tableau for $S \tnxTV{\emptyL} \exists_{\Phi_1} \Phi_2$ as follows.
Create a fresh tree node labeled by $S \tnxTV{\emptyL} \exists_{\Phi_1} \Phi_2$, having as its left child the root node of the successful TNF  timed tableau for $f_<(S) \tnxTV{\emptyL} \Phi_1$ and as its right child the root of the successful TNF timed tableau for $f_=(S) \tnxTV{\emptyL} \Phi_2$. The rule application associated with the new node is $(\exists,f)$.
The new tableau is clearly successful and TNF. 

Now assume $\Phi = \forall_{\Phi_1} \Phi_2$;
we first establish that there exists a function $g \in S \times \Rnneg \to \Rnneg$ such that, for all $s \in S$, $\tdelay \in \tdelays(s)$, $g(s,\tdelay) \leq \tdelay$, and such that if $g(s,\tdelay) = \tdelay$, $\tsucc(s,g(s,\tdelay)) \in \semTV{\Phi_2}$, and $\tsucc(s,g(s,\tdelay)) \in \semTV{\Phi_1}$ otherwise.
Fix arbitrary $s \in S$. Since $s \in \semTV{\forall_{\Phi_1} \Phi_2}$,
the definition of $\semTV{\forall_{\Phi_1} \Phi_2}$ guarantees that, for all $\tdelay \in \tdelays(s)$, either $\tsucc_{<}(s,\tdelay) \cap \semTV{\Phi_1} \neq \emptyset$ or $\tsucc(s,\tdelay) \in \semTV{\Phi_2}$. So, if $\tsucc_{<}(s,\tdelay) \cap \semTV{\Phi_1} \neq \emptyset$ there is some $\tdelay' < \tdelay$ such that $\tsucc(s,\tdelay') \in \semTV{\Phi_1}$, and we choose $g(s,\tdelay) = \tdelay'$. 
Otherwise, $\tsucc(s,\tdelay) \in \semTV{\Phi_2}$ and we can choose $g(s,\tdelay) = \tdelay$.
Given such a $g$, 
it is easy to see that $g_<(S) \subseteq \semTV{\Phi_1}$ and $g_=(S) \subseteq \semTV{\Phi_2}$; hence the induction hypothesis guarantees that successful TNF timed tableaux exist for $g_<(S) \tnxTV{\emptyL} \Phi_1$ and $g_=(S) \tnxTV{\emptyL} \Phi_2$.
Without loss of generality, assume that these tableaux have disjoint sets of proof nodes.
We now construct a successful TNF timed tableau for $S \tnxTV{\emptyL} \forall_{\Phi_1} \Phi_2$ as follows.
Create a fresh tree node labeled by $S \tnxTV{\emptyL} \forall_{\Phi_1} \Phi_2$, having as its left child the root node of the successful TNF  timed tableau for $g_<(S) \tnxTV{\emptyL} \Phi_1$ and as its right child the root of the successful TNF timed tableau for $g_=(S) \tnxTV{\emptyL} \Phi_2$. The rule application associated with the new node is $(\forall,g)$.
The new tableau is clearly successful and TNF.\qedhere
\end{proof}

The notion of tableau compliance (Definition~\ref{def:tableau-compliance}) generalizes to the timed setting in the obvious way.
We next establish the existence of successful TNF timed tableaux for valid sequents in the case where formulas have the form $\sigma Z . \Phi$, where $\Phi$ does not contain fixpoint subformulas, and have specific $\sigma$-compatible fixpoint orderings, analogous to Lemma~\ref{lem:single-fixpoint-completeness}.

\begin{lemma}[Timed single-fixpoint completeness]\label{lem:timed-single-fixpoint-completeness}
Let $\T$ be a TTS, and let $\Phi$, $Z$, $\V$, $\sigma$ and $S$ be such that $\Phi$ is a fixpoint-free timed mu-calculus formula and $S = \semTV{\sigma Z . \Phi}$.
Also let $(S, \prec)$ be a $\sigma$-compatible, total, qwf support ordering for $\semfZTV{\Phi}$. 
Then $S \tnxTV{\emptyL} \sigma Z . \Phi$ has a successful TNF timed tableau compliant with $(S, \prec)$.
\end{lemma}
\begin{proof}
Fix TTS $\T= \lts{\states{S}}$ of sort $\Sigma$, and let $\Phi, Z, \V, \sigma$ and $S$ be such that $\Phi$ is a fixpoint-free timed mu-calculus formula, and $S = \semTV{\sigma Z.\Phi}$.  Also let $(S, \prec)$ be a $\sigma$-compatible, total, qwf support ordering for $f_\Phi = \semfZTV{\Phi}$.
We must construct a successful TNF timed tableau for sequent $S \tnxTV{\emptyL} \sigma Z.\Phi$ that is compliant with $(S,\prec)$. The proof mirrors the one given for Lemma~\ref{lem:single-fixpoint-completeness} and consists of the following steps.
\begin{enumerate}
    \item\label{it:step-single-state-timed-tableau}
    For each $s \in S$ we use Lemma~\ref{lem:timed-fixpoint-free-completeness} to establish the existence of a successful TNF timed tableau for sequent $\{s\} \tnxT{\V_s}{\emptyL} \Phi$, where $\V_s = \V[Z := \preimg{{\prec}}{s}]$.
    \item\label{it:step-full-set-timed-tableau}
    We then construct a successful TNF timed tableau for sequent $S \tnxT{\V_S}{\emptyL} \Phi$, where $\V_S = \V[Z := \preimg{{\prec}}{S}]$, from the individual tableaux for the $s \in S$.
    \item\label{it:step-fixpoint-timed-tableau}
    We convert the tableau for $S \tnxT{\V_s}{\emptyL} \Phi$ into a successful TNF timed tableau for $S \tnxTV{\emptyL} \sigma Z.\Phi$ that is compliant with $\prec$.
\end{enumerate}


\paragraph{Step~\ref{it:step-single-state-timed-tableau} of proof outline:  construct tableau for $\{s\} \tnxT{\V_s}{\emptyL} \Phi$, where $s \in S$.}
This construction is completely analogous to the one in the proof of Lemma~\ref{lem:single-fixpoint-completeness}, using Lemma~\ref{lem:timed-fixpoint-free-completeness} instead of Lemma~\ref{lem:fixpoint-free-completeness}.  Let $\mathbb{T}_s = (\T,\V_s,\tree{T},\rho_s,\lambda_s)$ be the successful TNF timed tableau whose root is $\{s\} \tnxT{\V_s}{\emptyL} \Phi$, with $\tree{T} = (\node{N},\node{r},p,cs)$ the common tree shared by all these structurally equivalent tableaux, with $\seqfm(\node{n})$ and $\rn(\node{n})$ for $\node{n} \in \node{N}$ the common formulas and rule names each tableau includes in $\node{n}$.

\paragraph{Step~\ref{it:step-full-set-timed-tableau} of proof outline:  construct tableau for $S \tnxT{\V_S}{\emptyL} \Phi$.}
We now adapt the construction in the proof of Lemma~\ref{lem:single-fixpoint-completeness} to build a successful TNF timed tableau for $S \tnxT{\V_S}{\emptyL} \Phi$ satisfying the following:  if $s, s'$ and $\node{n}'$ are such that $\seqfm(\node{n}') = Z$ and $s' <:_{\node{n}',\node{r}} s$, then $s' \prec s$.  There are two cases to consider.
In the first case, $S = \emptyset$.  In this case, ${\prec} = \preimg{{\prec}}{S} = \emptyset$, and $\emptyset \tnxT{\V_S}{\emptyL} \Phi$ is valid and therefore, by Lemma~\ref{lem:timed-fixpoint-free-completeness}, has a successful TNF tableau.  Define $\mathbb{T}_S$ to be this tableau.  Note that since $S = \emptyset$ $T_S$ vacuously satisfies the property involving $<:$.

In the second case, $S \neq \emptyset$; we will construct $\mathbb{T}_S = (\T, \V_S, \tree{T}, \rho_S, \lambda_S)$ that is structurally equivalent to each $\mathbb{T}_s$ for $s \in S$.
As was the case in the proof of Lemma~\ref{lem:single-fixpoint-completeness} 
the idea is to appropriately ``merge" the individual tableaux $\mathbb{T}_s$ for the $s \in S$.
The construction uses a co-inductive strategy to define $\rho_S$ and $\lambda_S$ so that invariants \ref{inv:rule}--\ref{inv:state-set} in Lemma~\ref{lem:single-fixpoint-completeness} hold for $\node{n} \in \node{N}$.  Specifically, $\lambda_S(\node{r})$ is set to be sequent $S \tnxT{\emptyL}{\V_S} \Phi$; this ensures that invariants $\ref{inv:formula}$ and $\ref{inv:state-set}$ hold of $\node{r}$.  Then, for every internal node $\node{n}$ for which $\lambda_S(\node{n})$ has been defined and for which invariants~\ref{inv:formula} and~\ref{inv:state-set} hold, $\rho_S(\node{n})$ and $\lambda_S(\node{n}')$ for each child $\node{n}'$ of $\node{n}$ are defined so that invariant~\ref{inv:rule} holds of $\node{n}$ and invariants~\ref{inv:formula} and~\ref{inv:state-set} hold of each $\node{n}'$.  This processing of internal nodes is done using a case analysis on $\rn(\node{n})$.  The details are as follows, where we let $S_{\node{n}} = \seqst(\lambda_S(\node{n}))$ be the set of states in the sequent labeling $\node{n}$.
\begin{description}
    \item[$\rn(\node{n})\div$, or $\rn(\node{n}) \in \{ \land, \lor, [K{]}, \dia{K} \}$.] The constructions are the same as the ones in the proof of Lemma~\ref{lem:single-fixpoint-completeness}.
    
    \item[$\rn(\node{n}) = \exists$.] 
     In this case, $cs(\node{n}) = \node{n}_1\node{n}_2$ and $\seqfm(\node{n}) = \exists_{\seqfm(\node{n}_1)}\seqfm(\node{n}_2)$.
     We begin by constructing a function $f_{\node{n}} \in S_{\node{n}} \to \Rnneg$ such that for all $s \in S_{\node{n}}$, $f_{\node{n}}(s) \in \tdelays(s)$, and also such that $(f_\node{n})_<(S_\node{n}) \subseteq \bigcup_{t \in S} \seqst(\lambda_t(\node{n}_1))$ and $(f_\node{n})_=(S_\node{n}) \subseteq \bigcup_{t \in S} \seqst(\lambda_t(\node{n}_2))$.
     This function will then be used to define $\rho_S(\node{n}$, $\lambda_S(\node{n}_1)$ and $\lambda_S(\node{n}_2)$.
     So fix $s \in S_{\node{n}}$; we construct $f_{\node{n}}(s)$ based on the tableaux $\mathbb{T}_t$ whose sequence for $\node{n}$ contains $s$. To this end, define
     \[
     I_s = \{t \in S \mid s \in \seqst(\lambda_{t}(\node{n})) \}.
     \]
     Intuitively, $I_s \subseteq S$ contains all states $t$ whose tableau $\mathbb{T}_{t}$ contains state $s$ in $\node{n}$.
    Clearly $I_s$ is non-empty and thus contains a pseudo-minimum element $t$ (Lemma~\ref{lem:qwo-pseudo-minimum}).
    Let $f_{\node{n},t} \in \seqst(\lambda_t(\node{n})) \to \Rnneg$ be such that $\rho_t(\node{n}) = (\exists, f_{\node{n},t})$.
    We know that
    for all $s \in  \seqst(\lambda_t(\node{n}))$, $f_{\node{n},t}(s) \in \tdelays(s)$,  $\tsucc_{<}(s,f_{\node{n},t}(s)) \subseteq \seqst(\lambda_t(\node{n}_1))$, and $\{ \tsucc(s, f_{\node{n},t}(s)) \mid s \in \seqst(\lambda_t(\node{n})) \} \subseteq \seqst(\lambda_t(\node{n}_2))$.
    We now define $f_{\node{n}}(s) = f_{\node{n},t}(s)$
    Finally, we define $\rho_S(\node{n}) = (\exists, f_{\node{n}})$,
    $\lambda_S(\node{n}_1) = (f_\node{n})_<(S_\node{n}) \tnxT{\V_S}{\emptyL} \seqfm(\node{n}_1)$, and $\lambda_S(\node{n}_2) = (f_\node{n})_=(S_\node{n}) \tnxT{\V_S}{\emptyL} \seqfm(\node{n}_2)$.
    It can be seen that invariant~\ref{inv:rule} holds of $\node{n}$ and that \ref{inv:formula} and \ref{inv:state-set} hold of $\node{n}_1$ and $\node{n}_2$.

    \item[$\rn(\node{n}) = \forall$.]
    In this case, $cs(\node{n}) = \node{n}_1\node{n}_2$ and $\seqfm(\node{n}) = \forall_{\seqfm(\node{n}_1)}\seqfm(\node{n}_2)$.
    We begin by constructing a function $g_\node{n} \in S_{\node{n}} \times \Rnneg \to \Rnneg$ such that $(g_\node{n})_<(S_\node{n}) \subseteq \bigcup_{t \in S} \seqst(\lambda_t(\node{n}_1))$ and $(g_\node{n})_=(S_\node{n}) \subseteq \bigcup_{t \in S} \seqst(\lambda_t(\node{n}_2))$.
    This function will then be used to define $\rho_S(\node{n})$, $\lambda_S(\node{n}_1)$ and $\lambda_S(\node{n}_2)$ so that the desired invariants hold.
    So fix $s \in S_\node{n}$ and  $\tdelay \in \tdelays(s)$; we construct $g_{\node{n}}(s,\tdelays)$ based on the tableaux $\mathbb{T}_{t}$ whose sequent for $\node{n}$ contains $s$. To this end, define
    \[
    I_s = \{t \in S \mid s \in \seqst(\lambda_{t}(\node{n})) \}.
    \]
    Intuitively, $I_s \subseteq S$ contains all states $t$ whose tableau $\mathbb{T}_{t}$ contains state $s$ in $\node{n}$.
    Clearly $I_s$ is non-empty and thus contains a pseudo-minimum element $t$ with respect to $\prec$ (Lemma~\ref{lem:qwo-pseudo-minimum}).
    Let $g_{\node{n}, t} \in \seqst(\lambda_{t}(\node{n})) \times \states{S} \to \states{S}$ be such that $\rho_t(\node{n}) = (\forall,g_{\node{n},t})$.  We know that for all $\tdelay \in \tsucc(s)$, if $g_{\node{n},t}(s,\tdelay) < \tdelay$ then $\tsucc(s,g_{\node{n},t}(s,\tdelay)) \in \seqst(\lambda_t(\node{n}_1))$, and if $g_{\node{n},t}(s,\tdelay) = \tdelay$ then $\tsucc(s,g_{\node{n},t}(s,\tdelay)) \in \seqst(\lambda_t(\node{n}_2))$.
    We take $g_{\node{n}}(s,\tdelay) = g_{\node{n},t}(s,\tdelay)$.  
    Finally, we define $\rho(\node{n}) = (\forall, g_{\node{n}})$, $\lambda_S(\node{n}_1) = (g_\node{n})_<(S_\node{n}) \tnxT{\V_S}{\emptyL} \seqfm(\node{n}_1)$ and $\lambda_S(\node{n}_2) = (g_\node{n})_=(S_\node{n}) \tnxT{\V_S}{\emptyL} \seqfm(\node{n}_2)$.
     It can be seen that invariant~\ref{inv:rule} holds of $\node{n}$ and that \ref{inv:formula} and \ref{inv:state-set} hold of $\node{n}_1$ and $\node{n}_2$.
     
\end{description}
This construction ensures that Properties~\ref{inv:rule}--\ref{inv:state-set} hold for all $\node{n}$.

To establish that $\mathbb{T}_S$ is successful we must show that every leaf in $\mathbb{T}_S$ is successful.  The argument is identical to the proof in Lemma~\ref{lem:single-fixpoint-completeness}, Step~\ref{it:step-full-set-tableau}.

\paragraph{Step~\ref{it:step-fixpoint-timed-tableau} of proof outline:  construct tableau for $S \tnxT{\V}{\emptyL} \sigma Z.\Phi$.}
The construction is identical to the one in Lemma~\ref{lem:single-fixpoint-completeness}, Step~\ref{it:step-fixpoint-tableau}.
\qedhere
\end{proof}

\noindent
As an immediate corollary, we have the following.

\begin{corollary}\label{cor:timed-single-fixpoint-completeness}
Fix $\T$, and let $\Phi, Z, \V, \sigma$ and $S$ be such that $\Phi$ is a timed fixpoint-free formula and $S = \semTV{\sigma Z.\Phi}$.  Then $S \tnxTV{\emptyL} \sigma Z.\Phi$ has a successful tableau.
\end{corollary}
\begin{proof}
Follows from Lemma~\ref{lem:timed-single-fixpoint-completeness} and the fact that every $\sigma$-maximal support ordering $(S, \prec)$ for $\semfTV{Z}{\Phi}$ is total and qwf.
\qedhere
\end{proof}

We now generalize Lemma~\ref{lem:timed-single-fixpoint-completeness} to the multiple fixpoint case.  This lemma mirrors the analogous result (Lemma~\ref{lem:fixpoint-completeness}) proved earlier for the non-timed mu-calculus, and the proof is a straightforward extension of the earlier proof.

\begin{lemma}[Timed fixpoint completeness]\label{lem:timed-fixpoint-completeness}
Fix TTS $\T$, let $\Phi$ be a timed mu-calculus formula, and let $Z, \V, \sigma$ and $S$ be such that $S = \semTV{\sigma Z.\Phi}$.  Also let $(S, \prec)$ be a $\sigma$-compatible, total, qwf support ordering for $\semfTV{Z}{\Phi}$.  Then $S \tnxTV{\emptyL} \sigma Z.\Phi$ has a successful TNF timed tableau compliant with $(S, \prec)$.
\end{lemma}

\begin{proof}

Fix $\T = \lts{\states{S}}$ of timed sort $\Sigma$.  We prove the following: for all timed $\Phi$, and $Z, \V, \sigma$ and $S$ with
$S = \semTV{\sigma Z.\Phi}$,
and $\sigma$-compatible, total qwf support ordering $(S,\prec)$ for $\semfTV{Z}{\Phi}$, $S \tnxTV{\emptyL} \sigma Z.\Phi$ has a successful TNF timed tableau $\mathbb{T}_\Phi$ that is compliant with $(S,\prec)$.
To simplify notation we use the following abbreviations.
\begin{align*}
f_\Phi      &= \semfTV{Z}{\Phi}
\\
\V_X        &= \V[Z := \preimg{{\prec}}{X}]
\end{align*}
Note that $\V_S = \V[Z := \preimg{{\prec}}{S}]$.  When $s \in S$ we also write $\V_s$ in lieu of $V_{\{ s\}}$.

The proof proceeds by strong induction on the number of fixpoint subformulas of $\Phi$. There are two cases to consider.  In the first case, $\Phi$ contains no fixpoint formulas. Lemma~\ref{lem:timed-single-fixpoint-completeness} immediately gives the desired result.

In the second case, $\Phi$ contains at least one fixpoint subformula.  The outline of the proof follows the same lines as the one for the same case in Lemma~\ref{lem:fixpoint-completeness}).
\begin{enumerate}
    \item\label{it:step-decompose-timed} 
    We decompose $\Phi$ into $\Phi'$, which uses a new free variable $W$, and $\sigma' Z'.\Gamma$ in such a way that $\Phi = \Phi'[W:=\sigma' Z'.
    \Gamma]$.
    \item\label{it:step-outer-timed-tableau}
    We inductively construct a successful TNF timed tableau $\mathbb{T}_{\Phi'}$ for $S \tnxT{\V'}{\emptyL} \sigma Z.\Phi'$ that is compliant with $(S,\prec)$ where:
    \begin{align*}
        S'  &= \semT{\sigma'Z'.\Gamma}{\V_S} \\
        \V' &= \V[W:=S'].
    \end{align*}
    ($S'$ may be seen as the semantic content of $\sigma'Z'.\Gamma$ relevant for $\semTV{\sigma Z.\Phi}$.)
    \item\label{it:step-inner-timed-tableau}
    We construct a successful TNF timed tableau $\mathbb{T}_\Gamma$ satisfying a compliance-related property for $S' \tnxT{\V_S}{\emptyL} \sigma'Z'.\Gamma$ by merging inductively constructed tableaux involving subsets of $S'$.
    \item\label{it:step-timed-tableau-composition}
    We show how to compose $\mathbb{T}_\Phi$ and $\mathbb{T}_\Gamma$ to yield a successful TNF timed tableau for $S \tnxTV \sigma Z.\Phi$ that is compliant with $(S,\prec)$.
\end{enumerate}
We now work through each of these proof steps.

\paragraph{Step~\ref{it:step-decompose-timed} of proof outline:  decompose $\Phi$.}
Completely analogous to Step~\ref{it:step-decompose} in Lemma~\ref{lem:fixpoint-completeness}.  We recall the following function definitions from the proof of that lemma,
\begin{align*}
f(X,Y) &= \semT{\Phi'}{\V[Z, W := X, Y]}\\
g(X,Y)  &= \semT{\Gamma}{\V[Z, Z' := X, Y]},
\end{align*}
and also note that $f_\Phi = f[\sigma']g$.

\paragraph{Step~\ref{it:step-outer-timed-tableau} of proof outline: construct tableau for $S \tnxT{\V'}{\emptyL} \sigma Z.\Phi'$ that is compliant with $(S,\prec)$.}
Completely analogous to Step~\ref{it:step-outer-tableau} in Lemma~\ref{lem:fixpoint-completeness}.  Recall that $\V' = \V[W := S']$, where $S' = \semT{\sigma' Z'.\Gamma}{\V_S}$.

\paragraph{Step~\ref{it:step-inner-tableau} of proof outline:  construct tableau for $S' \tnxT{\V_S}{\emptyL} \sigma'Z'.\Gamma$.}

Since $\Gamma$ contains strictly fewer fixpoint subformulas than $\Phi$, the induction hypothesis guarantees the existence of certain successful TNF timed tableaux involving $\sigma' Z'.\Gamma$.
Following the proof of Lemma~\ref{lem:fixpoint-completeness} we use this fact to construct a successful tableau, $\mathbb{T}_\Gamma$ for $S' \tnxT{\V_S}{\emptyL} \sigma'Z'.\Gamma$ satisfying Property~\ref{goal:dependencies} defined in that earlier proof.
%
%
We begin by recalling the following definitions from that earlier proof, where $X \subseteq S$.
\begin{align*}
g_X &= g_{(\preimg{{\prec}}{X},\cdot)}
\\
S'_X &= \sigma' g_X
\end{align*}
It was also noted there that, for any $X \subseteq S$,
$
S'_X \subseteq S'.
$
We again take $(Q_\prec, \sqsubseteq)$ to be the quotient of $(S,\prec)$, with $[x] \in Q_\prec$ the equivalence class of $x \in S$.
We also let $(S',\prec')$ be the $\sigma'$-compatible, total qwf support ordering $g_S$
that is locally consistent with $(S,\prec)$, as guaranteed by Lemma~\ref{lem:fg-support}.
The construction we present below for $\mathbb{T}_\Gamma$ follows the same approach as the one in the proof of Lemma~\ref{lem:fixpoint-completeness}.
\begin{itemize}
    \item
    For each $Q \in Q_\prec$ we inductively construct a successful TNF timed tableau $\mathbb{T}_{\Gamma,Q}$ for sequent $S'_Q \tnxT{\V_Q}{\emptyL} \sigma' Z'.\Gamma$ that is compliant with a subrelation of $\prec'$.
    \item
    We then merge the individual $\mathbb{T}_{\Gamma,Q}$ to form a successful TNF timed tableau $\mathbb{T}'_\Gamma$ compliant with $\prec'$ whose root sequent contains as its state set the union of all the individual root-sequent state sets of the $\mathbb{T}_{\Gamma,Q}$.
    \item
    We perform a final operation to obtain $\mathbb{T}_\Gamma$.
\end{itemize}

The constructions of $\mathbb{T}_{\Gamma,Q}$ and $\mathbb{T}'_\Gamma = (\T,\V_S,\tree{T}_\Gamma,\rho'_\Gamma,\lambda'_\Gamma)$ are completely analogous to the corresponding constructions in the proof of Lemma~\ref{lem:fixpoint-completeness}.
In what follows we focus on the final step:  constructing $\mathbb{T}_\Gamma$ that is compliant with $\prec'$ and satisfies Property~\ref{goal:dependencies}.
We begin by noting that since $(S',\prec')$ is a $\sigma'$-compatible, total qwf support ordering for $g_S$, the induction hypothesis and Lemma~\ref{lem:structural-equivalence-of-TNF-tableaux} guarantee the existence of a successful TNF tableau 
$$
\mathbb{T}_{\Gamma,S} = (\tree{T}, \rho_{\Gamma,S}, \T, \V_S, \lambda_{\Gamma,S})
$$
for sequent $S' \tnxT{\V_S}{\emptyL} \sigma'Z'.\Gamma$ that is compliant with $\prec'$ and structurally equivalent to $\mathbb{T}_{\Gamma,Q}$ for any $Q \in Q_\prec$.  
There are two cases to consider.  In the first case, $S = \emptyset$.  In this case, $\mathbb{T}_{\Gamma,S}$ vacuously satisfies Property~\ref{goal:dependencies}, and we take $\mathbb{T}_\Gamma$ to be $\mathbb{T}_{\Gamma,S}$.

In the second case, $S \neq \emptyset$, and thus $Q_\prec \neq \emptyset$.  In this case, since it is not guaranteed that $\mathbb{T}_{\Gamma,S}$ satisfies \ref{goal:dependencies}, we follow the approach in the proof of Lemma~\ref{lem:fixpoint-completeness} by building $\mathbb{T}_\Gamma$ using a coinductive definition of $\rho_\Gamma$ and $\lambda_\Gamma$ that merges $\mathbb{T}'_\Gamma$ and $\mathbb{T}_{\Gamma,S}$ so that the following invariants are satisfied.
\begin{invariants}
    \item
      If $\rho_\Gamma(\node{n})$ is defined then the sequents assigned by $\lambda_\Gamma$ to $\node{n}$ and its children constitute a valid application of the rule $\rho_\Gamma(\node{n})$.
    \item
      $\seqfm(\lambda_\Gamma(\node{n})) = \seqfm(\node{n})$
    \item
      $\seqst(\lambda_\Gamma(\node{n})) \subseteq \seqst(\lambda'_\Gamma(\node{n})) \cup \seqst(\lambda_{\Gamma,S}(\node{n}))$
    \item
      $\seqdl(\lambda_\Gamma(\node{n})) = \seqdl_\Gamma(\node{n})$
\end{invariants}
The definitions  begin by assigning a value to $\lambda_\Gamma(\node{r}_\Gamma)$ so that invariants~\ref{inv:formula}--\ref{inv:definition-list} are satisfied.  The coinductive step then assumes that $\node{n}$ satisfies these invariants and defines $\rho_\Gamma(\node{n})$ and $\lambda_\Gamma$ for the children of $\node{n}$ so that \ref{inv:rule} holds for $\node{n}$ and \ref{inv:formula}--\ref{inv:definition-list} hold for each child.

To start the construction, define $\lambda_\Gamma(\node{r}_\Gamma) = S' \tnxT{\V_S}{\emptyL} \sigma'Z'.\Gamma$.  Invariants~\ref{inv:formula}--\ref{inv:definition-list} clearly hold of $\node{r}_\Gamma$, since $\seqst(\lambda'_\Gamma(\node{r}_\Gamma)) \subseteq S' = \seqst(\lambda_\Gamma(\node{r}_\Gamma))$.

For the coinductive step, assume that $\node{n}$ is such that $\lambda_\Gamma(\node{n})$ satisfies \ref{inv:formula}--\ref{inv:definition-list}; we must define $\rho_\Gamma(\node{n})$, and $\lambda_\Gamma$ for the children of $\node{n}$, so that \ref{inv:rule} holds for $\lambda_\Gamma(\node{n})$ and \ref{inv:formula}--\ref{inv:definition-list} holds for each child. All cases are analogous to the cases in the proof of Lemma~\ref{lem:fixpoint-completeness}, except those that involve rules $\forall, \exists$, which are given below.
\begin{description}
    \item[$\rn_\Gamma(\node{n}) = \exists$.] In this case $cs(\node{n}) = \node{n}_1\node{n}_2$, $\lambda_{\Gamma}(\node{n}) = S_{\node{n}} \tnxT{\V_S}{\Delta} \exists_{\seqfm_\Gamma(\node{n}_1)} \seqfm_\Gamma(\node{n}_2)$, with $\Delta = \seqdl_\Gamma(\node{n}) = \seqdl_\Gamma(\node{n}_1) = \seqdl_\Gamma(\node{n}_2)$. 
    We construct a function $f_{\Gamma,\node{n}} \in S_{\node{n}} \to \Rnneg$ such that for all $s \in S_{\node{n}}$, $f_{\Gamma,\node{n}}(s) \in \tdelays(s)$ and such that
    $(f_{\Gamma,\node{n}})_<(S_\node{n}) \subseteq \seqst(\lambda'_\Gamma(\node{n}_1)) \cup \seqst(\lambda_{\Gamma,S}(\node{n}_1))$ and
    $(f_{\Gamma,\node{n}})_=(S_\node{n}) \subseteq \seqst(\lambda'_\Gamma(\node{n}_2)) \cup \seqst(\lambda_{\Gamma,S}(\node{n}_2))$.
    Since $\mathbb{T}'_{\Gamma}$ and $\mathbb{T}_{\Gamma,S}$ are successful we know that $\exists$-functions $f'_{\Gamma,\node{n}} \in \seqst(\lambda'_{\Gamma}(\node{n})) \to \Rnneg$ and $f_{\Gamma,S,\node{n}} \in \seqst(\lambda_{\Gamma,S}(\node{n})) \to \Rnneg$, where $\rho'_\Gamma(\node{n}) = (\exists, f'_{\Gamma,\node{n}})$ and $\rho_{\Gamma,S}(\node{n}) = (\exists, f_{\Gamma,S,\node{n}})$, satisfy the following.
    \begin{itemize}
        \item 
        $(f'_{\Gamma,\node{n}})_<(\seqst(\lambda'_\Gamma(\node{n}))) = \seqst(\lambda'_{\Gamma}(\node{n}_1))$
        \item
        $(f'_{\Gamma,\node{n}})_=(\seqst(\lambda'_\Gamma(\node{n}))) = \seqst(\lambda'_{\Gamma}(\node{n}_2))$
        \item
        $(f_{\Gamma,S,\node{n}})_<(\seqst(\lambda_{\Gamma,S}(\node{n}))) \subseteq \seqst(\lambda_{\Gamma,S}(\node{n}_1))$
        \item
        $(f_{\Gamma,S,\node{n}})_=(\seqst(\lambda_{\Gamma,S}(\node{n}))) \subseteq \seqst(\lambda_{\Gamma,S}(\node{n}_2))$
    \end{itemize}
    We now define $f_{\Gamma,\node{n}}$ as follows.
    \[
      f_{\Gamma,\node{n}}(s) = \begin{cases}
          f'_{\Gamma,\node{n}}(s) & \text{if $s \in \seqst(\lambda'_{\Gamma}(\node{n}))$} \\
          f_{\Gamma,S,\node{n}}(s) & \text{if $s \in \seqst(\lambda_{\Gamma,S}(\node{n})) \setminus \seqst(\lambda'_{\Gamma}(\node{n}))$}
      \end{cases}
    \]
    Since $\node{n}$ satisfies \ref{inv:state-set} if follows that $S_\node{n} \subseteq \seqst(\lambda'_\Gamma(\node{n})) \cup \seqst(\lambda_{\Gamma,S}(\node{n}))$ and thus $f_{\Gamma,\node{n}}$ is well-defined.
    Finally, we set $\rho(\node{n}) = (\exists, f_{\Gamma,\node{n}})$ and
    \begin{align*}
    \lambda_\Gamma(\node{n}_1) &= (f_{\Gamma,\node{n}})_<(S_\node{n})
    \tnxT{\V_{S}}{\Delta} \seqfm_{\Gamma}(\node{n}_1) \\
    \lambda_\Gamma(\node{n}_2) &= (f_{\Gamma,\node{n}})_=(S_\node{n}) \tnxT{\V_{S}}{\Delta} \seqfm_{\Gamma}(\node{n}_2)
    \end{align*}
    It is clear that invariant~\ref{inv:rule} holds for $\node{n}$, while \ref{inv:formula}--\ref{inv:definition-list} hold for $\node{n}_1$ and $\node{n}_2$.
    
    \item[$\rn_\Gamma(\node{n}) = \forall$.] 
    In this case $cs(\node{n}) = \node{n}_1\node{n}_2$, $\lambda_{\Gamma}(\node{n}) = S_{\node{n}} \tnxT{\V_S}{\Delta} \forall_{\seqfm_\Gamma(\node{n}_1)} \seqfm_\Gamma(\node{n}_2)$, with $\Delta = \seqdl_\Gamma(\node{n}) = \seqdl_\Gamma(\node{n}_1) = \seqdl_\Gamma(\node{n}_2)$. 
    We construct a function $g_{\Gamma,\node{n}} \in S_{\node{n}} \times \Rnneg \to \Rnneg$ such that $g(s,\delta) \leq \delta$ for all $s \in S_\node{n}$ and $\delta \in \tdelays(s)$ and such that $(g_{\Gamma,\node{n}})_<(S_\node{n}) \subseteq \seqst(\lambda'_\Gamma(\node{n}_1)) \cup \seqst(\lambda_{\Gamma,S}(\node{n}_1))$  and $(g_{\Gamma,\node{n}})_=(S_\node{n}) \subseteq \seqst(\lambda'_\Gamma(\node{n}_2)) \cup \seqst(\lambda_{\Gamma,S}(\node{n}_2))$.
    Since $\mathbb{T}'_\Gamma$ and $\mathbb{T}_{\Gamma, S}$ are successful we know that $\forall$-functions $g'_{\Gamma, \node{n}} \in \seqst(\lambda'_{\Gamma}(\node{n})) \times \Rnneg \to \Rnneg$ and $g_{\Gamma, S, \node{n}} \in \seqst(\lambda_{\Gamma,S}(\node{n})) \times \Rnneg \to \Rnneg$, where $\rho'_\Gamma(\node{n}) = (\forall, g'_{\Gamma,\node{n}})$ and
    $\rho_{\Gamma,S}(\node{n}) = (\forall, g_{\Gamma, S, \node{n}})$, satisfy the following.
    \begin{itemize}
        \item 
        $(g'_{\Gamma,\node{n}})_<(\seqst(\lambda'_\Gamma(\node{n}))) = \seqst(\lambda'_\Gamma(\node{n}_1))$
        \item 
        $(g'_{\Gamma,\node{n}})_=(\seqst(\lambda'_\Gamma(\node{n}))) = \seqst(\lambda'_\Gamma(\node{n}_2))$
        \item 
        $(g_{\Gamma,S,\node{n}})_<(\seqst(\lambda_{\Gamma,S}(\node{n}))) =
        \seqst(\lambda_{\Gamma,S}(\node{n}_1))$
        \item 
        $(g_{\Gamma,S,\node{n}})_=(\seqst(\lambda_{\Gamma,S}(\node{n}))) =
        \seqst(\lambda_{\Gamma,S}(\node{n}_2))$
        
    \end{itemize}
    We define $g_{\Gamma,\node{n}}$ as follows.
    \[
    g_{\Gamma,\node{n}}(s,\tdelay) = \begin{cases}
        g'_{\Gamma,\node{n}}(s,\tdelay) & \text{if $s \in \seqst(\lambda'_{\Gamma}(\node{n}))$}\\
        g_{\Gamma,S,\node{n}}(s,\tdelay) & \text{if $s \in \seqst(\lambda_{\Gamma,S}(\node{n})) \setminus \seqst(\lambda'_{\Gamma}(\node{n}))$}
    \end{cases}
    \]
    Since $\node{n}$ satisfies \ref{inv:state-set} if follows that $S_\node{n} \subseteq \seqst(\lambda'_\Gamma(\node{n})) \cup \seqst(\lambda_{\Gamma,S}(\node{n}))$ and thus $g_{\Gamma,\node{n}}$ is well-defined.
    Finally, we set $\rho(\node{n}) = (\forall, g_{\Gamma,\node{n}})$, and
    \begin{align*}
        \lambda_\Gamma(\node{n}_1) &= (g_{\Gamma,\node{n}})_<(S_\node{n}) \tnxT{\V_{S}}{\Delta} \seqfm_{\Gamma}(\node{n}_1) \\
        \lambda_\Gamma(\node{n}_2) &= (g_{\Gamma,\node{n}})_=(S_\node{n}) \tnxT{\V_{S}}{\Delta} \seqfm_{\Gamma}(\node{n}_2) 
    \end{align*}
    It is clear that invariant~\ref{inv:rule} holds for $\node{n}$, while \ref{inv:formula}--\ref{inv:definition-list} hold for $\node{n}_1$ and $\node{n}_2$.

\end{description}
The same arguments in the proof of Lemma~\ref{lem:fixpoint-completeness} can be used to show  that $\mathbb{T}_\Gamma$ is successful, TNF and compliant with $\prec'$ and that it satisfies Property~\ref{goal:dependencies}.

\paragraph{Step~\ref{it:step-timed-tableau-composition} of proof outline:  construct tableau for $S \tnxTV \sigma Z.\Phi$.}
Analogous to Step~\ref{it:step-timed-tableau-composition} in the proof of Lemma~\ref{lem:fixpoint-completeness}.
\qedhere
\end{proof}

With these lemmas in hand we may now state and prove the completeness theorem.

\begin{theorem}[Completeness]\label{thm:timed-completeness}
Let $\T = \lts{S}$ be a TTS and $\V$ a valuation, and let $S$ be a set of states and $\Phi$ a timed mu-calculus formula such that $S \subseteq \semTV{\Phi}$.  Then $S \tnxTV{\emptyL} \Phi$ has a successful tableau.
\end{theorem}
\begin{proof}
Analogous to the proof of Theorem~\ref{thm:completeness}, but using Lemma~\ref{lem:timed-fixpoint-completeness} instead of Lemma~\ref{lem:fixpoint-completeness}. \qedhere
\end{proof}

\section{Conclusions and Future Work}\label{sec:Conclusions}

The work in this paper was motivated by a desire to give a sound and complete proof system for the timed mu-calculus for infinite-state systems.  We intended to do so by modifying an existing proof system for infinite-state systems and the untimed mu-calculus due to Bradfield and Stirling~\cite{Bra1991,BS1992}, but this proved difficult because of the delicacy of their soundness and completeness arguments.  Instead, we gave an alternative approach, based on explicit tableau constructions, for the untimed mu-calculus.  A hallmark of these constructions is the extensibility and their lack of dependence on infinitary logic.  We then showed how these constructions admitted modifications to the core proof system, including new proof-search strategies and new logical modalities, such as those in the timed mu-calculus.

Our proof techniques are based on a fundamental, lattice-theoretic result giving a new characterization in terms of a notion we call \emph{support orderings} for least and greatest fixpoints of monotonic functions over complete lattices whose carrier sets have form $2^S$ for some set $S$.
Using this approach, we are able to present a proof that, contrary to Bradfield and Stirling's original proof, does not require reasoning about ordinal unfoldings directly, and that is extensible to other termination conditions and modalities.  Our completeness results also rely on direct constructions of proof tableaux for valid sequents; this also facilitates extensibility.

We have illustrated extensibility of our approach by showing that the soundness proof straightforwardly carries over to the proof system where $\mu$-nodes with non-empty sets of states are always unfolded.
Additionally, we have presented a proof system for an extension of the mu-calculus with two timed modalities, and shown our soundness and completeness proof are extended straightforwardly to the latter setting.

\paragraph{Future work.}
The proof approach presented in this paper enables us to prove soundness of extensions and modifications of the proof system for infinite state systems.
In particular, soundness of the proof system for the timed mu-calculus opens up the way to model check such systems.
We plan to implement this proof system, dealing with sets of states symbolically, and thus extending model checkers for alternation free timed mu-calculi \cite{FC2014} to the full mu-calculus.

Furthermore, the proof of soundness and completeness for the timed mu-calculus only involves adding cases for the newly added operators to the proofs of results about the base proof system. 
Each of these cases follows the same line of reasoning as the results for other operators. 
We therefore expect that the soundness and completeness results can be generalized to only refer to properties about the local dependencies of the operators in the mu-calculus, and our support ordering results. 
It would be interesting to investigate this direction, and simplify the proof obligations for soundness and completeness when adding operators to the mu-calculus even further.

Another direction to be explored is how the soundness and completeness results in this paper can be adapted to the equational mu-calculus~\cite{CS1993}, and other equational theories such as Boolean equation systems~\cite{Mad1997} and parameterized Boolean equation systems~\cite{GM1999,GW2005}. In particular, the definition lists are, in essence, already part of the formalism, so the mechanism of unfolding in the proof system needs to be adapted to deal with this difference.

Finally, another common extension of the mu-calculus used in the timed setting is to add freeze quantification, see e.g.,~\cite{BCL2011}. 
The extension of our proof rules to this setting should be similarly straightforward as the extension we presented in this paper, although the definitions of the underlying transition systems would need extension to accommodate explicity clock variables..

\bibliographystyle{splncs04}
\bibliography{bibliography}

\begin{thebibliography}{10}
\providecommand{\url}[1]{\texttt{#1}}
\providecommand{\urlprefix}{URL }
\providecommand{\doi}[1]{https://doi.org/#1}

\bibitem{AD1994}
Alur, R., Dill, D.L.: A theory of timed automata. Theoretical Computer Science
  \textbf{126}(2),  183--235 (Apr 1994). \doi{10.1016/0304-3975(94)90010-8}

\bibitem{And1993}
Andersen, H.R.: Verification of {{Temporal Properties}} of {{Concurrent
  Systems}}. Ph.D. thesis, Aarhus University, {Denmark} (Jun 1993).
  \doi{10.7146/dpb.v22i445.6762}

\bibitem{bertot2013interactive}
Bertot, Y., Cast{\'e}ran, P.: Interactive theorem proving and program
  development: Coq’Art: the calculus of inductive constructions. Springer
  Science \& Business Media (2013)

\bibitem{bhat1995efficient}
Bhat, G., Cleaveland, R., Grumberg, O.: Efficient on-the-fly model checking for
  {CTL}$^*$. In: Proceedings of Tenth Annual IEEE Symposium on Logic in
  Computer Science. pp. 388--397. IEEE (1995). \doi{10.1109/LICS.1995.523273}

\bibitem{BCL2011}
Bouyer, P., Cassez, F., Laroussinie, F.: Timed modal logics for real-time
  systems. Journal of Logic, Language and Information  \textbf{20}(2),
  169--203 (Apr 2011). \doi{10.1007/s10849-010-9127-4}

\bibitem{Bra1993}
Bradfield, J.C.: A proof assistant for symbolic model-checking. In: {von
  Bochmann}, G., Probst, D.K. (eds.) Computer {{Aided Verification}}. pp.
  316--329. Lecture {{Notes}} in {{Computer Science}}, {Springer}, {Berlin,
  Heidelberg} (1993). \doi{10.1007/3-540-56496-9\_25}

\bibitem{BS1992}
Bradfield, J., Stirling, C.: Local model checking for infinite state spaces.
  Theoretical Computer Science  \textbf{96}(1),  157--174 (Apr 1992).
  \doi{10.1016/0304-3975(92)90183-G}

\bibitem{BS2001}
Bradfield, J., Stirling, C.: Modal logics and mu-calculi: An introduction. In:
  Bergstra, J.A., Ponse, A., Smolka, S.A. (eds.) Handbook of {{Process
  Algebra}}, pp. 293--330. {Elsevier Science}, {Amsterdam} (Jan 2001).
  \doi{10.1016/B978-044482830-9/50022-9}

\bibitem{Bra1991}
Bradfield, J.C.: Verifying Temporal Properties of Systems with Applications to
  Petri Nets. Ph.D. thesis, University of Edinburgh, Edinburgh (1991),
  \url{http://hdl.handle.net/1842/6565}

\bibitem{ciesielski1997set}
Ciesielski, K.: Set theory for the working mathematician, London Mathematical
  Society Student Texts, vol.~39. Cambridge University Press (1997)

\bibitem{clarke2018model}
Clarke~Jr, E.M., Grumberg, O., Kroening, D., Peled, D., Veith, H.: Model
  checking. MIT press (2018)

\bibitem{CS1993}
Cleaveland, R., Steffen, B.: A linear-time model-checking algorithm for the
  alternation-free modal mu-calculus. Formal Methods in System Design
  \textbf{2}(2),  121--147 (Apr 1993). \doi{10.1007/BF01383878}

\bibitem{constable1986implementing}
Constable, R., Allen, S., Bromley, H., Cleaveland, W., Cremer, J., Harper, R.,
  Howe, D., Knoblock, T., Mendler, N., Panangaden, P., et~al.: Implementing
  mathematics. Prentice-Hall (1986)

\bibitem{FC2014}
Fontana, P., Cleaveland, R.: The power of proofs: New algorithms for timed
  automata model checking. In: Formal {{Modeling}} and {{Analysis}} of {{Timed
  Systems}}. pp. 115--129. {Springer, Cham} (Sep 2014).
  \doi{10.1007/978-3-319-10512-3\_9}

\bibitem{Fon2014}
Fontana, P.C.: Towards a Unified Theory of Timed Automata. Ph.D. thesis,
  University of Maryland, College Park (2014),
  \url{http://drum.lib.umd.edu/handle/1903/15232}

\bibitem{GM1999}
Groote, J.F., Mateescu, R.: Verification of temporal properties of processes in
  a setting with data. In: Haeberer, A.M. (ed.) Algebraic {{Methodology}} and
  {{Software Technology}}. pp. 74--90. Lecture {{Notes}} in {{Computer
  Science}}, {Springer}, {Berlin, Heidelberg} (1999).
  \doi{10.1007/3-540-49253-4\_8}

\bibitem{GW2005}
Groote, J.F., Willemse, T.A.C.: Parameterised boolean equation systems.
  Theoretical Computer Science  \textbf{343}(3),  332--369 (Oct 2005).
  \doi{10.1016/j.tcs.2005.06.016}

\bibitem{Haz01}
Hazewinkel, M. (ed.): Encyclopaedia of Mathematics, Supplement {III}. Springer
  (2001)

\bibitem{jacobs2011introduction}
Jacobs, B., Rutten, J.: An introduction to (co)algebra and (co)induction. In:
  Advanced Topics in Bisimulation and Coinduction, Cambridge Tracts in
  Theoretical Computer Science, vol.~52. Cambridge: Cambridge University Press
  (2011)

\bibitem{jech:1978}
Jech, T.J.: Set Theory. Academic Press (1978)

\bibitem{jech2008axiom}
Jech, T.J.: The axiom of choice. Courier Corporation (2008)

\bibitem{Koz1983}
Kozen, D.: Results on the propositional {$\mu$}-calculus. Theoretical Computer
  Science  \textbf{27}(3),  333--354 (Jan 1983).
  \doi{10.1016/0304-3975(82)90125-6}

\bibitem{kunen:1980}
Kunen, K.: Set theory : an introduction to independence proofs. Studies in
  logic and the foundations of mathematics, Elsevier Science, Amsterdam,
  Lausanne, New York (1980), \url{http://opac.inria.fr/record=b1117475}

\bibitem{Mad1997}
Mader, A.H.: Verification of Modal Properties Using {{Boolean}} Equation
  Systems. Ph.D. thesis, Technische Universit\"at M\"unchen, {M\"unchen} (1997)

\bibitem{mateescu2003efficient}
Mateescu, R., Sighireanu, M.: Efficient on-the-fly model-checking for regular
  alternation-free mu-calculus. Science of Computer Programming
  \textbf{46}(3),  255--281 (2003)

\bibitem{SW1991}
Stirling, C., Walker, D.: Local model checking in the modal mu-calculus.
  Theoretical Computer Science  \textbf{89}(1),  161--177 (Oct 1991).
  \doi{10.1016/0304-3975(90)90110-4}

\bibitem{SE1989}
Streett, R.S., Emerson, E.A.: An automata theoretic decision procedure for the
  propositional mu-calculus. Information and Computation  \textbf{81}(3),
  249--264 (Jun 1989). \doi{10.1016/0890-5401(89)90031-X}

\bibitem{Tar1955}
Tarski, A.: A lattice-theoretical fixpoint theorem and its applications.
  Pacific Journal of Mathematics  \textbf{5}(2),  285--309 (1955),
  \url{https://projecteuclid.org/euclid.pjm/1103044538}

\bibitem{DC2005}
Zhang, D., Cleaveland, R.: Fast on-the-fly parametric real-time model checking.
  In: 26th {{IEEE International Real}}-{{Time Systems Symposium}}
  ({{RTSS}}'05). pp. 157--166 (Dec 2005). \doi{10.1109/RTSS.2005.22}

\end{thebibliography}

\appendix
\section{Relation between (extended) dependency orderings and Bradfield and Stirlings (extended) paths}

We recall Bradfield and Stirling's definition of (extended) paths~\cite{BS1992}, in the version as presented by Bradfield~\cite{Bra1991}, and show the correspondence with our (extended) dependency ordering.

Bradfield and Stirling define the notion of \emph{path} from one state in a node in a proof tree to another state in a descendant node in the proof tree. 
This is done on the basis of the proof rules used to generate the child nodes. The formal definition is as follows.
\begin{definition}[Path~{\cite[Definition 3.8]{Bra1991}}]\label{def:path}
    There is a \emph{path} from state $s$ at node $\node{n}$ to state $s'$ at node $\node{n}'$ in a tableau iff there is a finite sequence $(s, \node{n}) = (s_0, \node{n}_0), \ldots, (s_k,\node{n}_k) = (s', \node{n}')$ such that the following hold.
    \begin{enumerate}
        \item $\node{n}_{i+1}$ is a child of $\node{n}_i$ for all $0 \leq i < k$.
        \item If $\node{n}_i = S_i \tnxTV{\Delta_i} \Phi_i$ then $s_i \in S_i$ for all $0 \leq i \leq k$.
        \item If the rule applied to $\node{n}_i$ is $[K]$ then $s_i \xrightarrow{K} s_{i+1}$; if the rule applied to $\node{n}_i$ is $\langle K \rangle$ then $s_{i+1} = f(s_i)$; in all other cases $s_{i+1} = s_i$.
    \end{enumerate}
\end{definition}
\noindent
Bradfield writes $s @ \node{n} \pathto s' @ \node{n}'$ if there is a path from $s$ at $\node{n}$ to $s'$ at $\node{n}'$.  Note that for any $s$ and $\node{n} = S \tnxTVD \Phi$ such that $s \in S$, $s @ \node{n} \pathto s @ \node{n}$.

In the same definition Bradfield also defines the notion of an \emph{extended} path from $s$ at $\node{n}$ to $s'$ at $\node{n}'$.
\begin{definition}[Extended path~{\cite[Definition 3.8]{Bra1991}}]\label{def:epath}
    There is an \emph{extended path} from state $s$ at node $\node{n}$ to state $s'$ at node $\node{n}'$ in a tableau, denoted $s @ \node{n} \epathto s' @ \node{n}'$, as follows.
    \begin{enumerate}
        \item \label{def:epath-base}
              Either $s @ \node{n} \pathto s' @ \node{n}'$; or
        \item \label{def:epath-step}
              there is a node $\node{n}'' = S'' \tnxTV{\Delta''} U$ with $\node{n}'' \neq \node{n}$ and $\node{n}''\neq \node{n}'$,\footnote{Bradfield's original definition does not require $\node{n}'' \neq \node{n}$, however, he does assert that the extended paths for $\node{n}$ are defined in terms of extended paths from strict descendants, and this assertion does not hold when the requirement is omitted.}
              and a finite sequence of states $s_0, \ldots, s_k$, and a finite sequence of nodes $\node{n}_1, \ldots \node{n}_k$ such that the following hold.
              \begin{enumerate}
                  \item $s @ \node{n} \pathto s_0 @ \node{n}''$.
                  \item Each $\node{n}_i$ is a $\sigma$-leaf with companion node $\node{n}''$.
                  \item For $0 \leq i < k$, $s_i @ \node{n}'' \epathto s_{i+1} @ \node{n}_{i+1}$.
                  \item $s_k @ \node{n}'' \epathto s' @ \node{n}'$
              \end{enumerate}
    \end{enumerate}
\end{definition}
This definition extends the notion of path by allowing ``looping'' through companion nodes and their associated leaves within the subtree rooted at $\node{n}$.

Bradfield's original definition~\cite[Definition 3.9]{Bra1991} of the ordering of states in a node, $\sqsupset_{\node{n}}$ is the following.
\begin{definition}[{\cite[Definition 3.9]{Bra1991}}]\label{def:bs-order}
    Let $\node{n}$ be a node in a tableau of the form $S \tnxTVD U$ where $\Delta(U) = \mu Z. \Phi$; let $\{ \node{n}_0, \ldots \node{n}_k\}$ be the leaves for which $\node{n}$ is the companion. Define ordering $\sqsupset_{\node{n}} \;\subseteq S \times S$ as follows:  $s \sqsupset_{\node{n}} s'$ iff there exists $\node{n}_i$ such that $s @ \node{n} \epathto s' @ \node{n}_i$.
\end{definition}

We prove that our (extended) dependency orderings correspond to Bradfield's (extended) paths. Hence, our dependency orderings are simply an alternative characterization of the dependencies in a tableau.

It follows immediately from the definitions that our dependency ordering (Definition~\ref{def:dependency_ordering}) coincides with Bradfield's paths (Definition~\ref{def:path}).
\begin{lemma}\label{lem:child_order_path}
    Let $\node{n} = S\tnxTVD \Phi$, $\node{n}' = S' \tnxTV{\Delta'} \Phi'$ be nodes in a tableau with $s \in S$ and $s' \in S'$. We have $s@\node{n} \pathto s'@\node{n}'$ iff there is a sequence $(s, \node{n}) = (s_0, \node{n}_0), \ldots, (s_k,\node{n}_k) = (s', \node{n}')$ such that for all $0 \leq i < k$, $s_{i+1} <_{\node{n}_{i+1},\node{n}_{i}} s_i$.
\end{lemma}
\begin{proof}
    Immediate since the definition of $<_{\node{n}',\node{n}}$ corresponds directly with the three clauses in Definition~\ref{def:path}.\qedhere
\end{proof}

\begin{lemma}\label{lem:dependency_ordering_is_path}
    For nodes $\node{m} = S_{\node{m}} \tnxTV{\Delta_{\node{m}}} \Phi_{\node{m}}$ and $\node{n} = S_{\node{n}} \tnxTV{\Delta_{\node{n}}} \Phi_{\node{n}}$ in a tableau, we have that $s_{\node{n}'} \lessdot_{\node{n}',\node{n}} s_{\node{n}}$ if and only if $s_{\node{n}}@\node{n} \pathto s_{\node{n}'}@\node{n}'$.
\end{lemma}
\begin{proof}
    We prove both directions separately.
    \begin{itemize}
        \item[$\Rightarrow$]
            We prove that if $s_{\node{n}'} \lessdot_{\node{n}',\node{n}} s_{\node{n}}$, then $s_{\node{n}}@\node{n} \pathto s_{\node{n}'}@\node{n}'$ by induction on the definition of $\lessdot_{\node{n}',\node{n}}$.
            If $\node{n} = \node{n}'$ and $s = s'$ (first case of Definition~\ref{def:dependency_ordering}), the empty sequence $(s_{\node{n}},\node{n}) = (s_{\node{n}'},\node{n}')$ witnesses $s_{\node{n}}@\node{n} \pathto s_{\node{n}'}@\node{n}'$.
            
            Now, suppose there exists $\node{n}''$ and $s_{\node{n}''}$ such that $s_{\node{n}''} <_{\node{n}'',\node{n}} s_{\node{n}}$ and $s_{\node{n}'} \lessdot_{\node{n}',\node{n}''} s_{\node{n}''}$. 
            By the induction hypothesis, we have $s_{\node{n}''}@\node{n}'' \pathto s_{\node{n}'}@\node{n}'$, so there is a sequence $(s_{\node{n}''},\node{n}'') = (s_0, \node{n}_0), \ldots, (s_k, \node{n}_k) = (s_{\node{n}'}, \node{n}')$ with $s_{i+1} <_{\node{n}_{i+1},\node{n}_i} s_i$ for all $i < k$. The sequence $(s_{\node{n}},\node{n}), (s_0, \node{n}_0), \ldots, (s_k,\node{n}_k)$ now witnesses $s_{\node{n}}@\node{n} \pathto s_{\node{n}'}@\node{n}'$.

        \item[$\Leftarrow$] Suppose $s_{\node{n}}@\node{n} \pathto s_{\node{n}'}@\node{n}'$.
              We show that $s_{\node{n}'} \lessdot_{\node{n}',\node{n}} s_{\node{n}}$.
              Since $s_{\node{n}}@\node{n} \pathto s_{\node{n}'}@\node{n}'$, there must be a sequence
              $(s_{\node{n}},\node{n}) = (s_0,\node{n}_0), \ldots, (s_k,\node{n}_k) = (s_{\node{n}'},\node{n}')$
              such that for all $0 \leq i < k$, $s_{i+1} <_{\node{n}_{i+1},\node{n}_i} s_i$.
              By definition, then also $s_{i+1} \lessdot_{\node{n}_{i+1},\node{n}_i} s_i$,
              so it follows immediately from repeated application of Lemma~\ref{lem:pseudo-transitivity-of-dependency-ordering} that $s_{k} \lessdot_{\node{n}_{k},\node{n}_0} s_0$, thus $s_{\node{n}'} \lessdot_{\node{n}',\node{n}} s_{\node{n}}$. \qedhere
    \end{itemize}
\end{proof}

Furthermore, our extended dependency ordering (Definition~\ref{def:extended_path_ordering}) corresponds to Bradfield's extended paths (Definition~\ref{def:epath}).

\begin{lemma}\label{lem:extended_path_implies_extended_order}\label{lem:extended_order_implies_extended_path}
    Given two nodes $\node{n}' = S' \tnxTV{\Delta'} \Phi'$ and $\node{n} = S \tnxTVD \Phi$ in a tableau, and states $s \in S$, $s' \in S'$. Then $s' <:_{\node{n}',\node{n}} s$ if and only if $s@\node{n} \epathto s'@\node{n}'$.
\end{lemma}
\begin{proof}
We prove both directions separately.

\begin{itemize}
\item[$\Rightarrow$]
We proceed by induction on the definition of $<:_{\node{n}',\node{n}}$.
        
\begin{itemize}
    \item If $s' \lessdot_{\node{n}',\node{n}} s$, according to Lemma~\ref{lem:dependency_ordering_is_path}, we have $s@\node{n} \pathto s'@\node{n}'$, and according to Definition~\ref{def:epath}~(\ref{def:epath-base}) we have $s@\node{n} \epathto s'@\node{n}'$.
    \item Otherwise, there is a node $\node{n}'' = S'' \tnxTV{\Delta''} U$ with $\node{n}'' \neq \node{n}$, $\node{n}'' \neq \node{n}'$, and $s'' \in S$ such that for some $\overline{s} \in S''$, we have (a) $s'' \lessdot_{\node{n}'',\node{n}} s$, (b) for some $m \geq 1$, $\overline{s} <:_{\node{n}''}^{+} s''$, and (c) $s' <:_{\node{n}',\node{n}''} \overline{s}$.

    From (a) and Lemma~\ref{lem:dependency_ordering_is_path}, we get that $s@\node{n} \pathto s''@\node{n}''$.

    From (b) we find that $\overline{s} <:_{\node{n}''}^{+} s''$. Therefore, for some $m \geq 1$, $\overline{s} <:^m_{\node{n}''} s''$. Consider the smallest such $m$. Then there exist states $s_0, \ldots, s_m$ such that $\overline{s} = s_0$, $s_i <:_{\node{n}''} s_{i+1}$ for $0 \leq i < m$, and $s_m = s''$.

    By definition of $<:_{\node{n}''}$, there exist nodes $\node{n}_0, \ldots, \node{n}_{m-1}$ such that $s_i <:_{\node{n}_{i}, \node{n}'',} s_{i+1}$ for $0 \leq i < m$. From the induction hypothesis, it follows that $s_{i+1}@\node{n}'' \epathto s_i@\node{n}_i$ for $0 \leq i < m$.

    From (c) and the induction hypothesis we find that $s''@\node{n}'' \epathto s'@\node{n}'$.

    Hence, according to Definition~\ref{def:epath}~(\ref{def:epath-step}), we have $s@\node{n} \epathto s'@\node{n}'$.

\end{itemize}
    
\item[$\Leftarrow$]
We proceed by induction on the definition of $\epathto$.
\begin{itemize}
    \item If $s@\mathbf{n} \pathto s'@\node{n}'$, according to Lemma~\ref{lem:dependency_ordering_is_path}, we have $s' \lessdot_{\node{n}',\node{n}} s$, and according to Definition~\ref{def:extended_path_ordering}~(\ref{def:extended_path_ordering-base}), we have $s' <:_{\node{n}',\node{n}} s$.
    
    \item There is a node $\node{n}'' = S'' \tnxTV{\Delta''} U$ with $\node{n}'' \neq \node{n}$, $\node{n}'' \neq \node{n}'$, and a finite sequence of states $s_0, \ldots s_k$, and nodes $\node{n}_1, \ldots \node{n}_k$ such that (a) $s@\node{n} \pathto s_0@\node{n}''$, (b) all $\node{n}_i$ are companion leaves of $\node{n}''$, (c) for $0 \leq i < k$, $s_i@\node{n}'' \epathto s_{i+1}@\node{n}_{i+1}$, and (d) $s_k@\node{n}'' \epathto s'@\node{n}'$.

    From (a) and Lemma~\ref{lem:dependency_ordering_is_path} we get that $s_0 \lessdot_{\node{n}'',\node{n}} s$.

    Using (b), (c) and the induction hypothesis, we get for $0 \leq i < k$, $s_{i+1} <:_{\node{n}_{i+1},\node{n}''} s_i$. Thus, $s_{i+1} <:_{\node{n}''} s_i$ for each such $i$, and by definition, $s_k <:^{+}_{\node{n}''} s_0$.

    Using (d) and the induction hypothesis, we get $s' <:_{\node{n}',\node{n}''} s_k$.

    Now, according to Definition~\ref{def:extended_path_ordering}~(\ref{def:extended_path_ordering-step}), we have $s' <:_{\node{n}',\node{n}} s$. \qedhere
    \end{itemize}
\end{itemize}
\end{proof}

The following now follows immediately from Lemma~\ref{lem:extended_path_implies_extended_order}, Definition~\ref{def:bs-order} and Definition~\ref{def:extended_path_ordering}
\begin{corollary}\label{lem:bradfield_order_vs_extended_ordering}
    Let $\node{n} = S \tnxTV U$ be a companion node in a tableau, and let $\node{n}_0, \ldots, \node{n}_k$ be the companion leaves of $\node{n}$. Then, for $s, s' \in S$, we we have $s' <:_{\node{n}} s$ if and only if $s \sqsupset_{\node{n}} s'$.
\end{corollary}


\end{document}